%% file: arxiv_main.tex
\documentclass[11pt]{article}

\usepackage{helvet}

\usepackage{fancyhdr}
\usepackage{titlesec}
\titlespacing*{\section}{0pt}{10pt}{5pt}
\titlespacing*{\subsection}{0pt}{6pt}{3pt}

\usepackage{bm}
\usepackage{color}
\usepackage{xcolor}
\usepackage{amsmath}
\usepackage{amssymb}
\usepackage{amsfonts}
\usepackage{mathrsfs}
\usepackage{mathtools}
\usepackage{times}
\usepackage{graphicx}
\usepackage{bbold}

\usepackage{silence}
\usepackage{subcaption}
\usepackage{natbib}

\usepackage{hyperref}
\hypersetup{
    colorlinks=true,
    linkcolor=blue,
    filecolor=magenta,
    urlcolor=cyan,
    citecolor=blue
}

\graphicspath{{./figures/}}

\usepackage[plain,noend]{algorithm2e}
\makeatletter
\renewcommand{\algocf@captiontext}[2]{#1\algocf@typo. \AlCapFnt{}#2}

\def\@algocf@capt@plain{top}
\renewcommand{\algocf@makecaption}[2]{%
  \addtolength{\hsize}{\algomargin}%
  \sbox\@tempboxa{\algocf@captiontext{#1}{#2}}%
  \ifdim\wd\@tempboxa >\hsize%
    \hskip .5\algomargin%
    \parbox[t]{\hsize}{\algocf@captiontext{#1}{#2}}%
  \else%
    \global\@minipagefalse%
    \hbox to\hsize{\box\@tempboxa}%
  \fi%
  \addtolength{\hsize}{-\algomargin}%
}
\makeatother

\usepackage{soul}
\usepackage{tikz}
\usetikzlibrary{arrows.meta}
\usepackage{booktabs}
\usepackage{enumitem}

\newcommand{\indsubset}{\subset\hspace{-0.23cm}\subset}
\def\pr{\mathrm{pr}}
\def\mE{E}
\def\Cov{\mathrm{cov}}
\def\wh{\widehat}
\def\eE{\mathcal{E}}
\def\var{\mathrm{var}}

\usepackage{amsthm}
\newtheorem{theorem}{Theorem}[section]
\newtheorem{lemma}[theorem]{Lemma}

\newtheorem{proposition}[theorem]{Proposition}
\newtheorem{definition}[theorem]{Definition}
\newtheorem{assumption}[theorem]{Assumption}

\input{format/symbol}

\usepackage[top=1in, bottom=1in, left=1in, right=1in]{geometry}

\newcommand{\bG}{\mathbb{G}}
\newcommand{\bGb}{\mathbb{G}_b}
\newcommand{\bGn}{\mathbb{G}_n}

\RestyleAlgo{ruled}
\DontPrintSemicolon

\begin{document}

\begin{center}
\textbf{\large Multivariate Inference of Network Moments by Subsampling}\\
\medskip
Mingyu Qi\\
Department of Statistics, University of Virginia\\
\medskip
Chen-Wei Hua\\
School of Statistics, University of Minnesota\\
\medskip
Tianxi Li\\
School of Statistics, University of Minnesota\\
\medskip
Wen Zhou\\
Department of Biostatistics, School of Global Public Health, New York University\\
\end{center}

\begin{abstract}
Network moments--rescaled counts of motifs such as stars and triangles--are fundamental summaries of network structure, widely used in goodness-of-fit testing, model selection, and network comparison. While the univariate distribution of a single network moment can be approximated by subsampling, the consistency of subsampling for their {\it joint} distribution has remained unestablished. In this paper, we prove that node subsampling provides an asymptotically accurate approximation of the joint distribution of multiple network moments under a general sparse graphon model. The theoretical analysis requires a careful characterization of the dependence structure among network moments and the corresponding multivariate asymptotic convergence, going substantially beyond existing univariate results. Building on this foundation, we address a practically important open problem: two-sample testing between unmatchable networks with unequal edge densities. We propose a novel subsampling-based procedure that combines {\it sparsification} with a {\it sample-splitting} strategy. This yields the first subsampling-based inferential procedure valid for this setting, to our knowledge. We demonstrate the utility of multivariate subsampling inference through simulation studies and a real data application comparing coexpression networks of core and non-core genes in a study of parallel adaptation in Trinidadian guppies, where joint and conditional moment distributions reveal a structural difference that no marginal test can detect.
\end{abstract}

\input{1intro}

\input{2preliminary}

\input{3subsampling}

\input{comparison}

\input{4simulation}

\input{5realData}

\input{6Discussion}

\appendix

\section*{Appendix}

\renewcommand\thesection{S\Alph{section}}
\renewcommand\thesubsection{\thesection.\arabic{subsection}}
\counterwithout{figure}{section}
\counterwithout{table}{section}
\setcounter{table}{0}
\renewcommand\thefigure{S.\arabic{figure}}
\renewcommand\thetable{S.\arabic{table}}

\counterwithout{equation}{section}
\setcounter{equation}{0}
\makeatletter
\renewcommand{\theequation}{S.\@arabic\c@equation}
\makeatother

\input{Appendix/Appendix_add}

\makeatletter
\renewcommand{\thelemma}{SB.\@arabic\c@lemma}
\makeatother

\input{Appendix/Appendix_SuppleA1}

\makeatletter
\renewcommand{\thelemma}{SD.\@arabic\c@lemma}
\makeatother
\input{Appendix/Appendix_SuppleA2A3}

\makeatletter
\renewcommand{\thelemma}{SE.\@arabic\c@lemma}
\makeatother
\input{Appendix/Appendix_SuppleA4}

\makeatletter
\renewcommand{\thelemma}{SF.\@arabic\c@lemma}
\makeatother
\input{Appendix/Appendix_SuppleA5}

\makeatletter
\renewcommand{\thelemma}{SG.\@arabic\c@lemma}
\makeatother
\newpage
\input{Appendix/Appendix_SuppleA6}

\makeatletter
\renewcommand{\thelemma}{SH.\@arabic\c@lemma}
\makeatother
\input{Appendix/Appendix_SuppleA7}

\makeatletter
\renewcommand{\thelemma}{SI.\@arabic\c@lemma}
\makeatother
\input{Appendix/Appendix_SuppleA8_simu}

\newpage
\bibliographystyle{abbrvnat}
\bibliography{reference/paper-ref}

\end{document}

%% file: format/symbol.tex
\usepackage{tikz}
\usepackage{xparse}

\NewDocumentCommand{\xtriangle}{O{.65}}{%
    \tikz[baseline = -0.4ex,scale=#1]{%
        \foreach \x in{90,-30,210}%
        {%
            \fill[](\x:0.2)circle(1.5pt);%
        }%
        \draw[thin] (-30:0.2) -- (90:0.2) -- (210:0.2) -- cycle;%
    }%
}%

\NewDocumentCommand{\xwheelfour}{O{.6}}{%
    \tikz[baseline=-0.8ex,scale=#1]{%
        \def\r{0.25}
        \fill[] (0,0) circle (1.5pt);
        \foreach \a in {45,135,225,315} {%
            \fill[] (\a:\r) circle (1.5pt);
        }
        \draw[thin]
            (45:\r) -- (135:\r) -- (225:\r) -- (315:\r) -- cycle;
        \foreach \a in {45,135,225,315} {%
            \draw[thin] (0,0) -- (\a:\r);
        }
    }%
}

\NewDocumentCommand{\xline}{O{.5}}{%
    \tikz[baseline = -0.8ex,scale=#1]{%
        \foreach \x in{-90,30}%
        {%
            \fill[](\x:0.2)circle(1.5pt);%
        }%
        \draw[thin] (-90:0.2) -- (30:0.2);%
    }%
}%

\NewDocumentCommand{\xsquare}{O{.5}}{%
    \tikz[baseline=-0.5ex,scale=#1]{%
        \foreach \x in {45, 135, -135, -45}%
        {%
            \fill[] (\x:0.2828) circle (1.5pt); 
        }%
        \draw[thin] (45:0.2828) -- (135:0.2828) -- (-135:0.2828) -- (-45:0.2828) -- cycle;%
    }%
}%

\NewDocumentCommand{\xthreestar}{O{.6}}{%
    \tikz[baseline=-0.8ex,scale=#1]{%
        \pgfmathsetmacro{\armLength}{0.4 / sqrt(3)}
        \pgfmathsetmacro{\angleA}{30}
        \pgfmathsetmacro{\angleB}{-90}
        \pgfmathsetmacro{\angleC}{150}
        \draw[thin] (0,-.03) -- (\angleA:\armLength)   (0,-.03) -- (\angleB:\armLength)   (0,-.03) -- (\angleC:\armLength);
        \fill[] (\angleA:\armLength) circle (1.5pt);
        \fill[] (\angleB:\armLength) circle (1.5pt);
        \fill[] (\angleC:\armLength) circle (1.5pt);
        \fill[] (0,-.03) circle (1.5pt);
    }%
}

\NewDocumentCommand{\xtwoline}{O{.6}}{%
    \tikz[baseline = -0.8ex,scale=#1]{%
        \foreach \x in{30,-90,150}%
        {%
            \fill[](\x:0.2)circle(1.5pt);%
        }%
        \draw[thin] (150:0.2) -- (-90:0.2) -- (30:0.2);%
    }%
}%

\NewDocumentCommand{\xrectangle}{O{.8}}{%
    \tikz[baseline = -0.5ex,scale=#1]{%
        \foreach \x in{45,-45,135,225}%
        {%
            \fill[](\x:0.2)circle(1.5pt);%
        }%
        \draw[thin] (225:.2) rectangle (45:.2);%
    }%
}%

\NewDocumentCommand{\xreltriangle}{O{1.5}}{%
    \tikz[baseline = -0.6ex,scale=#1]{%
        \foreach \x in{30,-30,150,210}%
        {%
            \fill[](\x:0.2)circle(1.3pt);%
        }%
        \fill[](0,0)circle(1.3pt);%
        \draw[thin] (-30:0.2) -- (30:0.2) -- (210:0.2) -- (150:.2) -- cycle;%
    }%
}%

\NewDocumentCommand{\xdiagrectangle}{O{1}}{%
    \tikz[baseline = -0.5ex,scale=#1]{%
        \foreach \x in{45,-45,135,225}%
        {%
            \fill[](\x:0.2)circle(1.5pt);%
        }%
        \draw[thin] (225:.2) rectangle (45:.2);%
        \draw[thin]  (225:.2) -- (45:.2);%
    }%
}%

\NewDocumentCommand{\xreldiamond}{O{.8}}{%
    \tikz[baseline = -0.5ex,scale=#1]{%
        \foreach \x in{45,-45,135,225}%
        {%
            \fill[](\x:0.2)circle(1.5pt);%
        }%
        \fill[](0,0)circle(1.5pt);%
        \fill[]({sqrt(2)*0.2},0)circle(1.5pt);%
        \fill[](-{sqrt(2)*0.2},0)circle(1.5pt);%
        \draw[thin] (-45:0.2) -- (0:{sqrt(2)*0.2}) -- (45:0.2) -- (225:.2) -- (0:-{sqrt(2)*0.2}) -- (135:.2) -- cycle;%
    }%
}%

\NewDocumentCommand{\xtwodiamond}{O{.8}}{%
    \tikz[baseline = -0.5ex,scale=#1]{%
        \foreach \x in
        {(0:.12),
         (0:.3),
         (90:.2),
         (0:-.12),
         (0:-.3),
         (90:-.2)
        }%
        {%
            \fill[] \x circle(1.5pt);%
        }%
        \draw[thin] (0:0.3) -- (90:.2) -- (0:-0.3) -- (90:-.2) -- cycle;%
        \draw[thin] (0:0.12) -- (90:.2) -- (0:-0.12) -- (90:-.2) -- cycle;;%
    }%
}%

\NewDocumentCommand{\xrecanddia}{O{.8}}{%
    \tikz[baseline = -0.5ex,scale=#1]{%
        \foreach \x in
        {(0:0),
         (45:.2),
         (135:.2),
         (225:.2),
         (315:.2),
         (0:{sqrt(2)*0.2})
        }%
        {%
            \fill[] \x circle(1.5pt);%
        }%
        \draw[thin] (225:.2) rectangle (45:0.2);%
        \draw[thin] (0:0) -- (45:.2) -- (0:{sqrt(2)*0.2}) -- (315:.2) -- cycle;%
    }%
}%

\NewDocumentCommand{\xtworectangle}{O{.8}}{%
    \tikz[baseline = -0.5ex,scale=#1]{%
        \foreach \x in
        {(45:.2),
         (135:.2),
         (225:.2),
         (315:.2),
         ([shift={(0:{sqrt(2)*0.2})}]45:.2),
         ([shift={(0:{sqrt(2)*0.2})}]-45:.2)
        }%
        {%
            \fill[] \x circle(1.5pt);%
        }%
        \draw[thin] (225:.2) rectangle ([shift={(0:{sqrt(2)*0.2})}]45:.2);%
        \draw[thin] (45:.2) -- (315:.2);%
    }%
}%

\NewDocumentCommand{\xdiagdiamond}{O{.8}}{%
    \tikz[baseline = -0.5ex,scale=#1]{%
        \foreach \x in
        {(0:.2),
         (90:.2),
         (180:.2),
         (270:.2),
         (0,0)
        }%
        {%
            \fill[] \x circle(1.5pt);%
        }%
        \draw[thin] (0:.2) -- (90:.2) -- (180:.2) -- (270:.2) -- cycle;%
        \draw[thin] (90:.2) -- (270:.2);%
    }%
}%

\NewDocumentCommand{\xtwodiagrectangle}{O{.8}}{%
    \tikz[baseline = -0.5ex,scale=#1]{%
        \foreach \x in{45,-45,135,225}%
        {%
            \fill[](\x:0.2)circle(1.5pt);%
        }%
        \draw[thin] (225:.2) rectangle (45:.2);%
        \draw[thin]  (225:.2) -- (45:.2);%
        \draw[thin]  (-45:.2) -- (135:.2);%
    }%
}%

\NewDocumentCommand{\xthreetrangle}{O{.8}}{%
    \tikz[baseline = 0ex,scale=#1]{%
        \foreach \x in
        {(0:0),
         (60:.2),
         (120:.2),
         (0:.2),
         (0:-0.2)
        }%
        {%
            \fill[] \x circle(1.5pt);%
        }%
        \draw[thin] (0:.2) -- (60:.2) -- (120:.2) -- (0:-.2) -- cycle;%
        \draw[thin] (0:0) -- (60:.2);%
        \draw[thin] (0:0) -- (120:.2);%
    }%
}%

\NewDocumentCommand{\xlutriangle}{O{.8}}{%
    \tikz[baseline = 0ex,scale=#1]{%
        \foreach \x in
        {(0,0.3),
         (0.3,.3),
         (0,0)
        }%
        {%
            \fill[] \x circle(1.5pt);%
        }%
        \draw[thin] (0,0) -- (0,.3) -- (0.3,.3) -- cycle;%
    }%
}%

\NewDocumentCommand{\xrutriangle}{O{.8}}{%
    \tikz[baseline = 0ex,scale=#1]{%
        \foreach \x in
        {(0,0.3),
         (0.3,.3),
         (0.3,0)
        }%
        {%
            \fill[] \x circle(1.5pt);%
        }%
        \draw[thin] (0.3,0) -- (0,.3) -- (0.3,.3) -- cycle;%
    }%
}%

\NewDocumentCommand{\xlineuptwo}{O{.8}}{%
    \tikz[baseline = 0ex,scale=#1]{%
        \foreach \x in
        {(0,0.3),
         (0.3,0.3),
         (0.6,0.3)
        }%
        {%
            \fill[] \x circle(1.5pt);%
        }%
        \draw[thin] (0,0.3) -- (0.3,.3) -- (0.6,.3) -- cycle;%
    }%
}%

\NewDocumentCommand{\xlineupone}{O{.8}}{%
    \tikz[baseline = 0ex,scale=#1]{%
        \foreach \x in
        {(0,0.3),
         (0.3,0.3),
         (0.6,0.3)
        }%
        {%
            \fill[] \x circle(1.5pt);%
        }%
        \draw[thin] (0,0.3) -- (0.3,.3);%
    }%
}%

\NewDocumentCommand{\xrdtriangle}{O{.8}}{%
    \tikz[baseline = 0ex,scale=#1]{%
        \foreach \x in
        {(0,0),
         (0.3,.3),
         (0.3,0)
        }%
        {%
            \fill[] \x circle(1.5pt);%
        }%
        \draw[thin] (0,0) -- (.3,0) -- (0.3,.3) -- cycle;%
    }%
}%

\NewDocumentCommand{\xlinet}{O{.8}}{%
    \tikz[baseline = 0ex,scale=#1]{%
        \foreach \x in
        {(0,0),
         (0.3,.3),
         (0.6,0.3)
        }%
        {%
            \fill[] \x circle(1.5pt);%
        }%
        \draw[thin] (0,0) -- (.3,.3) -- (0.6,.3);%
    }%
}%

\NewDocumentCommand{\xlined}{O{.8}}{%
    \tikz[baseline = 0ex,scale=#1]{%
        \foreach \x in
        {(0,0),
         (0.3,.3),
         (0.6,0.3)
        }%
        {%
            \fill[] \x circle(1.5pt);%
        }%
        \draw[thin] (0,0) -- (.3,.3);%
    }%
}%

\NewDocumentCommand{\xlinef}{O{.8}}{%
    \tikz[baseline = 0ex,scale=#1]{%
        \foreach \x in
        {(0,0),
         (0,.3),
         (0.3,0.3)
        }%
        {%
            \fill[] \x circle(1.5pt);%
        }%
        \draw[thin] (0,0) -- (0,.3);%
    }%
}%

%% file: 1intro.tex
\section{Introduction}\label{sec:intro}

Networks, spanning diverse fields such as social sciences, biology, and computer science, are widely used as data structures for exploring complex systems. Statistical network analysis serves as a powerful toolset for uncovering patterns, structures, and dynamics within these networks, providing insights into phenomena
ranging from social interactions to biological processes \citep{barabasi2013network,newman2018networks}. Here, we are interested in characterizing a population of networks based on a single observed network, allowing for a broader understanding of the underlying structure and dynamics of complex systems.

Network motif counts, such as the number of triangles or stars, are crucial for understanding local structure and connectivity patterns within networks. By quantifying the prevalence of these motifs across a population of networks, we can discern common structural patterns and infer underlying mechanisms governing network formation and function \citep{borgs2010moments,bickel2011method}. For example, a high number of triangles in a social network may indicate the presence of tightly-knit communities or cliques, while an abundance of stars could suggest influential hubs connecting disparate network parts \citep{wasserman1994social}.
Moreover, local motif counts reveal global properties and facilitate inference across networks of varying sizes but within the same population. Motif counts therefore play a central role in goodness-of-fit testing and model selection \citep{gao2017testing,klusowski2020estimating,yuan2022testing}, as well as network
comparison tasks such as two-sample tests and correlation analysis \citep{ghoshdastidar2017two,maugis2020testing,mao2021testing,shao2022higher}.

At a high level, our statistical task is as follows. Given a network $G$, assumed to be generated from a population graphon model \citep{bickel2009nonparametric}, and a set of motifs $R_1, \ldots, R_m$ of interest, we seek to characterize the joint distribution of properly rescaled motif counts, or network moments, for random networks generated from the same graphon model.

One seemingly natural approach is to estimate the true graphon model and then directly derive or sample the desired distribution. However, accurately identifying the graphon function is challenging without making restrictive assumptions \citep{chan2014consistent, yang2014nonparametric} or resorting to computationally infeasible methods \citep{choi2014co,olhede2014network,gao2015rate}. While there are computationally feasible and accurate methods for estimating the connection probabilities of the given network $G$ \citep{chatterjee2015matrix,zhang2017estimating,li2023network}, they are not suitable for studying population distributions at the graphon level. Additionally, these estimation approaches still depend on certain structural regularity assumptions. Therefore, we turn to resampling strategies, which are generally considered flexible and versatile for approximating population distributions. Many studies have explored resampling inference methods in network problems, including cross-validation \citep{chen2018network, li2020network}, bootstrap \citep{levin2019bootstrapping,green2022bootstrapping}, subsampling \citep{bhattacharyya2015subsampling,zhang2022edgeworth,lunde2023subsampling}, and conformal inference \citep{lunde2023conformal}.

Among these methods, those most closely related to our study are the subsampling approaches of \cite{zhang2022edgeworth} and \cite{lunde2023subsampling} and the
bootstrap approaches of \cite{levin2019bootstrapping} and \cite{green2022bootstrapping}, all of which focus on the distribution of motif counts from the population model.
However, these studies address the resampling approximation of the \emph{marginal} distribution for a single motif, offering a necessarily limited view of network
structure. Marginal distributions characterize each motif in isolation and cannot capture the dependence structure among motifs, potentially leading to less informative
or even misleading inferences. Consider, for instance, comparing two gene coexpression networks from distinct gene sets. Examined separately, the 2-star ($\xtwoline$) and
3-star ($\xthreestar$) counts may appear indistinguishable between the two networks. Yet since a 2-star is a subgraph of a 3-star, the two counts are inherently correlated. Conditioning on the 2-star level, one network may exhibit a significantly elevated 3-star frequency, revealing a structural difference that no marginal test can detect. This issue is central to our real data application in Section~\ref{sec:RealData}, which underscores the need to analyze the \emph{joint} distribution of motif counts as a fundamental tool for multivariate network inference.

Beyond distributional characterization, a second important problem is \emph{two-sample testing}: given two networks $G$ and $G'$ on different node sets and potentially of different sizes, test whether they arise from the same underlying graphon. This problem, known as testing between \emph{unmatchable} networks, has received considerable attention. Methods based on the random dot product graph model \citep{young2007random} have been developed \citep{tang2017nonparametric,
agterberg2020nonparametric,alyakin2024correcting}, but their scope is restricted to that parametric model. In the more general graphon setting, \cite{ghoshdastidar2017two} and \cite{shao2022higher} have proposed testing procedures based on network moments.

A critical and practically important complication arises when the two networks have \emph{unequal} edge densities. In applications, networks derived from different groups, gene sets, time points, or experimental conditions naturally differ in both size and connectivity level --- there is generally no reason to expect their overall densities to match. Handling this setting is therefore not a minor technical
generalization but a practical necessity for real-world network comparison. \cite{shao2022higher} is the only existing work that handles the unequal-density case within the graphon model; their method constructs a second-order-correct statistic from the full network and achieves valid inference for a single moment when comparing two networks with unequal densities. However, it is restricted to univariate testing and is not applicable to subsampling procedures, so it cannot leverage the joint distribution of multiple moments. Thus, subsampling-based multivariate inference for unmatchable networks under unequal densities remains an open problem, despite its clear practical importance.

In this paper, we address both problems. Our contributions are twofold. First, we prove that node subsampling provides an asymptotically accurate approximation of the \emph{joint} distribution of multiple network moments under a general graphon model, extending the known effectiveness of subsampling from single to multiple motifs. Although using subsampled network moments to approximate their joint distribution is a natural extension of the marginal approach, the theoretical analysis is highly
nontrivial. It requires a careful characterization of the dependence structure among network motifs and the corresponding multivariate asymptotic convergence. Second, building on this foundation, we propose a novel subsampling-based two-sample testing procedure for comparing unmatchable networks with unequal densities using more advanced subsampling inference techniques. The key idea is \emph{sparsification}: by uniformly sparsifying each network to a common target density via random edge removal, we reduce the unequal-density problem to the equal-density case, enabling valid comparisons. To handle the unknown densities in practice, we further introduce a
\emph{sample-splitting} subsampling strategy, where part of each network is used to estimate the sparisification probabilities and the remainder is used to construct the test statistic. This yields the first subsampling-based inferential procedure, to our knowledge, that is valid for this practically important setting.

%% file: 2preliminary.tex
\section{Notations, motif counts and network moments}\label{subsec:motifs}
Throughout this paper, we denote the set $\{1, \ldots, n\} $ for any positive integer $n$ by $[n]$, and denote the cardinality of a set by $|\cdot|$.  More generally, given a sequence of quantities $\{a_1, \ldots, a_m\}$, we will denote it by $[a_m]$, when it is clear in context.  Let $G $ be an undirected unweighted graph whose
node set is  $V(G) = \{v_1, \ldots, v_n\}$ and edge set is $\eE(G) = \{(v_i,v_j): v_i,v_j \in V(G)\}$. 
Furthermore, denote the density of $G$ by $\wh{\rho}_G = |\eE(G)|/[n(n-1)]$.

A graph $S$ is a subgraph of $G$, written as $S  \subset  G$, if $V(S) \subset V(G) $ and  $\eE(S) \subset \eE(G)$.  In particular, a subgraph $S \subset G$ is called an induced subgraph of $G$, denoted by $S \indsubset G$, if for any $v_i, v_j \in V(S)$, $(v_i,v_j)\in \eE(S)$ whenever $(v_i,v_j)\in \eE(G)$. Lastly, two graphs $S$ and $G$ are isomorphic, denoted by $S \cong G$, when there exists a bijective function $\phi$: $V(S) \to V(G)$ such that $(v_i,v_j)\in \eE(S)$ if and only if edge $[\phi(v_i), \phi(v_j)] \in \eE(G)$.

A motif refers to a (usually simple) graph, such as an edge ($\xline$), a 2-star/V-shape ($\xtwoline$), a triangle ($\xtriangle$), or a 3-star ($\xthreestar$),  which forms the building blocks of larger graphs. In this study, we denote a motif by $R$, with $|V(R)| = r$ representing the number of nodes and $|\eE(R)| = \mathfrak{r}$ representing the number of edges. We focus exclusively on connected motifs, aligning with previous research \citep{bickel2011method,bhattacharyya2015subsampling,lunde2023subsampling}. For network $G$ and motif $R$, the motif count of $R$ in $G$ is defined as the number of subgraphs of $G$ that are isomorphic to $R$: 
\begin{equation}
\label{eq:MotifCountp} 
    X_R(G)=\big|\{S: S \subset G, S \cong  R\}\big|.
\end{equation}

This functional has received considerable attention in network analysis \citep{cook1971complexity,milo2002network,maugis2020testing,bhattacharya2022fluctuations}. Note that the subgraph $S$ need not be an induced subgraph of $G$. In contrast, the induced motif count is defined as
\begin{equation}
 \label{eq:ExactMotifCountp}
    \tilde{X}_R(G)=\big|\{S: S \indsubset G, S \cong  R\}\big|,
\end{equation}
which requires that the subgraph $S$ in the calculation must be an induced subgraph. These two definitions are essentially equivalent due to their linear mapping relations \citep{bickel2011method,maugis2020testing}. However, the non-induced counts \eqref{eq:MotifCountp} offer a more streamlined theoretical analysis \citep{zhang2022edgeworth}. Thus, following \cite{bickel2011method}, \cite{bhattacharyya2015subsampling} and \cite{zhang2022edgeworth}, we focus on the non-induced motif count \eqref{eq:MotifCountp} for our theoretical studies, but the distributional properties also hold for the induced motifs. 
Our data analyses in Section~\ref{sec:RealData} employ induced counts for a better interpretability.

The scale of motif counts is influenced by the size of both network and motif, making direct comparisons across networks of different sizes less informative. To avoid this, it is common to rescale the motif count. For a given motif $R$, the (sample) network moment of $R$ in a graph $G$ is defined as
\begin{equation*}
    U_R(G) = \binom{n}{r}^{-1}X_{R}(G).
\end{equation*} Several efficient computation strategies for network moments are outlined in \cite{ribeiro2010g}, \cite{gonen2011counting}, and \cite{maugis2020testing}.

%% file: 3subsampling.tex
\section{Node subsampling and its properties} \label{sec:subsampling}


\subsection{Subsampling under the sparse graphon model}
Before presenting our multivariate inference of network moments, we first outline the probabilistic framework that defines the network population and facilitates our analysis, namely the graphon framework adopted from \citep{hoover1979relations,aldous1981representations,bickel2009nonparametric}.

\begin{definition}[Sparse graphon model]
Let the graphon function $w : [0,1]^2 \to[0,1]$ be a nonnegative Lebesgue measurable function, such that $w(u,v)=w(v,u)$ for any $u, v \in [0,1]$, such that $\int_0^1 \int_0^1 w(u,v) \,du  \,dv=1$. Define a sequence of scalars $\rho_n \in [0,1]$. A random network is denoted as $\mathbb{G}_{n}\sim \rho_n w(u,v)$ if it is generated as follows.
\begin{enumerate}
    \item Generate $\{\xi_i\}_{i = 1}^n$ independently with
    \begin{equation}\label{eq:uniform-xi}
   \xi_i \sim \mathrm{Uniform}(0,1)
    \end{equation}
    \item For each node pair $(i, j)$, connect them independently with probability $\rho_n w(u,v)\mathbb{1}_{\{\rho_n w(u,v) \leq 1\}}$.
\end{enumerate}
\label{defi:graphon}
\end{definition}
The parameter $\rho_n$, governing network sparsity, typically tends towards $0$ at a specific rate. Similar to \cite{bickel2011method}, we always assume that $\rho_n w(u,v) \leq 1$ and ignore the constraint $\rho_nw(u,v)\le 1$.

\renewcommand{\thefootnote}{\fnsymbol{footnote}}

We assume that the observed network $G$ follows the sparse graphon model $\rho_nw(u,v)$. From $G$, our objective is to infer the distributional properties of network moments derived from this graphon model. Specifically, given a set of motifs $R_j$ for $j \in [m]$ and a sample size $b$ where $b < n$\footnote{In practical scenarios, $n$ is typically large, rendering the computation of network moments for $b \ge n$ infeasible, even without considering advanced inference tasks. Hence, we focus on the case where $b < n$.}, we aim to characterize the distribution of network moments $U_{R_j}(\bGb)$, for $\bGb$ drawn from $\rho_b w$. Our primary focus, as previously discussed, is on the joint distribution of $U_{R_j}(\bGb)$ for $j \in [m]$, rather than their marginal distributions.

It should be noted that though \citet{bickel2011method} established asymptotic distribution of full-network motif counts, this result cannot be directly used for practical inference as the correspondingly parameters are not available from a single network observation under the graphon model. Practical inference in this context thus has to rely on certain type of computation-intensive resampling procedures \citep{bhattacharyya2015subsampling, green2022bootstrapping, zhang2022edgeworth, lunde2023subsampling}, such as subsampling, bootstrap, and jacknife. Among all options, we focus on the subsampling, motivated by its flexibility and computational advantage for the network comparison problem (Section~\ref{sec:RealData}).


Consider an ideal scenario where the true graphon model $\rho_n w$ is known. In this context, we could approximate the distribution of $U_{R_j}(\bGb)$ for $j \in [m]$ directly using the Monte Carlo method: sampling $\bGb$ from the model,  computing the corresponding network moments, which give the empirical cumulative distribution functions. However, in our context, the graphon model is unknown, rendering the above procedure inapplicable. Nonetheless, if $n$ is sufficiently large, we can consider the graph $G$ as a discretized approximation of the true graphon, which allows for a feasible sampling procedure based on $G$ that resembles the Monte Carlo strategy. This insight forms the basis for the subsequent subsampling algorithm.



\begin{algorithm}
\caption{Uniform node subsampling for multivariate network moments}
\label{algo:subsampling}
\KwIn{Network $G$ of size $n$; motifs $[R_m] = \{R_1,\ldots, R_m\}$; replication number $N_{\mathrm{sub}}$; subsampling size $b$.}
\KwOut{$\wh{\rho}_G $; $\{U_{[R_m]}(G_b^{*(i)})\}_{i=1}^{N_{\mathrm{sub}}}$ for downstream inference tasks.}
\BlankLine
Calculate $\wh{\rho}_G = |\eE(G)|/[n(n-1)]$\;
\For{$i = 1, \dots, N_{\mathrm{sub}}$}{
    Randomly sample $b$ nodes (without replacement) from $[n]$ to be the subsampled set $\mathcal{S}$\;
    Set $G^{*(i)}_{b} \indsubset G$ to be the induced subgraph by $\mathcal{S}$\;
    Calculate the network moments of the subsampled graph $U_{R_j}(G^{*(i)}_{b})$ for $j\in [m]$\;
    Set the $m$-dimensional vector $U_{[R_m]}(G_b^{*(i)}) = (U_{R_1}(G^{*(i)}_{b}),\ldots, U_{R_m}(G^{*(i)}_{b}))$\;
}
\end{algorithm}

A crucial aspect of the subsampling approach is its emphasis on computing network moments within networks of size $b$ rather than $n$ during the generation of $U_{[R_m]}(G_b^{*(i)})$. Given that motif counting complexity typically increases superlinearly with network size \citep{ribeiro2010g}, this subsampling method emerges as a pivotal technique for addressing scalability in network inference tasks, allowing for the analysis of a large network $G$ by computing motif counts in a much smaller one of size $b$. Additionally, it is important to note that we keep the specific inference method for downstream tasks open in the algorithm, subsequent to obtaining the sample $\{U_{[R_m]}(G_b^{*(i)})\}_{i=1}^{N_{\mathrm{sub}}}$. This flexibility ensures that the process can accommodate any inference method the user prefers, ranging from intuitive visualization to more sophisticated testing procedures.

Subsampling procedures similar to ours have been explored by \cite{zhang2022edgeworth} and \cite{lunde2023subsampling}. However, as noted in Section~\ref{sec:intro}, those studies primarily focused on inferring individual network moments, particularly concerning the marginal distribution of single motifs. In contrast, we will examine the validity of our method on the joint distribution of network moments, laying the foundation for flexible multivariate inference on network moments. This generalization requires a precise characterization of the dependence between network moments, which is nontrivial when extending beyond marginal cases.


\subsection{The subsampling approximation to the joint distribution of network motifs}
Recall that $\mathbb{G}_{b}$ denotes a random graph generated from the graphon model with $b$ nodes, and $G$ denotes an {\it observed} graph with $n$ nodes. Moreover, we use $\mathbb{G}^{*}_{b}$ to denote a random subsample subgraph of $G$.  Our study will be based on a pre-defined set of motifs $\{R_1,\ldots,R_m\}$. Recall we use $[R_m]$ to denote $\{R_1,\ldots,R_m\}$ and $[t_m]$ to denote an $m$-dimensional vector of scalars $(t_1,\ldots,t_m)$. Define 
$$\Psi(x,y) = \left[\frac{y_1}{x^{\mathfrak{r}_1}}, \frac{y_2}{x^{\mathfrak{r}_2}}, \ldots, \frac{y_m}{x^{\mathfrak{r}_m}}\right]$$
for a scalar $x$ and $m$-dimensional $y$. Denote the $m$-dimensional vector ${U}_{[R_m]}(G)$ as the vector of motif counts for $[R_m]$ in $G$. In particular, we study the distribution of the normalized motif vector
$$\Psi(\wh{\rho}_G,{U}_{[R_m]}(\mathbb{G}^*_b)) =  \big[\wh{\rho}^{-\mathfrak{r}_1}_{G}U_{R_1}(\mathbb{G}^*_b), \ldots, \wh{\rho}^{-\mathfrak{r}_m}_{G}U_{R_m}(\mathbb{G}^*_b)\big]$$
given $G$.

Formally, conditioning on $\bGn=G$, we focus on the CDF of $\Psi(\wh{\rho}_G,{U}_{[R_m]}(\mathbb{G}^*_b))$ under the proper centering and scaling:
\begin{align}
    J^{[R_m]}_{*,n,b}([t_m]) &= \pr_*\left\{\sqrt{b}\big(\Psi(\wh{\rho}_G,{U}_{[R_m]}(\mathbb{G}^*_b)) - \Psi(\wh{\rho}_G,{U}_{[R_m]}(G)) \big)\le [t_m]\right\}\label{eq:jointcdfs}. 
\end{align} Our goal is to show that the subsampling distribution, viewed as a random probability measure (with respect to the randomness of $\bG_n$), effectively approximates the multivariate network moments distribution of $\bGb$ from the graphon model. For this purpose, we define the graphon sampling distribution as
\begin{align}
J^{[R_m]}_{b,c}([t_m])  &= ~\pr\left\{\sqrt{bc}\big(\Psi(\rho_b, {U}_{[R_m]}(\mathbb{G}_{b}))-\mE[\Psi(\rho_b, {U}_{[R_m]}(\mathbb{G}_{b}))]\big) \leq  [t_m] \right\},
\label{eq:true-jointcdfs}
\end{align}
where the scaling factor $c$ is introduced to correct the sample size difference, whose form will be provided later. Similar correction was also used by \citet{zhang2022edgeworth}. 

\begin{assumption}[Sparsity level]
\label{ass:rho_n_h_n}
Define $r = \max\{r_1,\ldots,r_m\} $ and $\mathfrak{r} = \max\{\mathfrak{r}_1,\ldots,\mathfrak{r}_m\}$. There exists a constant $c_1>1$ such that $n\rho^{4\mathfrak{r}}_{n} \geqslant c_1\log(n) $ for sufficiently large $n$. Furthermore, $b\rho^{r/2}_{n}  \to \infty$, and  $b\rho_{n}^{2\mathfrak{r}} \to \infty$ as $n\to \infty$.
\end{assumption}

\begin{assumption}[Subsampling size]
\label{ass:b}
The subsample size  $b\to \infty$ as $n\to \infty$ and  $\lim_{n\to \infty}b/n = c_2$ for a constant $c_2\in [0,1)$. 
\end{assumption}

\begin{assumption}[Non-degenerate moment]
\label{ass:non_degenerate}
As $n\to \infty$, the covariance matrix of $\sqrt{n}\Psi(\rho_n, {U}_{[R_m]}(\mathbb{G}_{n}))$ converges to a positive definite matrix.
\end{assumption}

Assumptions~\ref{ass:rho_n_h_n}--\ref{ass:non_degenerate} establish a transparent foundation for our multivariate framework. To clarify the explicit interplay between sampling and sparsity, we compare our conditions with the univariate frameworks of \citet{green2022bootstrapping}, \citet{zhang2022edgeworth}, and \citet{lunde2023subsampling}. First, by assuming $\lim_{n \to \infty} b/n \in [0, 1)$ (Assumption~\ref{ass:b}), our framework unifies the vanishing subsample regime ($b = o(n)$) of \citet{lunde2023subsampling} and the macroscopic regime ($b \asymp n$) of \citet{zhang2022edgeworth}. Second, our sparsity conditions (Assumption~\ref{ass:rho_n_h_n}) impose precise structural boundaries for both the full network and the subsample. For the full network, the threshold $n\rho_n^{4\mathfrak{r}} \ge c_1\log n$ ensures the uniform concentration of complex subgraph overlaps. Because we establish joint convergence \textit{with probability one} (Theorem~\ref{theo:consistent}), bounding the variance of cross-covariance estimators, which involves 4th-order moments and generates overlapping subgraphs with up to $4\mathfrak{r}$ edges, requires the $\log n$ factor to guarantee polynomial tail decay via the Borel-Cantelli lemma. The necessity of bounding these exact $4\mathfrak{r}$-edge overlaps is also recognized by \citet{green2022bootstrapping}, who impose a $b\rho_n^{4\mathfrak{r}} \to \infty$ requirement on their empirical graphon resample size. Notably, regarding the required subsample size, our dynamic bound ($b\rho_n^{2\mathfrak{r}} \to \infty$) is milder than the $b\rho_n^{4\mathfrak{r}} \to \infty$ requirement imposed by the empirical graphon bootstrap in \citet{green2022bootstrapping} (e.g., requiring $b\rho_n^6 \to \infty$ rather than $b\rho_n^{12} \to \infty$ for triangles). This explicit interplay guarantees that even sublinear subsamples are dense enough to prevent the joint covariance matrix from degenerating, whereas \citet{zhang2022edgeworth} express sparsity solely via $n$ by fixing $b \asymp n$, and \citet{lunde2023subsampling} absorb this dynamic into an implicit assumption. Finally, by deriving first-order joint consistency directly from primitive non-degeneracy conditions (Assumption~\ref{ass:non_degenerate}), our framework safely avoids the restrictive continuous Cram\'er condition required by \citet{zhang2022edgeworth}. This ensures our theory supports discrete graphons, such as the stochastic block model, without resorting to the high-level convergence assumptions of \citet{lunde2023subsampling}. We have the following property for our node subsampling distribution.

\begin{theorem}
\label{theo:consistent}
Under Assumptions~\ref{ass:rho_n_h_n}--\ref{ass:non_degenerate},
with probability one (with respect to the random sequence $\{\mathbb{G}_n\}$),
\begin{equation}
\label{eq:multi_consistent}
\begin{aligned}
    \sup\limits_{[t_m]\in \mathbb{R}^m}\Big| J^{[R_m]}_{*,n,b}([t_m]) -  J^{[R_m]}_{b,(1 - b/n)}([t_m]) \Big|   \to 0.
\end{aligned}
\end{equation}
\end{theorem}

Theorem~\ref{theo:consistent} bridges a critical theoretical gap in the statistical network literature. While \citet{bickel2011method} first established the weak convergence of full-network motif counts, the resulting asymptotic distribution is intractable for practical inference. Subsequent studies \citep{green2022bootstrapping, zhang2022edgeworth, lunde2023subsampling} bypassed this issue using subsampling or bootstrap approximations, but their theoretical guarantees were strictly limited to the \textit{univariate marginal distribution} $J^{\{R_1\}}_{b,1}(t_1)$. Theorem~\ref{theo:consistent} provides, to our knowledge, the first proof that node subsampling yields an asymptotically valid approximation for the {\it joint multivariate distribution} of motif counts under a general graphon model. Theoretically, extending these results to a multivariate setting fundamentally advances the literature in three directions:

\begin{enumerate}
    \item {\it Generalizing network finite-population U-statistics for multivariate moments}: Subsampling without replacement induces non-negligible dependence among sampled nodes, rendering standard infinite-population U-statistic theory \citep[e.g.,][]{serfling2009approximation} inapplicable. While \citet{zhang2022edgeworth} addressed this for univariate motifs using the finite-population framework of \citet{bloznelis2001orthogonal}, establishing joint multivariate convergence requires a substantial theoretical expansion. We bridge this gap by formally linking multivariate network moment estimators to finite-population U-statistics, enabling a rigorous characterization of their exact joint subsampling distribution.
    \item {\it Explicit characterization of heterogeneous covariances}: Evaluating the joint distribution requires characterizing the cross-covariance between distinct network motifs—a combinatorially demanding task compared to analyzing a single motif's variance. This covariance depends intricately on a vast array of partially overlapping subgraph structures (see details in Section~\ref{subsec:moreTheroemNMgraphon1}). Due to this structural complexity, explicit covariance expressions, let alone their asymptotic limits, are remarkably sparse in the literature \citep{bhattacharyya2015subsampling, maugis2020testing}. Leveraging graph limit theory \citep{lovasz2006limits, lovasz2012large}, we systematically characterize the asymptotic behavior of these cross-motif covariances, formally capturing how complex structural dependencies manifest as the network grows.
    \item {\it Joint distribution convergence from primitive conditions}: Recent work by \citet{lunde2023subsampling} provides a versatile and comprehensive subsampling framework for general univariate network statistics. To elegantly accommodate complex global functionals (such as eigenvalues), their methodology relies on a high-level condition, essentially assuming \textit{a priori} that both the full-sample and subsampled statistics converge to a well-behaved limiting distribution (Assumption 1 in their paper). While such high-level conditions are highly effective for broad classes of network functionals, imposing this assumption for network motifs would bypass the core theoretical challenge: mathematically characterizing the intricate cross-motif dependencies and joint asymptotic behaviors, especially when extending the analysis to multivariate cases. Furthermore, as \citet{lunde2023subsampling} note, relying on such an assumption imposes implicit, opaque restrictions on allowable network sparsity and subsampling rates. A primary contribution of our analysis is tackling this underlying structural challenge directly. Rather than treating convergence as a starting assumption, Theorem~\ref{theo:consistent} explicitly \textit{derives} the exact joint limiting distribution from first principles. By relying strictly on primitive, easily interpretable network properties, explicit sparsity bounds, subsample sizes, and non-degeneracy, our analysis executes the rigorous mathematical lifting required to prove multivariate consistency from the ground up.
\end{enumerate}

    For purely acyclic motifs, univariate frameworks achieve consistency under the sharper condition $b\rho_n \to \infty$ \citep{bickel2011method, zhang2022edgeworth}. While our unified bound (e.g., $b\rho_n^{2\mathfrak{r}} \to \infty$) is structurally conservative for individual trees, it is a mathematical necessity for establishing multivariate joint convergence with probability one. Although the cross-covariances of acyclic motifs converge under the milder $b\rho_n \to \infty$ rate, our proof of the joint convergence requires the Lindeberg-Feller condition for arbitrary combinations of heterogeneous motifs. Propogating the benefits of single acyclic motif to arbitrary combinations of motifs becomes intractable. Our proof thus relies on a deterministic, worst-case bound on the local motif counts. This uniform bounding overrides the topological benefits of acyclic structures, mathematically restricting the required global sparsity rate to $b\rho_n^{2\mathfrak{r}} \to \infty$.

%% file: comparison.tex
\section{Unmatchable Network Comparison with Unequal Densities}

Network comparison involves determining whether two or more networks originate from the
same underlying population, a question that has gained significant attention recently.
For example, studies like \cite{ghoshdastidar2018practical,maugis2020testing,yuan2023practical}
have focused on comparing two groups of networks, where each group contains a large number
of individual networks. In contrast, research such as
\cite{tang2017semiparametric,li2018two,liu2021meta,chatterjee2023two,du2023hypothesis}
has explored comparisons between two individual networks that share the same set of nodes,
often known as ``matchable networks''.

Comparing ``unmatchable'' networks---those differing in both size and node composition---introduces
further complications. Various methods \citep{tang2017nonparametric,agterberg2020nonparametric,alyakin2024correcting}
have been developed to address these challenges, particularly under the random dot product graph
model \citep{young2007random}. In the context of the more general graphon model,
\cite{ghoshdastidar2017two} and \cite{shao2022higher} have introduced hypothesis testing
procedures based on network moments, which can be incorporated into resampling methods.
However, these methods test the marginal distribution of each network moment separately,
forgoing the power gains and richer inference afforded by their joint distribution.
As one of its important applications, the subsampling approach studied in this paper
paves the way for comparing unmatchable networks through the lens of multivariate inference
on network moments.

Consider two unmatchable networks, $G$ and $G'$, with sizes $n$ and $b$, respectively.
Assume they are realizations of two graphon models $\bG \sim \rho_n w$ and $\bG' \sim \rho_{b} w'$. We aim to test $H_0: w = w'$ using multiple network moments jointly.
Given the motifs of interest, $R_1, \ldots, R_m$, we consider comparing the two networks according to the subsampling distribution of the motifs. 

A natural approach is as follows. When $b$ is sufficiently large, we subsample from $G$ using Algorithm~\ref{algo:subsampling} to approximate the true distribution of network moments \eqref{eq:true-jointcdfs} for a network of size $b$ drawn from graphon $w$, and then assess whether the observed network
moments of $G'$ are consistent with this distribution.

Unfortunately, the above strategy fails when $\rho_n \neq \rho_{b}$. To see why, Theorem~\ref{theo:consistent} requires that
$\Psi(\widehat{\rho}_{G},\, {U}_{[R_m]}(G^*_{b}))$ asymptotically match the distribution of $\Psi(\rho_{b},\, {U}_{[R_m]}(G'))$ under the null, up to the correction factor $\sqrt{1-b/n}$. Since $\rho_{b}$ is unknown in practice, one might replace it with the empirical density $\widehat{\rho}_{G'}$. However, this self-normalization distorts the null distribution. Moreover, naively
applying self-normalization to both $\Psi(\widehat{\rho}_{G^*_b},\, {U}_{[R_m]}(G^*_{b}))$
and $\Psi(\widehat{\rho}_{G'},\, {U}_{[R_m]}(G'))$ is equally invalid, as we demonstrate in Section~\ref{sec:necessary-sparsification}.
To the best of our knowledge, \cite{shao2022higher} is the only existing work that handles hypothesis testing of $H_0$ under unequal densities in the graphon model. However, their method constructs a second-order-correct statistic from the full network and is restricted to univariate testing; furthermore it is not applicable to subsampling procedures and cannot leverage their computational advantages. Thus, subsampling-based inference under unequal densities
remains a completely open problem.

Here, we propose a novel subsampling-based test to resolve this problem via a simple but powerful idea: \emph{sparsification}. Since the fundamental obstacle is the density mismatch between $G$ and $G'$, we convert both networks to a common target density $\rho^\dagger$ via uniform random edge removal,
while leaving their underlying graphons unchanged. Specifically, for a fixed $\rho^\dagger$ smaller than both $\rho_n$ and $\rho_{b}$, we independently retain each edge in $G$ with probability $\rho^\dagger \rho_n^{-1}$ and each edge in $G'$ with probability $\rho^\dagger \rho_{b}^{-1}$. The resulting networks then follow $\rho^\dagger \cdot w$ and $\rho^\dagger \cdot w'$, respectively, reducing the comparison to the equal-density case and enabling the direct application of our subsampling framework.

While the above idea is simple, a practical challenge remains: $\rho_n$ and $\rho_{b}$ are unknown, so the sparsification probabilities cannot be computed directly. We address this via a sample-splitting procedure: each network is randomly split into two subnetworks, $G^1$ and $G^2$, where $G^1$ is used to estimate the sparsification probabilities and $G^2$ is then sparsified to construct the network moments as the test statistic. For simplicity of exposition, we use an equal-size split, though this is not required. The full procedure is summarized in Algorithm~\ref{algo:sparsification_stat}.

\begin{algorithm}
\caption{Externally sparsified network moments $\bar{\Psi}_{\rho^\dagger}(G^1, G^2)$}
\label{algo:sparsification_stat}
\KwIn{Two networks $G^1$ and $G^2$; motifs $[R_m] = \{R_1,\ldots,R_m\}$; target density $\rho^\dagger$.}
\KwOut{A normalized $m$-dimensional statistic $\bar{\Psi}_{\rho^\dagger}(G^1, G^2)$.}
\BlankLine
Estimate the sparsification probability: $\widehat{p} = \min\!\bigl(1,\, \rho^\dagger / \widehat{\rho}_{G^1}\bigr)$.\;
Sparsify $G^2$ by independently removing each edge with probability $1 - \widehat{p}$,
yielding $\widetilde{G}^2$.\;
\Return $\bar{\Psi}_{\rho^\dagger}(G^1, G^2) = \Psi\!\bigl(\widehat{\rho}_{\widetilde{G}^2},\, U_{[R_m]}(\widetilde{G}^2)\bigr)$.\;
\end{algorithm}

The sparsified moments $\bar{{\Psi}}$ in Algorithm~\ref{algo:sparsification_stat} can be viewed as a stochastically manipulated version of the standard network moments used earlier. It is natural to conjecture that the subsampling inference remains valid for these generalized moments. Motivated by this, we propose the following subsampling comparison algorithm. In the final step of Algorithm~\ref{algo:sparsification_test}, any valid multivariate test for $H_0$ may be applied.

\begin{algorithm}
\caption{Subsampling comparison between networks with unequal densities}
\label{algo:sparsification_test}
\KwIn{Networks $G$ (size $n$) and $G'$ (size $b$); motifs $R_1,\ldots,R_m$; target density
$\rho^\dagger$; subsampling size $N_{\mathrm{sub}}$.}
\BlankLine
Randomly split the nodes of $G'$ into two equal subsets, inducing subgraphs
$G^{\prime 1}$ and $G^{\prime 2}$; compute
$\bar{\Psi}_{\rho^\dagger}(G^{\prime 1}, G^{\prime 2})$ via
Algorithm~\ref{algo:sparsification_stat}.\;
\For{$i = 1, \dots, N_{\mathrm{sub}}$}{
    Independently subsample two sets of $\lfloor b/2 \rfloor$ nodes from $[n]$, inducing
    subgraphs $G^{*(i1)}_{b/2}$ and $G^{*(i2)}_{b/2}$.\;
    Compute $\bar{\Psi}_{\rho^\dagger}(G^{*(i1)}_{b/2}, G^{*(i2)}_{b/2})$ via
    Algorithm~\ref{algo:sparsification_stat}.\;
}
Compare $\bar{\Psi}_{\rho^\dagger}(G^{\prime 1}, G^{\prime 2})$ against the
empirical distribution of
$\bigl\{\bar{\Psi}_{\rho^\dagger}(G^{*(i1)}_{b/2}, G^{*(i2)}_{b/2})\bigr\}_{i=1}^{N_{\mathrm{sub}}}$
via a multivariate statistical test.
\end{algorithm}

Denote $\eta_w = {E}[\Psi(\rho_b, U_{[R_m]}(\mathbb{G}'))]$, which is a vector that depends only on $w$, but not on $\rho_b$ (Section~\ref{subsec:moreTheroemNMgraphon1}). 
The following theorem establishes the validity of the above procedure. It shows that, under $H_0: w = w'$, the subsampling distribution of $\bar{\Psi}_{\rho^\dagger}(\mathbb{G}^{*(i1)}_{b/2}, \mathbb{G}^{*(i2)}_{b/2})$
consistently approximates the sampling distribution of $\bar{\Psi}_{\rho^\dagger}(\mathbb{G}^{\prime 1}, \mathbb{G}^{\prime 2})$.

\begin{theorem}
\label{them:sparsification_validity}
Suppose Assumptions~\ref{ass:rho_n_h_n}--\ref{ass:non_degenerate} hold. Under $H_0$,
let $G \sim \rho_n w$ and $G' \sim \rho_{b} w$. Assume $\rho^\dagger = \kappa\min(\rho_n, \rho_{b})$
for some constant $\kappa > 0$. 
Denote by $\bar{J}^{[R_m]}_{b,c'}$ the CDF of
$\sqrt{c'b/2}\bigl(\bar{\Psi}_{\rho^\dagger}(\mathbb{G}^{\prime 1}, \mathbb{G}^{\prime 2}) - \eta_w\bigr)$,
and by $\bar{J}^{[R_m]}_{*,n,b}$ the CDF of
$\sqrt{b/2}\bigl(\bar{\Psi}_{\rho^\dagger}(\mathbb{G}^{*(i1)}_{b/2}, \mathbb{G}^{*(i2)}_{b/2}) - \eta_w\bigr)$
conditional on $G$.
Then, with probability one (with respect to the random sequence $\{\mathbb{G}_n\}$),
\begin{equation}
\label{eq:sparsified_consistency}
    \sup_{[t_m]\in \mathbb{R}^m}\Bigl|
      \bar{J}^{[R_m]}_{*,n,b}([t_m])
      - \bar{J}^{[R_m]}_{b,1-b/n}([t_m])
    \Bigr| \to 0, \quad \text{as } n, b \to \infty.
\end{equation}
\end{theorem}

%% file: 4simulation.tex
 \section{Simulation} \label{sec:simulation}
\subsection{Evaluation of subsampling approximation accuracy}\label{subsec:pre}

We now employ numerical studies to assess the accuracy of approximating subsampling distributions by evaluating the finite sample approximation error given by the righthand side of \eqref{eq:multi_consistent} under Assumptions \ref{ass:rho_n_h_n}-\ref{ass:non_degenerate}. Specifically, using networks generated from graphon models, we calculate the empirical Kolmogorov-Smirnov distance between $\widehat{J}^{[R_m]}_{*,n,b}$ and $\widehat{J}^{[R_m]}_{b,(1 -b/n)}$, which are the empirical cumulative distribution functions corresponding to \eqref{eq:jointcdfs} and \eqref{eq:true-jointcdfs}, respectively. We focus on the performance for $m=1$ (marginal distribution) and $2$ (bi-variate joint distribution), considering three basic motifs: $\xtwoline$ (2-star), $\xtriangle$ (triangle), and $\xthreestar$ (3-star). The experimental setups are detailed below:
\begin{itemize}
\item The true network models: two graphons from previous studies \citep{green2022bootstrapping,zhang2022edgeworth,lunde2023subsampling} are used.
\begin{enumerate}
    \item Graphon 1 (smooth): $ w(u, v) \propto \exp\{-25(u - v)^2/2\}$.
    \item Graphon 2 (nonsmooth): $w(u, v) \propto  0.5 \cos[0.1 \{\left(u -0.5\right)^2+ \left(v -0.5\right)^2\} + 0.01]$ $\cdot \max(u, v)^{2/3} + 0.4$.
\end{enumerate}
\item The network and subsampling sizes: $n$ varies from $2,000$ to $16,000$ and $b = \lceil n^{2/3} \rceil$.
\item Sparsity levels: Two sparsity levels are considered $\rho_n = 0.25n^{-0.1}$ and $\rho_n=0.25n^{-0.25}$.
\end{itemize}

Notice that the simulation configurations are designed to align with our theoretical assumptions. In particular, the subsampling size and the graphons are selected to satisfy Assumptions~\ref{ass:rho_n_h_n} and \ref{ass:non_degenerate}.  Moreover, $\rho_n = 0.25n^{-0.1}$ satisfies the density requirement in Assumption~\ref{ass:b}, whereas $\rho_n = 0.25n^{-0.25}$ falls outside that regime for the triangle and 3-star motifs. Examining both cases allows us to assess the extent to which the theoretical conclusions remain informative in a slightly sparser setting than covered by the assumptions.

For each configuration, the true cumulative distribution function is approximated by the empirical cumulative distribution function from network moments of size-$b$ networks sampled from the true model. To assess the approximation error, measured by the Kolmogorov-Smirnov distance, we generate a size-$n$ network from the true model and use the empirical cumulative distribution function of the subsampled $\{U_{[R_m]}(G_b^{*(i)})\}_{i=1}^{N_{\mathrm{sub}}}$ from Algorithm~\ref{algo:subsampling} with $N_{\mathrm{sub}}=2,000$. This process is replicated $50$ times, and we report the average approximation errors from these replications as the performance metric.


Figure~\ref{fig:simub06p01} displays the log-scale approximation errors for both the marginal and pairwise joint distributions under two graphon models at a sparsity level of $\rho_n = 0.25n^{-0.1}$. The errors across all evaluated cumulative distribution functions exhibit a clear decreasing trend. With both axes labeled on a log scale, this decreasing trend appears nearly linear as expected. The rate at which the marginal distribution errors decrease roughly aligns with the findings of \cite{zhang2022edgeworth}. Although the joint distributions show a slightly slower decrease in errors, the overall pattern remains the same. Both graphons, smooth or nonsmooth, demonstrate similar decreasing patterns, suggesting the subsampling method's robustness to graphon smoothness.


\begin{figure}[ht]
\centering
\begin{subfigure}[t]{0.22\textwidth}
\centering
\includegraphics[width=\textwidth]{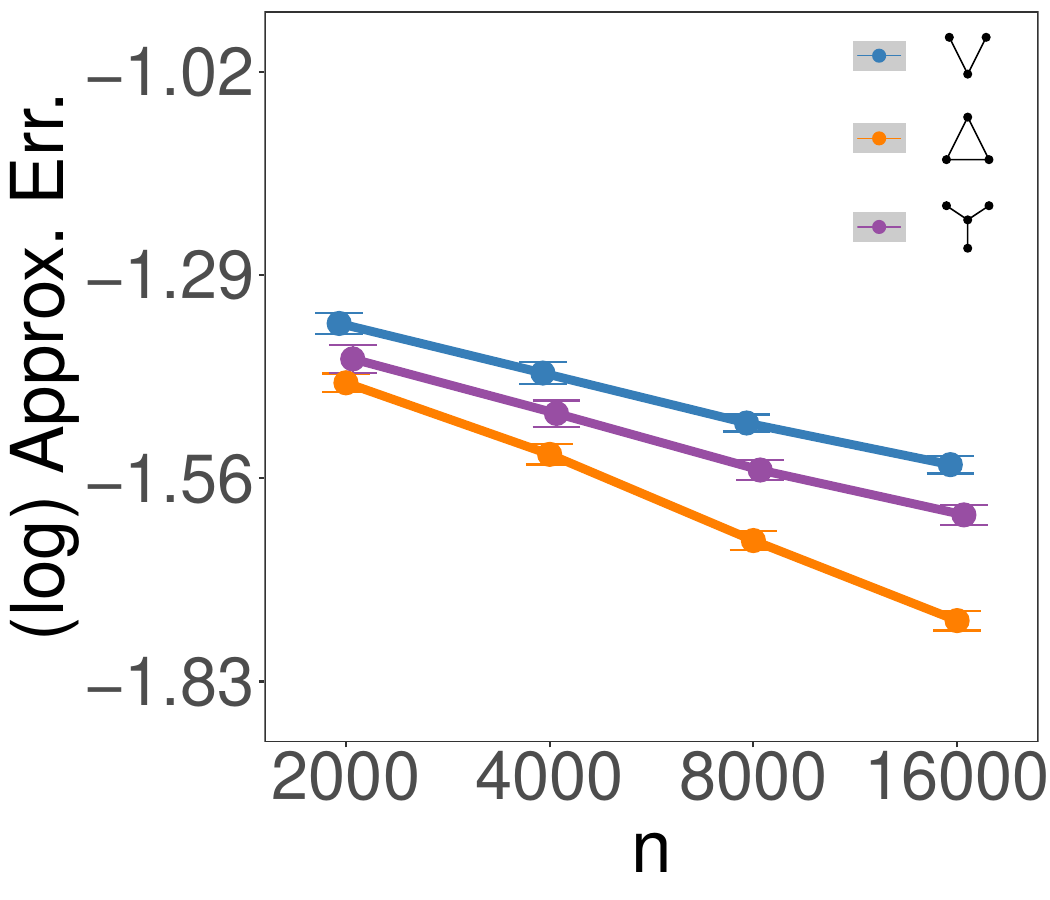}
\captionsetup{justification=centering}
\caption{\scriptsize Graphon 1: marginal cumulative distribution functions}
\label{subfig:loguni_g2}
\end{subfigure}
\hfill
\begin{subfigure}[t]{0.22\textwidth}
\centering
\includegraphics[width=\textwidth]{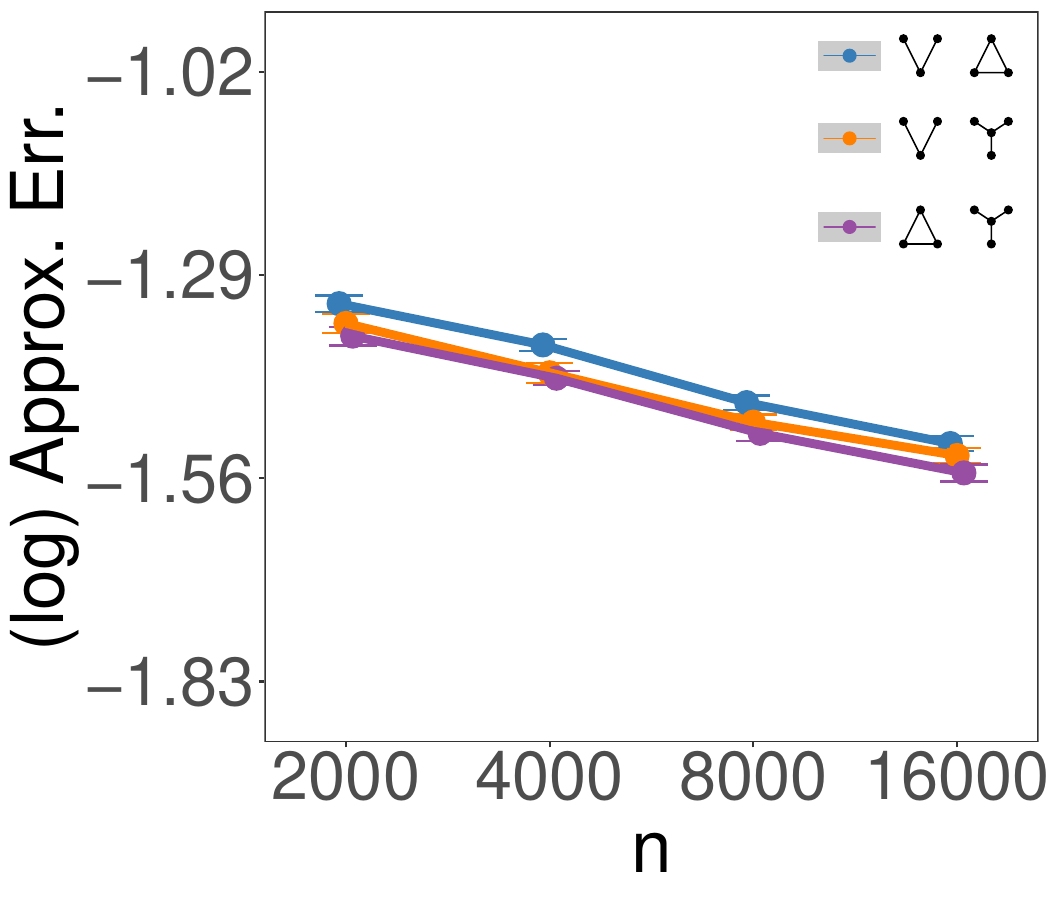}
\captionsetup{justification=centering}
\caption{\scriptsize Graphon 1: joint cumulative distribution functions}
\label{subfig:logbi_g2}
\end{subfigure}
\hfill
\begin{subfigure}[t]{0.22\textwidth}
\centering
\includegraphics[width=\textwidth]{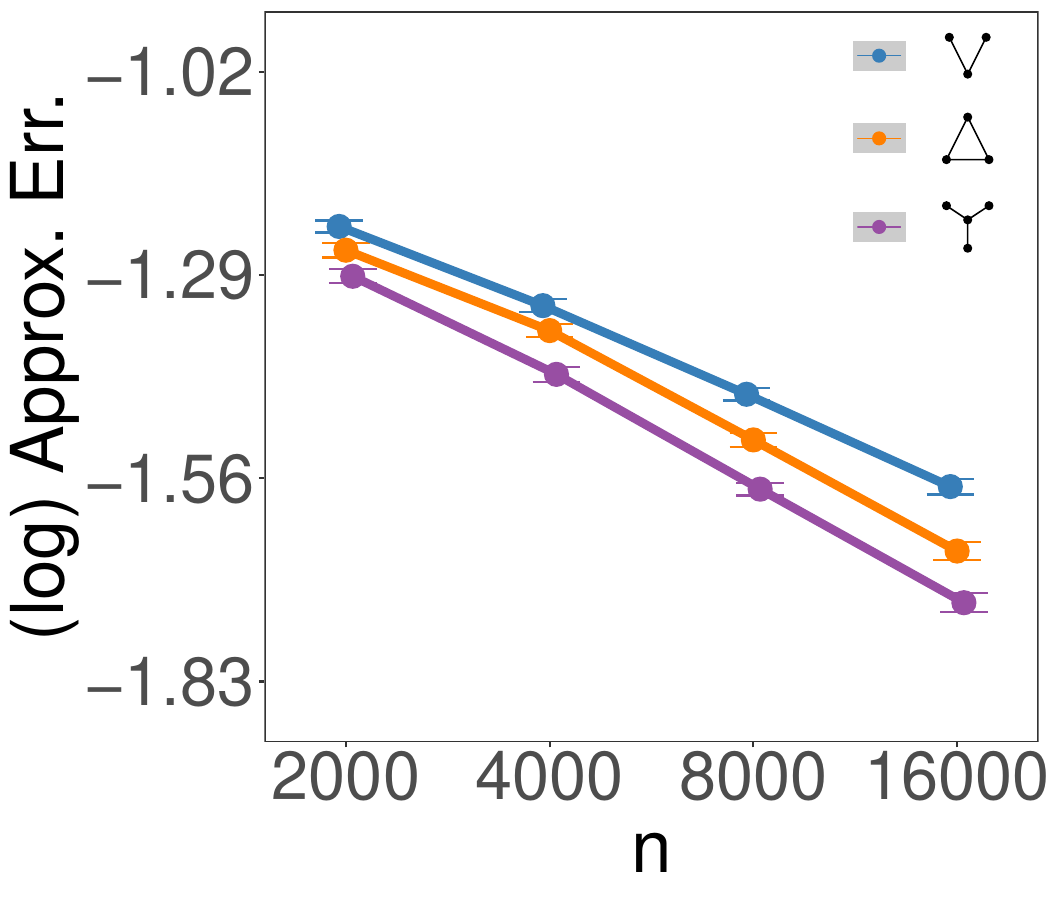}
\captionsetup{justification=centering}
\caption{\scriptsize Graphon 2: marginal cumulative distribution functions}
\label{subfig:loguni_g1}
\end{subfigure}
\hfill
\begin{subfigure}[t]{0.22\textwidth}
\centering
\includegraphics[width=\textwidth]{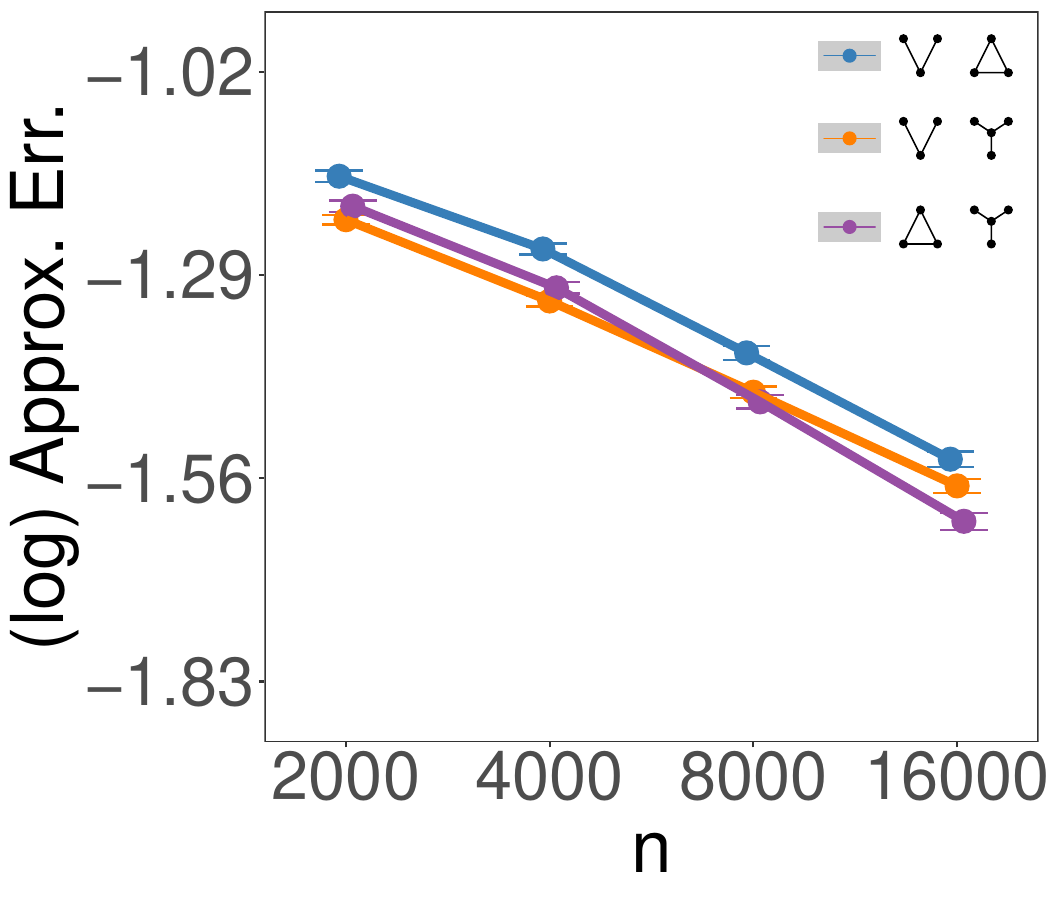}
\captionsetup{justification=centering}
\caption{\scriptsize Graphon 2: joint cumulative distribution functions}
\label{subfig:logbi_g1}
\end{subfigure}
\caption{Empirical approximation errors of the cumulative distribution functions under the sparsity level $\rho_n = 0.25n^{-0.1}$.}
\label{fig:simub06p01}
\end{figure}

Figure~\ref{fig:simub06p02} presents results under a sparser setting with $\rho_n=0.25n^{-0.25}$. The pattern remains consistent with previous results, though the errors are slightly higher due to the increased sparsity. The variation in numerical values across different motifs is more pronounced, yet the overall trend remains the same. It is important to note that excessive sparsity can weaken the signal-to-noise ratio to the extent that the approximation may fail, a known issue in network resampling methods \citep{zhang2022edgeworth,green2022bootstrapping,lunde2023subsampling}. We explore such an overly sparse scenario in Section~\ref{subsec:addsimu}. Additional results for experiments with a subsampling size of $b=\lceil 2n^{1/2} \rceil$ are also available in Section~\ref{subsec:addsimu}.


\begin{figure}[ht]
\centering
\begin{subfigure}[t]{0.22\textwidth}
\centering
\includegraphics[width=\textwidth]{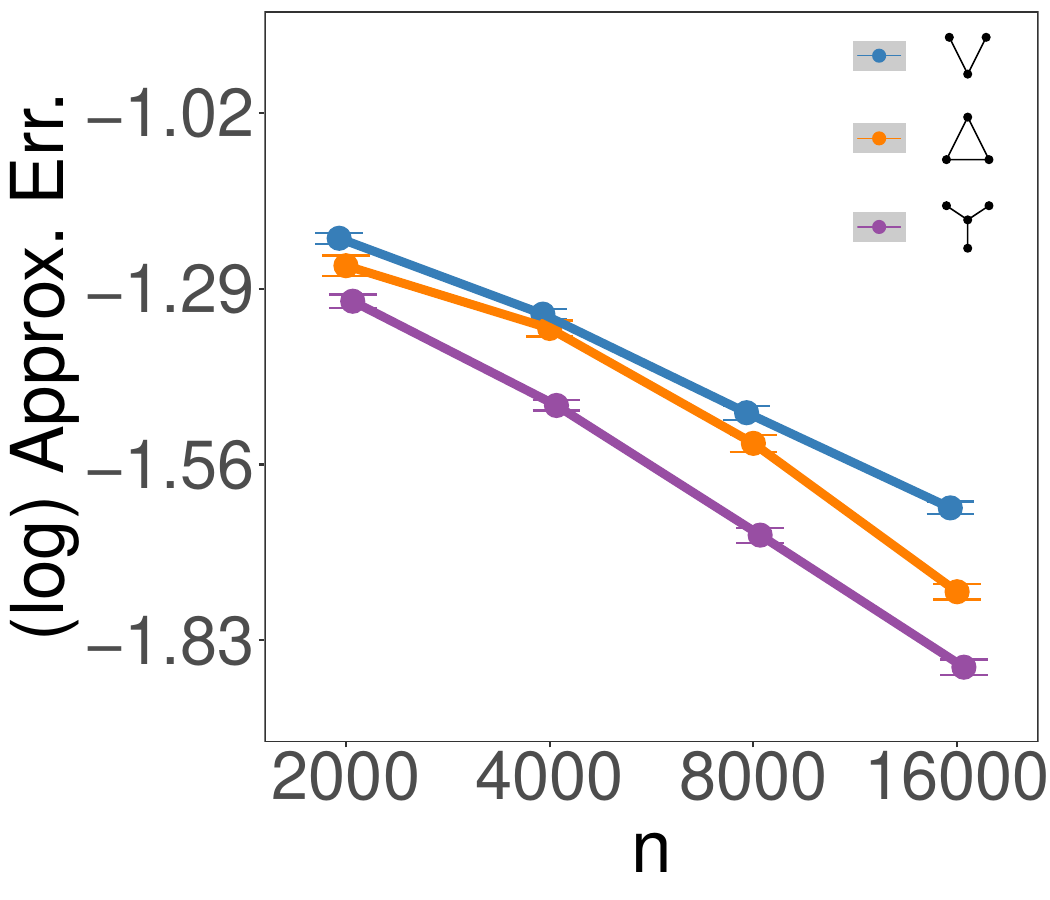}
\captionsetup{justification=centering}
\caption{\scriptsize Graphon 1: marginal cumulative distribution functions}
\label{subfig:loguni_g2_b06p02}
\end{subfigure}
\hfill
\begin{subfigure}[t]{0.22\textwidth}
\centering
\includegraphics[width=\textwidth]{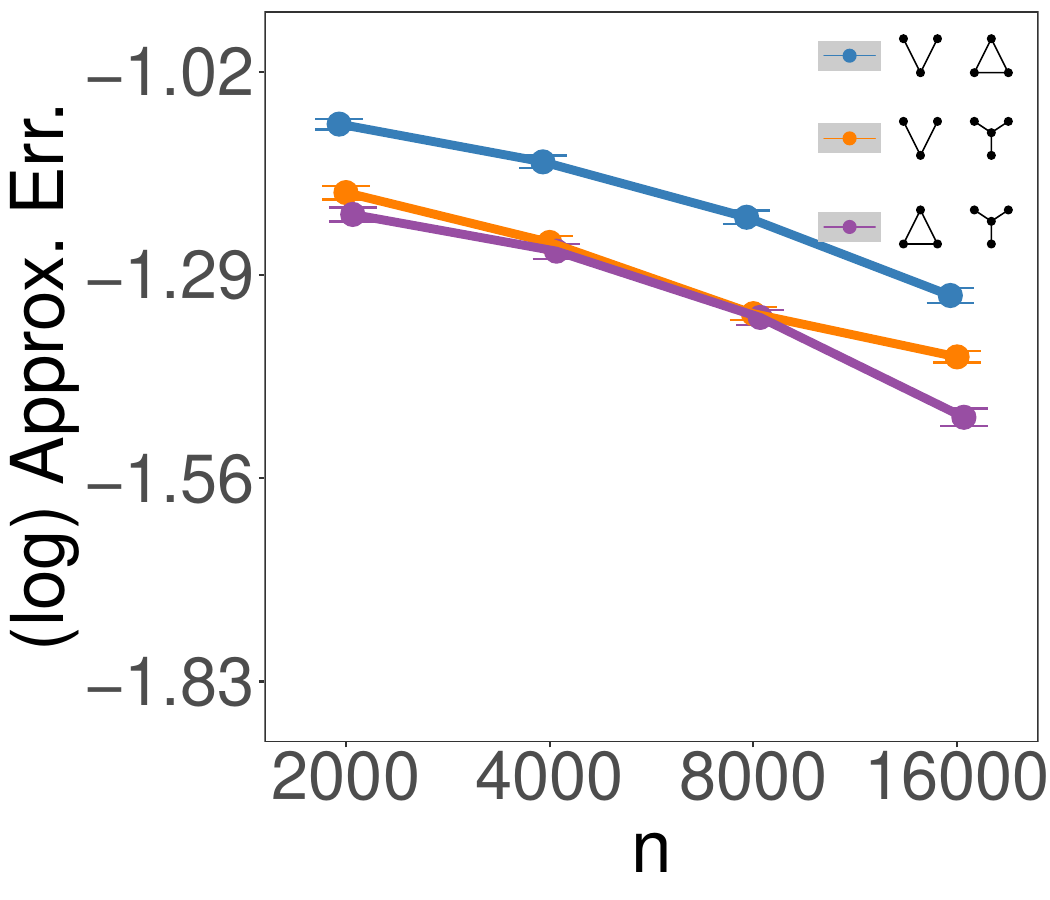}
\captionsetup{justification=centering}
\caption{\scriptsize Graphon 1: joint cumulative distribution functions}
\label{subfig:logbi_g2_b06p02}
\end{subfigure}
\hfill
\begin{subfigure}[t]{0.22\textwidth}
\centering
\includegraphics[width=\textwidth]{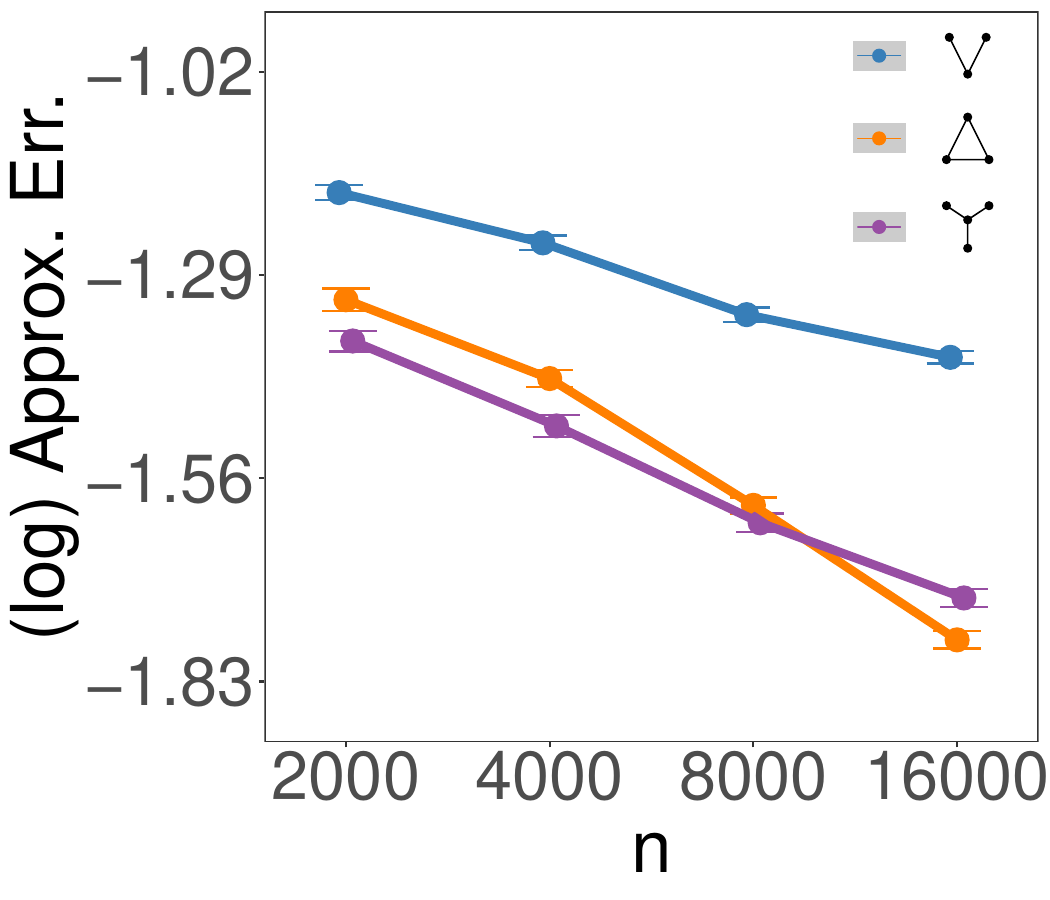}
\captionsetup{justification=centering}
\caption{\scriptsize Graphon 2: marginal cumulative distribution functions}
\label{subfig:loguni_g1_b06p02}
\end{subfigure}
\hfill
\begin{subfigure}[t]{0.22\textwidth}
\centering
\includegraphics[width=\textwidth]{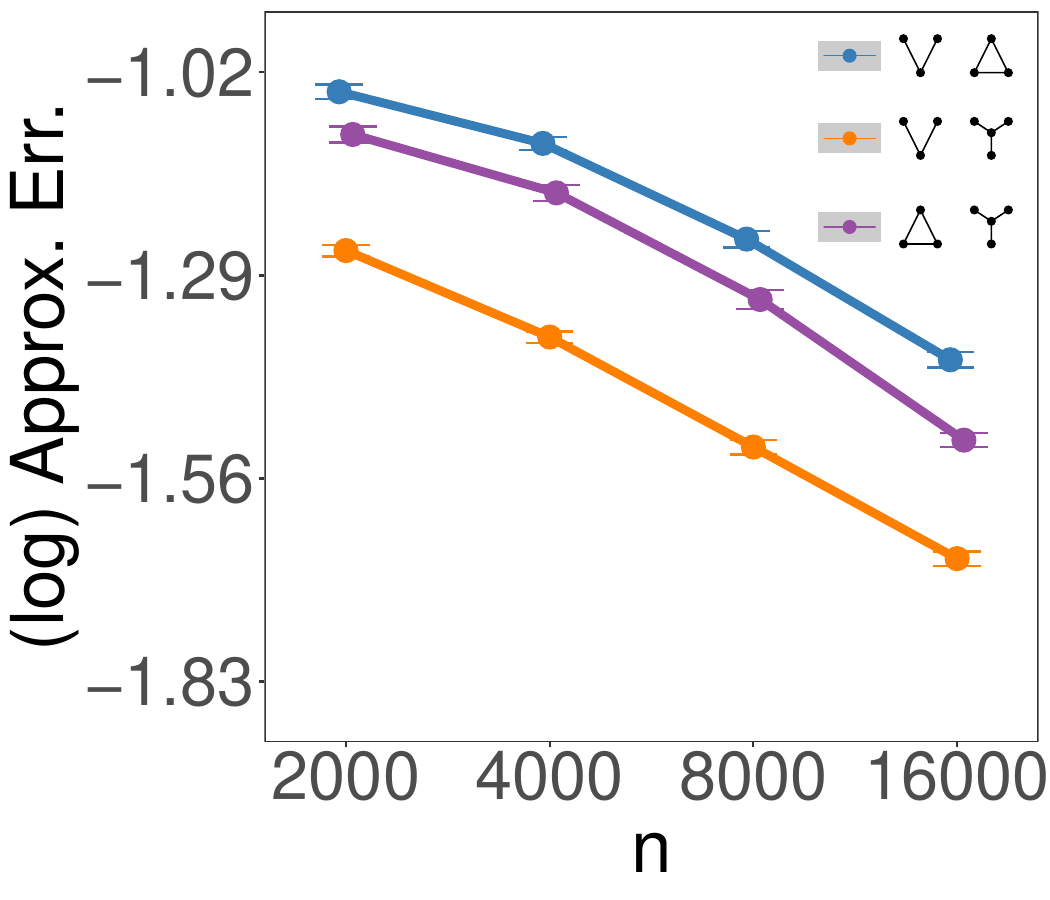}
\captionsetup{justification=centering}
\caption{\scriptsize Graphon 2: joint cumulative distribution functions}
\label{subfig:bi_logg1_b06p02}
\end{subfigure}
\caption{Empirical approximation errors of the cumulative distribution functions under the sparsity level $\rho_n = 0.25n^{-0.25}$.}
\label{fig:simub06p02}
\end{figure}

\subsection{Evaluation of the density-matching two-sample test}\label{subsec:two-sample}
Next, we evaluate Algorithm~\ref{algo:sparsification_test} in terms of type~I error control and
power. We set $n = 8,000$, $b = 500$, $\rho^\dagger = \kappa\min(\rho_n,\rho_b)$ with $\kappa = 0.7$,
and $N_{\mathrm{sub}} = 4,000$. Algorithm~\ref{algo:sparsification_test} is agnostic to the choice of final test. We
consider two options:
\begin{enumerate}
    \item {Mahalanobis test.} Let $\hat{{\eta}}_w$ and $\widehat{\Sigma}$ denote
    the sample mean and covariance of
    $\{\bar{\Psi}_{\rho^\dagger}(G^{*(i1)}_{b/2}, G^{*(i2)}_{b/2})\}_{i=1}^{N_{\mathrm{sub}}}$.
    Define the Mahalanobis distances
    \[
      D_i^* = \bigl(\bar{\Psi}_{\rho^\dagger}(G^{*(i1)}_{b/2}, G^{*(i2)}_{b/2})
               - \hat{\eta}_w\bigr)^\top
              \widehat{\Sigma}^{-1}
              \bigl(\bar{\Psi}_{\rho^\dagger}(G^{*(i1)}_{b/2}, G^{*(i2)}_{b/2})
               - \hat{\eta}_w\bigr),
    \]
    and $D_0$ analogously for $\bar{\Psi}_{\rho^\dagger}(G^{\prime 1}, G^{\prime 2})$.
    The p-value is $N_{\mathrm{sub}}^{-1}\sum_{i=1}^{N_{\mathrm{sub}}}\mathbf{1}(D_i^* > D_0)$.
    This test fully exploits the joint distribution of the $m$ motifs.
    \item {Cauchy combination test} \citep{liu2020cauchy}. This combines the $m$
    marginal p-values
    \[
      p_k = \frac{1}{N_{\mathrm{sub}}}\sum_{i=1}^{N_{\mathrm{sub}}}
            {1}\!(
              |\bar{\Psi}_{\rho^\dagger}(G^{*(i1)}_{b/2}, G^{*(i2)}_{b/2})_k
               - \bar{\eta}_{wk}|
              > |\bar{\Psi}_{\rho^\dagger}(G^{\prime 1}, G^{\prime 2})_k
               - \bar{\eta}_{wk}|
            ), \quad k = 1,\ldots,m.
    \]
    This test controls type~I error but does not exploit dependence among motifs.
\end{enumerate}

To evaluate our method, we consider the following simulation setup. Let $w_1$ and $w_2$ be two graphons.
The large network $G$ is generated from $\rho_n \cdot w_1$, whereas the small network $G'$ is generated from
$\rho_b \cdot \{(1-t)w_1 + t w_2\}$, with mixing parameter $t \in \{0, 0.5, 1\}$. The null hypothesis corresponds to $t=0$. 
We report rejection proportions based on $1,000$ replications at level $\alpha = 0.05$; under
the null, these estimate the type~I error rate, and under the alternatives $t>0$, they estimate power. Two settings are considered below.

\begin{table}[h]
\centering
\resizebox{\textwidth}{!}{%
\begin{tabular}{lcccccc}
 & \multicolumn{2}{c}{$t = 0$}
 & \multicolumn{2}{c}{$t = 0.5$}
 & \multicolumn{2}{c}{$t = 1$} \\ 
 & Mahalanobis & Cauchy
 & Mahalanobis & Cauchy
 & Mahalanobis & Cauchy \\ 
2-star, 3-star
& 0.042 & 0.049 & 0.223 & 0.196 & 0.876 & 0.820 \\ 
3-star, triangle
& 0.051 & 0.053 & 0.888 & 0.873 & 1     & 1     \\ 
2-star, triangle
& 0.051 & 0.057 & 0.885 & 0.876 & 1     & 1     \\ 
2-star, 3-star, triangle
& 0.046 & 0.049 & 0.855 & 0.850 & 1     & 1     \\ 
\end{tabular}
}
\caption{Rejection rates under the mixed model of Graphon~1 and Graphon~2.}
\label{tab:simu3gLgS}
\end{table}

{Setting 1: Synthetic graphons.} We take $w_1$ and $w_2$ to be Graphons~1 and ~2 from Section~\ref{subsec:pre}, with
$\rho_n = 0.25n^{-0.1}$ and $\rho_b = 0.25b^{-0.1}$. Results are reported in
Table~\ref{tab:simu3gLgS}. Both tests control type~I error at the nominal level, supporting the validity of Algorithm~\ref{algo:sparsification_test}.  Power increases with $t$, and the Mahalanobis
test is slightly more powerful than the Cauchy test. This is expected: since the two graphons differ in the marginal distributions of all three motifs, so little efficiency is lost by not fully exploiting their joint dependence.

Setting 2: {\it Data-driven} graphons. We next take $w_1$ and $w_2$ to be the graphons estimated from the two coexpression networks in
Section~\ref{sec:RealData} using the network mixing algorithm of \citet{li2023network}. The corresponding results are shown in Table~\ref{tab:simu3PCG}. At $t = 0.5$, the Mahalanobis test substantially outperforms the Cauchy test. A possible explanation is that the two graphons differ mainly through the \emph{joint} distribution of the 2-star and 3-star counts, rather than through their marginal distributions, as suggested in Section~\ref{sec:RealData}. Thus, a test that ignores
dependence information may suffer a substantial loss of power.

\begin{table}[h]
\centering
\resizebox{\textwidth}{!}{%
\begin{tabular}{lcccccc}
 & \multicolumn{2}{c}{$t = 0$}
 & \multicolumn{2}{c}{$t = 0.5$}
 & \multicolumn{2}{c}{$t = 1$} \\ 
 & Mahalanobis & Cauchy
 & Mahalanobis & Cauchy
 & Mahalanobis & Cauchy \\ 
2-star, 3-star
& 0.059 & 0.043 & 0.316 & 0.008 & 1 & 1 \\ 
3-star, triangle
& 0.052 & 0.045 & 0.165 & 0.076 & 1 & 1 \\ 
2-star, triangle
& 0.051 & 0.048 & 0.599 & 0.074 & 1 & 1 \\ 
2-star, 3-star, triangle
& 0.060 & 0.046 & 0.806 & 0.054 & 1 & 1 \\ 
\end{tabular}
}
\caption{Rejection rates under the mixed model of graphons learned from the coexpression networks in Section~\ref{sec:RealData}.}
\label{tab:simu3PCG}
\end{table}

Both experiments confirm that Algorithm~\ref{algo:sparsification_test}
controls type~I error at the nominal level under unequal network densities, supporting the proposed
density-matching approach for the two-sample problem. Across both settings, the
Mahalanobis test is at least as powerful as the Cauchy combination test, and
substantially more powerful when the signal lies in the joint distribution of motifs rather
than in their marginal distributions. This underscores a key advantage of our multivariate
subsampling framework: by testing the joint distribution of network moments, it can exploit
dependence among motifs that univariate or marginal  procedures inevitably fail to exploit.

%% file: 5realData.tex
\section{Comparison of Coexpression Networks of Core versus Non-Core Genes for Evolutionary Adaptation}
\label{sec:RealData}

\cite{fischer2021nonparallel} investigated the predictability of gene expression evolution during parallel adaptation across independent lineages of Trinidadian guppies
(\textit{Poecilia reticulata}), a model system in evolutionary biology in which multiple river drainages have independently colonized low-predation environments, giving rise to parallel phenotypic changes in life history, morphology, and behaviour. By comparing transcriptional mechanisms within and across lineages, they found that parallel phenotypic adaptation is associated with largely nonparallel gene expression changes: the vast majority of differentially expressed genes are lineage-specific, while a small number of genes are differentially expressed in the same direction across independent drainages.
These shared genes, referred to as \emph{core genes}, may represent a set of transcriptional targets repeatedly recruited during early-stage adaptation, distinct from the larger, lineage-specific \emph{non-core} gene set. Whether the coexpression networks of core and non-core genes share the same connection pattern, or whether core genes form a more tightly connected regulatory subnetwork reflecting their role as a shared adaptive hub, is an open biological question that our method is well-positioned to address.

We emphasize that the coexpression network considered here is an estimated summary of population-level gene-gene dependence, rather than a directly observed biological interaction graph. Such networks are widely used for exploratory and functional analyses, but their scientific interpretation depends on the validity of the preprocessing and network-construction steps, as well as on structural assumptions of the underlying network \citep{magwene2004estimating,langfelder2008wgcna,marbach2012wisdom,hill2016inferring,wang2021network}. Our inferential target is therefore a feature of the dependence structure represented by the estimated network, conditional on the chosen preprocessing and network-learning procedure. This interpretation is most meaningful when the measured units are sufficiently comparable to support the dependence structure and the network-learning procedure yields a stable approximation to the population coexpression pattern rather than sample-specific noise. We also note that formally characterizing the uncertainty from network estimation into the downstream inferential step would be valuable, but is beyond the scope of the present paper.

\begin{figure}[h]
    \centering
  \begin{subfigure}[b]{0.4\textwidth}
\centering
\includegraphics[width=\textwidth]{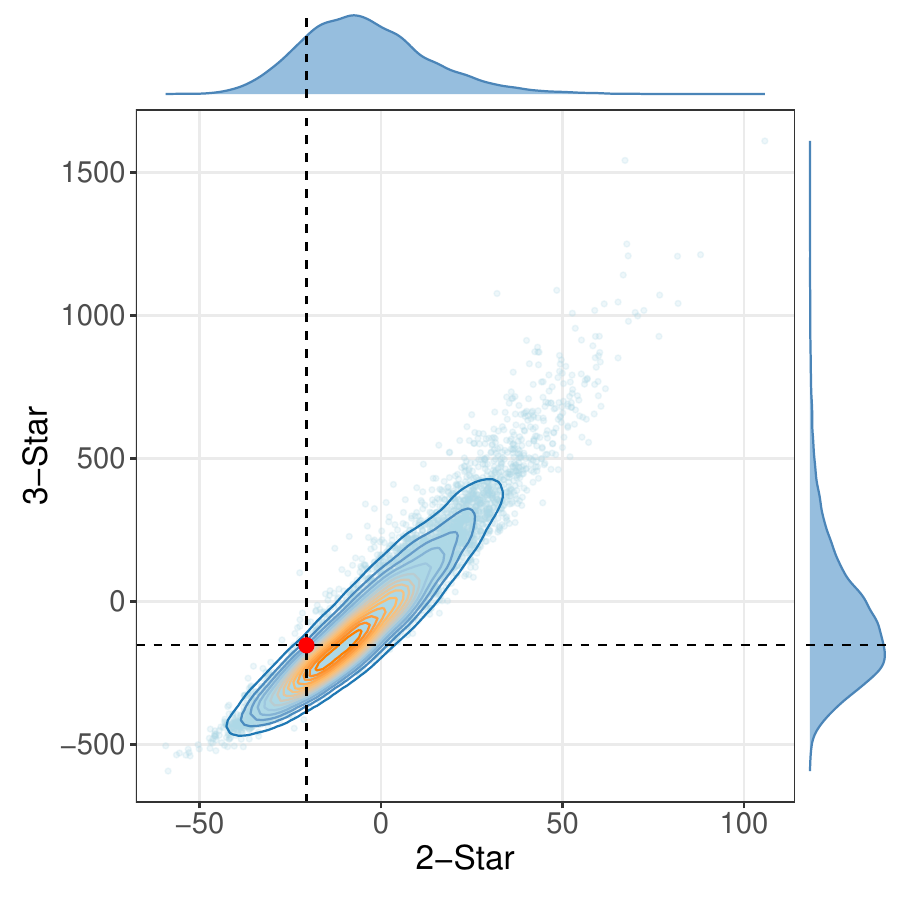}
\caption{Joint distribution of  $\xtwoline$ and $\xthreestar$. }
\label{fig:guppy-3s-2s}
\end{subfigure}
\hspace{-.15cm}
\begin{subfigure}[b]{0.4\textwidth}
\centering
\includegraphics[width=\textwidth]{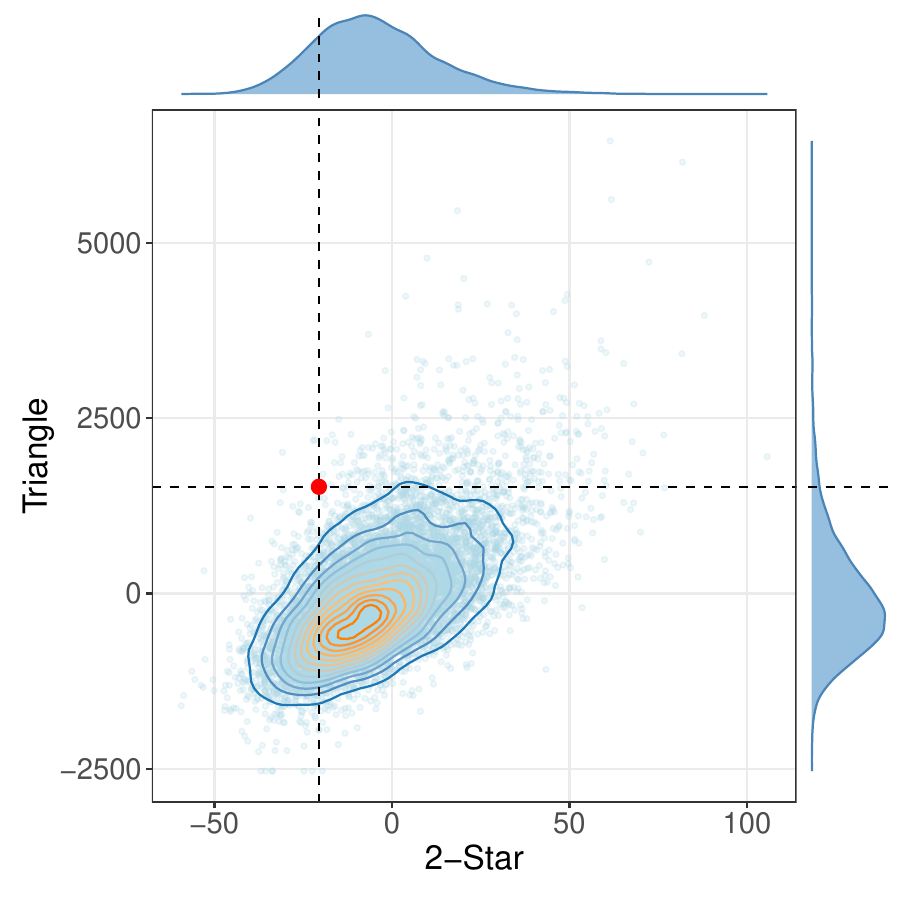}
\caption{Joint distribution of  $\xtwoline$ and $\xtriangle$. }
\label{fig:guppy-tri-2s}
\end{subfigure}
\vspace{0.1cm}
  \begin{subfigure}[b]{0.4\textwidth}
\centering
\includegraphics[width=\textwidth]{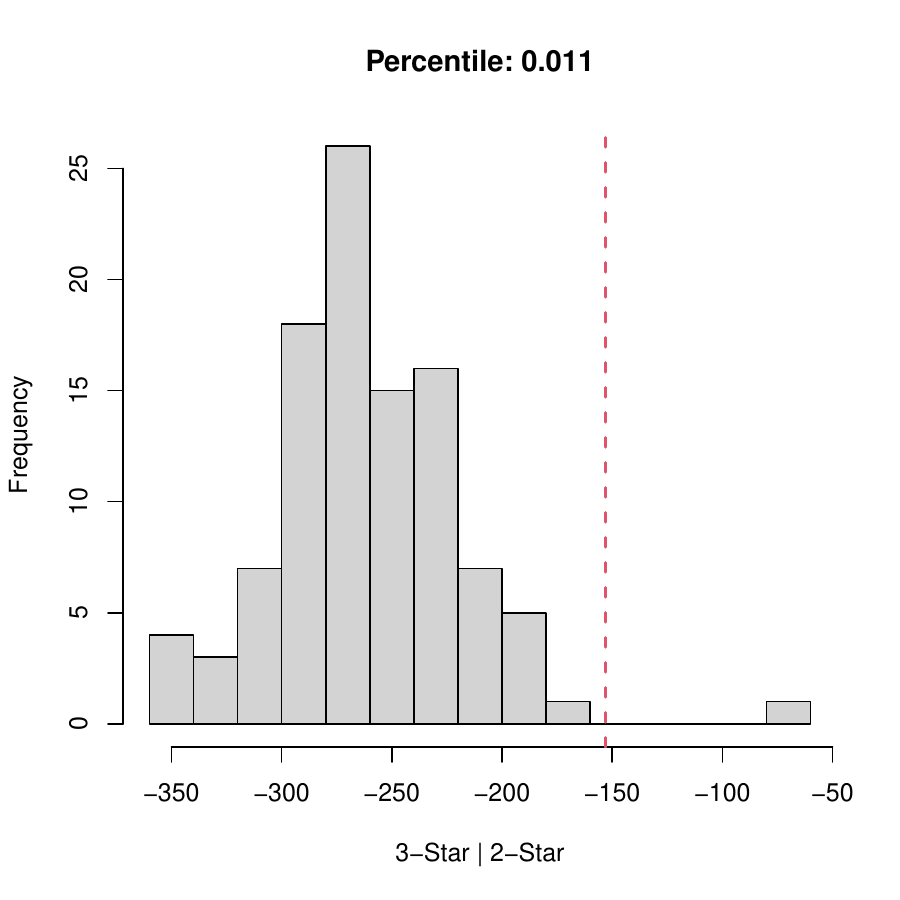}
\caption{Conditional distribution of  $\xthreestar | \xtwoline$. }
\label{fig:conditio-guppy-3s-2s}
\end{subfigure}
\hspace{-.15cm}
\begin{subfigure}[b]{0.4\textwidth}
\centering
\includegraphics[width=\textwidth]{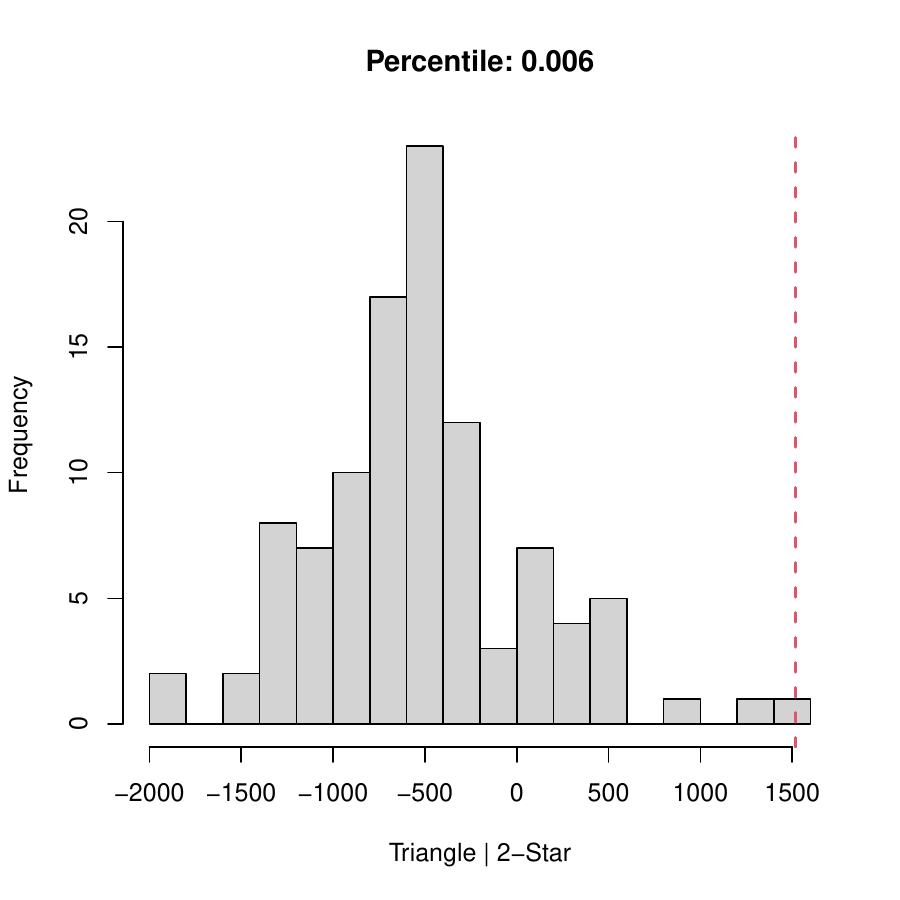}
\caption{Conditional distribution of  $\xtriangle | \xtwoline$. }
\label{fig:condition-guppy-tri-2s}
\end{subfigure}

    \caption{The blue points in panels (a) and (b) depict the subsampling distributions of network moments from the coexpression network of non-core genes. The red points represent the observed network moments in the core gene coexpression network. Panels (c) and (d) give the conditional subsampling distributions of 3-star and triangle, conditioning on the level of 2-stars in the core gene network. The red dotted line indicates the values of 3-star and triangle in the core gene network.}
    \label{fig:guppy}
\end{figure}

With this interpretation in mind, we apply our method to compare the coexpression networks of core and non-core genes. We construct the networks using the method of
\citet{cai2016large}, which provides a principled approach to recovering a sparse correlation-based dependence graph in high dimensions. Specifically, we form gene-wise
adjacency matrices by testing whether pairwise correlations are zero, controlling the false discovery rate at $0.05$, and the resulting binary adjacency matrices represent the coexpression networks \citep{magwene2004estimating}. The $16{,}485$ non-core genes form the larger network $G$, and the $618$ core genes form the smaller network $G'$. In addition to the large size imbalance, the two networks have substantially different densities: $0.0062$ for the non-core network and $0.0039$ for the core network. Since our primary interest is in comparing their connection patterns rather than their overall densities, we apply Algorithm~\ref{algo:sparsification_test} with $\rho^\dagger = 0.0035$ and $N_{\mathrm{sub}} = 10{,}000$, testing the 2-star ($\xtwoline$), 3-star ($\xthreestar$), and triangle ($\xtriangle$) motifs.

The Mahalanobis test and Cauchy combination test yield p-values of $0.020$ and $0.260$, respectively. The Cauchy test fails to detect any signal, as it cannot leverage the dependence among motifs. The Mahalanobis test, in contrast, provides clear evidence of a structural difference between the two networks. To understand this finding further, Figure~\ref{fig:guppy} displays the two-dimensional joint distributions of $(\xtwoline, \xthreestar)$ and $(\xtwoline, \xtriangle)$, as well as the corresponding
conditional distributions.

From a marginal perspective, the core network does not differ from the non-core network in either the 2-star or 3-star moments (Figure~\ref{fig:guppy-3s-2s}). This explains why the Cauchy combination test, which aggregates marginal evidence, finds no signal. However, since the 2-star is an induced subgraph of the 3-star, the two moments are inherently correlated, as confirmed by the shape of the subsampling cloud in Figure~\ref{fig:guppy-3s-2s}. Viewed jointly, the core network point lies on the boundary of the non-core subsampling distribution, suggesting a subtle but coherent departure in the joint structure of these two moments.

We further examine the conditional distribution of the 3-star moment given the 2-star level (Figure~\ref{fig:conditio-guppy-3s-2s}), restricting attention to subsampled networks whose 2-star count matches that of the core network. The approximate conditional p-value is $0.011$: conditioning on the 2-star level, the core gene network has a significantly elevated 3-star frequency relative to the non-core network. An analogous analysis for triangles, based on Figures~\ref{fig:guppy-tri-2s}
and~\ref{fig:condition-guppy-tri-2s}.

In conclusion, after accounting for density differences, the two networks exhibit comparable marginal 2-star levels. Yet the core gene network has a significantly higher
density of 3-stars and triangles conditional on 2-stars, indicating a more intensive and tightly connected interaction pattern among core genes. This is consistent with the biological hypothesis that core genes, as a repeatedly recruited set of adaptive targets, form a more cohesive regulatory subnetwork than the broader, lineage-specific non-core gene pool. This example also vividly illustrates the value of multivariate inference: the signal here is invisible to any marginal test, but is clearly revealed by the joint and conditional distributions of network moments.

%% file: 6Discussion.tex
\section{Discussion} \label{sec:discussion}

We have demonstrated that network node subsampling provides asymptotically valid inference for the joint distribution of multiple network moments. Building on this,
we proposed a subsampling-based two-sample testing procedure, based on network splitting and sparsification, to compare unmatchable networks with unequal densities; to our knowledge, this is the first inferential procedure applicable to this setting. As illustrated in the real data application, comparing the joint subsampling
distributions of network moments yields richer inference than the marginal testing studied in prior work.

Several directions could extend this work. A natural next step is to investigate whether higher-order accuracy of the joint subsampling distribution is achievable,
analogous to the Edgeworth corrections developed for the univariate case. From a computational perspective, evaluating network moments exactly is expensive
for large networks; developing scalable approximations and understanding their effect on downstream inference would substantially broaden the practical reach of
our framework.

%% file: Appendix/Appendix_add.tex
The appendix is organized as follows. Section~\ref{subsec:addresMC} collects additional properties of motif counts
used in subsequent proofs. Section~\ref{subsec:moreTheroemNMgraphon1} studies the statistical properties
of $J^{\{R_1,\ldots,R_m\}}_{b,c}$ as defined in Equation~\eqref{eq:true-jointcdfs} in the main paper.
While related results have been established by \citet{bickel2011method} and
\citet{maugis2020testing}, our main contribution here is the derivation of the
analytic form of the asymptotic covariance
\[
\lim_{n \to \infty} \operatorname{Cov}\!\left[
  \sqrt{n}\,\rho^{-\mathfrak{r}}_{n} U_{R}(\mathbb{G}_{n}),\;
  \sqrt{n}\,\rho^{-\mathfrak{r}'}_{n} U_{R'}(\mathbb{G}_{n})
\right].
\]
These results characterize the variance structure of the joint moment distribution
and establish conditions under which $J_{b,c}$ is non-degenerate.

Section~\ref{subsec:addresQuantify1} investigates statistical properties of
$J^{\{R_1,\ldots,R_m\}}_{*,n,b}$ as defined in Equation~\eqref{eq:jointcdfs}.
Specifically, we analyze quantities such as
${\mE}_*\!\left[U_{R}(\mathbb{G}^{*}_{b})\right]$ and
$\operatorname{Cov}_*\!\left[U_{R}(\mathbb{G}^{*}_{b}),\, U_{R'}(\mathbb{G}^{*}_{b})\right]$,
as well as their scaled limits:
\[
{\mE}_*\!\left[\rho^{-\mathfrak{r}}_{b}\, U_{R}(\mathbb{G}_{b})\right]
\quad\text{and}\quad
\lim_{b \to \infty} \rho^{-(\mathfrak{r}+\mathfrak{r}')}_{n}
\operatorname{Cov}_*\!\left[\sqrt{b}\, U_{R}(\mathbb{G}^{*}_{b}),\,
\sqrt{b}\, U_{R'}(\mathbb{G}^{*}_{b})\right].
\]
This analysis enables us to relax non-degeneracy assumptions commonly imposed
in earlier work \citep{zhang2022edgeworth,lunde2023subsampling}.

Section~\ref{subsec:asym_dis_UR} derives the asymptotic distribution of
$J^{\{R_1,\ldots,R_m\}}_{*,n,b}$. Following the approach of \citet{zhang2022edgeworth},
we adopt the finite-population U-statistic framework of \citet{bloznelis2001orthogonal}.
By modeling $\sqrt{b_n}\,\rho^{-\mathfrak{r}}_{n}U_{R}(\mathbb{G}^{*}_{b_n})$ as a
finite-population U-statistic and verifying a smoothness condition, a non-lattice
condition, and a Lindeberg--Feller-type condition, we establish the multivariate
asymptotic distribution of $J_{*,n,b}$ for multiple motifs $R_1, \ldots, R_m$.
To our knowledge, this is the first result to establish such asymptotic properties
in the multiple-motif setting.

Building on these results and the asymptotic theory for
$J^{\{R_1,\ldots,R_m\}}_{b,c}$ developed in \citet{bickel2011method}, we analyze
the Kolmogorov--Smirnov distance and prove Theorem~\ref{theo:consistent}. For
the reader's convenience, the overall proof workflow is illustrated in
Figure~\ref{fig:proof}.

\begin{center}
    \begin{tikzpicture}[line width = 1pt, x = .5cm, y = .48cm]
        \tikzset{main/.style = {draw,align = center,minimum width = 8em},every node/.style = {scale = 0.8}}
            \node[main] (a) at (0,0) {Section~\ref{subsec:consis}:$ \sup\limits_{t \in \mathbb{R}}\big| J^R_{*,n,b}(t) - J^R_{b,c}(t)
       \big| \rightarrow 0.$};

            \node[main] (b) at (-11,-4) {Section~\ref{subsec:asym_dis_UR}: asymptotic distribution of $J_{*,n,b}$};

            \node[main,align = center] (c) at (0,-4) {\citet{bickel2011method} studied the \\asymptotic distribution of $J_{b,c}$ . };

            \node[main,align = left] (d) at (11,-4) { $E_*\stackrel{P}{\longrightarrow}E$, $Cov_*\stackrel{P}{\longrightarrow}Cov$};
    
            \node[main,align = left, anchor = north] (e) at (-11,-8) {Section~\ref{subsec:asym_dis_UR}: model $\sqrt{b_n}\rho^{-\mathfrak{r}}_{n}U_{R}(\mathbb{G}^{*}_{b_n})$\\ as a symmetric finite population statistic. \\Verify smoothness condition, non-lattice,\\
            and a Lindeberg-Feller typed condition.};

            \node[main,align = left, anchor = north] (f) at (11,-8) {Section~\ref{subsec:moreTheroemNMgraphon1}: statistical properties of  $J_{b,c}$. \\ \( {\mE}[U_{R}(\mathbb{G}_{n})] \) and \( \operatorname{Cov}[U_{R}(\mathbb{G}_{n}), U_{R'}(\mathbb{G}_{n})] \) \\as well as their limits.\\
            Section~\ref{subsec:addresQuantify1}:statistical properties of $J_{*,n,b}$.\\ \( {\mE}_*\left[U_{R}(\mathbb{G}^{*}_{b}) \right] \) and \( \operatorname{Cov}_*\left[U_{R}(\mathbb{G}^{*}_{b}), U_{R'}(\mathbb{G}^{*}_{b}) \right] \)\\as well as their limits.};

    
            \node[main,align = left] (h) at (0,4) {Section~\ref{subsec:consis}:$\sup\limits_{(t_1,\ldots,t_m)\in \mathbb{R}^m}\Big| J^{\{R_1,\ldots,R_m\}}_{*,n,b}(t_1,\ldots,t_m) -  J^{\{R_1,\ldots,R_m\}}_{b,(1 - b/n)}(t_1,\ldots,t_m) \Big|   \to 0$ (Theorem~\ref{theo:consistent})\\
            Section~\ref{subsec:consis1}: empirical version of Theorem~\ref{theo:consistent} }; 
    
            \draw[-{Triangle},line width = 1pt] (b) -- (a);
            \draw[-{Triangle},line width = 1pt] (c) -- (a);
            \draw[-{Triangle},line width = 1pt] (d) -- (a);
            \draw[-{Triangle},line width = 1pt] (e.north) -- (b.south)node[midway,right] {Results in \cite{bloznelis2001orthogonal}.};
            \draw[-{Triangle},line width = 1pt] (f) -- (d);
            \draw[-{Triangle},line width = 1pt] (a) -- (h)node[midway,right,align = left]{By employing the Cramér-Wold device.};
            \end{tikzpicture}    
    \captionof{figure}{\label{fig:proof}Proof sketch diagram for Theorem~\ref{theo:consistent}.}
\end{center}

The technical results in Sections~\ref{subsec:moreTheroemNMgraphon1},
\ref{subsec:addresQuantify1}, and~\ref{subsec:asym_dis_UR} are proved in
Sections~\ref{sec:proof for nmg1}, \ref{sec:proof for q1},
and~\ref{sec:proof for asymdis}, respectively. The proof of
Theorem~\ref{them:sparsification_validity} is given in
Section~\ref{sec:proof-sparsification}. Section~\ref{sec:necessary-sparsification}
provides empirical experiments demonstrating that sparsification is necessary
to ensure inference validity when comparing two networks with unequal densities.
The main numerical results of the paper are presented in
Sections~\ref{subsec:consis} and~\ref{subsec:consis1}, with additional simulation
results collected in Section~\ref{subsec:addsimu}.

\section{Supporting propositions, lemmas, and additional theoretical results}\label{sec:additionTheorem}


For the ease of notation, we define
\begin{equation}
\label{eq:sgraphon}
    h_n(u,v) = \rho_n w(u,v)\mathbb{1}_{\{\rho_n w(u,v) \leq 1\}}.
\end{equation}
If a network $\bGn \sim h_n$, we denote it by $\mathbb{G}_{n}^{h_{n}}$ for simplicity. The network moment $U_{R}(\mathbb{G}^{*}_{b})$ is a function of $\mathbb{G}^{*}_{b}$, and $\mathbb{G}^{*}_{b}$ can be viewed as a conditional random variable. 



Note that $U_R(\mathbb{G}^*_b)$ is a finite population U-statistic \citep{zhang2022edgeworth} and network $G$ can be treated as a finite population: $G= \left\{\mathfrak{v}_1, \cdots, \mathfrak{v}_n\right\}$,  where each unit $ \mathfrak{v}_i$ represent the adjacency information between the $i$th node and others. The finite population U-statistic has been studied in \cite{zhao1990normal,bloznelis2001orthogonal,bloznelis2002edgeworth}, which is defined as follows.

\begin{definition}[Finite population U-statistic]
\label{def:finiteU}
Let  $\mathcal{V} = (\mathfrak{v}_1,\cdots,\mathfrak{v}_n)$ be a finite population consisting of $n$ units. 
    Let $T = t(\mathbb{V}_1,\cdots,\mathbb{V}_b)$ denote a statistic based on simple random sample $\mathbb{V}_1,\cdots,\mathbb{V}_b$ drawn without replacement from $\mathcal{V}$. If the kernel function $t$ is invariant under permutations of its arguments, then $T$ is called a finite population U-statistic.
\end{definition}

\subsection{Properties of motif counts} \label{subsec:addresMC}

In this section, we introduce two useful features of motif counts. The first one is the relationship between motif counts and graph injective homomorphisms:
\begin{lemma}[Proposition 1 of \cite{amini2012counting}]
\label{countinjhom}
For any motif $R$ and graph $G$,
\begin{equation*}
    X_R(G) = \mathrm{inj}(R,G)/|\mathrm{Aut}(R)|,
\end{equation*}
\end{lemma}
where  $\mathrm{inj}(R,G)$ denotes the number of injective graph homomorphisms \citep{lovasz2006limits}, and $\mathrm{Aut}(R)$ denotes the set of all automorphisms of $R$. A mapping $\phi$: $V(R) \rightarrow V(G)$  is a graph homomorphism if $(v_i,v_j)\in \eE(R)$ implies $[\phi(v_i), \phi(v_j)] \in \eE(G)$, and it is an injective graph homomorphism if $\phi(v_i) = \phi(v_j)$ implies $v_i = v_j$. On the other hand, $\mathrm{Aut}(R)$ is the set of all permutations $\psi$ of the node set $V(R)$ such that $(x,y) \in \eE(R)$ if and only if $[\psi(x),\psi(y)] \in \eE(R)$.
More discussions on $\mathrm{Aut}(R)$  are provided in \cite{rodriguez2014automorphism}. 

Let $\mathcal{S}_{R, R'}$  denote the set of all unlabeled graphs that can be formed from $R$ and $R'$. That is, 
\begin{equation}
    \label{eq:defs}
    \mathcal{S}_{R, R'} = \Big\{S \subset K_{r + r'} :V(S) = V(R_1) \cup V(R_2), \eE(S) = \eE(R_1) \cup \eE(R_2), R_1 \cong R, R_2 \cong R'\Big\},
\end{equation}
where $K_n$ denotes the complete graph of size $n$. Furthermore, $\mathcal{S}_{R, R'}$ can be partitioned into disjoint sets $\mathcal{S}^{(q)}_{R, R'}$ based on the number of merged nodes $q$, where $
    \mathcal{S}^{(q)}_{R, R'} = \big\{S :S\subset  \mathcal{S}_{R, R'}, |V(S)| = r + r' - q \big\}$. Lastly, for each  $S \subset \mathcal{S}_{R, R'}$, we define a constant $c_S$ as 
\begin{equation}
    \label{eq:defcs}
    c_S = \Big|\left\{ (R_1,R_2) \subset S :V(S) = V(R_1) \cup V(R_2), \eE(S) = \eE(R_1) \cup \eE(R_2), R_1 \cong R, R_2 \cong R'\right\}\Big|.
\end{equation} 
Following \cite{maugis2020testing}, we use two examples to explain above definitions. In the first example, let $R$ be a $\xtriangle$ and $R'$ also be a $\xtriangle$. Then the set $\mathcal{S}_{R, R'}$ can be constructed as $\left\{\xtriangle \,\xtriangle, \xreltriangle[1.05], \xdiagrectangle[.85], \xtriangle\right\}$. Each element in $\mathcal{S}_{R, R'}$  can be obtained by building blocks based on $R$ and $R'$. Let $R_1$ be a copy of $R$, and $R_2$ be a copy to $R'$. The pattern $\xtriangle \,\xtriangle$ can be built by either put $R_1$ in the left side or in the right side. Thus, $c_{\xtriangle[.5]\xtriangle[.5]} = 2$. Similarly,  $c_{\,\xreltriangle[.75]} = 2$, $c_{\,\xdiagrectangle[.75]} = 2$ and $c_{\xtriangle[.5]} = 1$. Generally speaking, $c_S$ denotes the number of ways $S$ can be built from copies of  $R$ and $R'$. Based on the number of merged nodes, we have $S_{R,R'}^{(0)} = \left\{\xtriangle \,\xtriangle\right\}$,$S_{R,R'}^{(1)} = \left\{\xreltriangle[.95]\right\}$, $S_{R,R'}^{(2)} = \left\{\xdiagrectangle\right\}$, and $S_{R,R'}^{(3)} = \left\{\xtriangle\right\}$. For the second example, let $R$ be a $\xrectangle$ and $R'$ be a $\xrectangle$. Then  $S_{R,R'} = \{S_{R,R'}^{(0)},S_{R,R'}^{(1)},S_{R,R'}^{(2)},S_{R,R'}^{(3)},S_{R,R'}^{(4)}\}$, with $S_{R,R'}^{(0)} = \left\{\xrectangle\,\xrectangle\right\}$, $S_{R,R'}^{(1)} = \left\{\xreldiamond\right\}$, $S_{R,R'}^{(2)} = \left\{\xtwodiamond,\, \xrecanddia,\, \xtworectangle\right\}$, $S_{R,R'}^{(3)} = \left\{\xthreetrangle[1],\, \xdiagdiamond[.75],\, \xtwodiagrectangle[.9]\right\}$, $S_{R,R'}^{(4)} = \left\{\xrectangle\right\}$. Correspondingly, $c_{\xrectangle[.5]\xrectangle[.5]} = 2$, $c_{\xreldiamond[.5]} = 2$, $c_{\xtwodiamond[.5]} = 6$, $c_{\xrecanddia[.5]} = 2$, $c_{\xtworectangle[.5]} = 2$, $c_{\xthreetrangle[.75]} = 2$, $c_{\xdiagdiamond[.5]} = 6$, $c_{\xtwodiagrectangle[.5]} = 6$ and $c_{\xrectangle[.5]} = 1$.

We are in position to introduce the second feature regarding the linearity of motif counts. 
\begin{lemma}[Lemma 1 in \cite{maugis2020testing}]
 \label{lem: linerityMC}
For any two motifs $R$ and $R'$. \begin{equation}
    \label{eq:linerityMC}
        X_R(G)X_{R'}(G) = \sum_{S\in \mathcal{S}_{R, R'}}c_{S}X_S(G) = \sum_{q=0}^{\min\{r,r'\}}\sum_{S\in \mathcal{S}_{R, R'}^{(q)}}c_{S}X_S(G).
    \end{equation}
\end{lemma}
As noted in \cite{maugis2020testing}, $X_R(G)X_{R'}(G)$ involves counting pairs of motifs, and could be recovered by counting the number of the all motifs that are formed by using one copy of $R$ and one copy of $R'$ as building blocks. This is the intuition of Lemma \ref{lem: linerityMC}. Moreover, \eqref{eq:linerityMC} provides flexibility as it does not depend on the generation mechanism of $G$. 


\subsection{Statistical properties of network moments of graphs under the sparse graphon model } \label{subsec:moreTheroemNMgraphon1}
Following \cite{bickel2011method}, we define the following quantities:
\begin{equation}
    \label{eq:PR}
    \begin{aligned}
          P_{h_n}(R)  &= \int_{[0,1]^r} \prod_{(v_i,v_j) \in \eE(R)} h_n\left(\xi_i, \xi_j\right) \prod_{v_i \in V(R)} d \xi_i, \\
  P_{w}(R) &= \int_{[0,1]^r} \prod_{(v_i,v_j) \in \eE(R)} w\left(\xi_i, \xi_j\right) \prod_{v_i \in V(R)} d \xi_i.
    \end{aligned}
\end{equation} 
Lemma \ref{lem: Exp_Var_graphon} below documents some fundamental properties of network moment $U_{R}(\mathbb{G}_{n})$.
\begin{lemma}
\label{lem: Exp_Var_graphon}
For any motif $R$,
\begin{equation}
\label{eq:exp_graphon}
\mE\big[U_{R}(\mathbb{G}_{n})\big] = \frac{r!}{|\mathrm{Aut}(R)|} P_{h_n}(R).
\end{equation}
Moreover, for a pair of motifs $R$ and $R'$, using the definitions in \eqref{eq:defs} and \eqref{eq:defcs},
\begin{equation}
\label{eq:cov_graphon}
\begin{aligned}
\Cov\big[U_{R}(\mathbb{G}_{n}),U_{R'}(\mathbb{G}_{n})\big]  
=&~ \binom{n}{r}^{-1}\binom{n}{r'}^{-1}\sum_{q=1}^{\min\{r,r'\}}\sum_{S\in \mathcal{S}_{R, R'}^{(q)}}c_S\mE\big[X_{S}(\mathbb{G}_{n})\big]\quad \\&-  \Big[\dfrac{\binom{n}{r'}}{\binom{n-r}{r'}}-1 \Big]\binom{n}{r}^{-1}\binom{n}{r'}^{-1} \sum_{S\in \mathcal{S}_{R, R'}^{(0)}}c_S \mE\big[X_{S}(\mathbb{G}_{n})\big]
\end{aligned}
\end{equation}
in which recall that $r$ and $r'$ are the number of nodes in $R$ and $R'$, respectively.
\end{lemma}

We focus on the statistical properties of $\rho^{-\mathfrak{r}}_{n}U_{R}(\mathbb{G}_{n})$ rather than $U_{R}(\mathbb{G}_{n})$ because both the expectation and variance of $U_{R}(\mathbb{G}_{n})$ shrink to zero when $\rho_n$ converges to zero under the sparse graphon model. 
\begin{proposition}
   \label{prop: lim_var_graphon}
   For any motif $R$, 
\begin{equation}
\label{eq:lim_exp_graphon}
\begin{aligned}
      \mE\big[\rho^{-\mathfrak{r}}_{n}U_{R}(\mathbb{G}_{n})\big] = \eta_w(R) = \frac{r!}{|\mathrm{Aut}(R)|}P_{w}(R).
\end{aligned}
\end{equation}
Furthermore, consider motifs $R$ and $R'$ with sizes $r \leqslant r'$. Assume that $n\rho^{r/2}_n \rightarrow \infty$, then
     \begin{equation}
        \label{eq: lim_var_graphon}
       \begin{aligned}
      &\lim\limits_{n \to \infty}\Cov\big[\sqrt{n}\rho^{-\mathfrak{r}}_{n}U_{R}(\mathbb{G}_{n}), \sqrt{n}\rho^{-\mathfrak{r}'}_{n}U_{R'}(\mathbb{G}_{n})\big] = \sum_{S\in\mathcal{S}_{R,R'}^{(1)}}\frac{c_Sr!r'!}{|\mathrm{Aut}(S)|}P_{w}(S) - \sum_{S\in\mathcal{S}_{R,R'}^{(0)}} \frac{c_Sr!r'!rr'}{|\mathrm{Aut}(S)|}P_{w}(S).
\end{aligned}
     \end{equation}
\end{proposition}
The right-hand side in \eqref{eq: lim_var_graphon} describes the limit of covariance for any pair of motifs, and it includes the variance of $\sqrt{n}\rho^{-\mathfrak{r}}_{n}U_{R}(\mathbb{G}_{n})$ as a special case. When the limit of variance is non-zero, we say $\rho^{-\mathfrak{r}}_{n}U_{R}(\mathbb{G}_{n})$ is \emph{non-degenerate}.

\subsection{Statistical properties of network moments of subsampled graphs}\label{subsec:addresQuantify1}

Let $\mathcal{S}(\mathbb{G}^*_b)$ denote the collection of all possible  instantiations of $\mathbb{G}^*_b$. For a fixed node $v \in V(G)$, we use $\mathbb{G}^{v*}_{b}$ to denote a randomly induced subgraph of $G$ based on the fixed node $v$ and other $b-1$ nodes randomly drawn without replacement from $V(G)\setminus v$. Similarly, we use $\mathcal{S}(\mathbb{G}^{v*}_b)$ to denote the sample space of $\mathbb{G}^{v*}_b$. Let $G^{v*}_{b} \in  \mathcal{S}(\mathbb{G}^{v*}_b)$ be one instantiation. We use $\mathbb{G}^{v**}_{b,r}$ to denote a randomly induced subgraph of $G^{v*}_{b}$ based on node $v$ and other $r-1$ nodes randomly drawn without replacement from $V(\mathbb{G}^{*}_{b}) \setminus v$, and use $\mathcal{S}(\mathbb{G}^{v**}_{b,r})$ to denote the set contains all possible  $\mathbb{G}^{v**}_{b,r}$. The following lemma provides a few useful identities to be used in later proofs.

\begin{lemma}
\label{lem: usefuleqs}
For any network $\mathcal{G}$ and motif $R$, the following identities hold:
\begin{equation}
\label{eq:iden1}
\sum_{\mathcal{G} \in \mathcal{S}(\mathbb{G}^*_b)} X_R(\mathcal{G}) = \binom{n-r}{b-r} X_R(G),
\end{equation}
\begin{equation}
\label{eq:iden2}
\Big|\{S: S \subset G^{v*}_{b}, v \in V(S), S \cong R\}\Big| = \sum_{\mathcal{G} \in \mathcal{S}(\mathbb{G}^{v**}_{b,r})} X_R(\mathcal{G}),
\end{equation}
\begin{equation}
\label{eq:iden3}
\Big|\{S: S \subset G, v \in V(S), S \cong R\}\Big| = \sum_{\mathcal{G} \in \mathcal{S}(\mathbb{G}^{v*}_{r})} X_R(\mathcal{G}),
\end{equation}
\begin{equation}
\label{eq:iden4}
\sum_{G_b^* \in \mathcal{S}(\mathbb{G}_{b}^{v*})}\Big[\sum_{\mathcal{G} \in \mathcal{S}(\mathbb{G}^{v**}_{b,r})} X_R(\mathcal{G})\Big] = \binom{n-r}{b-r}\Big|\{S: S \subset G, v \in V(S), S \cong R\}\Big|,
\end{equation}
and 
\begin{equation}
\label{eq:iden5}
\sum_{i=1}^{n}\sum_{\mathcal{G} \in \mathcal{S}(\mathbb{G}^{v_i*}_{r})} X_R(\mathcal{G}) = \sum_{i=1}^{n}\Big|\{S: S \subset G, v_i \in V(S), S \cong R\}\Big| = rX_R(G).
\end{equation}

\end{lemma}

We now introduce the following extension of the results in \cite{bhattacharyya2015subsampling}. 
\begin{lemma}
        \label{lem: exp_cov_subsample}
       Given the network $G$, for any motif $R$,
        \begin{equation}        
\label{eq:exp_subsample}
     \mE_*\left[U_{R}(\mathbb{G}^{*}_{b}) \right] = U_{R}(G).
\end{equation}

And for any two motifs $R$ and $R'$ with $r + r' < b$, 
    \begin{equation}
\label{eq:cov_subsample}
\begin{aligned}
       \Cov_{*}\left[U_{R}(\mathbb{G}^{*}_{b}),U_{R'}(\mathbb{G}^{*}_{b}) \right]   =& \binom{b}{r}^{-1}\binom{b}{r'}^{-1} \sum_{q=0}^{\min\{r,r'\}}\sum_{S\in \mathcal{S}_{R, R'}^{(q)}}C_{S}\binom{b}{s}\binom{n}{s}^{-1}X_{S}(G)- U_{R}(G)U_{R'}(G),
\end{aligned}
\end{equation}
where $s = |V(S)| = r + r' - q$. Moreover, suppose that $G \sim \mathbb{G}_n$. Then   \begin{equation}
\label{eq:exp_subsample_conditional}
   E[U_R(\mathbb{G}^{h_n}_b)] = E \{\mE_*[U_{R}(\mathbb{G}^{(*\mathbb{G}_n = G)}_{b}) ]\} = E \{\mE_*[U_{R}(\mathbb{G}^{*}_{b}) ]\},
\end{equation}
\begin{equation}
 \label{eq:cov_subsample_conditional}
\begin{aligned}
   &\Cov[U_R(\mathbb{G}^{h_n}_b), U_{R'}(\mathbb{G}^{h_n}_b)] = \Cov\{E_*[U_R(\mathbb{G}^*_b)], E_*[U_{R'}(\mathbb{G}^*_b)]\}  +   E\{\Cov_{*}[U_{R}(\mathbb{G}^{*}_{b}),U_{R'}(\mathbb{G}^{*}_{b})]\}.
\end{aligned}
\end{equation}
    \end{lemma}

The next result is about the impact one the variance and covariance scale due to the subsampling.
\begin{lemma}
\label{coro:lim_var_subsample}
Let $R$ and $R'$ be two motifs with $\max\{r,r'\} \leq r_1$ and  $ \max\{\mathfrak{r},\mathfrak{r}'\} \leq \mathfrak{r}_1$. Suppose that Assumption \ref{ass:rho_n_h_n} holds after replacing $r$ by $r_1$ and $\mathfrak{r}$ by $\mathfrak{r}_1$, and Assumption \ref{ass:b} holds. Then
\begin{equation*}
\begin{aligned}
    &\lim\limits_{ b\to \infty}\rho^{-(\mathfrak{r}+\mathfrak{r}')}_{n}\Cov_*\big[\sqrt{b}U_{R}(\mathbb{G}^{*}_{b}),\sqrt{b}U_{R'}(\mathbb{G}^{*}_{b})  \big]  = \big(1 - c_2\big)\lim\limits_{b \to \infty}\rho^{-(\mathfrak{r}+\mathfrak{r}')}_{b}\Cov\big[\sqrt{b}U_{R}(\mathbb{G}_{b}),\sqrt{b}U_{R'}(\mathbb{G}_{b})\big]
\end{aligned}
\end{equation*}
with probability one.
\end{lemma}

The following proposition extends the results on finite population statistics from \cite{bloznelis2001orthogonal, bloznelis2002edgeworth} to the context of network subsampling.

\begin{proposition}
\label{prop:stats_prop_finite_U_network}
\hfill

\begin{enumerate}[label=(\alph*)]
\item \label{eq:p2a} The Hoeffding's decomposition of $ U_{R}(\mathbb{G}^{*}_{b})$ is 
\begin{equation}
\label{eq:hoeffding_UR}
U_{R}(\mathbb{G}^{*}_{b}) =E_* [U_{R}(\mathbb{G}^{*}_{b})]+\sum_{1 \leqslant i \leqslant b} g_{1,R}\left(\mathbb{V}_i\right)+\sum_{1 \leqslant  i<j \leqslant b} g_{2,R}\left(\mathbb{V}_i, \mathbb{V}_j\right)+\cdots,
\end{equation}
where 
\begin{equation}
\label{eq:g1}
\begin{aligned}
      g_{1,R}(\mathbb{V}_1) =&\frac{r!(n-r-1)!}{b(n-2)!}\sum_{\mathcal{G} \in \mathcal{S}(\mathbb{G}^{\mathbb{V}_1*}_{r})}  X_{R}(\mathcal{G}) - \frac{r(n-1)}{b(n-r)}U_{R}(G) =  \frac{(n-1)}{b}[U_R(G) - U_R(G\setminus \mathbb{V}_1)],
\end{aligned}
\end{equation}
with 
\begin{equation}
\label{eq:varg1}
\begin{aligned}
    & \Cov_{\mathbb{V}_1*}[ g_{1,R}(\mathbb{V}_1) , g_{1,R'}(\mathbb{V}_1)]
   = \frac{r!(n-r-1)!}{b(n-2)!} \frac{r'!(n-r'-1)!}{b(n-2)!}\sum_{k=0}^{\min\{r,r'\}} \sum_{S\in \mathcal{S}_{R, R'}^{(q)}} c_S\dfrac{nq-rr'}{n^2}X_S(G).
\end{aligned}
\end{equation}
Furthermore, we have
\begin{equation}
\label{eq:var_sumg1}
      \Cov_*\left[\sum_{1\leqslant i \leqslant b} g_{1,R}(\mathbb{V}_i),
      \sum_{1\leqslant i \leqslant b} g_{1,R'}(\mathbb{V}_i)\right] = \frac{b(n-b)}{(n-1)}\Cov_*[g_{1,R}(\mathbb{V}_1),
    g_{1,R'}(\mathbb{V}_1)],
\end{equation}
and as $n,b \rightarrow \infty$
\begin{equation}
\label{eq:lim_var_g1}
     \lim\limits_{b,n\to\infty}  \var_*[\sum_{1 \leqslant i \leqslant b}g_{1,R}\left(\mathbb{V}_1\right)] = 0.
\end{equation}
\item  For two motifs  $R$ and $R'$, $U_{R}(\mathbb{G}^{*}_{b}) + U_{R'}(\mathbb{G}^{*}_{b})$ is also a  symmetric finite population statistic with the following Hoeffding's decomposition
\begin{equation*}
\begin{aligned}
  U_{R}(\mathbb{G}^{*}_{b}) + U_{R'}(\mathbb{G}^{*}_{b}) =&E_* [U_{R}(\mathbb{G}^{*}_{b}) + U_{R'}(\mathbb{G}^{*}_{b}) ]+\sum_{1 \leqslant i \leqslant b} g_{1,R,R'}\left(\mathbb{V}_i\right)+\sum_{1 \leqslant  i<j \leqslant b} g_{2,R,R'}\left(\mathbb{V}_i, \mathbb{V}_j\right)+\cdots,  
\end{aligned}
\end{equation*}
where 
\begin{equation}
\label{eq:sg1}
\begin{aligned}
      g_{1,R,R'}(\mathbb{V}_1) &=  g_{1,R}(\mathbb{V}_1) +  g_{1,R'}(\mathbb{V}_1).
\end{aligned}
\end{equation}

Moreover, the variance of linear parts satisfies:
\begin{equation}
\label{eq:svar_sum_g1}
 \var_*\sum_{1 \leqslant i \leqslant b} g_{1,R,R'}\left(\mathbb{V}_i\right) = \frac{b(n-b)}{(n-1)} \var_*\left[g_{1,R,R'}\left(\mathbb{V}_1\right)\right],
\end{equation}
\begin{equation}
\label{eq:lim_svar_sumg1}
     \lim\limits_{b,n\to\infty}  \var_*\left[\sum_{1 \leqslant i \leqslant b} g_{1,R,R'}\left(\mathbb{V}_i\right)\right] = 0.
\end{equation}
\end{enumerate}
\end{proposition}

\subsection{Asymptotic distribution of network moments of subsampled graphs} \label{subsec:asym_dis_UR}

Using the tools in \cite{bloznelis2001orthogonal}, we derive the following results for the subsampled moments. 
\begin{theorem}
    \label{theo:subsample_distribution}  Suppose that $\{G^{(n)}\}_{n=1}^{\infty}$ is a sequence of networks, where $G^{(n)} \sim \mathbb{G}_{n}$. 
    \hfill

    \begin{enumerate}[label=(\alph*)] 
        \item \label{theo:subsample_distribution_a}
       The Hoeffding's decomposition of $\sqrt{b_n}\rho^{-\mathfrak{r}}_{n}U_{R}(\mathbb{G}^{*}_{b_n})$ is
 \begin{equation}
\label{eq:hoeffding_U_R}
\begin{aligned}
\sqrt{b_n}\rho_{n}^{-\mathfrak{r}}  U_{R}(\mathbb{G}^{*}_{b_n}) = & \sqrt{b_n}\rho_{n}^{-\mathfrak{r}}  U_{R}[G^{(n)}]+\sum_{1 \leqslant i \leqslant b_n} \sqrt{b_n}\rho_{n}^{-\mathfrak{r}} g_{1,R}\left(\mathbb{V}_i\right)+ \Delta[\sqrt{b_n} \rho_{n}^{-\mathfrak{r}}U_{R}(\mathbb{G}^{*}_{b_n})].
\end{aligned}
\end{equation}

For any network sequence, the following conditions hold with probability one.
       \begin{enumerate}[label=(\roman*)] 
            \item Under Assumptions \ref{ass:rho_n_h_n}, \ref{ass:b} and \ref{ass:non_degenerate}, we have
            \begin{equation}
\label{eq:assump1}
     \lim\limits_{n \to \infty} E_*\Delta^2[\sqrt{b_n} \rho_{n}^{-\mathfrak{r}}U_{R}(\mathbb{G}^{*}_{b_n})] = 0, 
\end{equation}
\begin{equation}
\label{eq:assump2}
       0 < c_3 \leqslant \var_*\big[\sqrt{b_n}\rho_{n}^{-\mathfrak{r}} U_{R}(\mathbb{G}^{*}_{b_n})\big] \leqslant c_4 < \infty \hspace{0.2cm} \text{for some $c_3,c_4 > 0$}.
\end{equation}     
            \item  Under Assumption \ref{ass:rho_n_h_n}, for every $\epsilon > 0$,
            \begin{equation}
    \label{eq:assump3}
    \lim\limits_{n \to \infty}  b_nE_* 
    \big[b_n\rho_{n}^{-2\mathfrak{r}} g^2_{1,R}(\mathbb{V}_1)\mathbb{1}_{\{b_n\rho_{n}^{-2\mathfrak{r}} g^2_{1,R}(\mathbb{V}_1)> \epsilon\}}\big] = 0,
\end{equation}
        \end{enumerate}

Consequently, if Assumptions \ref{ass:rho_n_h_n}, \ref{ass:b} and \ref{ass:non_degenerate} hold,  with probability one, 
        \begin{equation}
\label{eq:unit_subsample_distribution_with_truepn}
\sqrt{b}\big[\rho_{n}^{-\mathfrak{r}}U_{R}(\mathbb{G}^{*}_{b}) -\rho_{n}^{-\mathfrak{r}}U_{R}(G)\big] \rightarrow  \mathcal{N}(0,\sigma^2_{*R})~\text{in distribution},
\end{equation}

\item \label{theo:subsample_distribution_b} Let  $\{R_1,\cdots,R_m\}$ be $m$ motifs with $\max\{r_1,\cdots,r_m\} \leq r$ and $ \max\{\mathfrak{r}_1,\cdots,\mathfrak{r}_m\} \leq \mathfrak{r}$.  Suppose Assumptions \ref{ass:rho_n_h_n}, \ref{ass:b} and \ref{ass:non_degenerate} hold. With probability one (with respect to the random sequence $\{\mathbb{G}_n\})$, 
\begin{equation}
\label{eq:joint_subsample_distribution_with_truepn}
\begin{aligned}
    &\sqrt{b}\Big\{\big[\rho^{-\mathfrak{r}_1}_{n}U_{R_1}(\mathbb{G}^{*}_{b}), \cdots, \rho^{-\mathfrak{r}_m}_{n}U_{R_m}(\mathbb{G}^{*}_{b})\big]- \big[\rho^{-\mathfrak{r}_1}_{n}U_{R_1}(G), \cdots, \rho^{-\mathfrak{r}_m}_{n}U_{R_m}(G)\big]\Big\} \\ \rightarrow &~ \mathcal{N}\big[0, \Sigma_{*[R_m]}\big] ~\text{in distribution},
\end{aligned}
\end{equation}
    \end{enumerate}
\end{theorem}

\subsection{Consistency of empirical distribution} \label{subsec:consistencyemperical}

Consider the following empirical cumulative distribution function
\begin{equation}
\begin{aligned}
      \widehat{J}^{\{R_1,\cdots,R_m\}}_{*,n,b}(t_1,\cdots,t_m)  &=  \frac{1}{N} \sum_{i=1}^{N} \mathbb{1} \Big\{\sqrt{b}\big[\wh{\rho}^{-\mathfrak{r}_1}_{G}U_{R_1}(G^{*(i)}_b)- \wh{\rho}^{-\mathfrak{r}_1}_{G}U_{R_1}(G)\big]\leq t_1, \\ &\qquad \qquad \cdots, \sqrt{b}\big[\wh{\rho}^{-\mathfrak{r}_m}_{G}U_{R_m}(G^{*(i)}_b)-\wh{\rho}^{-\mathfrak{r}_m}_{G}U_{R_m}(G)\big] \leq t_m\Big\}.
\end{aligned}
\end{equation}

The following consistency result is developed based on \cite{lunde2023subsampling}:
\begin{lemma}
\label{theo:consistent_empir}
For $\{R_1,\cdots,R_m\}$ with $\max\{r_1,\cdots,r_m\} \leq r$ and $ \max\{\mathfrak{r}_1,\cdots,\mathfrak{r}_m\} \leq \mathfrak{r}$. Under Assumptions  \ref{ass:rho_n_h_n}-\ref{ass:non_degenerate},  with probability one:
\begin{equation*}
    \begin{aligned}
     &\sup\limits_{[t_m] \in \mathbb{R}^m}\Big|  \widehat{J}^{\{R_1,\cdots,R_m\}}_{*,n,b}(t_1,\cdots,t_m)
       - J^{\{R_1,\cdots,R_m\}}_{*,n,b}(t_1,\cdots,t_m)\Big|   \to 0.
    \end{aligned}
\end{equation*}
\end{lemma}

%% file: Appendix/Appendix_SuppleA1.tex
\section{Two examples for Definition \ref{eq:defs}}
\label{sec:twoexamples}
Following \cite{maugis2020testing}, we use two examples to explain this definition. In the first example, let $R$ be a $\xtriangle$ and $R'$ also be a $\xtriangle$. Then the set $\mathcal{S}_{R, R'}$ can be constructed as $\left\{\xtriangle \,\xtriangle, \xreltriangle[1.05], \xdiagrectangle[.85], \xtriangle\right\}$. Each element in $\mathcal{S}_{R, R'}$  can be obtained by building blocks based on $R$ and $R'$. Let $R_1$ be a copy of $R$, and $R_2$ be a copy to $R'$. The pattern $\xtriangle \,\xtriangle$ can be built by either put $R_1$ in the left side or in the right side. Thus, $c_{\xtriangle[.5]\xtriangle[.5]} = 2$. Similarly,  $c_{\,\xreltriangle[.75]} = 2$, $c_{\,\xdiagrectangle[.75]} = 2$ and $c_{\xtriangle[.5]} = 1$. Generally speaking, $c_S$ denotes the number of ways $S$ can be built from copies of  $R$ and $R'$. Based on the number of merged nodes, we have $S_{R,R'}^{(0)} = \left\{\xtriangle \,\xtriangle\right\}$,$S_{R,R'}^{(1)} = \left\{\xreltriangle[.95]\right\}$, $S_{R,R'}^{(2)} = \left\{\xdiagrectangle\right\}$, and $S_{R,R'}^{(3)} = \left\{\xtriangle\right\}$. For the second example, let $R$ be a $\xrectangle$ and $R'$ be a $\xrectangle$. Then  $S_{R,R'} = \{S_{R,R'}^{(0)},S_{R,R'}^{(1)},S_{R,R'}^{(2)},S_{R,R'}^{(3)},S_{R,R'}^{(4)}\}$, with $S_{R,R'}^{(0)} = \left\{\xrectangle\,\xrectangle\right\}$, $S_{R,R'}^{(1)} = \left\{\xreldiamond\right\}$, $S_{R,R'}^{(2)} = \left\{\xtwodiamond,\, \xrecanddia,\, \xtworectangle\right\}$, $S_{R,R'}^{(3)} = \left\{\xthreetrangle[1],\, \xdiagdiamond[.75],\, \xtwodiagrectangle[.9]\right\}$, $S_{R,R'}^{(4)} = \left\{\xrectangle\right\}$. Correspondingly, $c_{\xrectangle[.5]\xrectangle[.5]} = 2$, $c_{\xreldiamond[.5]} = 2$, $c_{\xtwodiamond[.5]} = 6$, $c_{\xrecanddia[.5]} = 2$, $c_{\xtworectangle[.5]} = 2$, $c_{\xthreetrangle[.75]} = 2$, $c_{\xdiagdiamond[.5]} = 6$, $c_{\xtwodiagrectangle[.5]} = 6$ and $c_{\xrectangle[.5]} = 1$. 

As noted in \cite{maugis2020testing}, $X_R(G)X_{R'}(G)$ involves counting pairs of motifs, and could be recovered by counting the number of the all motifs that are formed by using one copy of $R$ and one copy of $R'$ as building blocks. This is the intuition of Lemma \ref{lem: linerityMC}. Moreover, the equation in \eqref{eq:linerityMC} provides flexibility as it does not depend on the generation mechanism of $G$.

%% file: Appendix/Appendix_SuppleA2A3.tex
\allowdisplaybreaks

\section{Proofs for Section \ref{subsec:moreTheroemNMgraphon1}}
\label{sec:proof for nmg1}
\subsection{Proof of Lemma \ref{lem: Exp_Var_graphon}}

\begin{proof}
 Lemma \ref{lem: Exp_Var_graphon} mostly follows the results in \cite{bhattacharyya2015supplement,maugis2020testing,bhattacharya2022fluctuations}. We provide the proof here for completeness.

\begin{enumerate}[label=\roman*)]
    \item First, by \eqref{eq:ref_exp_U_graphon} in Lemma~\ref{lem:lemma-ass}, $E[U_{R}(\mathbb{G}_{n})]  = \binom{n}{r}^{-1} X_R(K_{n})P_{h_n}(R)$, such  that 
    \begin{equation}
\label{eq:aut}
X_R(K_{n}) \stackrel{\text{\eqref{eq:bb07} in Lemma~\ref{lem:lemma-ass}}}{=} \binom{n}{r}X_R(K_r) \stackrel{\text{\eqref{eq:bha27} in Lemma~\ref{lem:lemma-ass}}}{=}\binom{n}{r} \frac{r!}{|\mathrm{Aut}(R)|}
\end{equation} and 
$E[U_{R}(\mathbb{G}_{n})] = \binom{n}{r}^{-1} X_R(K_{n})P_{h_n}(R) = \frac{r!}{|\mathrm{Aut}(R)|} P_{h_n}(R).$
These give \eqref{eq:exp_graphon} directly.


\item To show \eqref{eq:cov_graphon}, we start with
\begin{equation}
\label{eq:varref}
\begin{aligned}
&\Cov\big[U_{R}(\mathbb{G}_{n}),U_{R'}(\mathbb{G}_{n})\big] = E\big[U_{R}(\mathbb{G}_{n})U_{R'}(\mathbb{G}_{n})\big] -E\big[U_{R}(\mathbb{G}_{n})\big]E\big[U_{R'}(\mathbb{G}_{n})\big]
\\ \stackrel{\text{\eqref{eq:ref_exp_U_graphon}}}{=} & \binom{n}{r}^{-1}\binom{n}{r'}^{-1}E\big[X_R(\mathbb{G}_{n})X_{R'}(\mathbb{G}_{n})\big]-\Big\{E\big[U_{R}(\mathbb{G}_{n})\big]E\big[U_{R'}(\mathbb{G}_{n})\big]\Big\}
\\ \stackrel{\text{\eqref{eq:linerityMC}}}{=}&\binom{n}{r}^{-1}\binom{n}{r'}^{-1}E\big[\sum_{q=0}^{\min\{r,r'\}}\sum_{S\in \mathcal{S}_{R, R'}^{(q)}}c_SX_{S}(\mathbb{G}_{n})\big] - E\big[U_{R}(\mathbb{G}_{n})\big]E\big[U_{R'}(\mathbb{G}_{n})\big].
\end{aligned}
\end{equation}

The result in \cite{bhattacharyya2015subsampling} (see Section~\ref{subsec:addresMC})implies that
\begin{equation*}
{\small
\begin{aligned}
\Cov\big[U_{R}(\mathbb{G}_{n}), U_{R'}(\mathbb{G}_{n})\big] =& \binom{n}{r}^{-1}\binom{n}{r'}^{-1} \sum_{q=1}^{\min\{r,r'\}}\sum_{S\in \mathcal{S}_{R, R'}^{(q)}}c_SE\big[X_{S}(\mathbb{G}_{n})\big] - \left[1-\dfrac{\binom{n-r}{r'}}{\binom{n}{r'}}\right]  \cdot E\big[U_{R}(\mathbb{G}_{n})\big]E\big[U_{R'}(\mathbb{G}_{n})\big].
\end{aligned}}
\end{equation*}

Combining with \eqref{eq:varref}, we obtain 
\begin{equation}
\label{eq:ref_from_bickel}
\begin{aligned}
   E\big[U_{R}(\mathbb{G}_{n})\big]E\big[U_{R'}(\mathbb{G}_{n})\big] = \dfrac{\binom{n}{r'}}{\binom{n-r}{r'}} \binom{n}{r}^{-1}\binom{n}{r'}^{-1} \sum_{S\in \mathcal{S}_{R, R'}^{(0)}}c_SE\big[X_{S}(\mathbb{G}_{n})\big].
\end{aligned}
\end{equation}

Finally, we have
\begin{equation*}
\begin{aligned}
     \Cov\big[U_{R}(\mathbb{G}_{n}),U_{R'}(\mathbb{G}_{n})\big] 
     \stackrel{\text{\eqref{eq:varref}}}{=}&\binom{n}{r}^{-1}\binom{n}{r'}^{-1}\sum_{q=0}^{\min\{r,r'\}}\sum_{S\in \mathcal{S}_{R, R'}^{(q)}}c_SE\big[X_{S}(\mathbb{G}_{n})\big]\\&-  E\big[U_{R}(\mathbb{G}_{n})\big]E\big[U_{R'}(\mathbb{G}_{n})\big]
    \\ \stackrel{\text{\eqref{eq:ref_from_bickel}}}{=} & \binom{n}{r}^{-1}\binom{n}{r'}^{-1}\sum_{q=1}^{\min\{r,r'\}}\sum_{S\in \mathcal{S}_{R, R'}^{(q)}}c_SE\big[X_{S}(\mathbb{G}_{n})\big]  \\&+   \binom{n}{r}^{-1}\binom{n}{r'}^{-1}\sum_{S\in \mathcal{S}_{R, R'}^{(0)}}c_SE\big[X_{S}(\mathbb{G}_{n})\big] \\&- \dfrac{\binom{n}{r'}}{\binom{n-r}{r'}} \binom{n}{r}^{-1}\binom{n}{r'}^{-1} \sum_{S\in \mathcal{S}_{R, R'}^{(0)}}c_SE\big[X_{S}(\mathbb{G}_{n})\big] \\
       = & \binom{n}{r}^{-1}\binom{n}{r'}^{-1}\sum_{q=1}^{\min\{r,r'\}}\sum_{S\in \mathcal{S}_{R, R'}^{(q)}}c_SE\big[X_{S}(\mathbb{G}_{n})\big] \\&- \Big(\dfrac{\binom{n}{r'}}{\binom{n-r}{r'}}-1\Big)\binom{n}{r}^{-1}\binom{n}{r'}^{-1} \sum_{S\in \mathcal{S}_{R, R'}^{(0)}}c_SE\big[X_{S}(\mathbb{G}_{n})\big],
\end{aligned}
\end{equation*}
which gives \eqref{eq:cov_graphon}.
\end{enumerate}
\end{proof}

\begin{lemma}\label{lem:lemma-ass}
\begin{equation}
\label{eq:ref_exp_U_graphon}
\begin{aligned}
E\big[U_{R}(\mathbb{G}_{n})\big]  =  \binom{n}{r}^{-1} E\big[X_R(\mathbb{G}_{n})\big]= \binom{n}{r}^{-1} X_R(K_{n})P_{h_n}(R),
\end{aligned}
\end{equation}
where $K_{n}$ denotes a complete graph of size $n$, and $ P_{h_n}(R)$ is defined in \eqref{eq:PR} . 
\begin{equation}
\label{eq:bb07}
    X_R(K_{n}) = \binom{n}{r}X_R(K_r),
\end{equation}
\begin{equation}
\label{eq:bha27}
     X_R(K_r) = r!/|\mathrm{Aut}(R)|.
\end{equation}

\end{lemma}
The \eqref{eq:ref_exp_U_graphon} is proved in \cite{maugis2020testing} (See their Equation (1)),  \eqref{eq:bb07} is used in \cite{bollobas2007metrics}, and \eqref{eq:bha27} is proved in  \cite{bhattacharya2022fluctuations} (see their Equation (2.7)).

\subsection{Proof of Proposition \ref{prop: lim_var_graphon}}
   
 
\begin{proof}
Proof of \eqref{eq:lim_exp_graphon}:
 \begin{equation*}
\begin{aligned}
      E\big[\rho^{-\mathfrak{r}}_{n}U_{R}(\mathbb{G}_{n})\big]& =  \rho^{-\mathfrak{r}}_{n}E\big[U_{R}(\mathbb{G}_{n})\big] \stackrel{\text{\eqref{eq:exp_graphon}}}{=}  \frac{\rho^{-\mathfrak{r}}_{n}r!}{|\mathrm{Aut}(R)|} P_{h_n}(R) \\
      & \stackrel{\text{\eqref{eq:PR}}}{=} \frac{\rho^{-\mathfrak{r}}_{n}r!}{|\mathrm{Aut}(R)|} \int_{[0,1]^r} \prod_{(v_i,v_j) \in \eE(R)} h_n\left(\xi_i, \xi_j\right) \prod_{v_i \in V(R)} d \xi_i \\
      & = \frac{r!\rho^{-\mathfrak{r}}_{n}\rho^{\mathfrak{r}}_n}{|\mathrm{Aut}(R)|}\int_{[0,1]^r} \prod_{(v_i,v_j) \in \eE(R)} w\left(\xi_i, \xi_j\right)\mathbb{1}_{\{\rho_nw(\xi_i,\xi_j) \leqslant 1\}} \prod_{v_i \in V(R)} d \xi_i\\
      & = \frac{r!}{|\mathrm{Aut}(R)|}\int_{[0,1]^r} \prod_{(v_i,v_j) \in \eE(R)} w\left(\xi_i, \xi_j\right) \prod_{v_i \in V(R)} d \xi_i \stackrel{\text{\eqref{eq:PR}}}{=}  \frac{r!}{|\mathrm{Aut}(R)|}P_{w}(R).
\end{aligned}
\end{equation*}  
Proof of \eqref{eq: lim_var_graphon}: we decompose the covariance as
\begin{equation*}
\begin{aligned}
     &\Cov\Big[\sqrt{n}\rho^{-\mathfrak{r}}_{n}U_{R}(\mathbb{G}_{n}), \sqrt{n}\rho^{-\mathfrak{r}'}_{n}U_{R'}(\mathbb{G}_{n})\Big] = n\rho^{-(\mathfrak{r} + \mathfrak{r}')}_{n}\Cov\Big[U_{R}(\mathbb{G}_{n}), U_{R'}(\mathbb{G}_{n})\Big]
       \\   \stackrel{\text{\eqref{eq:cov_graphon}}}{=} &~ n\rho^{-(\mathfrak{r} + \mathfrak{r}')}_{n}\binom{n}{r}^{-1}\binom{n}{r'}^{-1}\sum_{q=1}^{\min\{r,r'\}}\sum_{S\in \mathcal{S}_{R, R'}^{(q)}}c_SE\big[X_{S}(\mathbb{G}_{n})\big]
       \\&- n\rho^{-(\mathfrak{r} + \mathfrak{r}')}_{n}\Big[\dfrac{\binom{n}{r'}}{\binom{n-r}{r'}}-1\Big]\binom{n}{r}^{-1}\binom{n}{r'}^{-1} \sum_{S\in \mathcal{S}_{R, R'}^{(0)}}c_SE\big[X_{S}(\mathbb{G}_{n})\big]
        \\:=& ~ \mathrm{I}  - \mathrm{II} .
\end{aligned}
\end{equation*}

Let $|V(S)| = s$ and $|E(S)| = \mathfrak{s}$, we have
\begin{equation*}
\begin{aligned}
     \mathrm{I}&=n\rho^{-(\mathfrak{r} + \mathfrak{r}')}_{n}\binom{n}{r}^{-1}\binom{n}{r'}^{-1}\sum_{q=1}^{\min\{r,r'\}}\sum_{S\in \mathcal{S}_{R, R'}^{(q)}}c_SE\big[X_{S}(\mathbb{G}_{n})\big]
    \\&= \sum_{q=1}^{\min\{r,r'\}}\sum_{S\in \mathcal{S}_{R, R'}^{(q)}}c_S\rho^{-(\mathfrak{r} + \mathfrak{r}')}_{n}\dfrac{nr!(n-r)!}{n!}\dfrac{r'!(n-r')!}{n!}\binom{n}{s}E\big[U_{S}(\mathbb{G}_{n})\big]
       \\ &\stackrel{\text{\eqref{eq:lim_exp_graphon}}}{=}\sum_{q=1}^{\min\{r,r'\}}\sum_{S\in \mathcal{S}_{R, R'}^{(q)}}c_S\rho^{\mathfrak{s}-(\mathfrak{r} + \mathfrak{r}')}_{n}\dfrac{nr!(n-r)!}{n!}\dfrac{r'!(n-r')!}{n!}\dfrac{n!}{s!(n-s)!}
       \frac{s!}{|\text{Aut}(S)|}P_{w}(S)
        \\ & =\sum_{q=1}^{\min\{r,r'\}}\sum_{S\in \mathcal{S}_{R, R'}^{(q)}}\rho^{\mathfrak{s}-(\mathfrak{r} + \mathfrak{r}')}_{n}\dfrac{(n-r')!}{(n-1)\cdots(n-r+1)(n-s)!}\frac{c_Sr!r'!}{|\text{Aut}(S)|}P_{w}(S).
\end{aligned}
\end{equation*}

The quantities $r$, $r'$, $c_S$, $|\text{Aut}(S)|$, and $P_{w}(S)$ are invariant of $n$. The quantities $\rho^{\mathfrak{s}-(\mathfrak{r} + \mathfrak{r}')}_{n}$ and $(n-r')!/[(n-1)\cdots(n-r+1)(n-s)!]$  change with $n$. Now, we consider these two quantities based on the number of merged nodes $q$. 
\begin{itemize}
\item When $q = 1$, we have $\mathfrak{s}=\mathfrak{r} + \mathfrak{r}'$ and $s = r + r' - 1$. The following quantity
\begin{equation*}
    \dfrac{(n-r')!}{(n-1)\cdots(n-r+1)(n-s)!} = \dfrac{(n-r')(n-r'-1)\cdots(n-r-r'+2)}{(n-1)\cdots(n-r+1)}
\end{equation*}
has $r-1$ items including $n$ in both numerator and denominator. Thus,
\begin{equation*}
    \rho^{\mathfrak{s}-(\mathfrak{r} + \mathfrak{r}')}_{n}\dfrac{(n-r')!}{(n-1)\cdots(n-r+1)(n-s)!}  =  1 + o(1).
\end{equation*}

\item When $q = 2$, we have $\mathfrak{s} = \mathfrak{r} + \mathfrak{r}' - 1$ because one edge is merged. The following quantity
\begin{equation*}
    \dfrac{(n-r')!}{(n-1)\cdots(n-r+1)(n-s)!} = \dfrac{(n-r')(n-r'-1)\cdots(n-r-r'+3)}{(n-1)\cdots(n-r+1)}
\end{equation*}
has $r-2$ items with $n$ in numerator, and $r-1$ items with $n$ in denominator. As $n\rho_n \to \infty$,
\begin{equation*}
    \rho^{\mathfrak{s}-(\mathfrak{r} + \mathfrak{r}')}_{n}\dfrac{(n-r')!}{(n-1)\cdots(n-r+1)(n-s)!} = O(\frac{1}{n\rho_n}) = o(1).
\end{equation*}

\item When $2 < q < \min\{r,r'\}$, at most $q(q-1)/2$ edges are merged.  The following quantity
\begin{equation*}
    \dfrac{(n-r')!}{(n-1)\cdots(n-r+1)(n-s)!} = \dfrac{(n-r')(n-r'-1)\cdots(n-r-r'+(q+1))}{(n-1)\cdots(n-r+1)}
\end{equation*}
has $r-q$ items with $n$ in numerator and $r-1$ items with $n$ in denominator. Since $n^{(q-1)}\rho^{(q(q-1)/2)}_n = (n\rho^{q/2}_n)^{(q-1)} \rightarrow \infty$
\begin{equation*}
    \rho^{\mathfrak{s}-(\mathfrak{r} + \mathfrak{r}')}_{n}\dfrac{(n-r')!}{(n-1)\cdots(n-r+1)(n-s)!} = O(\frac{1}{(n\rho^{q/2}_n)^{(q-1)}}) = o(1).
\end{equation*}
\end{itemize} Therefore, $\mathrm{I}\to\sum_{S\in\mathcal{S}_{R,R'}^{(1)}}(c_Sr!r'!)(|\text{Aut}(S)|)^{-1}P_{w}(S)$ as $n\to \infty$.
 
Now we turn to Part $\mathrm{II}$. Since $s = r + r'$ and $\mathfrak{s}=\mathfrak{r} + \mathfrak{r}'$ when $q = 0$, we have
\begin{equation*}
\begin{aligned}
     \Big[\dfrac{n\binom{n}{r'}}{\binom{n-r}{r'}}-n\Big]\binom{n}{r}^{-1}\binom{n}{r'}^{-1}\binom{n}{s} =& \dfrac{r!r'!}{(r+r')!}\bigg\{b\Big[1 - \dfrac{(n-r)!(n-r')!}{n!(n-r-r'
       )!}\Big]\bigg\}\\=&
       \dfrac{r!r'!}{(r+r')!}\Big[\dfrac{n(n-1)\cdots(n-r+1)- (n-r')\cdots(n-r-r'+1)}{(n-1)\cdots(n-r+1)}\Big].
\end{aligned}
\end{equation*}

Consequently,
\begin{align*}
  n(n-1)\cdots(n-r+1) &= n^{r} - \dfrac{(r-1)r}{2}n^{r-1} + o(n^{r-1}), \\
   (n-r')\cdots(n-r-r'+1) &= n^{r} - \dfrac{(2r'+r-1)r}{2}n^{r-1} + o(n^{r-1}), \\
    (n-1)\cdots(n-r+1) &= n^{r-1} + o(n^{r-1}).
\end{align*}
It is easy to see that 
 \begin{equation}
\label{eq:lhop}
\begin{aligned}
     & \lim\limits_{n \to \infty } \Big(\dfrac{n\binom{n}{r'}}{\binom{n-r}{r'}}-n\Big)\binom{n}{r}^{-1}\binom{n}{r'}^{-1}\binom{n}{s}
        \\=& \lim\limits_{n \to \infty }  \dfrac{r!r'!}{(r+r')!}\Big[\dfrac{n(n-1)\cdots(n-r+1)- (n-r')\cdots(n-r-r'+1)}{(n-1)\cdots(n-r+1)}\Big]  \\=& \lim\limits_{n \to \infty } \dfrac{r!r'!}{(r+r')!}\dfrac{ rr' n^{r-1} + o( n^{r-1})}{ n^{r-1} + o(n^{r-1})} =  \dfrac{r!r'!}{(r+r')!}rr'.
\end{aligned}
\end{equation}
Hence, 
\begin{equation*}
\begin{aligned}
   \lim\limits_{n \to \infty}\mathrm{II}&=   \lim\limits_{n \to \infty}n\rho^{-(\mathfrak{r} + \mathfrak{r}')}_{n}\Big(\dfrac{\binom{n}{r'}}{\binom{n-r}{r'}}-1\Big)\binom{n}{r}^{-1}\binom{n}{r'}^{-1} \sum_{S\in \mathcal{S}_{R, R'}^{(0)}}c_SE\big[X_{S}(\mathbb{G}_{n})\big]
   \\&=   \lim\limits_{n \to \infty}n\rho^{\mathfrak{s}-(\mathfrak{r} + \mathfrak{r}')}_{n}\Big(\dfrac{\binom{n}{r'}}{\binom{n-r}{r'}}-1\Big)\binom{n}{r}^{-1}\binom{n}{r'}^{-1}\binom{n}{s}  \sum_{S\in \mathcal{S}_{R, R'}^{(0)}}c_SE\big[\rho^{-\mathfrak{s}}_{n}U_{S}(\mathbb{G}_{n})\big]
   \\&=   \lim\limits_{n \to \infty}n\Big(\dfrac{\binom{n}{r'}}{\binom{n-r}{r'}}-1\Big)\binom{n}{r}^{-1}\binom{n}{r'}^{-1}\binom{n}{s}  \sum_{S\in \mathcal{S}_{R, R'}^{(0)}}c_SE\big[\rho^{-\mathfrak{s}}_{n}U_{S}(\mathbb{G}_{n})\big]
   \\ &\stackrel{\text{\eqref{eq:lim_exp_graphon}}}{=}\lim\limits_{n \to \infty}n\Big(\dfrac{\binom{n}{r'}}{\binom{n-r}{r'}}-1\Big)\binom{n}{r}^{-1}\binom{n}{r'}^{-1}\binom{n}{s}  \sum_{S\in \mathcal{S}_{R, R'}^{(0)}}c_S \frac{s!}{|\text{Aut}(S)|}P_{w}(S)\\&\stackrel{\text{\eqref{eq:lhop}}}{=}\sum_{S\in\mathcal{S}_{R,R'}^{(0)}} \frac{c_Sr!r'!rr'}{|\text{Aut}(S)|}P_{w}(S).
\end{aligned}
\end{equation*}

Combining the above results, we have
\begin{equation*}
\begin{aligned}
     \lim\limits_{n \to \infty}\Cov\Big[\sqrt{n}\rho^{-\mathfrak{r}}_{n}U_{R}(\mathbb{G}_{n}), \sqrt{n}\rho^{-\mathfrak{r}'}_{n}U_{R'}(\mathbb{G}_{n})\Big] 
      = &\lim\limits_{n \to \infty}\mathrm{I} -  \lim\limits_{n \to \infty}\mathrm{II}\\
    =&\sum_{S\in\mathcal{S}_{R,R'}^{(1)}}\frac{c_Sr!r'!}{|\text{Aut}(S)|}P_{w}(S) - \sum_{S\in\mathcal{S}_{R,R'}^{(0)}} \frac{c_Sr!r'!rr'}{|\text{Aut}(S)|}P_{w}(S),
\end{aligned}
\end{equation*}
which gives \eqref{eq: lim_var_graphon}.

\end{proof}

\counterwithout{equation}{section}
\setcounter{equation}{0}
\makeatletter
\renewcommand{\theequation}{SE.\@arabic\c@equation}
\renewcommand{\thelemma}{SE.\@arabic\c@lemma} 
\makeatother

\section{Proofs for Section \ref{subsec:addresQuantify1}}
\label{sec:proof for q1}
\subsection{Proof of Lemma \ref{lem: usefuleqs}}

\begin{proof} 

The following equation from \cite{maugis2020testing} is used in this proof.
\begin{equation}
\label{eq:ref}
    \binom{n}{r}U_R(G) = X_R(G) = \sum_{R_c \in \{S: S \subset G, S \cong  R\}}1 = \big|\{S: S \subset G , S \cong  R\}\big|.
\end{equation}

We now prove the identities in Lemma \ref{lem: usefuleqs} one by one.
\begin{enumerate}[label=\roman*)]
    \item For any $R_c \in \{S: S \subset G, S \cong  R\}$, recall that $R_{c} \subset G_b^*$ if both $V(R_{c}) \subset V(G_b^*)$ and $\eE(R_{c} ) \subset \eE(G_b^*)$. Furthermore, since $R_c \subset G$ and $G_b^* \indsubset G$, $V(R_{c}) \subset V(G_b^*)$ implies  $\eE(R_{c} ) \subset \eE(G_b^*)$. Thus,
\begin{equation}
\label{eq:ref1}
    \mathbb{1}_{\{R_{c} \subset G_b^* \}} = 1 \quad \text{if and only if} \quad V(R_c) \subset V(G_b^* ).
\end{equation}
Now let us consider drawing $b$ nodes from $V(G)$ by first selecting all nodes in $V(R_{c})$, and then randomly drawing $b-r$ nodes without replacement from $V(G)\backslash V(R_{c})$. There are $\binom{n-r}{b-r}$ ways to draw these $b$ nodes. Thus,
\begin{equation}
    \label{eq:ref2}
    \sum\limits_{\mathcal{G}  \in \mathcal{S}(\mathbb{G}^{*}_{b})}\mathbb{1}_{\{R_{c} \subset \mathcal{G} \}} = \binom{n-r}{b-r}.
\end{equation}
Consequently,
\begin{equation*}
    \begin{aligned}
      \sum_{\mathcal{G}  \in \mathcal{S}(\mathbb{G}^{*}_{b})} X_R(\mathcal{G}) \stackrel{\text{\eqref{eq:ref}}}{=}&\sum\limits_{\mathcal{G}  \in \mathcal{S}(\mathbb{G}^{*}_{b})}\sum_{R_c \in \{S: S \subset \mathcal{G}, S \cong  R\}}1 =\sum\limits_{\mathcal{G}  \in \mathcal{S}(\mathbb{G}^{*}_{b})}\sum_{R_c \in \{S: S \subset G, S \cong  R\}}\mathbb{1}_{\{R_c \subset \mathcal{G} \}} 
      \\ = &\sum_{R_c \in \{S: S \subset G, S \cong  R\}}\sum\limits_{\mathcal{G} \in \mathcal{S}(\mathbb{G}^{*}_{b})}\mathbb{1}_{\{R_c \subset \mathcal{G}\}}
      \\ \stackrel{\text{ \eqref{eq:ref2}}}{=} &\sum_{R_c \in \{S: S \subset G, S \cong  R\}}\binom{n-r}{b-r}   \stackrel{\text{\eqref{eq:ref}}}{=} \binom{n-r}{b-r} X_R(G),
        \end{aligned}  
\end{equation*}
which gives \eqref{eq:iden1}.

\item For any $G^{v**}_{b,r} \in \mathcal{S}(\mathbb{G}^{v**}_{b,r})$, if  $S \subset G^{v**}_{b,r}$ and $|V(S)| = |V(G^{v**}_{b,r})| $, we have $V(S) = V(G^{v**}_{b,r})$. In addition, as $v \in V(G^{v**}_{b,r})$, we have 
\begin{equation}
    \label{eq:ref3}
    \{S : S \subset G^{v**}_{b,r}, S \cong R \} = \{S : S \subset G^{v**}_{b,r}, v\in V(S), S \cong R \}.
\end{equation}

Let $R_{c} \in \{S: S \subset G, S \cong  R\}$. Suppose that $R_{c} \subset G^{v*}_{b}$ and $v \in V(R_{c})$. As every $G^{v**}_{b,r}$ is an induced subgraph, we have
\begin{equation}
    \label{eq:ref4}
 \sum\limits_{\mathcal{G} \in \mathcal{S}(\mathbb{G}^{v**}_{b,r})}\mathbb{1}_{\{R_{c} \subset \mathcal{G}\}}= \sum\limits_{\mathcal{G} \in \mathcal{S}(\mathbb{G}^{v**}_{b,r})}\mathbb{1}_{\{V(R_{c}) = V(\mathcal{G})\}} = 1.
\end{equation}
Consequently,
\begin{equation*}
\begin{aligned} 
   &  \Big|\{S: S \subset G^{v*}_{b}, v \in V(S), S \cong  R\}\Big| \\= &\sum\limits_{R_c \in \{S: S \subset G^{v*}_{b}, v \in V(S), S \cong  R\}}1 
     \stackrel{\text{\eqref{eq:ref4}}}{=}  \sum\limits_{R_c \in \{S: S \subset G^{v*}_{b}, v \in V(S), S \cong  R\}}\Big(\sum\limits_{\mathcal{G} \in \mathcal{S}(\mathbb{G}^{v**}_{b,r})}\mathbb{1}_{\{R_{c} \subset \mathcal{G}\}}\Big)
     \\ = &\sum\limits_{\mathcal{G} \in \mathcal{S}(\mathbb{G}^{v**}_{b,r})}\Big(\sum\limits_{R_c \in \{S: S \subset G^{v*}_{b}, v \in V(S), S \cong  R\}}\mathbb{1}_{\{R_{c} \subset \mathcal{G}\}}\Big)
          \\ = &\sum\limits_{\mathcal{G} \in \mathcal{S}(\mathbb{G}^{v**}_{b,r})}\Big(\sum\limits_{R_c \in \{S: S \subset \mathcal{G}, v \in V(S), S \cong  R\}}1\Big)
     \\ \stackrel{\text{\eqref{eq:ref}}}{=} &\sum\limits_{\mathcal{G} \in \mathcal{S}(\mathbb{G}^{v**}_{b,r})}\Big|\{S : S \subset \mathcal{G}, v\in V(S), S \cong R \}\Big|
      \\\stackrel{\text{\eqref{eq:ref3}}}{=} &\sum\limits_{\mathcal{G} \in \mathcal{S}(\mathbb{G}^{v**}_{b,r})}\Big|\{S : S \subset \mathcal{G},S \cong R \}\Big| \stackrel{\text{\eqref{eq:ref}}}{=} \sum_{\mathcal{G}\in \mathcal{S}(\mathbb{G}^{v**}_{b,r})}X_R(\mathcal{G}),
\end{aligned}
\end{equation*}
leading to \eqref{eq:iden2}.

\item If $|V(S)| = r$ and $S \subset G^{v*}_{r}$, we have $V(S) = V(G^{v*}_{r})$. Thus,
$\{S : S \subset G^{v*}_{r}, S \cong R \} = \{S : S \subset G^{v*}_{r}, v\in V(S), S \cong R \}.$ Let $R_{c} \in \{S: S \subset G, S \cong  R\}$. Because $G^{v*}_{r} \indsubset G$, there exist only one $G^{v*}_{r} \in \mathcal{S}(\mathbb{G}^{v*}_{r})$ such that $V(R_{c}) = V(G^{v*}_{r})$. Also, $V(R_{c}) = V(G^{v*}_{r})$ implies $E(R_{c}) \subset E(G^{v*}_{r})$. Hence,
$\sum\limits_{\mathcal{G} \in \mathcal{S}(\mathbb{G}^{v*}_{r})}\mathbb{1}_{\{R_{c} \subset \mathcal{G}\}} = \sum\limits_{\mathcal{G} \in \mathcal{S}(\mathbb{G}^{v*}_{r})}\mathbb{1}_{\{V(R_{c}) \subset V(\mathcal{G})\}} = 1.$
Consequently,
\begin{equation*}
\begin{aligned}
     \Big|\{S: S \subset G, v \in V(S), S \cong  R\}\Big| &= \sum\limits_{R_c \in \{S: S \subset G, v \in V(S), S \cong  R\}}1 
     \\&= \sum\limits_{R_c \in \{S: S \subset G, v \in V(S), S \cong  R\}}\Big(\sum\limits_{\mathcal{G} \in \mathcal{S}(\mathbb{G}^{v*}_{r})}\mathbb{1}_{\{R_{c} \subset \mathcal{G}\}}\Big)
     \\& = \sum\limits_{\mathcal{G} \in \mathcal{S}(\mathbb{G}^{v*}_{r})}\Big(\sum\limits_{R_c \in \{S: S \subset G, v \in V(S), S \cong  R\}}\mathbb{1}_{\{R_{c} \subset \mathcal{G}\}}\Big)
     \\& = \sum\limits_{\mathcal{G} \in \mathcal{S}(\mathbb{G}^{v*}_{r})}\Big|\{S : S \subset \mathcal{G}, v\in V(S), S \cong R \}\Big|
      \\& = \sum\limits_{\mathcal{G} \in \mathcal{S}(\mathbb{G}^{v*}_{r})}\Big|\{S : S \subset \mathcal{G},S \cong R \}\Big| = \sum_{\mathcal{G}\in \mathcal{S}(\mathbb{G}^{v*}_{r})}X_R(\mathcal{G}),
\end{aligned}
\end{equation*}
which leads to \eqref{eq:iden3}
\item Let $R_{c} \in \{S: S \subset G, v \in V(S), S \cong  R\}$, we have 
\begin{equation}
\label{eq:ref5}
\begin{aligned}
    \sum\limits_{\mathcal{G} \in \mathcal{S}(\mathbb{G}_{b}^{*})}\mathbb{1}_{\{R_{c} \subset \mathcal{G}\}} &= \sum\limits_{\mathcal{G} \in \mathcal{S}(\mathbb{G}_{b}^{*}), v \in V(\mathcal{G} )}\mathbb{1}_{\{R_{c} \subset \mathcal{G}\}} + \sum\limits_{\mathcal{G} \in \mathcal{S}(\mathbb{G}_{b}^{*}),v \notin V(\mathcal{G}) }\mathbb{1}_{\{R_{c} \subset \mathcal{G}\}} \\
    &= \sum\limits_{\mathcal{G} \in \mathcal{S}(\mathbb{G}_{b}^{*}), v \in V(\mathcal{G} )}\mathbb{1}_{\{R_{c} \subset \mathcal{G}\}} + 0 \\=& \sum\limits_{\mathcal{G} \in \mathcal{S}(\mathbb{G}_{b}^{v*})}\mathbb{1}_{\{R_{c} \subset \mathcal{G}\}}.
\end{aligned}
\end{equation}
Consequently,
\begin{equation}
    \begin{aligned}
       \sum\limits_{G_{b}^{*} \in \mathcal{S}(\mathbb{G}_{b}^{v*})}\Big(\sum_{\mathcal{G}\in \mathcal{S}(\mathbb{G}_{b,r}^{v**})}  X_R(\mathcal{G})\Big) &\stackrel{\text{\eqref{eq:iden2}}}{=}     \sum\limits_{\mathcal{G} \in \mathcal{S}(\mathbb{G}_{b}^{v*})} \Big|\{S: S \subset \mathcal{G} , v \in V(S), S \cong  R\}\Big|
       \\&\hspace{0.05cm}= \sum\limits_{\mathcal{G} \in \mathcal{S}(\mathbb{G}_{b}^{v*})} \sum\limits_{R_{c} \in \{S: S \subset G, v \in V(S), S \cong  R\}}\mathbb{1}_{\{R_c \subset \mathcal{G}\}}
       \\&\hspace{0.05cm} =  \sum\limits_{R_{c} \in \{S: S \subset G, v \in V(S), S \cong  R\}}\sum\limits_{\mathcal{G} \in \mathcal{S}(\mathbb{G}_{b}^{v*})}\mathbb{1}_{\{R_c \subset \mathcal{G}\}}
       \\&\stackrel{\text{\eqref{eq:ref5}}}{=}   \sum\limits_{R_{c} \in \{S: S \subset G, v \in V(S), S \cong  R\}}\sum\limits_{\mathcal{G} \in \mathcal{S}(\mathbb{G}_{b}^{*})}\mathbb{1}_{\{R_c \subset \mathcal{G}\}}
         \\&\stackrel{\text{\eqref{eq:ref2}}}{=}   \sum\limits_{R_{c} \in \{S: S \subset G, v \in V(S), S \cong  R\}}\binom{n-r}{b-r}
         \\&\hspace{0.05cm} = \binom{n-r}{b-r}\Big|\{S: S \subset G, v \in V(S), S \cong R\}\Big|,
    \end{aligned}
\end{equation}
which gives \eqref{eq:iden4}.

\item For the last identity, we have
\begin{equation*}
    \begin{aligned}
        \sum_{i=1}^{n}\Big|\{S: S \subset G, v_i \in V(S), S \cong  R\}\Big| &\stackrel{\text{\eqref{eq:ref}}}{=}    \sum_{i=1}^{n}\sum\limits_{R_{c} \in \{S: S \subset G, v_i \in V(S), S \cong  R\}}1
        \\& =\sum_{i=1}^{n}\sum\limits_{R_{c} \in \{S: S \subset G, S \cong  R\}}\mathbb{1}_{\{v_i \in V(R_c)\}}
        \\& = \sum\limits_{R_{c} \in \{S: S \subset G, S \cong  R\}}\sum_{i=1}^{n}\mathbb{1}_{\{v_i \in V(R_c)\}}
        \\&= \sum\limits_{R_{c} \in \{S: S \subset G, S \cong  R\}}r \\&= r\Big|\{S: S \subset G, S \cong  R\}\Big| = rX_R(G),
    \end{aligned}
\end{equation*}
which gives \eqref{eq:iden5}.
\end{enumerate}
\end{proof}

\subsection{Proof of Lemma  \ref{lem: exp_cov_subsample}}

\begin{proof}
We will prove the stated identities of Lemma~\ref{lem: exp_cov_subsample} one by one.

\begin{enumerate}[label=\roman*)]
    \item Following the results in \cite{maugis2020testing}, we have
         \begin{equation*}
    \begin{aligned}
      &E_*\Big[U_{R}(\mathbb{G}^{*}_{b})\Big] \stackrel{\text{\eqref{eq:ref}}}{=} E_*\Big[\binom{b}{r}^{-1}X_R(\mathbb{G}^{*}_{b})\Big] = \binom{b}{r}^{-1}E_*\Big[X_{R}(\mathbb{G}^{*}_{b})\Big] \\= & \binom{b}{r}^{-1} \binom{n}{b}^{-1} \sum_{\mathcal{G} \in \mathcal{S}(\mathbb{G}^*_b)} X_R(\mathcal{G})
       \stackrel{\text{\eqref{eq:iden1}}}{=} \binom{b}{r}^{-1} \binom{n}{b}^{-1}\binom{n-r}{b-r} X_R(G)\\= &\binom{n}{r}^{-1}X_R(G) 
       \stackrel{\text{\eqref{eq:ref}}}{=} U_{R}(G),
        \end{aligned}  
\end{equation*}
which gives \eqref{eq:exp_subsample}.

\item We start by showing that
\begin{equation}
  \label{eq:U_Gbn_Gn}
\begin{aligned}
       E\Big[U_{R}(\mathbb{G}_n)\Big]  \stackrel{\text{\eqref{eq:exp_graphon}}}{=} &\frac{r!}{|\mathrm{Aut}(R)|} P_{h_n}(R) \stackrel{\text{\eqref{eq:aut}}}{=}  \binom{b}{r}^{-1}X_R(K_{b})  P_{h_n}(R) \\ \stackrel{\text{\eqref{eq:ref_exp_U_graphon}}}{=} &E[U_R(\mathbb{G}^{h_n}_b)].
\end{aligned}
\end{equation}

Hence, $E\{E_*[U_{R}(\mathbb{G}^{*}_{b})]\}= E\{E_*[U_{R}(\mathbb{G}^{(*)}_{b})]\}= E[U_{R}(\mathbb{G}_n)] =E[U_R(\mathbb{G}^{h_n}_b)]$, where the second and third identifies are from \eqref{eq:exp_subsample} and \eqref{eq:U_Gbn_Gn}, respectively. This gives \eqref{eq:exp_subsample_conditional}.

\item Following \cite{bhattacharyya2015subsampling,maugis2020testing}, we have
\begin{equation*}
    \begin{aligned}
\Cov_{*}\Big[U_{R}(\mathbb{G}^{*}_{b}),U_{R'}(\mathbb{G}^{*}_{b})\Big]      &
\stackrel{\text{\eqref{eq:ref}}}{=}  \Cov_{*}\Big[\binom{b}{r}^{-1}X_R(\mathbb{G}^{*}_{b}), \binom{b}{r'}^{-1}X_{R'}( \mathbb{G}^{*}_{b})\Big] \\
        & =  \bigg\{\binom{b}{r}^{-1}\binom{b}{r'}^{-1}
E_*\Big[X_{R}(\mathbb{G}^{*}_{b})X_{R'}(\mathbb{G}^{*}_{b})\Big]\bigg\} -E_*\Big[U_{R}(\mathbb{G}^{*}_{b})\Big]E_*\Big[U_{R'}(\mathbb{G}^{*}_{b})\Big].
           \end{aligned}
         \end{equation*}
For the first term, we have
 \begin{align*}
&E_*\Big[X_R(\mathbb{G}^{*}_{b})X_{R'}(\mathbb{G}^{*}_{b})\Big]
         \stackrel{\text{\eqref{eq:linerityMC}}}{=}E_*\Big[
     \sum_{q=0}^{\min\{r,r'\}}\sum_{S\in \mathcal{S}_{R, R'}^{(q)}}c_{S}X_S(\mathbb{G}^{*}_{b})
        \Big]  \\
        =&
     \sum_{q=0}^{\min\{r,r'\}}\sum_{S\in \mathcal{S}_{R, R'}^{(q)}}c_{S}E_*[X_S(\mathbb{G}^{*}_{b})]=
     \sum_{q=0}^{\min\{r,r'\}}\sum_{S\in \mathcal{S}_{R, R'}^{(q)}}c_{S}\binom{b}{s}E_*\Big[U_{S}(\mathbb{G}^{*}_{b})\Big] \\ \stackrel{\text{\eqref{eq:exp_subsample}}}{=}&
     \sum_{q=0}^{\min\{r,r'\}}\sum_{S\in \mathcal{S}_{R, R'}^{(q)}}c_{S}\binom{b}{s}U_{S}(G).
    \end{align*}

For the second term, we have
$$E_*\big[U_R(\mathbb{G}^{*}_{b})\big]E_*\big[U_{R'}(\mathbb{G}^{*}_{b})\big]  \stackrel{\text{\eqref{eq:exp_subsample}}}{=} U_R(G)U_{R'}(G).$$ 

Thus,
\begin{equation*}
\begin{aligned}
              \Cov_{*}\Big[U_{R}(\mathbb{G}^{*}_{b}),U_{R'}(\mathbb{G}^{*}_{b})\Big]  = &\binom{b}{r}^{-1}\binom{b}{r'}^{-1} \sum_{q=0}^{\min\{r,r'\}}\sum_{S\in \mathcal{S}_{R, R'}^{(q)}}c_{S}\binom{b}{s}U_{S}(G) - U_{R}(G)U_{R'}(G),
\end{aligned}
         \end{equation*}
which gives \eqref{eq:cov_subsample}. As a special case,  \allowdisplaybreaks
\begin{equation*}
    \begin{aligned}
  \var_{*}\Big[U_{R}(\mathbb{G}^{*}_{b})\Big]    =  &\Big\{\binom{b}{r}^{-2}\sum_{q=0}^{r}\sum_{S\in \mathcal{S}_{R, R}^{(q}}c_S\binom{b}{2r-q}U_{S}(G)\Big\} -  [U_{R}(G)]^2
  \\ = & \binom{b}{r}^{-2}\sum_{q=0}^{r}\sum_{S\in \mathcal{S}_{R, R}^{(q}}c_S \dfrac{\binom{b}{2r-q}}{\binom{n}{2r-q}} X_S(G)  -  [U_{R}(G)]^2 
        \\ = & \binom{b}{r}^{-2}\sum_{q=0}^{r}\sum_{S\in \mathcal{S}_{R, R}^{(q}}c_S  \dfrac{\binom{n-2r+q}{b-2r+q}}{\binom{n}{b}} X_S(G) -  [U_{R}(G)]^2
\\ = &   \binom{b}{r}^{-2}\Big[\sum_{S\in \mathcal{S}_{R, R}^{(0)}}c_S  \dfrac{\binom{n-2r}{b-2r}}{\binom{n}{b}} X_S(G) +  \sum_{q=1}^{r}\sum_{S\in \mathcal{S}_{R, R}^{(q}}c_S  \dfrac{\binom{n-2r+q}{b-2r+q}}{\binom{n}{b}} X_S(G)\Big] -  [U_{R}(G)]^2,
           \end{aligned}
         \end{equation*}
matching the result of \cite{bhattacharyya2015supplement}.
\item It remains to   the total covariance in terms of network node subsampling. First, it holds that $\Cov\{E_*[U_R(\mathbb{G}^*_b)], E_*[U_{R'}(\mathbb{G}^*_b)]\}=\Cov[U_R(\mathbb{G}_n), U_{R'}(\mathbb{G}_n)]
=E[U_R(\mathbb{G}_n)U_{R'}(\mathbb{G}_n)]-E[U_R(\mathbb{G}_n)]E[U_{R'}(\mathbb{G}_n)]$, where the first equality follows \eqref{eq:exp_subsample}. 

Second, we have \allowdisplaybreaks 
\begin{align*}
     & E\Big\{\Cov_{*}\big[U_{R}(\mathbb{G}^{*}_{b}),U_{R'}(\mathbb{G}^{*}_{b})\big]\Big\}  \\
 \stackrel{\text{\eqref{eq:cov_subsample}}}{=}  & E\bigg\{\binom{b}{r}^{-1}\binom{b}{r'}^{-1} \sum_{q=0}^{\min\{r,r'\}}\sum_{S\in \mathcal{S}_{R, R'}^{(q)}}c_{S}\binom{b}{s}U_{S}(\mathbb{G}_n)- U_{R}(\mathbb{G}_n)U_{R'}(\mathbb{G}_n)\bigg\}
 \\ =&\binom{b}{r}^{-1}\binom{b}{r'}^{-1} \sum_{q=0}^{\min\{r,r'\}}\sum_{S\in \mathcal{S}_{R, R'}^{(q)}}c_{S}\binom{b}{s}E[U_{S}(\mathbb{G}_n)] -  E\big[U_R(\mathbb{G}_n)U_{R'}(\mathbb{G}_n) \big] .\end{align*}
Thus,  \allowdisplaybreaks 
\begin{align*}
    & \Cov\Big\{E_*\big[U_R(\mathbb{G}^*_b)\big], E_*\big[U_{R'}(\mathbb{G}^*_b)\big]\Big\}  +   E\Big\{\Cov_{*}\big[U_{R}(\mathbb{G}^{*}_{b}),U_{R'}(\mathbb{G}^{*}_{b})\big]\Big\}\\
    = & \binom{b}{r}^{-1}\binom{b}{r'}^{-1} \sum_{q=0}^{\min\{r,r'\}}\sum_{S\in \mathcal{S}_{R, R'}^{(q)}}c_{S}\binom{b}{s}E[U_{S}(\mathbb{G}_n)] -  E\Big[U_R(\mathbb{G}_n) \Big]E\Big[ U_{R'}(\mathbb{G}_n) \Big]
       \\  \stackrel{\text{\eqref{eq:U_Gbn_Gn}}}{=} & \binom{b}{r}^{-1}\binom{b}{r'}^{-1} \sum_{q=0}^{\min\{r,r'\}}\sum_{S\in \mathcal{S}_{R, R'}^{(q)}}c_{S}\binom{b}{s}E[U_{S}(\mathbb{G}^{h_n}_b)] -  E\Big[U_R(\mathbb{G}^{h_n}_b) \Big]E\Big[ U_{R'}(\mathbb{G}^{h_n}_b) \Big]\\
        \stackrel{\text{\eqref{eq:ref}}}{=} & \binom{b}{r}^{-1}\binom{b}{r'}^{-1} \sum_{q=0}^{\min\{r,r'\}}\sum_{S\in \mathcal{S}_{R, R'}^{(q)}}c_{S}E[X_{S}(\mathbb{G}^{h_n}_b)] -  E\Big[U_R(\mathbb{G}^{h_n}_b) \Big]E\Big[ U_{R'}(\mathbb{G}^{h_n}_b) \Big]
       \\
       = &E\Big[ \binom{b}{r}^{-1}\binom{b}{r'}^{-1} \sum_{q=0}^{\min\{r,r'\}}\sum_{S\in \mathcal{S}_{R, R'}^{(q)}}c_{S}X_{S}(\mathbb{G}^{h_n}_b)\Big] -  E\Big[U_R(\mathbb{G}^{h_n}_b) \Big]E\Big[ U_{R'}(\mathbb{G}^{h_n}_b) \Big]
            \\  \stackrel{\text{\eqref{eq:linerityMC}}}{=} & E\Big[ \binom{b}{r}^{-1}\binom{b}{r'}^{-1} X_R(\mathbb{G}^{h_n}_b) X_{R'}(\mathbb{G}^{h_n}_b) \Big] -  E\Big[U_R(\mathbb{G}^{h_n}_b) \Big]E\Big[ U_{R'}(\mathbb{G}^{h_n}_b) \Big]\\
            = & \Cov\Big[U_R(\mathbb{G}^{h_n}_b), U_{R'}(\mathbb{G}^{h_n}_b)\Big],
\end{align*}
which gives \eqref{eq:cov_subsample_conditional}.
\end{enumerate}
\end{proof}

%% file: Appendix/Appendix_SuppleA4.tex
\subsection{Proof of Proposition \ref{prop:stats_prop_finite_U_network}}

Let $T$ denote a general finite population U-statistic. The following Hoeffding's decomposition represents $T$ as the sum of mutually uncorrelated U-statistics of increasing order:
\begin{equation}
\label{eq:hoef_fu}
T=\mE_*(T)+\sum_{1 \leqslant i \leqslant b} g_1\left(\mathbb{V}_i\right)+\sum_{1 \leqslant  i<j \leqslant b} g_2\left(\mathbb{V}_i, \mathbb{V}_j\right)+\cdots.
\end{equation}

\cite{bloznelis2001orthogonal,bloznelis2002edgeworth} showed that this decomposition is unique and orthogonal, which implies that $\{g_i\}_{i = 1}^b$ are centered and satisfy
\begin{equation}
\label{eq:kernal_g}
\mE_*[g_i(\mathbb{V}_1,\cdots,\mathbb{V}_i)\mid \mathbb{V}_1,\cdots, \mathbb{V}_{i-1}] = 0.
\end{equation}

Additionally, \cite{bloznelis2001orthogonal} (see their  Equation (2.3)) also showed that
\begin{equation}
\label{eq:g_H}
    g_1\left(\mathbb{V}_1\right)=\frac{n-1}{n-b} h_1\left(\mathbb{V}_1\right), 
\end{equation}
where $h_1\left(\mathbb{V}_1\right)=\mE_*[T- \mE_*(T) \mid \mathbb{V}_1]$.

Since the network $G$ can be treated as a population $G= \left\{\mathfrak{v}_1, \cdots, \mathfrak{v}_n\right\}$,  the subsampled network $\mathbb{G}_b^*$ is uniquely determined by a random sample $\left\{\mathbb{V}_1, \cdots, \mathbb{V}_b\right\}$. Thus, $U_{R}(\mathbb{G}^{*}_{b})$ is a statistic based on $\mathbb{G}_b^*$, and is invariant of its permutation. Thus, $U_{R}(\mathbb{G}^{*}_{b})$ is a finite population U-statistic by definition \ref{def:finiteU}.  We next present the following auxiliary lemma, whose proof is given in Section \ref{C1}.
\begin{lemma}
\label{cor:cov_subsample_v1}
For any motifs $R$ and $R'$,
\begin{equation}
\label{eq:cor_cov}
    \begin{aligned}
        \Cov_{\mathbb{V}_1*}\Big[\sum_{\mathcal{G}\in \mathcal{S}(\mathbb{G}^{\mathbb{V}_1*}_{r})}  X_R(\mathcal{G}),\sum_{\mathcal{G}\in \mathcal{S}(\mathbb{G}^{\mathbb{V}_1*}_{r'})}  X_{R'}(\mathcal{G})\Big] = \sum_{q=0}^{\min\{r,r'\}} \sum_{S\in \mathcal{S}_{R, R'}^{(q)}} c_S\dfrac{nq-rr'}{n^2}X_S(G)   
    \end{aligned}
\end{equation}
\end{lemma}

Now we start to prove Proposition \ref{prop:stats_prop_finite_U_network}.

\begin{proof}[of Proposition \ref{prop:stats_prop_finite_U_network}]
We start by proving the results in part \ref{theo:subsample_distribution_a}.
\begin{enumerate}[label={\ref{theo:subsample_distribution_a}.\roman*}]
    \item We first show  \eqref{eq:g1}. From \eqref{eq:g_H} we have  
  \begin{equation}
    \label{eq:lem5}
    \begin{aligned}
      h_{1,R}\left(\mathbb{V}_1\right)&=E_*\big[U_{R}(\mathbb{G}^{*}_{b})- E_* [U_{R}(\mathbb{G}^{*}_{b})] \mid \mathbb{V}_1\big] = E_*\big[U_{R}(\mathbb{G}^{*}_{b}) \mid \mathbb{V}_1\big] - E_*\big[U_{R}(\mathbb{G}^{*}_{b})\big], \\
g_{1,R}\left(\mathbb{V}_1\right)&=\frac{n-1}{n-b} h_{1,R}\left(\mathbb{V}_1\right).
    \end{aligned}
\end{equation}

We focus on $h_{1,R}\left(\mathbb{V}_1\right)$ first. Recall that $\mathbb{G}^{v*}_{b}$ denotes a random induced graph of $G$ with node $v$ and other $b-1$ random nodes drawn without replacement from $V(G)\setminus v$. Thus,
\begin{equation}
\label{countcondi}
    \begin{aligned}
       U_{R}(\mathbb{G}^{*}_{b})\mid (\mathbb{V}_1 = \mathfrak{v}_1) =&  U_{R}(\mathbb{G}^{v_1*}_{b})
        = \binom{b}{r}^{-1}X_R(\mathbb{G}^{v_1*}_{b}) 
       \\\stackrel{\text{\eqref{eq:ref}}}{=} &\binom{b}{r}^{-1} \Big|\{S: S \subset \mathbb{G}^{v_1*}_{b}, S \cong  R\}\Big|.
        \end{aligned} 
\end{equation}

Next, we partition $\{S: S \subset \mathbb{G}^{v_1*}_{b}, S \cong  R\}$ by $\{S: S \subset \mathbb{G}^{v_1*}_{b},v_1 \notin V(S), S \cong  R\}$ and $\{S: S \subset \mathbb{G}^{v_1*}_{b}, v_1 \in V(S), S \cong  R\}$, 
which leads to the decomposition of $ U_{R}(\mathbb{G}^{*}_{b}) \mid (\mathbb{V}_1 = \mathfrak{v}_1)$ as
\begin{equation}
\label{eq:motif_v1_ref}
    \begin{aligned}
           U_{R}(\mathbb{G}^{*}_{b}) \mid \mathfrak{v}_1 =& \binom{b}{r}^{-1} \Big|\{S: S \subset \mathbb{G}^{v_1*}_{b}, S \cong  R\}\Big|
             \\=&\binom{b}{r}^{-1} \left\{\Big|\{S: S \subset \mathbb{G}^{v_1*}_{b},v_1 \notin V(S),, S \cong  R\}\Big| + \Big|\{S: S \subset \mathbb{G}^{v_1*}_{b}, v_1 \in V(S), S \cong  R\}\Big|\right\}
           \\=&\binom{b}{r}^{-1} \left\{\Big|\{S: S \subset \mathbb{G}^{v_1*}_{b}\setminus v_1, S \cong  R\}\Big| + \Big|\{S: S \subset \mathbb{G}^{v_1*}_{b}, v_1 \in V(S), S \cong  R\}\Big|\right\}
           \\ \stackrel{\text{\eqref{eq:ref}}}{=} &\binom{b}{r}^{-1}X_R(\mathbb{G}^{v_1*}_{b}\setminus v_1) +  \binom{b}{r}^{-1}\Big|\{S: S \subset \mathbb{G}^{v_1*}_{b}, v_1 \in V(S), S \cong  R\}\Big|
          \\ \stackrel{\text{\eqref{eq:iden2}}}{=} & \binom{b}{r}^{-1}X_R(\mathbb{G}^{v_1*}_{b}\setminus v_1)+  \binom{b}{r}^{-1}\sum_{\mathcal{G}\in \mathcal{S}(\mathbb{G}^{v_1**}_{b,r})} X_R(\mathcal{G})
          \\ := &~ \mathrm{I}  + \mathrm{II},
    \end{aligned}
\end{equation}
where $\mathbb{G}^{v_1*}_{b}\setminus v_1$ is a  a randomly induced subsampled graph based on $b -1$ nodes that are randomly drawn without replacement from $V(G) \setminus v_1$. Let $G' = G \setminus v_1$ be the network after removing node $v_1$ and all edges involving $v_1$ from $G$. Then $\mathbb{G}^{v_1*}_{b}\setminus v_1$ is essentially a randomly induced graph $\mathbb{G'}_{b-1}^*$. In addition, we use $E_{*\setminus v_1}$ to indicate probability calculations with respect to other $b-1$ random nodes without $v_1$.

\allowdisplaybreaks
In \eqref{eq:motif_v1_ref}, term I admits
\begin{equation*}
    \begin{aligned}
       E_{*\setminus v_1}\big[ \mathrm{I}\big]= &E_{*\setminus v_1}\Big[\binom{b}{r}^{-1}X_R(\mathbb{G}^{v_1*}_{b}\setminus v_1)\Big]
      =  \frac{b-r}{b}E_{*\setminus v_1}\Big[\binom{b-1}{r}^{-1}X_R(\mathbb{G'}_{b-1}^*)\Big]
     \\ =  &\frac{b-r}{b}E_{*\setminus v_1}\Big[U_{R}(\mathbb{G'}_{b-1}^*)\Big]
      \stackrel{\text{ \eqref{eq:exp_subsample}}}{=} \frac{b-r}{b}U_{R}(G') \\= &\frac{b-r}{b}U_{R}(G \setminus v_1) \stackrel{\text{ \eqref{eq:ref}}}{=} \frac{b-r}{b}\binom{n-1}{r}^{-1}X_{R}(G \setminus v_1)
      \\\stackrel{\text{ \eqref{eq:ref}}}{=} &\frac{b-r}{b}\binom{n-1}{r}^{-1}\Big(\Big|\{S: S \subset G, S \cong  R\}\Big| - \Big|\{S: S \subset G, v_1 \in V(S), S \cong  R\}\Big|\Big)\\
      = &\frac{b-r}{b}\binom{n-1}{r}^{-1}\Big(X_R(G ) -\Big|\{S: S \subset G, v_1 \in V(S), S \cong  R\}\Big|\Big)
      \\ \stackrel{\text{\eqref{eq:iden3}}}{=} &\frac{b-r}{b}\binom{n-1}{r}^{-1}\Big(X_R(G ) - \sum_{\mathcal{G}\in \mathcal{S}(\mathbb{G}^{v_1*}_{r})}  X_R(\mathcal{G})\Big).
    \end{aligned}
\end{equation*}
\allowdisplaybreaks

For term II, it holds that
\begin{equation*}
    \begin{aligned}
        E_{*\setminus v_1}\big(\mathrm{II}\big)= &E_{*\setminus v_1}\Big[ \binom{b}{r}^{-1}\sum_{\mathcal{G}\in \mathcal{S}(\mathbb{G}^{v_1**}_{b,r})} X_R(\mathcal{G})\Big]
        = \binom{b}{r}^{-1}\binom{n-1}{b-1}^{-1}\sum_{G_b^* \in \mathcal{S}(\mathbb{G}_{b}^{v_1*})}\Big(\sum_{\mathcal{G} \in \mathcal{S}(\mathbb{G}^{v_1**}_{b,r})} X_R(\mathcal{G})\Big) 
       \\\stackrel{\text{\eqref{eq:iden4}}}{=} &\binom{b}{r}^{-1}\binom{n-1}{b-1}^{-1}\binom{n-r}{b-r}\Big|\{S: S \subset G, v_1 \in V(S), S \cong R\}\Big|
       \\\stackrel{\text{\eqref{eq:iden3}}}{=} &\frac{r!(n-r)!}{b(n-1)!} \sum_{\mathcal{G} \in \mathcal{S}(\mathbb{G}^{v_1*}_{r})}  X_R( \mathcal{G}).
    \end{aligned}
\end{equation*}

\allowdisplaybreaks
Putting these two parts together, we have
\begin{equation*}
    \begin{aligned}
         &E_{*} \big[ U_{R}(\mathbb{G}^{*}_{b})\mid \mathfrak{v}_1\big] =E_{*\setminus v_1} \big[  U_{R}(\mathbb{G}^{*}_{b}) \mid \mathfrak{v}_1 \big] = E_{*\setminus v_1} \big[ \mathrm{I}  + \mathrm{II} \big] 
          \\
          =&~ \frac{b-r}{b}\binom{n-1}{r}^{-1}\Big[X_R(G ) - \sum_{\mathcal{G}\in \mathcal{S}(\mathbb{G}^{v_1*}_{r})}  X_R(\mathcal{G})\Big] + \frac{r!(n-r!)}{b(n-1)!} \sum_{\mathcal{G}\in \mathcal{S}(\mathbb{G}^{v_1*}_{r})}  X_R(\mathcal{G})\\
          =&~ \frac{(b-r)n}{b(n-r)}\binom{n}{r}^{-1}X_R(G ) - \frac{(b-r)r!(n-r-1)}{b(n-1)!}\sum_{\mathcal{G}\in \mathcal{S}(\mathbb{G}^{v_1*}_{r})}  X_R(\mathcal{G}) \\&+  \frac{(n-r)r!(n-r-1)!}{b(n-1)!} \sum_{\mathcal{G}\in \mathcal{S}(\mathbb{G}^{v_1*}_{r})}  X_R(\mathcal{G})
          \\
          =&~\frac{(b-r)n}{b(n-r)}U_{R}(G )+\frac{r!(n-r-1)!(n-b)}{b(n-1)!}\sum_{\mathcal{G}\in \mathcal{S}(\mathbb{G}^{v_1*}_{r})}  X_R(\mathcal{G}).
    \end{aligned}
\end{equation*}

On the other hand, \eqref{eq:exp_subsample} indicates that
$E_*\big[U_{R}(\mathbb{G}^{*}_{b}) \big]  = U_{R}(G)$. Therefore,  
\begin{equation*}
    \begin{aligned}
h_{1,R}\left(\mathbb{V}_1\right)\stackrel{\text{\eqref{eq:lem5}}}{=}&E_*\big[U_{R}(\mathbb{G}^{*}_{b})- E_* [U_{R}(\mathbb{G}^{*}_{b})] \mid \mathbb{V}_1\big] = E_*\big[U_{R}(\mathbb{G}^{*}_{b})\mid \mathbb{V}_1\big] - E_*\big[U_{R}(\mathbb{G}^{*}_{b})\big]
        \\= &\frac{(b-r)n}{b(n-r)}U_{R}(G )+\frac{r!(n-r-1)!(n-b)}{b(n-1)!}\sum_{\mathcal{G}\in \mathcal{S}(\mathbb{G}^{\mathbb{V}_1*}_{r})}  X_R(\mathcal{G}) - U_{R}(G )
        \\= &\frac{r!(n-r-1)!(n-b)}{b(n-1)!}\sum_{\mathcal{G}\in \mathcal{S}(\mathbb{G}^{\mathbb{V}_1*}_{r})}  X_R(\mathcal{G}) - \frac{(n-b)r}{b(n-r)}U_{R}(G),
    \end{aligned}
\end{equation*}
 and 
$$g_{1,R}\left(\mathbb{V}_1\right)\stackrel{\text{\eqref{eq:lem5}}}{=}\frac{n-1}{n-b} h_{1,R}\left(\mathbb{V}_1\right) =\frac{r!(n-r-1)!}{b(n-2)!}\sum_{\mathcal{G}\in \mathcal{S}(\mathbb{G}^{\mathbb{V}_1*}_{r})}  X_R(\mathcal{G}) - \frac{(n-1)r}{b(n-r)}U_{R}(G).$$

The first term satisfies 
\begin{equation*}
    \begin{aligned}
        &\frac{r!(n-r-1)!}{b(n-2)!}\sum_{\mathcal{G}\in \mathcal{S}(\mathbb{G}^{\mathbb{V}_1*}_{r})}  X_{R}(\mathcal{G})
        \\\stackrel{\text{ \eqref{eq:iden3}}}{=} &~\frac{r!(n-r-1)!}{b(n-2)!}\Big[\big|\{S: S \subset G, \mathbb{V}_1 \in V(S), S \cong R\}\big|\Big]
        \\=&~\frac{r!(n-r-1)!}{b(n-2)!}\Big[\big|\{S: S \subset G,  S \cong R\}\big| - \big|\{S: S \subset G, \mathbb{V}_1 \notin V(S), S \cong R\}\big|\Big]
          \\=&~\frac{r!(n-r-1)!}{b(n-2)!}\Big[X_R(G) - X_R(G\setminus \mathbb{V}_1)\Big]
        \\=&~\frac{r!(n-r-1)!}{b(n-2)!}X_R(G) - \frac{r!(n-r-1)!}{b(n-2)!}X_R(G\setminus \mathbb{V}_1).
           \end{aligned}
\end{equation*}

Consequently,
\begin{equation}
\label{eq:assump3ref1}
    \begin{aligned}
        &g_{1,R}(\mathbb{V}_1) = \frac{r!(n-r-1)!}{b(n-2)!}\sum_{\mathcal{G}\in \mathcal{S}(\mathbb{G}^{\mathbb{V}_1*}_{r})}  X_{R}(\mathcal{G}) -  \frac{(n-1)r}{b(n-r)}U_{R}(G)
        \\ = &~\frac{r!(n-r-1)!}{b(n-2)!}X_R(G) - \frac{r!(n-r-1)!}{b(n-2)!}X_R(G\setminus \mathbb{V}_1)  -   \frac{(n-1)r}{b(n-r)}U_{R}(G)
        \\ = &~\frac{n(n-1)}{b(n-r)}U_{R}(G) -  \frac{(n-1)r}{b(n-r)}U_{R}(G)   - \frac{(n-1)r!(n-r-1)!}{b(n-1)!}X_R(G\setminus \mathbb{V}_1) 
          \\ = &~\frac{(n-1)}{b}[U_R(G) - U_R(G\setminus \mathbb{V}_1)],
       \end{aligned}
\end{equation}
which gives \eqref{eq:g1}. 

We now show $g_{1,R}\left(\mathbb{V}_1\right)$ has mean zero.  
\begin{equation}
    \label{eq:expg1}
    \begin{aligned}
E_{*}\big[g_{1,R}\left(\mathbb{V}_1\right)\big] =& E_{\mathbb{V}_1*}\big[g_{1,R}\left(\mathbb{V}_1\right)\big] \\=& E_{\mathbb{V}_1*}\big[\frac{r!(n-r-1)!}{b(n-2)!}\sum_{\mathcal{G}\in \mathcal{S}(\mathbb{G}^{\mathbb{V}_1*}_{r})}  X_R(\mathcal{G}) - \frac{(n-1)r}{b(n-r)}U_{R}(G)\big]\\
       =& \frac{r!(n-r-1)!}{b(n-2)!}E_{\mathbb{V}_1*}\big[\sum_{\mathcal{G}\in \mathcal{S}(\mathbb{G}^{\mathbb{V}_1*}_{r})}  X_R(\mathcal{G})\big] - \frac{(n-1)r}{b(n-r)}U_{R}(G)\\
       =& \frac{r!(n-r-1)!}{nb(n-2)!}\sum\limits_{i=1}^{n}\big[\sum_{\mathcal{G}\in \mathcal{S}(\mathbb{G}^{v_i*}_{r})}  X_R(\mathcal{G})\big] - \frac{(n-1)r}{b(n-r)}U_{R}(G)\\
       \stackrel{\text{\eqref{eq:iden5}}}{=}& \frac{r!(n-r-1)!}{nb(n-2)!}rX_R(G)- \frac{(n-1)r}{b(n-r)}U_{R}(G) \\=& \frac{r!(n-r-1)!}{nb(n-2)!}\frac{rn!}{r!(n-r)!}U_{R}(G) - \frac{(n-1)r}{b(n-r)}U_{R}(G)\\
       =& \frac{(n-1)r}{b(n-r)}U_{R}(G) - \frac{(n-1)r}{b(n-r)}U_{R}(G) = 0.
    \end{aligned}
\end{equation}

\item Now we proceed to prove \eqref{eq:varg1}. We use $\var_{\mathbb{V}_1*}$ and $\Cov_{\mathbb{V}_1*}$ to indicate probability calculations with respect to random $\mathbb{V}_1$. Because the randomness in  $\Cov_{*}[ g_{1,R}(\mathbb{V}_1) , g_{1,R'}(\mathbb{V}_1)]$ is from the random node $\mathbb{V}_1$,  we have
   $\Cov_{*}\big[ g_{1,R}(\mathbb{V}_1) , g_{1,R'}(\mathbb{V}_1)\big] = \Cov_{\mathbb{V}_1*}\big[ g_{1,R}(\mathbb{V}_1) , g_{1,R'}(\mathbb{V}_1)\big]$.
Thanks to Lemma \ref{cor:cov_subsample_v1}, we have \allowdisplaybreaks
\begin{align}
\label{eq:cov_v1_g1rr}
   \nonumber  &  \Cov_{\mathbb{V}_1*}\Big[ g_{1,R}(\mathbb{V}_1) , g_{1,R'}(\mathbb{V}_1)\Big] \\ \stackrel{\text{\eqref{eq:g1}}}{=} &\Cov_{\mathbb{V}_1*}\Big[\frac{r!(n-r-1)!}{b(n-2)!}\sum_{\mathcal{G}\in \mathcal{S}(\mathbb{G}^{\mathbb{V}_1*}_{r})}  X_R(\mathcal{G}) , \frac{r'!(n-r'-1)!}{b(n-2)!}\sum_{\mathcal{G}\in \mathcal{S}(\mathbb{G}^{\mathbb{V}_1*}_{r'})}  X_{R'}(\mathcal{G})\Big]\\
      \nonumber    = & \frac{r!(n-r-1)!}{b(n-2)!} \frac{r'!(n-r'-1)!}{b(n-2)!}\Cov_{\mathbb{V}_1*}\Big[\sum_{\mathcal{G}\in \mathcal{S}(\mathbb{G}^{\mathbb{V}_1*}_{r})}  X_R(\mathcal{G}) ,\sum_{\mathcal{G}\in \mathcal{S}(\mathbb{G}^{\mathbb{V}_1*}_{r'})}  X_{R'}(\mathcal{G})\Big]\\
\stackrel{\text{\eqref{eq:cor_cov}}}{=} &\frac{r!(n-r-1)!}{b(n-2)!} \frac{r'!(n-r'-1)!}{b(n-2)!}\sum_{q=0}^{\min\{r,r'\}} \sum_{S\in \mathcal{S}_{R, R'}^{(q)}} c_S\dfrac{nq-rr'}{n^2}X_S(G).
    \end{align} 

Thus, \eqref{eq:varg1} follows. As a special case, $\var_{\mathbb{V}_1*}\left[g_{1,R}\left(\mathbb{V}_1\right)\right]  = \{[r !(n-r-1) !]/[b(n-2) !]\}^2 \Big[\sum_{q=0}^{r}\sum_{S\in S_{R, R}^{(k)}}c_S (nk-r^2)n^{-2}X_S(G)\Big]$.
  


\item Now we continue to prove \eqref{eq:var_sumg1}. Recall that the subscript in $\mathbb{V}_i*$ indicates the randomness from  random node $\mathbb{V}_i$, and the subscripts in $\mathbb{V}_i*$ and $\mathbb{V}_j*$ indicate that the randomness are from random nodes $\mathbb{V}_i$ and $\mathbb{V}_j$. Notice that $(n-1)rU_{R}(G)/b(n-r)$ and $(n-1)r'U_{R'}(G)/b(n-r')$ are two constants when $G$ is  given. We first decompose the covariance by \allowdisplaybreaks
\begin{align}
\label{eq:break_cov_sum_g}
   \nonumber 
    & \Cov_*\Big[\sum_{1\leqslant i \leqslant b} g_{1,R}(\mathbb{V}_i),
      \sum_{1\leqslant i \leqslant b} g_{1,R'}(\mathbb{V}_i)\Big]\\
 \nonumber   =&\sum_{i=1}^{b} \sum_{j=1}^{b}  \Cov_{\mathbb{V}_i,\mathbb{V}_j*}\Big[g_{1,R}(\mathbb{V}_i),
     g_{1,R'}(\mathbb{V}_j)\Big] \\
    \nonumber   \stackrel{\text{\eqref{eq:g1}}}{=} &\dfrac{r !(n-r-1)!}{b(n-2)!}\dfrac{r'!(n-r'-1)!}{b(n-2)!} \sum_{i=1}^{b} \sum_{j=1}^{b} \Cov_{\mathbb{V}_i,\mathbb{V}_j*}\Big[\sum_{\mathcal{G}\in \mathcal{S}(\mathbb{G}^{\mathbb{V}_i*}_{r})}  X_R(\mathcal{G}),\sum_{\mathcal{G}\in \mathcal{S}(\mathbb{G}^{\mathbb{V}_j*}_{r'})}  X_{R'}(\mathcal{G}) \Big]
  \\    \nonumber  = & \dfrac{r !(n-r-1)!}{b(n-2)!}\dfrac{r'!(n-r'-1)!}{b(n-2)!} \sum_{i=1}^{b} \bigg\{\Cov_{\mathbb{V}_i*}\Big[\sum_{\mathcal{G}\in \mathcal{S}(\mathbb{G}^{\mathbb{V}_i*}_{r})}  X_R(\mathcal{G}),\sum_{\mathcal{G}\in \mathcal{S}(\mathbb{G}^{\mathbb{V}_i*}_{r'})}  X_{R'}(\mathcal{G}) \Big] 
   \nonumber    \\
      \nonumber  &+ \sum_{j=1,j\neq i}^{b} \Cov_{\mathbb{V}_i,\mathbb{V}_j*}\Big[\sum_{\mathcal{G}\in \mathcal{S}(\mathbb{G}^{\mathbb{V}_i*}_{r})}  X_R(\mathcal{G}),\sum_{\mathcal{G}\in \mathcal{S}(\mathbb{G}^{\mathbb{V}_j*}_{r'})}  X_{R'}(\mathcal{G}) \Big]\bigg\} \\
    := & \dfrac{r !(n-r-1)!}{b(n-2)!}\dfrac{r'!(n-r'-1)!}{b(n-2)!} \sum_{i=1}^{b} \Big(\mathrm{I}(i) + \mathrm{II}(i)\Big).
  \end{align} 
 
Furthermore, \allowdisplaybreaks
\begin{equation}
\label{eq:cov_sum_g_vi}
    \begin{aligned}
       \mathrm{I}(i) &= \Cov_{\mathbb{V}_i*}\Big[\sum_{\mathcal{G}\in \mathcal{S}(\mathbb{G}^{\mathbb{V}_i*}_{r})}  X_R(\mathcal{G}),\sum_{\mathcal{G}\in \mathcal{S}(\mathbb{G}^{\mathbb{V}_i*}_{r'})}  X_{R'}(\mathcal{G}) \Big]
        \\&=  E_{\mathbb{V}_i*}\Big[\sum_{\mathcal{G}\in \mathcal{S}(\mathbb{G}^{\mathbb{V}_i*}_{r})}  X_R(\mathcal{G})\sum_{\mathcal{G}\in \mathcal{S}(\mathbb{G}^{\mathbb{V}_i*}_{r'})}  X_{R'}(\mathcal{G})\Big] - E_{\mathbb{V}_i*}\Big[\sum_{\mathcal{G}\in \mathcal{S}(\mathbb{G}^{\mathbb{V}_i*}_{r})}  X_R(\mathcal{G})\Big]E_{\mathbb{V}_i*}\Big[\sum_{\mathcal{G}\in \mathcal{S}(\mathbb{G}^{\mathbb{V}_i*}_{r'})}  X_{R'}(\mathcal{G})\Big]\\ 
     &= \dfrac{1}{n} \sum_{a=1}^{n} \Big[\sum_{\mathcal{G}\in \mathcal{S}(\mathbb{G}^{v_a*}_{r})}  X_R(\mathcal{G})\sum_{\mathcal{G}\in \mathcal{S}(\mathbb{G}^{v_a*}_{r'})}  X_{R'}(\mathcal{G})\Big] - \Big[\dfrac{1}{n} \sum_{a=1}^{n}\sum_{\mathcal{G}\in \mathcal{S}(\mathbb{G}^{v_a*}_{r})}  X_R(\mathcal{G})\Big]\Big[\dfrac{1}{n} \sum_{a=1}^{n}\sum_{\mathcal{G}\in \mathcal{S}(\mathbb{G}^{v_a*}_{r'})}  X_{R'}(\mathcal{G})\Big]\\ 
    &\hspace{-0.08cm}\stackrel{\text{\eqref{eq:iden5}}}{=}  \dfrac{1}{n} \sum_{a=1}^{n} \Big[\sum_{\mathcal{G}\in \mathcal{S}(\mathbb{G}^{v_a*}_{r})}  X_R(\mathcal{G})\sum_{\mathcal{G}\in \mathcal{S}(\mathbb{G}^{v_a*}_{r'})}  X_{R'}(\mathcal{G})\Big] - \frac{ rX_R(G)}{n}\frac{ r'X_{R'}(G)}{n}.
    \end{aligned}
\end{equation} For part II$(i)$, first we have, 
\begin{align}\allowdisplaybreaks
\label{eq:cov_sum_x_vi_vj}
\nonumber &\Cov_{\mathbb{V}_i,\mathbb{V}_j*}\Big[\sum_{\mathcal{G}\in \mathcal{S}(\mathbb{G}^{\mathbb{V}_i*}_{r})}  X_R(\mathcal{G}),\sum_{\mathcal{G}\in \mathcal{S}(\mathbb{G}^{\mathbb{V}_j*}_{r'})}  X_{R'}(\mathcal{G}) \Big]
  \\ \nonumber =&  E_{\mathbb{V}_i,\mathbb{V}_j*}\Big[\sum_{\mathcal{G}\in \mathcal{S}(\mathbb{G}^{\mathbb{V}_i*}_{r})}  X_R(\mathcal{G})\sum_{\mathcal{G}\in \mathcal{S}(\mathbb{G}^{\mathbb{V}_j*}_{r'})}  X_{R'}(\mathcal{G})\Big] - E_{\mathbb{V}_i*}\Big[\sum_{\mathcal{G}\in \mathcal{S}(\mathbb{G}^{\mathbb{V}_i*}_{r})}  X_R(\mathcal{G})\Big]E_{\mathbb{V}_j*}\Big[\sum_{\mathcal{G}\in \mathcal{S}(\mathbb{G}^{\mathbb{V}_j*}_{r'})}  X_{R'}(\mathcal{G})\Big]
    \\ \nonumber =&  E_{\mathbb{V}_i,\mathbb{V}_j*}\Big[\sum_{\mathcal{G}\in \mathcal{S}(\mathbb{G}^{\mathbb{V}_i*}_{r})}  X_R(\mathcal{G})\sum_{\mathcal{G}\in \mathcal{S}(\mathbb{G}^{\mathbb{V}_j*}_{r'})}  X_{R'}(\mathcal{G})\Big] -  \Big[\dfrac{1}{n} \sum_{k=1}^{n}\sum_{\mathcal{G}\in \mathcal{S}(\mathbb{G}^{v_k*}_{r})}  X_R(\mathcal{G})\Big]\Big[\dfrac{1}{n} \sum_{k=1}^{n}\sum_{\mathcal{G}\in \mathcal{S}(\mathbb{G}^{v_k*}_{r'})}  X_{R'}(\mathcal{G})\Big] 
    \\ \nonumber & \hspace{-0.43cm}\stackrel{\text{\eqref{eq:iden5}}}{=} \dfrac{1}{n(n-1)}\sum_{a = 1}^{n} \sum_{b = 1, b\neq a}^{n} \Big[\sum_{\mathcal{G}\in \mathcal{S}(\mathbb{G}^{v_a*}_{r})}  X_R(\mathcal{G})\sum_{\mathcal{G}\in \mathcal{S}(\mathbb{G}^{v_b*}_{r'})}  X_{R'}(\mathcal{G})\Big]  - \dfrac{rX_R(G) r'X_{R'}(G) }{n^2}\\ \nonumber 
=&  \dfrac{1}{n(n-1)}\sum_{a = 1}^{n} \sum_{\mathcal{G}\in \mathcal{S}(\mathbb{G}^{v_a*}_{r})}  X_R(\mathcal{G}) \Big[\sum_{b = 1, b\neq a}^{n}\sum_{\mathcal{G}\in \mathcal{S}(\mathbb{G}^{v_b*}_{r'})}  X_{R'}(\mathcal{G})\Big] - \dfrac{rX_R(G) r'X_{R'}(G) }{n^2}\\ \nonumber 
=&  \dfrac{1}{n(n-1)}\sum_{a = 1}^{n} \sum_{\mathcal{G}\in \mathcal{S}(\mathbb{G}^{v_a*}_{r})}  X_R(\mathcal{G}) \Big[\sum_{b = 1}^{n}\sum_{\mathcal{G}\in \mathcal{S}(\mathbb{G}^{v_b*}_{r'})}  X_{R'}(\mathcal{G}) - \sum_{\mathcal{G}\in \mathcal{S}(\mathbb{G}^{v_a*}_{r'})}  X_{R'}(\mathcal{G})\Big]  - \dfrac{rX_R(G) r'X_{R'}(G) }{n^2}\\ \nonumber 
=&\dfrac{1}{n(n-1)}\sum_{a = 1}^{n} \sum_{\mathcal{G}\in \mathcal{S}(\mathbb{G}^{v_a*}_{r})}  X_R(\mathcal{G}) \Big[ rX(R',G) - \sum_{\mathcal{G}\in \mathcal{S}(\mathbb{G}^{v_a*}_{r'})}  X_{R'}(\mathcal{G}) \Big]  - \dfrac{rX_R(G) r'X_{R'}(G) }{n^2}\\ \nonumber 
=&\dfrac{ rX_{R'}(G) }{n(n-1)}\sum_{a = 1}^{n} \sum_{\mathcal{G}\in \mathcal{S}(\mathbb{G}^{v_a*}_{r})}  X_R(\mathcal{G})- \dfrac{1}{n(n-1)}\sum_{a = 1}^{n} \sum_{\mathcal{G}\in \mathcal{S}(\mathbb{G}^{v_a*}_{r})}  X_R(\mathcal{G})\sum_{\mathcal{G}\in \mathcal{S}(\mathbb{G}^{v_a*}_{r'})}  X_{R'}(\mathcal{G})  - \dfrac{rX_R(G) r'X_{R'}(G) }{n^2}\\
&\hspace{-0.43cm}\stackrel{\text{\eqref{eq:iden5}}}{=}\dfrac{rX_R(G) r'X_{R'}(G) }{n^2(n-1)}- \dfrac{1}{n(n-1)}\sum_{a = 1}^{n} \sum_{\mathcal{G}\in \mathcal{S}(\mathbb{G}^{v_a*}_{r})}  X_R(\mathcal{G})\sum_{\mathcal{G}\in \mathcal{S}(\mathbb{G}^{v_a*}_{r'})}  X_{R'}(\mathcal{G}).
\end{align}

As a result, we have \allowdisplaybreaks
\begin{align*} 
          \mathrm{II}(i) &= \sum_{j=1,j\neq i}^{b} \Cov_{\mathbb{V}_i,\mathbb{V}_j*}\Big[\sum_{\mathcal{G}\in \mathcal{S}(\mathbb{G}^{\mathbb{V}_i*}_{r})}  X_R(\mathcal{G}),\sum_{\mathcal{G}\in \mathcal{S}(\mathbb{G}^{\mathbb{V}_j*}_{r'})}  X_{R'}(\mathcal{G}) \Big]\\&\hspace{-0.08cm}\stackrel{\text{\eqref{eq:cov_sum_x_vi_vj}}}{=}\sum_{j=1,j\neq i}^{b} \Big\{\dfrac{rX_R(G) r'X(R',G) }{n^2(n-1)}- \dfrac{1}{n(n-1)}\sum_{a = 1}^{n} \sum_{\mathcal{G}\in \mathcal{S}(\mathbb{G}^{v_a*}_{r})}  X_R(\mathcal{G})\sum_{\mathcal{G}\in \mathcal{S}(\mathbb{G}^{v_a*}_{r'})}  X_{R'}(\mathcal{G})\Big\} \\&=
          (b-1)\Big\{\dfrac{rX_R(G) r'X(R',G) }{n^2(n-1)}- \dfrac{1}{n(n-1)}\sum_{a = 1}^{n} \sum_{\mathcal{G}\in \mathcal{S}(\mathbb{G}^{v_a*}_{r})}  X_R(\mathcal{G})\sum_{\mathcal{G}\in \mathcal{S}(\mathbb{G}^{v_a*}_{r'})}  X_{R'}(\mathcal{G})\Big\}.
\end{align*}
By adding I$(i)$ and II$(i)$ following previous results, we have
\begin{align*}
     \mathrm{I}(i) + \mathrm{II}(i)
    &= \dfrac{n-b}{n-1} \Big\{\dfrac{1}{n}\sum_{a = 1}^{n} \sum_{\mathcal{G}\in \mathcal{S}(\mathbb{G}^{v_a*}_{r})}  X_R(\mathcal{G})\sum_{\mathcal{G}\in \mathcal{S}(\mathbb{G}^{v_a*}_{r'})}  X_{R'}(\mathcal{G}) - \dfrac{rX_R(G) r'X_{R'}(G) }{n^2}\Big\}\\
    &\stackrel{\text{\eqref{eq:cov_sum_g_vi}}}{=} \dfrac{n-b}{n-1} \Cov\Big[\sum_{\mathcal{G}\in \mathcal{S}(\mathbb{G}^{\mathbb{V}_i*}_{r})}  X_R(\mathcal{G}),\sum_{\mathcal{G}\in \mathcal{S}(\mathbb{G}^{\mathbb{V}_i*}_{r'})}  X_{R'}(\mathcal{G})\Big].
\end{align*}
\allowdisplaybreaks
Thus,
\begin{align}\label{eq:cov_sum_g1_more}
 & \Cov_*\Big[\sum_{1\leqslant i \leqslant b} g_{1,R}(\mathbb{V}_i),
      \sum_{1\leqslant i \leqslant b} g_{1,R'}(\mathbb{V}_i)\Big] \\
   \nonumber   \stackrel{\text{\eqref{eq:break_cov_sum_g}}}{=} & \dfrac{r !(n-r-1)!}{b(n-2)!}\dfrac{r'!(n-r'-1)!}{b(n-2)!} \sum_{i=1}^{b} \Big(\mathrm{I}(i) + \mathrm{II}(i)\Big)
      \\ \nonumber  = &
      \dfrac{n-b}{n-1}\dfrac{r !(n-r-1)!}{b(n-2)!}\dfrac{r'!(n-r'-1)!}{b(n-2)!} \sum_{i=1}^{b} \Cov\Big[\sum_{\mathcal{G}\in \mathcal{S}(\mathbb{G}^{\mathbb{V}_i*}_{r})}  X_R(\mathcal{G}),\sum_{\mathcal{G}\in \mathcal{S}(\mathbb{G}^{\mathbb{V}_i*}_{r'})}  X_{R'}(\mathcal{G}) \Big]\\ \nonumber  = &\dfrac{b(n-b)}{n-1}\dfrac{r !(n-r-1)!}{b(n-2)!}\dfrac{r'!(n-r'-1)!}{b(n-2)!} \Cov\Big[\sum_{\mathcal{G}\in \mathcal{S}(\mathbb{G}^{\mathbb{V}_1*}_{r})}  X_R(\mathcal{G}),\sum_{\mathcal{G}\in \mathcal{S}(\mathbb{G}^{\mathbb{V}_1*}_{r'})}  X_{R'}(\mathcal{G}) \Big]\\  \nonumber 
    \stackrel{\text{\eqref{eq:g1}}}{=} & \dfrac{b(n-b)}{n-1}\Cov_*\Big[g_{1,R}(\mathbb{V}_1),
    g_{1,R'}(\mathbb{V}_1)\Big], 
  \end{align}
which gives \eqref{eq:var_sumg1}.  
\item To show \eqref{eq:lim_var_g1}, we start by decomposing the variance: \allowdisplaybreaks
\begin{align*} 
   & \var_*\sum_{1 \leqslant i \leqslant b} g_{1,R}\left(\mathbb{V}_i\right)   \stackrel{\text{\eqref{eq:var_sumg1}}}{=}  \frac{b(n-b)}{(n-1)}\var_*\Big[g_{1,R}\left(\mathbb{V}_1\right)\Big]
   \\\stackrel{\text{\eqref{eq:varg1}}}{=} &\frac{b(n-b)}{(n-1)}\left(\frac{r !(n-r-1) !}{b(n-2) !}\right)^2 \Big[\sum_{q=0}^{r}\sum_{S\in S_{R,R}^{(q)}}c_S\dfrac{nq-r^2}{n^2}X_S(G)\Big]\\
     = &  \frac{(n-b)}{b(n-1)}\left(\frac{r !(n-r-1) !}{(n-2) !}\right)^2 \Big\{\Big[\sum_{q=0}^{r}\sum_{S\in S_{R,R}^{(q)}}c_S\dfrac{q}{n}X_S(G)\Big] - \Big[\sum_{q=0}^{r}\sum_{S\in S_{R,R}^{(q)}}c_S\dfrac{r^2}{n^2}X_S(G)\Big]\Big\}\\
     = &  \frac{(n-b)}{b(n-1)}\left(\frac{r !(n-r-1) !}{(n-2) !}\right)^2 \Big[\sum_{q=0}^{r}\sum_{S\in\mathcal{S}_{R,R}^{(q)}}c_S\dfrac{q}{n}X_S(G)\Big] \\&-  \frac{(n-b)}{b(n-1)}\left(\frac{r !(n-r-1) !}{(n-2) !}\right)^2\Big[\sum_{q=0}^{r}\sum_{S\in\mathcal{S}_{R,R}^{(q)}}c_S\dfrac{r^2}{n^2}X_S(G)\Big]
     \\ = & \hspace{0.3cm} \mathrm{I} - \mathrm{II}.
                   \end{align*}    
For term I, we have 
\begin{equation*}
    \begin{aligned}
    &\frac{(n-b)}{b(n-1)}\left(\frac{r !(n-r-1) !}{(n-2) !}\right)^2  \Big[\sum_{q=0}^{r}\sum_{S\in\mathcal{S}_{R,R}^{(q)}}c_S\dfrac{q}{n}X_S(G)\Big]\\
      =  & \frac{(n-b)}{b(n-1)}\left(\frac{r(n-1)}{(n-r)} \binom{n-1}{r-1}^{-1}\right)^2\Big[\sum_{q=0}^{r}\sum_{S\in \mathcal{S}_{R,R}^{(q)}}\dfrac{c_Sq}{n}
   \binom{n}{2r-q}U_{S}(G)
   \Big] \\
    =  & \frac{(n-b)}{b(n-1)}\left(\frac{r(n-1)}{(n-r)}\right)^2\Big[\sum_{q=0}^{r}\sum_{S\in \mathcal{S}_{R,R}^{(q)}}\dfrac{c_Sq}{n}
    \binom{n-1}{r-1}^{-2}\binom{n}{2r-q}U_{S}(G)
    \Big] \\
    = & \frac{(n-b)}{b(n-1)}\left(\frac{r(n-1)}{(n-r)}\right)^2 \Big[\sum_{q=1}^{r}\sum_{S\in \mathcal{S}_{R,R}^{(q)}} 
   \dfrac{c_S q}{n} \left(\frac{(r-1) !(n-r) !}{(n-1) !}\right)^2 \dfrac{n!}{(2r-q)!(n-2r+q)} U_{S}(G)\Big] \\
    = & \frac{(n-b)}{b(n-1)} \left(\frac{r(n-1)}{(n-r)}\right)^2 \Big[
        \sum_{q=1}^{r}\sum_{S\in \mathcal{S}_{R,R}^{(q)}} \frac{q[(r-1) !]^2 c_S}{(2 r-q) !} \frac{(n-r)(n-r-1) \cdots(n-2 r+q+1)}{(n-1)(n-2) \cdots(n-r+1)} U_{S}(G) \Big].
\end{aligned}
\end{equation*}

Notice that  $ (n-r)(n-r-1) \cdots(n-2 r+q+1)/(n-1)(n-2) \cdots(n-r+1)$ has $r-q$ items in the numerator and $r-1$ items in the denominator, and  $U_{S}(G) < 1$ for all $S$. Thus,
\begin{equation}
\label{eq:limI}
    \begin{aligned}
     &\lim\limits_{b,n\to\infty}   \frac{(n-b)}{b(n-1)}\Big[\frac{r !(n-r-1) !}{(n-2) !}\Big]^2  \Big[\sum_{q=0}^{r}\sum_{S\in\mathcal{S}_{R,R}^{(q)}}c_S\dfrac{q}{n}X_S(G)\Big]\\
     =& \lim\limits_{b,n\to\infty}  \frac{(n-b)}{b(n-1)} \Big[\frac{r(n-1)}{(n-r)}\Big]^2 
        \sum_{q=1}^{r}\sum_{S\in \mathcal{S}_{R,R}^{(q)}} \frac{q[(r-1) !]^2 c_S}{(2 r-q) !} \frac{(n-r)(n-r-1) \cdots(n-2 r+q+1)}{(n-1)(n-2) \cdots(n-r+1)} U_{S}(G) \\
        =&  \lim\limits_{b,n\to\infty} \frac{1}{b}\sum_{S \in \mathcal{S}_{R,R}^{(1)}} \frac{c_S(r !)^2 }{(2 r-1) !}  = 0.
    \end{aligned}
\end{equation}

For term II,
\begin{equation}
\label{eq:limII}
    \begin{aligned}
         &\frac{(n-b)}{b(n-1)}\left(\frac{r !(n-r-1) !}{(n-2) !}\right)^2\Big[\sum_{q=0}^{r}\sum_{S\in\mathcal{S}_{R,R}^{(q)}}c_S\dfrac{r^2}{n^2}X_S(G)\Big]\\
    = & \frac{(n-b)}{b(n-1)}\left(\frac{r !(n-r-1) !}{(n-2) !}\frac{r}{n} \right)^2\Big[\sum_{q=0}^{r}\sum_{S\in\mathcal{S}_{R,R}^{(q)}}c_SX_S(G)\Big] \stackrel{\text{\eqref{eq:linerityMC}}}{=} \left(\frac{r !(n-r-1) !}{b(n-2) !}\frac{r}{n} \right)^2X_R(G)X_R(G)\\
    =  & \frac{(n-b)}{b(n-1)}\left(\frac{r(n-1)}{(n-r)}\right)^2  \binom{n}{r}^{-2}X_R(G)X_R(G) =  \frac{(n-b)}{b(n-1)}\left(\frac{r(n-1)}{(n-r)}U_{R}(G)\right)^2\\
    \to &~ 0
    \end{aligned}
\end{equation}
as  $n,b \to \infty$.  Therefore, $\lim\limits_{b,n\to\infty}  \var_*\left[\sum_{1 \leqslant i \leqslant b}g_{1,R}\left(\mathbb{V}_1\right)\right] = \lim\limits_{b,n\to\infty}    \mathrm{I}  - \mathrm{II}  = 0$, which gives \eqref{eq:lim_var_g1}.
\end{enumerate}

Next, we continue to prove part \ref{theo:subsample_distribution_b} based on the results in part \ref{theo:subsample_distribution_a}.
\begin{enumerate}[label= \ref{theo:subsample_distribution_b}.\roman*]
    \item We can verify \eqref{eq:sg1} similarly  to \eqref{eq:g1}. As $U_{R}(\mathbb{G}^{*}_{b}) + U_{R'}(\mathbb{G}^{*}_{b})$ is a  symmetric finite population statistic, \eqref{eq:g_H} implies that 
    $$
    \begin{aligned}
    h_{1,R,R'}\left(\mathbb{V}_1\right)&=E_*\Big[U_{R}(\mathbb{G}^{*}_{b}) + U_{R'}(\mathbb{G}^{*}_{b})- E_* \big[U_{R}(\mathbb{G}^{*}_{b}) + U_{R'}(\mathbb{G}^{*}_{b})\big] \mid \mathbb{V}_1\Big]
    \end{aligned}
    $$ and $$\begin{aligned}g_{1,R,R'}\left(\mathbb{V}_1\right)&=\frac{n-1}{n-b} h_{1,R,R'}\left(\mathbb{V}_1\right).
    \end{aligned}
    $$ 
    By the linearity of conditional expectation, we have 
    \begin{align*}
   & g_{1,R,R'}(\mathbb{V}_1) \\
    \stackrel{\text{\eqref{eq:g_H}}}{=} &\frac{n-1}{n-b}h_{1,R,R'}\left(\mathbb{V}_1\right)\\ 
    =& \frac{n-1}{n-b}\left(E_*\Big\{U_{R}(\mathbb{G}^{*}_{b}) + U_{R'}(\mathbb{G}^{*}_{b})- E_* \big[U_{R}(\mathbb{G}^{*}_{b}) + U_{R'}(\mathbb{G}^{*}_{b})\big] \mid \mathbb{V}_1\Big\}\right)\\
    =& \frac{n-1}{n-b}\left(E_*\Big\{U_{R}(\mathbb{G}^{*}_{b})- E_* \big[U_{R}(\mathbb{G}^{*}_{b})\big] \mid \mathbb{V}_1\Big\} + E_*\Big\{ U_{R'}(\mathbb{G}^{*}_{b})- E_* \big[ U_{R'}(\mathbb{G}^{*}_{b})\big] \mid \mathbb{V}_1\Big\}\right)\\
        \stackrel{\text{\eqref{eq:lem5}}}{=}&\frac{n-1}{n-b}h_{1,R}(\mathbb{V}_1) + \frac{n-1}{n-b}h_{1,R'}(\mathbb{V}_1)\\
        \stackrel{\text{\eqref{eq:lem5}}}{=}&g_{1,R}(\mathbb{V}_1) + g_{1,R'}(\mathbb{V}_1).
    \end{align*}

\item  \eqref{eq:svar_sum_g1} can be verified according to  \eqref{eq:var_sumg1} as follows.
   \begin{equation*}
    \begin{aligned}
      & \var_*\Big[\sum_{1\leqslant i \leqslant b} g_{1,R, R'}(\mathbb{V}_i)\Big] = \var_*\Big[\sum_{1\leqslant i \leqslant b} g_{1,R}(\mathbb{V}_i) + \sum_{1\leqslant i \leqslant b}g_{1,R'}(\mathbb{V}_i)\Big]\\
    = &\var_*\Big[\sum_{1\leqslant i \leqslant b}g_{1,R}(\mathbb{V}_i)\Big] + 
     \var_*\Big[\sum_{1\leqslant i \leqslant b} g_{1,R'}(\mathbb{V}_i)\Big] +
      2\Cov_*\Big[\sum_{1\leqslant i \leqslant b} g_{1,R}(\mathbb{V}_i),
      \sum_{1\leqslant i \leqslant b} g_{1,R'}(\mathbb{V}_i)\Big]\\
     \stackrel{\text{\eqref{eq:var_sumg1}}}{=}&\dfrac{b(n-b)}{n-1}\Big\{\var_*\Big[g_{1,R}(\mathbb{V}_1)\Big]+\var_*\Big[g_{1,R'}(\mathbb{V}_1)\Big] \Big\} +
      2\Cov_*\Big[\sum_{1\leqslant i \leqslant b} g_{1,R}(\mathbb{V}_i),
      \sum_{1\leqslant i \leqslant b} g_{1,R'}(\mathbb{V}_i)\Big] \\
     \stackrel{\text{\eqref{eq:cov_sum_g1_more}}}{=}&\dfrac{b(n-b)}{n-1}\Big\{\var_*\Big[g_{1,R}(\mathbb{V}_1)\Big]+\var_*\Big[g_{1,R'}(\mathbb{V}_1)\Big] + 2\Cov_*\Big[g_{1,R}(\mathbb{V}_1),
    g_{1,R'}(\mathbb{V}_1)\Big]\Big\}\\
     =&  \frac{b(n-b)}{(n-1)}\var_*\Big[g_{1,R, R'}\left(\mathbb{V}_1\right)\Big].
    \end{aligned}
\end{equation*}
\item Now we turn to prove  \eqref{eq:lim_svar_sumg1}. First, we have 
\begin{equation*}
\begin{aligned}
      & \var_*\sum_{1 \leqslant i \leqslant b} g_{1,R,R'}\left(\mathbb{V}_i\right) \stackrel{\text{ \eqref{eq:var_sumg1}}}{=} \frac{b(n-b)}{(n-1)}\var_*\Big[g_{1,R,R'}\left(\mathbb{V}_1\right)\Big] 
        \\ =&\frac{b(n-b)}{(n-1)}\var_*\Big[g_{1,R}(\mathbb{V}_1)\Big]+ \frac{b(n-b)}{(n-1)}\var_*\Big[g_{1,R'}(\mathbb{V}_1)\Big] 
        + \frac{2b(n-b)}{(n-1)}\Cov_*\Big[g_{1,R}(\mathbb{V}_1),
    g_{1,R'}(\mathbb{V}_1)\Big] \\\stackrel{\text{\eqref{eq:var_sumg1}}}{=}
       & \var_*\sum_{1 \leqslant i \leqslant b} g_{1,R}\left(\mathbb{V}_i\right)+ \var_*\sum_{1 \leqslant i \leqslant b} g_{1,R'}\left(\mathbb{V}_i\right)  + \frac{2b(n-b)}{(n-1)}\Cov_*\Big[g_{1,R}(\mathbb{V}_1),
    g_{1,R'}(\mathbb{V}_1)\Big]\\
        = &\mathrm{I} +\mathrm{II}  +\mathrm{III} .
\end{aligned}
\end{equation*}

For terms I and II, we have 
\begin{equation}
\label{eq:limIandII}
\begin{aligned}
    \lim\limits_{b,n\to\infty} \mathrm{I} &\stackrel{\text{ \eqref{eq:limI}}}{=} 0 \\
    \lim\limits_{b,n\to\infty} \mathrm{II} & \stackrel{\text{ \eqref{eq:limII}}}{=} 0.
\end{aligned}
\end{equation}
 
We next focus on the behavior of term III: \allowdisplaybreaks
\begin{align*} 
  \mathrm{III} =&\Cov_*\Big[g_{1,R}(\mathbb{V}_1),
    g_{1,R'}(\mathbb{V}_1)\Big] =   \Cov_{\mathbb{V}_1*}\Big[g_{1,R}(\mathbb{V}_1),
    g_{1,R'}(\mathbb{V}_1)\Big] \\= &\frac{r!(n-r-1)!}{b(n-2)!} \frac{r'!(n-r'-1)!}{b(n-2)!}\Cov\Big[\sum_{\mathcal{G}\in \mathcal{S}(\mathbb{G}^{\mathbb{V}_1*}_{r})}  X_R(\mathcal{G}) ,\sum_{\mathcal{G}\in \mathcal{S}(\mathbb{G}^{\mathbb{V}_1*}_{r'})}  X_{R'}(\mathcal{G})\Big]\\
        \stackrel{\text{\eqref{eq:cov_v1_g1rr}}}{=} &\frac{r!(n-r-1)!}{b(n-2)!} \frac{r'!(n-r'-1)!}{b(n-2)!}\sum_{q=0}^{\min\{r,r'\}} \sum_{S\in \mathcal{S}_{R,R'}^{(q)}} c_S\dfrac{nq-rr'}{n^2}X_S(G) \\
      = & \dfrac{r !(n-r-1)!}{b(n-2)!}\dfrac{r'!(n-r'-1)!}{b(n-2)!} \Big\{\Big[\sum_{q = 0}^{\min\{r,r'\}}\sum_{S \in \mathcal{S}_{R,R'}^{(q)}} c_S \frac{q}{n} X_S(G)\Big] - \frac{rr'}{n^2}X_R(G)X_{R'}(G)\Big\} \\
     = &\sum_{q = 0}^{\min\{r,r'\}}\sum_{S \in \mathcal{S}_{R,R'}^{(q)}}  c_S \dfrac{r !(n-r-1)!}{b(n-2)!}\dfrac{r'!(n-r'-1)!}{b(n-2)!}\frac{q}{n} X_S(G)
     \\& - \frac{rr'(n-1)(n-1)}{b^2(n-r)(n - r')}U_{R}(G)U_{R'}(G)
\end{align*} 

As $U_{R}(G)U_{R'}(G)\le 1$, we have $$\lim\limits_{b,n\to\infty}\frac{rr'(n-1)(n-1)}{b^2(n-r)(n - r')}U_{R}(G)U_{R'}(G) = 0.$$ On the other hand, 
\begin{equation*}
\begin{aligned}
    &\sum_{q = 0}^{\min\{r,r'\}}\sum_{S \in \mathcal{S}_{R,R'}^{(q)}}c_S \dfrac{r !(n-r-1)!}{b(n-2)!}\dfrac{r'!(n-r'-1)!}{b(n-2)!}\frac{q}{n} X_S(G)\\
     = &\sum_{q = 0}^{\min\{r,r'\}}\sum_{S \in \mathcal{S}_{R,R'}^{(q)}}c_S \dfrac{r !(n-r-1)!}{b(n-2)!}\dfrac{r'!(n-r'-1)!}{b(n-2)!}\frac{q}{n} \binom{n}{r+r'-q}U_{S}(G)\\
    = &\sum_{q = 1}^{\min\{r,r'\}}\sum_{S \in \mathcal{S}_{R ,R'}^{(q)}}c_S \dfrac{r !(n-r-1)!}{b(n-2)!}\dfrac{r'!(n-r'-1)!}{b(n-2)!}\frac{q}{n}\dfrac{n!}{(r+r'-q) !(n-r-r'+q)!}U_{S}( G)\\
    = &\sum_{q = 1}^{\min\{r,r'\}}\sum_{S \in \mathcal{S}_{R, R'}^{(q)}}c_S \dfrac{qr !r' !}{b^2(r+r'-q) !} \dfrac{(n-1)(n-r'-1)\cdots(n-r'-r+q+1)}{(n-2)(n-3)\cdots(n-r)}U_{S}(G)
\end{aligned}
\end{equation*}

Notice that $(n-1)(n-r'-1)\cdots(n-r'-r+q+1)/(n-2)(n-3)\cdots(n-r)$ has $r -q$ items in the numerator and $r -1$ items in the denominator, and $U_{S}(G) \leqslant 1$, we have 
\begin{equation*}
\begin{aligned}
      &\sum_{q = 0}^{\min\{r,r'\}}\sum_{S \in \mathcal{S}_{R R'}^{(q)}}c_S \dfrac{r !(n-r-1)!}{b(n-2)!}\dfrac{r'!(n-r'-1)!}{b(n-2)!}\frac{q}{n} X_S(G) \\ 
         \stackrel{\text{$n \to \infty$}}{\rightarrow}  &\sum_{S \in \mathcal{S}_{R R'}^{(1)}}c_S \dfrac{r !r' !}{b^2(r+r'-1) !} U_{S}(G)    \stackrel{\text{$b \to \infty$}}{\rightarrow} 0.
\end{aligned}
\end{equation*}

 Therefore, 
\begin{equation*}
\begin{aligned}
       & \lim\limits_{b,n\to\infty}  var_*\Big[\sum_{1 \leqslant i \leqslant b} g_{1,R,R'}\left(\mathbb{V}_i\right)\Big]= \lim\limits_{b,n\to\infty} \mathrm{I}  + \mathrm{II} + \mathrm{III}
        \stackrel{\text{\eqref{eq:limIandII}}}{=}\lim\limits_{b,n\to\infty}\mathrm{III}= 0, 
\end{aligned}
\end{equation*}   
leading to \eqref{eq:lim_svar_sumg1}.
\end{enumerate}
\end{proof}

\subsection{Proof of Lemma \ref{cor:cov_subsample_v1}}
\label{C1}
\begin{proof}
    To begin with, we have 
\begin{equation*}
    \begin{aligned}
     &\Cov_{\mathbb{V}_1*}\Big[\sum_{\mathcal{G}\in \mathcal{S}(\mathbb{G}^{\mathbb{V}_1*}_{r})}  X_R(\mathcal{G}),\sum_{\mathcal{G}\in \mathcal{S}(\mathbb{G}^{\mathbb{V}_1*}_{r'})}  X_{R'}(\mathcal{G})\Big] \\
 \stackrel{\text{\eqref{eq:iden3}}}{=}& \Cov_{\mathbb{V}_1*}\Big[\big|\{S: S \subset G, \mathbb{V}_1 \in V(S), S \cong  R\}\big|,\big|\{S: S \subset G, \mathbb{V}_1 \in V(S), S \cong  R'\}\big|\Big] \\
    \stackrel{\text{\eqref{eq:motif_v1_ref}}}{=} & \Cov_{\mathbb{V}_1*}\Big[X_R\left(G\right)-X_R\left(G \setminus \mathbb{V}_1\right), X_{R'}\left(G\right)-X_{R'}\left(G \setminus \mathbb{V}_1\right)\Big] \\
    = &\Cov_{\mathbb{V}_1*}\Big[X_R\left(G \setminus \mathbb{V}_1\right), X_{R'}\left( G \setminus \mathbb{V}_1\right)\Big] \\
=&E_{\mathbb{V}_1*}\Big[X_R\left(G \setminus \mathbb{V}_1\right)X_{R'}\left(G \setminus \mathbb{V}_1\right)\Big]- E_{\mathbb{V}_1*}\Big[X_R\left(G \setminus \mathbb{V}_1\right)\Big]  E\Big[X_{R'}\left(G \setminus \mathbb{V}_1\right)\Big].
\end{aligned}
\end{equation*}

For the first term, 
\begin{equation*}
   \begin{aligned}
& E_{\mathbb{V}_1*}\Big[X_R\left(G \setminus \mathbb{V}_1\right)X_{R'}\left(G \setminus \mathbb{V}_1\right)\Big]=\frac{1}{n} \sum_{i=1}^n\Big[X_R\left(G \setminus v_i\right)X_{R'}\left(G \setminus v_i\right)\Big] 
\\ \stackrel{\text{\eqref{eq:linerityMC}}}{=}  &\frac{1}{n} \sum_{i=1}^n\Big[\sum_{S \in \mathcal{S}_{R, R'}} c_S X_S\left(G \setminus v_i\right)\Big] =\frac{1}{n} \sum_{S \in \mathcal{S}_{R, R'}} c_S \sum_{i=1}^n X_S\left(G \setminus v_i\right) \\
   =&  \frac{1}{n} \sum_{S \in  \mathcal{S}_{R, R'}} c_S \sum_{i=1}^{n}\Big[X_S(G)-\big|\{H:   H\subset G, v_i \in V(H), H \cong  S\}\big|\Big]\\
   = & \frac{1}{n} \sum_{S \in  \mathcal{S}_{R, R'}} c_S \sum_{i=1}^{n}X_S(G)
   -\frac{1}{n} \sum_{S \in  \mathcal{H}_{R, R'}} c_S \sum_{i=1}^{n}\big|\{H: H \subset G, v_i \in V(H), H \cong  S\}\big|\\
 \stackrel{\text{\eqref{eq:iden5}}}{=}  & \sum_{S \in  \mathcal{S}_{R, R'}} c_SX_S(G) - \sum_{S \in  \mathcal{S}_{R, R'}} c_S\frac{s}{n}X_S(G).
\end{aligned} 
\end{equation*}
Regarding the second component, we have
\begin{equation*}
    \begin{aligned}
       &E_{\mathbb{V}_1*}\Big[X_R\left(G \setminus \mathbb{V}_1\right)\Big]  E_{\mathbb{V}_1*}\Big[X_{R'}\left(G \setminus \mathbb{V}_1\right)\Big]  = \frac{1}{n} \sum_{i=1}^n X_R\left(G \setminus v_i\right) \frac{1}{n} \sum_{i=1}^n X_{R'}\left(G \setminus v_i\right) \\
     \stackrel{\text{\eqref{eq:motif_v1_ref}}}{=}&\frac{1}{n}\Big[\sum_{i=1}^n \left(X_R\left(G\right)-\Big|\{S: S \subset G, v_i \in V(S), S \cong  R\}\Big| \right)\Big] \cdot\\&\frac{1}{n}\Big[\sum_{i=1}^n \left(X_{R'}\left(G\right)-\Big|\{S: S \subset G, v_i \in V(S), S \cong  R'\}\Big| \right)\Big]  \\
 \stackrel{\text{\eqref{eq:iden5}}}{=} & \frac{1}{n}\Big[nX_R\left(G\right)- rX_R\left(G\right) \Big] \frac{1}{n}\Big[nX_{R'}\left(G\right)- r'X_{R'}\left(G\right)\Big] \\=& \frac{\left(n-r\right)}{n}\frac{\left(n-r'\right)}{n} X_R\left(G\right) X_{R'}\left(G\right) = \left(1-\frac{r + r'}{n}+\frac{rr'}{n^2}\right)X_R\left(G\right) X_{R'}\left(G\right) 
    \\\stackrel{\text{\eqref{eq:linerityMC}}}{=} &\left(1-\frac{r + r'}{n}+\frac{rr'}{n^2}\right)\sum_{S \in {\mathcal{S}_{R, R'}}} c_S X_S(G).
    \end{aligned}
\end{equation*}

Thus,
\begin{equation*}
    \begin{aligned}
     &\hspace{0.4cm} \Cov_{\mathbb{V}_1*}\Big[\sum_{\mathcal{G}\in \mathcal{S}(\mathbb{G}^{\mathbb{V}_1*}_{r})}  X_R(\mathcal{G}),\sum_{\mathcal{G}\in \mathcal{S}(\mathbb{G}^{\mathbb{V}_1*}_{r'})}  X_{R'}(\mathcal{G})\Big] \\&\hspace{-0.4cm} =\hspace{0.4cm} 
     E_{\mathbb{V}_1*}\Big[X_{R}\left(G \setminus \mathbb{V}_1\right)X_{R'}\left(G \setminus \mathbb{V}_1\right)\Big]- E_{\mathbb{V}_1*}\Big[X_{R}\left(G \setminus \mathbb{V}_1\right)\Big]  E_{\mathbb{V}_1*}\Big[X_{R'}\left(G \setminus \mathbb{V}_1\right)\Big]\\
     &\hspace{-0.45cm}\stackrel{\text{\eqref{eq:linerityMC}}}{=}\sum_{S \in  \mathcal{S}_{R, R'}} c_SX_S(G) - \sum_{S \in  \mathcal{S}_{R, R'}} c_S\frac{s}{n}X_S(G)-\left(1-\frac{r + r'}{n}+\frac{rr'}{n^2}\right)\sum_{S \in {\mathcal{S}_{R, R'}}} c_S X_S(G)\\
      & \hspace{-0.45cm}= \hspace{0.2cm}\sum_{S \in \mathcal{S}_{R, R'}}\left(\frac{r+r'-s}{n}-\frac{r  r'}{n^2}\right) c_S X _S(G)
 = \hspace{0.2cm} \sum_{q = 0}^{\min\{r,r'\}}\sum_{S \in \mathcal{S}_{R, R'}^{(q)}} c_S \frac{nq-r r'}{n^2} X_S(G).
\end{aligned}
\end{equation*} 
\end{proof}


%% file: Appendix/Appendix_SuppleA5.tex
 
\section{Proofs for Results in Section \ref{subsec:asym_dis_UR}}
\label{sec:proof for asymdis}

\subsection{Proof of Theorem \ref{theo:subsample_distribution}}

We start by introducing the setup of finite population asymptotic described in \cite{bloznelis2001orthogonal}: Suppose that there exist a sequence of finite populations $\{\mathcal{V}^{(n)}\}$, where $\mathcal{V}^{(n)} = \{\mathfrak{v}_1 \cdots \mathfrak{v}_{n}\}$ with $n \to \infty$. Consequently, $\{T_n\}$ is a sequence of finite population  U-statistics, where $T_n = t_n(\mathbb{V}_{1}, \cdots,\mathbb{V}_{b_n})$ is based on samples $\{\mathbb{V}_{1}, \cdots,\mathbb{V}_{b_n}\}$ drawn without replacement from $\mathcal{V}^{(n)}$, with $b_n \to \infty$ as $n  \to \infty$.

We first present a few auxiliary lemmas. Proofs for Lemmas \ref{coro:lim_var_subsample} and \ref{lem: lim_netmoment_graphcon} are deferred to Sections \ref{D1} and \ref{D2}, respectively.
\begin{lemma}(Proposition 3 in \cite{bloznelis2001orthogonal})
    \label{lem: normal_UR_subsample}
The Hoeffding's decomposition of $T_n$ is
  $$T_n = E_* [T_n] + \sum\limits_{1 \leqslant i \leqslant b_n} T_{1,n}(\mathbb{V}_i) +  \Delta(T_n),$$
 where $\sum\limits_{1 \leqslant i \leqslant b_n} T_{1,n}(\mathbb{V}_i)$ is the linear term, and $\Delta(T_n)$ is the remainder. Suppose that
\begin{enumerate}[label= \arabic*)]
 \item  \label{eq:ori_normal1} $E_*\Delta^2(T_n) = o(1)$.
  \item \label{eq:ori_normal2} There exist constants $c_1,c_2 > 0$  such that $0 < c_1 \leqslant \var_*(T_n)  \leqslant c_2< \infty$,
  \item \label{eq:ori_normal3} For every $\epsilon > 0$, $\lim\limits_{n\to\infty} b_n E_*[T^2_{1,n}\left(\mathbb{V}_1\right)\mathbb{1}_{\left\{T^2_{1,n}\left(\mathbb{V}_1\right)> \epsilon\right\}}] = 0$.
\end{enumerate}
Then $(T_n - E_*T_n )/(\var_*(T_n))$ is asymptotically standard normal. 
Note that the subscript $*$ is not used in \cite{bloznelis2001orthogonal}, and we add it here to distinguish the sourse of randomness.
\end{lemma}
\begin{lemma}[\cite{bloznelis2002edgeworth}]
 \label{lem: l2og}
 Let \begin{equation}
\label{eq:hoef_fu}
 \begin{aligned}
   T= &E_* T+\sum_{1 \leqslant i \leqslant b} g_1\left(\mathbb{V}_i\right)+\sum_{1 \leqslant  i<j \leqslant b} g_2\left(\mathbb{V}_i, \mathbb{V}_j\right)+\cdots\\
   =& E_* T+ S_1 + S_2 + S_3 +\cdots 
 \end{aligned}
\end{equation}
be  the Hoeffding's decomposition for a general finite population U-statistic $T$. Then we have
 \begin{equation}
 \label{eq:ref_b2002}
     E_*[S_aS_b] = 0, \hspace{0.5cm} \text{ if $a \neq b$.}
     \end{equation}

\end{lemma}

\begin{lemma}[Theorem 1 of \cite{bickel2011method}]
\label{lem: samplingdistribution_bickel}
Let $\iint w^{2}(u,v)dudv < \infty$. 
\begin{enumerate}[label= \alph*)]
    \item If $(n-1)\rho_n \rightarrow \infty$,
    \begin{equation}
\label{eq:bickelsamplingref1}
\begin{aligned}
  \frac{\wh{\rho}_{\mathbb{G}_{n}}}{\rho_n}& \to 1~ \text{in probability},\\
\text{$\sqrt{n}\Big(\frac{\wh{\rho}_{\mathbb{G}_{n}}}{\rho_n} - 1\Big)$} &\to \mathcal{N}(0, \sigma^2)~ \text{in distribution},
\end{aligned}
\end{equation}
for some $ \sigma^2 > 0$.
\item \label{item:b} For any motif $R$, assume that $\iint w^{2\mathfrak{r}}(u,v)dudv < \infty$, also $\rho_n = \omega(n^{-1})$ if $R$ is acyclic and $\rho_n = \omega(n^{-2/r})$ otherwise. Then
\begin{equation}
\label{eq:greensamplingref4}
  \sqrt{n}\Big[\rho_n^{-\mathfrak{r}}U_{R}(\mathbb{G}_{n})-  \rho_n^{-\mathfrak{r}}\mE[U_{R}(\mathbb{G}_{n})]\Big] \to \mathcal{N}\big(0, \sigma^2_R\big) ~ \text{in distribution}
\end{equation}
\begin{equation}
    \label{eq:bickelsamplingref2}
    \begin{aligned}
          &\wh{\rho}^{-\mathfrak{r}}_{\mathbb{G}_{n}}U_{R}(\mathbb{G}_{n}) \to \rho_n^{-\mathfrak{r}}\mE[U_{R}(\mathbb{G}_{n})]
          ~\text{in probability}\\
          &\sqrt{n}\Big[{\rho}^{-\mathfrak{r}}_{_{n}}U_{R}(\mathbb{G}_{n})-  \rho_n^{-\mathfrak{r}}\mE[U_{R}(\mathbb{G}_{n})]\Big]   \to \mathcal{N}\big(0, \sigma^2_R\big)~ \text{in distribution},
    \end{aligned}
\end{equation}
where $\sigma^2_R$ is defined in Assumption \ref{ass:non_degenerate}.

\item More generally, for $m$ motifs $R_1,\cdots,R_m$ with sizes $\max\{r_1,\cdots,r_m\} \leqslant r $,
\begin{equation}
\label{eq:bickelsamplingref3}
\begin{aligned}
     &\sqrt{n}\Big\{\big[{\rho}^{-\mathfrak{r}_1}_{{n}}U_{R_1}(\mathbb{G}_{n}),\cdots,{\rho}^{-\mathfrak{r}_m}_{{n}}U_{R_m}(\mathbb{G}_{n})\big] -  \big[\rho_n^{-\mathfrak{r}_1}\mE[U_{R_1}(\mathbb{G}_{n})],\cdots,\rho_n^{-\mathfrak{r}_m}\mE[U_{R_m}(\mathbb{G}_{n})]\big]\Big\} \\
     \to&~  \mathcal{N}\big(0, \Sigma_{[R_m]}\big)~ \text{in distribution}, 
\end{aligned}
\end{equation}
where $\Sigma_{[R_m]}$ is the asymptotic covariance matrix.
\end{enumerate}
\end{lemma}

\begin{lemma}
\label{lem: lim_netmoment_graphcon}
        For any motif $R$, under Assumptions \ref{ass:rho_n_h_n}, with probability one,
\begin{equation}
\label{eq:lim_rUR_subsample}
\lim\limits_{n \to \infty} \rho^{-\mathfrak{r}}_nU_R(\mathbb{G}_{n}) = \frac{r!}{|\mathrm{Aut}(R)|}P_w(R).
\end{equation}
\end{lemma}

Now we proceed to prove Theorem~\ref{theo:subsample_distribution}.

\begin{proof}[of Theorem~\ref{theo:subsample_distribution}] 
 We start with part \ref{theo:subsample_distribution_a}, and will prove the results one by one. For notational simplicity,  we write $G = G^{(n)}$, $b = b_n$. 
\begin{enumerate}[label= \ref{theo:subsample_distribution_a}.\roman*]
    \item  Since 
  $U_{R}(\mathbb{G}^{*}_{b})$ is a finite population U-statistic, $\sqrt{b}\rho_n^{-\mathfrak{r}} U_{R}(\mathbb{G}^{*}_{b})$ is also a finite population U-statistic. \eqref{eq:hoeffding_U_R} can be obtained by
\begin{equation}
\label{eq:hoeffding_pUR}
\begin{aligned}
\sqrt{b}\rho_{n}^{-\mathfrak{r}}  U_{R}(\mathbb{G}^{*}_{b}) \stackrel{\text{\eqref{eq:hoeffding_UR}}}{=} &\sqrt{b}\rho_{n}^{-\mathfrak{r}} \Big\{E_* [ U_{R}(\mathbb{G}^{*}_{b})]+\sum_{1 \leqslant i \leqslant b} g_{1,R}\left(\mathbb{V}_i\right)+\sum_{1 \leqslant  i<j \leqslant b} g_{2,R}\left(\mathbb{V}_i, \mathbb{V}_j\right)  +\cdots\Big\}\\
=& \sqrt{b}\rho_{n}^{-\mathfrak{r}} \Big\{E_* [ U_{R}(\mathbb{G}^{*}_{b})]+\sum_{1 \leqslant i \leqslant b} g_{1,R}\left(\mathbb{V}_i\right)\Big\} + \Delta[\sqrt{b} \rho_{n}^{-\mathfrak{r}}U_{R}(\mathbb{G}^{*}_{b})]\\
=&E_* [\sqrt{b}\rho_{n}^{-\mathfrak{r}}  U_{R}(\mathbb{G}^{*}_{b})]+\sum_{1 \leqslant i \leqslant b} \sqrt{b}\rho_{n}^{-\mathfrak{r}} g_{1,R}\left(\mathbb{V}_i\right) + \Delta[\sqrt{b} \rho_{n}^{-\mathfrak{r}}U_{R}(\mathbb{G}^{*}_{b})]\\
\stackrel{\text{\eqref{eq:exp_subsample}}}{=} & \sqrt{b}\rho_{n}^{-\mathfrak{r}}  U_{R}(G)+\sum_{1 \leqslant i \leqslant b} \sqrt{b}\rho_{n}^{-\mathfrak{r}} g_{1,R}\left(\mathbb{V}_i\right) + \Delta[\sqrt{b} \rho_{n}^{-\mathfrak{r}}U_{R}(\mathbb{G}^{*}_{b})].
\end{aligned}
\end{equation}
\item Proof of \eqref{eq:assump1}: To begin with, we have 
\begin{equation}\label{tref2}
\begin{aligned}
   E_*\Big(\Big\{\Delta[\sqrt{b} \rho_{n}^{-\mathfrak{r}}U_{R}(\mathbb{G}_{b}^{*})]\Big\} \Big\{\sum_{1 \leqslant i \leqslant b} \sqrt{b}\rho_{n}^{-\mathfrak{r}} g_{1,R}\left(\mathbb{V}_i\right)\Big\} \Big) \stackrel{\text{\eqref{eq:ref_b2002}}}{=} 0. 
\end{aligned}
\end{equation}
Consequently,  
\begin{equation}
\label{addref9}
    \begin{aligned}
       & E_* \Big[\Delta^2[\sqrt{b} \rho_{n}^{-\mathfrak{r}}U_{R}(\mathbb{G}^{*}_{b})] \Big] 
       \stackrel{\text{\eqref{eq:hoeffding_U_R}}}{=}E_*\Big[  \sqrt{b}\rho_{n}^{-\mathfrak{r}}  U_{R}(\mathbb{G}^{*}_{b}) -E_* [\sqrt{b}\rho_{n}^{-\mathfrak{r}}  U_{R}(\mathbb{G}^{*}_{b})] - \sum_{1 \leqslant i \leqslant b} \sqrt{b}\rho_{n}^{-\mathfrak{r}} g_{1,R}\left(\mathbb{V}_i\right)\Big]^2 \\
    =&E_*\Big[  \sqrt{b}\rho_{n}^{-\mathfrak{r}}  U_{R}(\mathbb{G}^{*}_{b}) -E_* [\sqrt{b}\rho_{n}^{-\mathfrak{r}}  U_{R}(\mathbb{G}^{*}_{b})]\Big]^2 +E_*\Big[  \sum_{1 \leqslant i \leqslant b} \sqrt{b}\rho_{n}^{-\mathfrak{r}} g_{1,R}\left(\mathbb{V}_i\right)\Big]^2 \\
    &- 2E_*\Big[ \Big( \sqrt{b}\rho_{n}^{-\mathfrak{r}}  U_{R}(\mathbb{G}^{*}_{b}) -E_* [\sqrt{b}\rho_{n}^{-\mathfrak{r}}  U_{R}(\mathbb{G}^{*}_{b})] \Big)\Big( \sum_{1 \leqslant i \leqslant b} \sqrt{b}\rho_{n}^{-\mathfrak{r}} g_{1,R}\left(\mathbb{V}_i\right)\Big)\Big] \\
    \stackrel{\text{\eqref{eq:hoeffding_U_R}}}{=}&\var_*\Big[ \sqrt{b}\rho_{n}^{-\mathfrak{r}} U_{R}(\mathbb{G}^{*}_{b})\Big]+E_*\Big[  \sum_{1 \leqslant i \leqslant b} \sqrt{b}\rho_{n}^{-\mathfrak{r}} g_{1,R}\left(\mathbb{V}_i\right)\Big]^2 \\
    &- 2E_*\Big[ \Big( \Delta[\sqrt{b} \rho_{n}^{-\mathfrak{r}}U_{R}(\mathbb{G}^{*}_{b})] + \sum_{1 \leqslant i \leqslant b} \sqrt{b}\rho_{n}^{-\mathfrak{r}} g_{1,R}\left(\mathbb{V}_i\right)\Big)\Big( \sum_{1 \leqslant i \leqslant b} \sqrt{b}\rho_{n}^{-\mathfrak{r}} g_{1,R}\left(\mathbb{V}_i\right)\Big)\Big] \\
    \stackrel{\text{\eqref{tref2}}}{=}&\var_*\Big[ \sqrt{b}\rho_{n}^{-\mathfrak{r}} U_{R}(\mathbb{G}^{*}_{b})\Big]-E_*\Big[  \sum_{1 \leqslant i \leqslant b} \sqrt{b}\rho_{n}^{-\mathfrak{r}} g_{1,R}\left(\mathbb{V}_i\right)\Big]^2 \\\stackrel{\text{\eqref{eq:expg1}}}{=}&\var_*\Big[ \sqrt{b}\rho_{n}^{-\mathfrak{r}} U_{R}(\mathbb{G}^{*}_{b})\Big] -\var_*\Big[   \sum_{1 \leqslant i \leqslant b} \sqrt{b}\rho_{n}^{-\mathfrak{r}} g_{1,R}\left(\mathbb{V}_i\right)\Big] = \mathrm{I} - \mathrm{II}.
    \end{aligned}
\end{equation}

Term I is the variance of the network moment. For the convenience of later analysis, we study the more general covariance term here for any two motifs $R$ and $R'$.
    \begin{equation}
\label{eq:more_cov_subsample}
 \begin{aligned}
     &\Cov_*\Big[\sqrt{b}\rho^{-\mathfrak{r}}_{n}U_{R}(\mathbb{G}^{*}_{b}),\sqrt{b}\rho^{-\mathfrak{r}'}_{n} U_{R'}(\mathbb{G}^{*}_{b})\Big] = b\rho^{-(\mathfrak{r}+\mathfrak{r}')}_{n}\Cov_*\Big[U_{R}(\mathbb{G}^{*}_{b}), U_{R'}(\mathbb{G}^{*}_{b})\Big]
     \\  \stackrel{\text{\eqref{eq:cov_subsample}}}{=}&b\rho^{-(\mathfrak{r}+\mathfrak{r}')}_{n}\Big\{\binom{b}{r}^{-1}\binom{b}{r'}^{-1} \sum_{q=0}^{\min\{r,r'\}}\sum_{S\in \mathcal{S}_{R, R'}^{(q)}}c_S\binom{b}{s}U_{S}(G)- U_{R}(G)U_{R'}(G)\Big\} 
     \\= &b\rho^{-(\mathfrak{r}+\mathfrak{r}')}_{n}\Big\{\binom{b}{r}^{-1}\binom{b}{r'}^{-1} \sum_{q=0}^{\min\{r,r'\}}\sum_{S\in \mathcal{S}_{R, R'}^{(q)}}c_S\binom{b}{s}U_{S}(G) - \binom{n}{r}^{-1}\binom{n}{r'}^{-1}X_R(G)X_{R'}(G)\Big\} 
     \\   \stackrel{\text{\eqref{eq:linerityMC}}}{=} & b\rho^{-(\mathfrak{r}+\mathfrak{r}')}_{n}\Big\{\binom{b}{r}^{-1}\binom{b}{r'}^{-1} \sum_{q=0}^{\min\{r,r'\}}\sum_{S\in \mathcal{S}_{R, R'}^{(q)}}c_S\binom{b}{s}U_{S}(G) 
   \\&-\binom{n}{r}^{-1}\binom{n}{r'}^{-1}\sum_{q=0}^{\min\{r,r'\}} \sum_{S\in\mathcal{S}_{R,R'}^{(q)}}c_S \binom{n}{s}U_{S}(G)\Big\}
   \\= &\sum_{q=0}^{\min\{r,r'\}} \sum_{S\in \mathcal{S}_{R, R'}^{(q)}} \dfrac{c_Sr!r'!}{s!\rho^{\mathfrak{r}+\mathfrak{r}'}_{n}}\Big(\dfrac{(b-r)!(b-r')!}{(b-1)!(b-s)!} - \dfrac{b(n-r)!(n-r')!}{n!(n-s)!}\Big) U_{S}(G)\\
    &\hspace{-0.5cm}=\sum_{q=0}^{\min\{r,r'\}} \sum_{S\in \mathcal{S}_{R, R'}^{(q)}} \dfrac{c_Sr!r'!}{s!\rho^{\mathfrak{r}+\mathfrak{r}'}_{n}}\Big(\dfrac{(b-r)!(b-r')!n!(n-s)! - b!(b-s)!(n-r)!(n-r')!] }{(b-1)!(b-s)!n!(n-s)!} \Big) U_{S}(G).
\end{aligned}   
\end{equation}

As a special case of $R=R'$, we have 
\begin{equation}
    \label{eq:var_b_pn_UR}
    \begin{aligned}
         \mathrm{I} = &\var_*\Big[\sqrt{b}\rho^{-\mathfrak{r}}_{n}U_{R}(\mathbb{G}^{*}_{b})\Big] \\=&\sum_{q=0}^{r} \sum_{S\in S_{R,R}^{(q)}} \dfrac{c_Sr!r!}{s!\rho^{2\mathfrak{r}}_{n}}\Big(\dfrac{(b-r)!(b-r)!n!(n-s)! - b!(b-s)!(n-r)!(n-r)!] }{(b-1)!(b-s)!n!(n-s)!} \Big) U_{S}(G).
    \end{aligned}
\end{equation}

\allowdisplaybreaks
For term II,  based on Proposition \ref{prop:stats_prop_finite_U_network}, we have
\begin{align*}
    &\var_* \Big[\sqrt{b}\rho_{n}^{-\mathfrak{r}}\sum g_{1,R}\left(\mathbb{V}_i\right) \Big] = b\rho_{n}^{-2\mathfrak{r}}\var_* \Big[\sum g_{1,R}\left(\mathbb{V}_i\right) \Big] 
 \stackrel{\text{\eqref{eq:var_sumg1}}}{=} \rho_{n}^{-2\mathfrak{r}}b \frac{b(n-b)}{(n-1)}\var_* \Big[g_{1,R}\left(\mathbb{V}_1\right) \Big]  \\\stackrel{\text{\eqref{eq:varg1}}}{=} &\rho_{n}^{-2\mathfrak{r}}\dfrac{b^2(n-b)}{(n-1)}  \Big[\dfrac{r!(n-r-1)!}{b(n-2)!} \Big]^{2} \sum_{q=0}^{r} \sum_{S\in\mathcal{S}_{R,R}^{(q)}} c_S \dfrac{nq-r^2}{n^2} X_S(G) \\
  = & \rho_{n}^{-2\mathfrak{r}}\sum_{q=0}^{r} \sum_{S\in \mathcal{S}_{R,R}^{(q)}} c_S \dfrac{b^2(n-b)}{(n-1)} \dfrac{r!(n-r-1)!r!(n-r-1)!}{b^2(n-2)!(n-2)!}\dfrac{(nq-r^2)}{n^2}\dfrac{n!}{(2r-q)!(n-2r+q)!} U_{S}(G) \\
= & \rho_{n}^{-2\mathfrak{r}}\sum_{q=0}^{r} \sum_{S\in\mathcal{S}_{R,R}^{(q)}}\dfrac{ c_Sr!r!}{(2r-q)!}  \Big[\dfrac{n-b}{n-1}\dfrac{n(n-1)}{n^2}\dfrac{(n-r-1)\cdots(n-2r+q+1)}{(n-2)\cdots(n-r)}(nq-r^2) \Big] U_{S}(G) \\
  = &\rho_{n}^{-2\mathfrak{r}}\sum_{q=0}^{r} \sum_{S\in\mathcal{S}_{R,R}^{(q)}} \dfrac{c_Sr!r!}{(2r-q)!}  \Big[\dfrac{(n-b)(n-r-1)\cdots(n-2r+q+1)(nq-r^2)}{n(n-2)\cdots(n-r)} \Big] U_{S}(G)\\
  = &\rho_{n}^{-2\mathfrak{r}}\sum_{q=0}^{r} \sum_{S\in\mathcal{S}_{R,R}^{(q)}} \dfrac{c_Sr!r!}{s!}  \Big[\dfrac{(n-b)(n-r-1)\cdots(n-s+1)(nq-r^2)}{n(n-2)\cdots(n-r)} \Big] U_{S}(G).
\end{align*}

Therefore, by combining term I and term II, we have
\begin{equation}
\label{eq: addref9_1}
    \begin{aligned}
       &E_* \big\{\Delta^2[\sqrt{b} \rho_{n}^{-\mathfrak{r}}U_{R}(\mathbb{G}^{*}_{b})] \big\} \stackrel{\text{\eqref{addref9}}}{=} \mathrm{I} - \mathrm{II} = \sum_{q=0}^{r} \sum_{S\in\mathcal{S}_{R,R}^{(q)}} a_S y_{n,S} z_{n,S},
    \end{aligned}
\end{equation}
where 
\begin{equation}
    \begin{aligned}
        & a_S = \dfrac{c_Sr!r!}{s}, \\
        &  y_{n,S} = \Big[\dfrac{(b-r)!(b-r)!n!(n-s)! - b!(b-s)!(n-r)!(n-r)!] }{(b-1)!(b-s)!n!(n-s)!} \\&\hspace{1cm}-\dfrac{(n-b)(n-r-1)\cdots(n-s+1)(nq-r^2)}{n(n-2)\cdots(n-r)} \Big],\\
        &  z_{n,S} = \rho_{n}^{-2\mathfrak{r}}U_{S}(G).
    \end{aligned}
\end{equation}

For a given $S$, $a_s$ is a fixed quantity. Now we focus on the limiting behavior of $y_{n,S}z_{n,S}$ for different numbers of the merged nodes $q$.

\begin{itemize}
    \item When $q = 0$, we have $\mathfrak{s} = 2\mathfrak{r}$ and $s = 2r$.  We start with the first part of $y_{n,S}$.
     \begin{align*}
        &\dfrac{(b-r)!(b-r)!n!(n-s)! - b!(b-s)!(n-r)!(n-r)! }{(b-1)!(b-s)!n!(n-s)!}\\
      =  &  \dfrac{[n\cdots (n-r+1)(b-r)\cdots (b-2r+1)] -[b\cdots (b-r+1) (n-r)\cdots (n-2r+1)]}{(b-1)\cdots (b-r+1)n\cdots (n-r+1)}\\
      =  & \dfrac{\mathrm{IV} - \mathrm{V}}{\mathrm{VI}}.
    \end{align*}

Now we study each component as follows:
    \begin{align*}
   \mathrm{IV}= & b^{r}n^r - (r+r+1+\cdots r+r-1)b^{r-1}n^r -[1+2+\cdots +(r-1)]b^{r}n^{r-1}\\
   &+O(b^{r-2}n^r) + O(b^{r}n^{r-2}) + O(b^{r-1}n^{r-1}) +  o(b^{r-1}n^{r-1})
     \\=& b^{r}n^r - \dfrac{(r+2r-1)r}{2}b^{r-1}n^r -  \dfrac{(r-1)r}{2}b^{r}n^{r-1}
     \\  &+O(b^{r-2}n^r) + O(b^{r}n^{r-2}) + O(b^{r-1}n^{r-1}) +  o(b^{r-1}n^{r-1}).\\
      \mathrm{V}= & b^{r}n^r - (r+r+1+\cdots r+r-1)n^{r-1}b^r -[1+2+\cdots +(r-1)]n^{r}b^{r-1}\\
   &+O(b^{r-2}n^r) + O(b^{r}n^{r-2}) + O(b^{r-1}n^{r-1}) +  o(b^{r-1}n^{r-1})
     \\=& b^{r}n^r - \dfrac{(r+2r-1)r}{2}n^{r-1}b^r -  \dfrac{(r-1)r}{2}n^{r}b^{r-1}
      \\&+O(b^{r-2}n^r) + O(b^{r}n^{r-2}) + O(b^{r-1}n^{r-1}) +  o(b^{r-1}n^{r-1}).\\
\end{align*}
Consequently,
  \begin{equation*}
\begin{aligned}
        \mathrm{IV} - \mathrm{V} = &b^{r}n^r - \dfrac{(r+2r-1)r}{2}b^{r-1}n^r -  \dfrac{(r-1)r}{2}b^{r}n^{r-1}
     \\&+O(b^{r-2}n^r) + O(b^{r}n^{r-2}) + O(b^{r-1}n^{r-1})\\
     & - \big[b^{r}n^r - \dfrac{(r+2r-1)r}{2}n^{r-1}b^r -  \dfrac{(r-1)r}{2}n^{r}b^{r-1}
     \\&+O(b^{r-2}n^r) + O(b^{r}n^{r-2}) + O(b^{r-1}n^{r-1})\big]\\
     = &\dfrac{-(r+2r-1)r}{2}n^{r-1}b^{r-1}(n-b) + \dfrac{(r-1)r}{2}n^{r-1}b^{r-1}(n-b)\\
     & + O(b^{r-2}n^r) + O(n^{r-2}b^r)  + O(b^{r-1}n^{r-1}) \\
     = & (-r^2)n^{r-1}b^{r-1}(n-b) + O(b^{r-2}n^r) + O(n^{r-2}b^r)  + O(b^{r-1}n^{r-1}).
\end{aligned}
\end{equation*}

Thus,
\begin{equation}
\label{eq:limref}
    \begin{aligned}
   \mathrm{VI} = & b^{r-1}n^r + o(b^{r-1}n^r).\\
     \dfrac{  \mathrm{IV} - \mathrm{V}}{  \mathrm{VI}} = & (-r^2)\frac{n-b}{n} + o(1).
    \end{aligned}
\end{equation}

For the second part of $y_{n,S}$, we have
\begin{align*}
& \Big[\dfrac{(n-b)(n-r-1)\cdots(n-s+1)(nq-r^2)}{n(n-2)\cdots(n-r)} \Big]\\
= &\Big[\dfrac{(n-b)(n-r-1)(n-r-2)\cdots(n-2r+1)(0-r^2)}{n(n-2)(n-3)\cdots(n-r)} \Big] \\
= &\Big[(1 - \dfrac{b}{n})(1 - \dfrac{r+1}{n-2})(1 - \dfrac{r+1}{n-3})\cdots(1 - \dfrac{r+1}{n-r})(-r^2) \Big]\\
= & -r^2 \Big[ \dfrac{n-b}{n} + o(1) \Big].
\end{align*}

Therefore, we have 
\begin{equation*}
    y_{n,S} =  (-r^2)\frac{n-b}{n} + o(1) - (-r^2)\frac{n-b}{n} + o(1) = o(1).
\end{equation*}

To understand $z_{n,S}$, recall that $G \sim \mathbb{G_{n}}$.   Lemma \ref{lem: lim_netmoment_graphcon} implies that under Assumption \ref{ass:rho_n_h_n}, with probability one,
\begin{equation*}
    \lim\limits_{n \to \infty} \rho_{n}^{-\mathfrak{s}}U_{S}(\mathbb{G}_{n}) = \frac{s!}{|\mathrm{Aut}(S)|}P_w(S).
\end{equation*}
Thus, we assert that $y_{n,S}z_{n,S} \to 0$ with probability one. 
 
 \item When $q = 1$, we have $\mathfrak{s} = 2\mathfrak{r}$ and $s = 2r - 1$. As before, we study the first part of $y_{n,S}$.
 \begin{align*}
        &\dfrac{[(b-r)!(b-r)!n!(n-s)! - b!(b-s)!(n-r)!(n-r)!] }{(b-1)!(b-s)!n!(n-s)!}\\
      =  & \dfrac{[(b-r)!(b-r)!n!(n-2r+1)! - b!(b-2r+1)!(n-r)!(n-r)!] }{(b-1)!(b-2r+1)!n!(n-2r+1)!} = 1 - \dfrac{b}{n} + o(1).
    \end{align*}

For the second part of $y_{n,S}$, we have
\begin{align*}
& \Big[\dfrac{(n-b)(n-r-1)\cdots(n-s+1)(nq-r^2)}{n(n-2)\cdots(n-r)} \Big]\\
= & \Big[\dfrac{(n-b)(n-r-1)(n-r-2)\cdots(n-2r+2)(n-r^2)}{n(n-2)(n-3)\cdots(n-r)} \Big]  \\
= & 1 - \frac{b}{n} + o(1).
\end{align*}
Therefore, when  $q = 1$, we have $y_{n,S} =o(1)$.  Thus, we also have $y_{n,S}z_{n,S} \to 0$ with probability one.
 \item When $q > 1$, the first part of $y_{n,S}$ is
\begin{align*}
        \dfrac{[(b-r)!(b-r)!n!(n-s)! - b!(b-s)!(n-r)!(n-r)!] }{(b-1)!(b-s)!n!(n-s)!}= O\big[\frac{1}{b^{(q-1)}}\big].
    \end{align*}

The second part of $y_{n,S}$ is
\begin{align*}
& \dfrac{(n-b)(n-r-1)\cdots(n-s+1)(nq-r^2)}{n(n-2)\cdots(n-r)} \\
= &\dfrac{(n-b)(n-r-1)(n-r-2)\cdots(n-2r+q+1)(nq-r^2)}{n(n-2)(n-3)\cdots(n-r)} \\
= & O\big[\frac{(n-b)}{n^{q}}\big].
\end{align*}

Therefore, $y_{n,s} = o(1)$. We have $y_{n,S}z_{n,S} \to 0$ for every $q \geq 2$, .

\end{itemize}

Finally, because $r$ is a constant and $\mathcal{S}_{R,R}$ is a fixed set given $R$. For any random network sequence $\{G^{(n)}\}$,  with probability one, 
\begin{equation*}
   \lim\limits_{n \to \infty}E_* \big\{ \Delta^2[\sqrt{b} \rho_{n}^{-\mathfrak{r}}U_{R}(\mathbb{G}^{*}_{b})]\big\} = \lim\limits_{n \to \infty}\sum_{q=0}^{r} \sum_{S\in\mathcal{S}_{R,R}^{(q)}} a_S y_{n,s} z_{n,s} = 0.
\end{equation*}
\item Now we want to show  \eqref{eq:assump2}, which is related to non-degeneration, and is termed as the non-lattice assumption in \cite{zhang2022edgeworth}. From Lemma \ref{coro:lim_var_subsample}, under Assumptions \ref{ass:rho_n_h_n} and \ref{ass:b}, with probability one,
\begin{equation*}
\begin{aligned}
 &\lim\limits_{ b\to \infty}\rho^{-2\mathfrak{r}}_{n}\var_*\big[\sqrt{b}U_{R}(\mathbb{G}^{*}_{b})\big] = \big(1 - c_2\big)\lim\limits_{b \to \infty}\rho^{-2\mathfrak{r}}_{b}\var\big[\sqrt{b}U_{R}(\mathbb{G}_{b})\big].\\
\end{aligned}
\end{equation*}

Since $c_2 < 1$ is a constant, \eqref{eq:assump2} holds by Assumption \ref{ass:non_degenerate}.

\item Next, we want to show \eqref{eq:assump3}, the Lindeberg-Feller typed condition. To verify this, We want to show that 
 $$b\rho_{n}^{-2\mathfrak{r}}g^2_{1,R}(\mathbb{V}_1) = o(1).$$
 Let's begin by considering the following expression:
\begin{equation}
\begin{aligned}
\label{eq:b_rho_g1square}
b\rho_{n}^{-2\mathfrak{r}}g^2_{1,R}(\mathbb{V}_1)  &\stackrel{\text{ \eqref{eq:g1}}}{=}  \frac{\rho_{n}^{-2\mathfrak{r}}}{b}\bigg\{  \frac{r(n-1)}{(n-r)}\binom{n-1}{r-1}^{-1}\sum_{\mathcal{G}\in \mathcal{S}(\mathbb{G}^{V_1*}_{r})}  X_R(\mathcal{G}) - \frac{(n-1)r}{(n-r)}U_{R}(G)\bigg\}^2.
\end{aligned}
\end{equation}

Recall that $K_r$ denotes a complete graph of $r$ nodes. Clearly, for any $\mathcal{G}\in \mathcal{S}(\mathbb{G}^{V_1*}_{r})$, $$X_R(\mathcal{G}) \leqslant X_R(K_{r}).$$ Equation (2.7) in \cite{bhattacharya2022fluctuations} shows that $X_R(K_{r}) = r!/|\mathrm{Aut}(R)|$. Therefore,
\begin{equation}
\label{tref6}
    \begin{aligned}
        &\frac{r!(n-r-1)!}{(n-2)!}\sum_{\mathcal{G}\in \mathcal{S}(\mathbb{G}^{V_1*}_{r})}  X_R(\mathcal{G}) =  \frac{r(n-r)}{(n-1)}\binom{n-1}{r-1}^{-1}\sum_{\mathcal{G}\in \mathcal{S}(\mathbb{G}^{V_1*}_{r})}  X_R(\mathcal{G})\\
        \leqslant  & \frac{r(n-r)}{(n-1)}\binom{n-1}{r-1}^{-1}\sum_{\mathcal{G}\in \mathcal{S}(\mathbb{G}^{V_1*}_{r})}  \left(\dfrac{r!}{|\mathrm{Aut}(R)|}\right)\mathbb{1}_{\{R \subset \mathcal{G}\}}\\
        =  & \frac{r(n-r)}{(n-1)}\left(\dfrac{r!}{|\mathrm{Aut}(R)|}\right)\binom{n-1}{r-1}^{-1}\sum_{\mathcal{G}\in \mathcal{S}(\mathbb{G}^{V_1*}_{r})}  \mathbb{1}_{\{R \subset \mathcal{G}\}}
        \leqslant \frac{r(n-r)}{(n-1)}\left(\dfrac{r!}{|\mathrm{Aut}(R)|}\right).
    \end{aligned}
\end{equation}
By \eqref{eq:b_rho_g1square}, 
\begin{equation*}
\begin{aligned}
    b\rho_{n}^{-2\mathfrak{r}}g^2_{1,R}(\mathbb{V}_1)  
   &  \stackrel{\text{ \eqref{eq:b_rho_g1square}}}{\leqslant }  \frac{\rho_{n}^{-2\mathfrak{r}}}{b} \bigg\{\Big[ \frac{r(n-1)}{(n-r)}\binom{n-1}{r-1}^{-1}\sum_{\mathcal{G}\in \mathcal{S}(\mathbb{G}^{V_1*}_{r})}  X_R(\mathcal{G})\Big]^2 + \left[\frac{(n-1)r}{(n-r)}U_{R}( G)\right]^2\bigg\}\\
   & \stackrel{\text{\eqref{tref6}}}{\leqslant} \frac{\rho_{n}^{-2\mathfrak{r}}}{b} \bigg\{\Big[\frac{r(n-r)}{(n-1)}\left(\dfrac{r!}{|\mathrm{Aut}(R)|}\right)\Big]^2 + \left[\frac{(n-1)r}{(n-r)}U_{R}( G)\right]^2\bigg\}.
\end{aligned}
\end{equation*}

Given that both $r$ and$|\mathrm{Aut}(R)|$ are constants, and considering $U_{R}(G) \leqslant 1$, it follows that if $\rho^{-2\mathfrak{r}}_{n} /b\to 0$ , then for any given $\epsilon > 0$, there exists a $K > 0$ such that when $k > K$, 
\begin{equation}
    \label{eq:addref16}
    \mathbb{1}_{\left\{b\rho_{n}^{-2\mathfrak{r}}g^2_{1,R}(\mathbb{V}_1)> \epsilon\right\}} = 0.
\end{equation}

The above arguments indicate
$$ \lim\limits_{n \to \infty}  bE_* 
    \big[b\rho_{n}^{-2\mathfrak{r}} g^2_{1,R}(\mathbb{V}_1)\mathbb{1}_{\{b\rho_{n}^{-2\mathfrak{r}} g^2_{1,R}(\mathbb{V}_1)> \epsilon\}}\big] = 0.$$

\item To prove \eqref{eq:unit_subsample_distribution_with_truepn}, recall that  \eqref{lem: exp_cov_subsample}  implies $E_*[U_{R}(\mathbb{G}^{*}_{b})] =  U_{R}(G)$. When \eqref{eq:assump1}, \eqref{eq:assump2}, and \eqref{eq:assump3} hold, Lemma \ref{lem: normal_UR_subsample} implies that 
\begin{equation*}
    \begin{aligned}
        &\dfrac{\sqrt{b}\big[\rho_{n}^{-\mathfrak{r}}U_{R}(\mathbb{G}^{*}_{b}) -\rho_{n}^{-\mathfrak{r}}U_{R}(G^{(n)})\big] }{\var_*(\sqrt{b}\rho_{n}^{-\mathfrak{r}}U_{R}(\mathbb{G}^{*}_{b})) }
    \end{aligned}
\end{equation*}
is  asymptotically standard normal. 
\end{enumerate}

Now we proceed to prove Part \ref{theo:subsample_distribution_b}. Give the $m$ motifs, $R_1, \cdots, R_m$ be  $m$, consider the following linear combination 
$$ \Theta^{(a_1,\cdots,a_m)}_{R_1, \ldots, R_m} = a_1\sqrt{b}\rho^{-\mathfrak{r}_1}_{n}U_{R_1}(\mathbb{G}^{*}_{b}) + \cdots + a_m\sqrt{b}\rho^{-\mathfrak{r}_m}_{n}U_{R_m}(\mathbb{G}^{*}_{b})$$ 
where $a_1,\cdots,a_m$ are constants. For simplicity, denote $\Theta = \Theta^{(a_1,\cdots,a_m)}_{R_1, \ldots, R_m}$, which is a symmetric finite population statistic with Hoeffding's decomposition 
$$\Theta = E_* (\Theta )+\sum_{1 \leqslant i \leqslant b} g_{1, \Theta}\left(\mathbb{V}_i\right)+ \Delta(\Theta).$$
We want to show that every linear combination $\Theta$ is asymptotically normal. Following Lemma \ref{lem: normal_UR_subsample}, we need to verify the following conditions:
\begin{equation}
\label{eq:massump1}
     \lim\limits_{b \to \infty} E_*\Delta^2(  \Theta) = 0,  \hspace{0.2cm} \text{where $\Delta( \Theta) =  \Theta -E_* ( \Theta ) -\sum_{1 \leqslant i \leqslant b} g_{1, \Theta}\left(\mathbb{V}_i\right)$.}
\end{equation}
\begin{equation}
\label{eq:massump2}
    0 < c_1 \leqslant \lim_{b \to \infty}\var_*(\Theta) \leqslant c_2 < \infty, \hspace{1cm} \text{for some $c_1, c_2 > 0$.}
\end{equation}
\begin{equation}
\label{eq:massump3}
 \text{For every $\epsilon > 0$} \hspace{1cm} \lim_{b \to \infty}  bE_* 
    \big[g^2_{1, \Theta}\left(\mathbb{V}_1\right)\mathbb{1}_{\left\{g^2_{1, \Theta}\left(\mathbb{V}_1\right) > \epsilon \right\}}] = 0.
\end{equation}
\begin{enumerate}[label= \ref{theo:subsample_distribution_b}.\roman*]
  \item To show \eqref{eq:massump1},  we first consider a special case that we have only a pair of motifs $R$ and $R'$ whose linear combination is defined by two coefficients $\alpha$ and $\beta$. In this case, 
\begin{equation}
    \label{eq:theta_RR'}
  \Theta= \Theta^{(\alpha,\beta)}_{R,R'} = \alpha\sqrt{b}\rho^{-\mathfrak{r}}_{n}U_{R}(\mathbb{G}^{*}_{b}) + \beta\sqrt{b}\rho^{-\mathfrak{r'}}_{n}U_{R'}(\mathbb{G}^{*}_{b}).
\end{equation}

The Hoeffding's decomposition of $  \Theta^{(\alpha,\beta)}_{R,R'}$  can be expressed as follows
\begin{equation}
\label{eq:hoeffing_theta_RR'}
    \begin{aligned}
      \Theta^{(\alpha,\beta)}_{R,R'}= E_* \big( \Theta^{(\alpha,\beta)}_{R,R'}\big)+\sum_{1 \leqslant i \leqslant b} g_{1,     \Theta^{(\alpha,\beta)}_{R,R'}}\left(\mathbb{V}_i\right)+\sum_{1 \leqslant  i<j \leqslant b} g_{2,     \Theta^{(\alpha,\beta)}_{R,R'}}\left(\mathbb{V}_i, \mathbb{V}_j\right)+\cdots.
    \end{aligned}
\end{equation}

The Proposition \ref{prop:stats_prop_finite_U_network} implies 
\begin{equation}
    \label{eq:exp_two_R}
   E_* ( \Theta^{(\alpha,\beta)}_{R,R'}) = \alpha\sqrt{b}\rho^{-\mathfrak{r}}_{n}E_* \big[U_{R}(\mathbb{G}^{*}_{b})\big] + \beta \sqrt{b}\rho^{-\mathfrak{r'}}_{n}E_* \big[U_{R'}(\mathbb{G}^{*}_{b})\big],
\end{equation}

and 
\begin{equation}
    \label{eq:sum_g_two_R}
   \sum_{1 \leqslant i \leqslant b} g_{1, \Theta^{(\alpha,\beta)}_{R,R'}}\left(\mathbb{V}_i\right) \stackrel{\text{\eqref{eq:sg1}}}{=}  \alpha\sqrt{b}\rho^{-\mathfrak{r}}_{n}\sum_{1 \leqslant i \leqslant b} g_{1,R}\left(\mathbb{V}_i\right)+ \beta \sqrt{b}\rho^{-\mathfrak{r'}}_{n}\sum_{1 \leqslant i \leqslant b} g_{1,R'}\left(\mathbb{V}_i\right),
\end{equation}
where $g_{1,R}\left(\mathbb{V}_i\right)$ and $g_{1,R'}\left(\mathbb{V}_i\right)$ are defined in \eqref{eq:g1}.  Note that 
\begin{equation}
\label{eq:delta_two_RR'}
    \begin{aligned}
     \Delta[  \Theta^{(\alpha,\beta)}_{R,R'}] &=  \Theta^{(\alpha,\beta)}_{R,R'} -E_* \big[ \Theta^{(\alpha,\beta)}_{R,R'}\big] -\sum_{1 \leqslant i \leqslant b} g_{1, \Theta^{(\alpha,\beta)}_{R,R'}}\left(\mathbb{V}_i\right)   \\
     &\stackrel{\text{\eqref{eq:theta_RR'}}}{=}\alpha\sqrt{b}\rho^{-\mathfrak{r}}_{n}U_{R}(\mathbb{G}^{*}_{b}) + \beta\sqrt{b}\rho^{-\mathfrak{r'}}_{n}U_{R'}(\mathbb{G}^{*}_{b}) -E_* \big[ \Theta^{(\alpha,\beta)}_{R,R'}\big] -\sum_{1 \leqslant i \leqslant b} g_{1, \Theta^{(\alpha,\beta)}_{R,R'}}\left(\mathbb{V}_i\right)   \\
& \stackrel{\text{\eqref{eq:exp_two_R}}}{=}\alpha\sqrt{b}\rho^{-\mathfrak{r}}_{n}U_{R}(\mathbb{G}^{*}_{b}) + \beta\sqrt{b}\rho^{-\mathfrak{r'}}_{n}U_{R'}(\mathbb{G}^{*}_{b}) \\&\hspace{0.5cm}- 
\alpha\sqrt{b}\rho^{-\mathfrak{r}}_{n}E_* \big[U_{R}(\mathbb{G}^{*}_{b})\big] - \beta \sqrt{b}\rho^{-\mathfrak{r'}}_{n}E_* \big[U_{R'}(\mathbb{G}^{*}_{b})\big]-\sum_{1 \leqslant i \leqslant b} g_{1, \Theta^{(\alpha,\beta)}_{R,R'}}\left(\mathbb{V}_i\right)\\
& \stackrel{\text{\eqref{eq:sum_g_two_R}}}{=}\alpha\sqrt{b}\rho^{-\mathfrak{r}}_{n}U_{R}(\mathbb{G}^{*}_{b}) + \beta\sqrt{b}\rho^{-\mathfrak{r'}}_{n}U_{R'}(\mathbb{G}^{*}_{b})\\&\hspace{0.5cm} - 
\alpha\sqrt{b}\rho^{-\mathfrak{r}}_{n}E_* \big[U_{R}(\mathbb{G}^{*}_{b})\big] - \beta \sqrt{b}\rho^{-\mathfrak{r'}}_{n}E_* \big[U_{R'}(\mathbb{G}^{*}_{b})\big]\\&\hspace{0.5cm}-
\alpha\sqrt{b}\rho^{-\mathfrak{r}}_{n}\sum_{1 \leqslant i \leqslant b} g_{1,R}\left(\mathbb{V}_i\right)- \beta \sqrt{b}\rho^{-\mathfrak{r'}}_{n}\sum_{1 \leqslant i \leqslant b} g_{1,R'}\left(\mathbb{V}_i\right)\\& = \alpha\sqrt{b}\rho^{-\mathfrak{r}}_{n}U_{R}(\mathbb{G}^{*}_{b})- 
\alpha\sqrt{b}\rho^{-\mathfrak{r}}_{n}E_* \big[U_{R}(\mathbb{G}^{*}_{b})\big] -
\alpha\sqrt{b}\rho^{-\mathfrak{r}}_{n}\sum_{1 \leqslant i \leqslant b} g_{1,R}\left(\mathbb{V}_i\right)\\&\hspace{0.5cm}+ \beta\sqrt{b}\rho^{-\mathfrak{r'}}_{n}U_{R'}(\mathbb{G}^{*}_{b}) - \beta \sqrt{b}\rho^{-\mathfrak{r'}}_{n}E_* \big[U_{R'}(\mathbb{G}^{*}_{b})\big]- \beta \sqrt{b}\rho^{-\mathfrak{r'}}_{n}\sum_{1 \leqslant i \leqslant b} g_{1,R'}\left(\mathbb{V}_i\right)\\
& \stackrel{\text{\eqref{eq:hoeffding_pUR}}}{=}\alpha \Delta(\sqrt{b}\rho^{-\mathfrak{r}}_{n}U_{R}(\mathbb{G}^{*}_{b})) + \beta \Delta(\sqrt{b}\rho^{-\mathfrak{r'}}_{n}U_{R'}(\mathbb{G}^{*}_{b})).
    \end{aligned}
\end{equation}

Consequently,
\begin{equation*}
\begin{aligned}      &\hspace{0.6cm}E_*\Big[\Delta^2(\Theta^{(\alpha,\beta)}_{R,R'})\Big] \stackrel{\text{\eqref{eq:delta_two_RR'}}}{\leq} 2E_*\Big[\alpha \Delta(\sqrt{b}\rho^{-\mathfrak{r}}_{n}U_{R}(\mathbb{G}^{*}_{b}))\Big]^2  +  2E_*\Big[\Delta(\sqrt{b}\rho^{-\mathfrak{r'}}_{n}U_{R'}(\mathbb{G}^{*}_{b}))\Big]^2.
\end{aligned}
\end{equation*}

From \eqref{eq:assump1}, under Assumptions  \ref{ass:rho_n_h_n}-\ref{ass:non_degenerate}, we have $$
        \lim\limits_{b \to \infty }  E_*\Big[\Delta^2(\Theta^{(\alpha,\beta)}_{R,R'})\Big] \leqslant \lim\limits_{b \to \infty }  2E_*\Big[\alpha \Delta(\sqrt{b}\rho^{-\mathfrak{r}}_{n}U_{R}(\mathbb{G}^{*}_{b}))\Big]^2  +  \lim\limits_{b \to \infty } 2E_*\Big[\beta \Delta(\sqrt{b}\rho^{-\mathfrak{r'}}_{n}U_{R'}(\mathbb{G}^{*}_{b}))\Big]^2  =0.$$

More generally, for $m$ motifs, we have
\begin{equation}
\label{eq:linear_combine_UR}
   E_* \big[ \Theta\big] = a_1E_* \big[\sqrt{b}\rho^{-\mathfrak{r}_1}_{n}U_{R_1}(\mathbb{G}^{*}_{b})\big] + \cdots + a_mE_* \big[\sqrt{b}\rho^{-\mathfrak{r}_m}_{n}U_{R_m}(\mathbb{G}^{*}_{b})\big].
\end{equation}

On the other hand, we have
\begin{equation}
\label{eq:linear_combine_sum_g}
   \sum_{1 \leqslant i \leqslant b} g_{1, \Theta}\left(\mathbb{V}_i\right)\stackrel{\text{\eqref{eq:sum_g_two_R}}}{=}  a_1\sqrt{b}\rho^{-\mathfrak{r}_1}_{n}\sum_{1 \leqslant i \leqslant b} g_{1,R_1}\left(\mathbb{V}_i\right) + \cdots + a_m \sqrt{b}\rho^{-\mathfrak{r}_m}_{n}\sum_{1 \leqslant i \leqslant b} g_{1,R_m}\left(\mathbb{V}_i\right).
\end{equation}

Similar to the derivation of \eqref{eq:delta_two_RR'}, we  have:
\begin{equation*}
    \Delta(\Theta) = a_1\Delta[\sqrt{b} \rho_{n}^{-\mathfrak{r}_1}U_{R_1}(\mathbb{G}_{b}^{*})] + \cdots a_m\Delta[\sqrt{b} \rho_{n}^{-\mathfrak{r}_m}U_{R_m}(\mathbb{G}_{b}^{*})].
\end{equation*}

Thus, the accuracy of approximation of the linear part could be bounded as:
\begin{equation}
\begin{aligned}
    \lim\limits_{b \to \infty} E_*\Delta^2( \Theta) \leqslant &m \Big\{ 
  a^2_1\lim\limits_{b \to \infty} E_*\Delta^2[\sqrt{b} \rho_{n}^{-\mathfrak{r}_1}U_{R_1}(\mathbb{G}_{b}^{*})] + \\& \cdots +  a^2_m\lim\limits_{b \to \infty} E_*\Delta^2[\sqrt{b} \rho_{n}^{-\mathfrak{r}_m}U_{R_m}(\mathbb{G}_{b}^{*})]  \Big\}
  \stackrel{\text{\eqref{eq:assump1}}}{=} 0.
\end{aligned}
\end{equation}
\item Now we prove  \eqref{eq:massump2}. Again, we first consider the case of a pair of motifs $R,R'$ to illustrate the procedure.
\begin{equation}
    \begin{aligned}
        \var_*\big[  \Theta^{(\alpha,\beta)}_{R,R'}\big] 
         = & \var_*\big[  \alpha\sqrt{b}\rho^{-\mathfrak{r}}_{n}U_{R}(\mathbb{G}^{*}_{b}) + \beta\sqrt{b}\rho^{-\mathfrak{r'}}_{n}U_{R'}(\mathbb{G}^{*}_{b})\big] \\
          =& \alpha^2\var_*\big[ \sqrt{b}\rho^{-\mathfrak{r}}_{n}U_{R}(\mathbb{G}^{*}_{b})\big]+\beta^2\var_*\big[\sqrt{b}\rho^{-\mathfrak{r'}}_{n}U_{R'}(\mathbb{G}^{*}_{b})\big]\\&+2 \alpha\beta\operatorname{Cov}_*\big[\sqrt{b}\rho^{-\mathfrak{r}}_{n}U_{R}(\mathbb{G}^{*}_{b}), \sqrt{b}\rho^{-\mathfrak{r'}}_{n}U_{R'}(\mathbb{G}^{*}_{b})\big].
    \end{aligned}
\end{equation}

By Lemma~\ref{coro:lim_var_subsample},  under Assumptions \ref{ass:rho_n_h_n} and \ref{ass:b}, with probability one,
\begin{equation}
 \lim\limits_{b \to \infty} \var_*\big[  \Theta^{(\alpha,\beta)}_{R,R'}\big]  = (1-c_2)\lim_{b \to \infty}\var\Big[\alpha\sqrt{b}\rho^{-\mathfrak{r}}_{b}U_{R}(\mathbb{G}_{b})+\beta\sqrt{b}\rho^{-\mathfrak{r'}}_{b}U_{R'}(\mathbb{G}_{b})\Big].
\end{equation}

Therefore, for any sequence of networks, condition in \eqref{eq:massump2} holds with probability one if
\begin{equation}
\label{assumpb2ref1}
    0 < c_1 \leqslant \lim_{b \to \infty}\var\Big[\alpha\sqrt{b}\rho^{-\mathfrak{r}}_{b}U_{R}(\mathbb{G}_{b})+\beta\sqrt{b}\rho^{-\mathfrak{r'}}_{b}U_{R'}(\mathbb{G}_{b})\Big]  \leqslant c_2 < \infty.
\end{equation}

Now we define:
\begin{align*}
\tilde{\sigma}^2_R &= \lim_{b \to \infty}\var\left[ \sqrt{b}\rho^{-\mathfrak{r}}_{b}U_{R}(\mathbb{G}_{b})\right], \\
\tilde{\sigma}^2_{R'} &= \lim_{b \to \infty} \var\left[ \sqrt{b}\rho^{-\mathfrak{r'}}_{b}U_{R'}(\mathbb{G}_{b})\right], \\
\tilde{\sigma}_{R,R'} &= \lim_{b \to \infty} \Cov\left[ \sqrt{b}\rho^{-\mathfrak{r}}_{b}U_{R}(\mathbb{G}_{b}),  \sqrt{b}\rho^{-\mathfrak{r'}}_{b}U_{R'}(\mathbb{G}_{b})\right],
\end{align*}
and recall that Proposition~\ref{prop: lim_var_graphon} implies that 
$\tilde{\sigma}^2_R$, $\tilde{\sigma}^2_{R'}$, and $\tilde{\sigma}^2_{R,R'}$ are the constants if $b\rho_b^{\max{\{r,r'\}}/2} \to \infty$. Thus, we have
\begin{equation*}
    \begin{aligned}
       \lim_{b \to \infty}\var\Big[\alpha\sqrt{b}\rho^{-\mathfrak{r}}_{b}U_{R}(\mathbb{G}_{b})+\beta\sqrt{b}\rho^{-\mathfrak{r'}}_{b}U_{R'}(\mathbb{G}_{b})\Big] &=  \tilde{\sigma}^2_R\alpha^2 +  2\tilde{\sigma}^2_{R,R'}\alpha\beta + \tilde{\sigma}^2_{R'}\beta^2\\
        &= \beta^2\left[\tilde{\sigma}^2_{R}\frac{\alpha^2}{\beta^2} + 2\tilde{\sigma}_{R,R'}\frac{\alpha}{\beta} +\tilde{\sigma}^2_{R'}\right].
    \end{aligned}
\end{equation*}

Lemma~\ref{lem: samplingdistribution_bickel} implies that $(\sqrt{b}\rho^{-\mathfrak{r}}_{b}U_{R}(\mathbb{G}_{b}),\sqrt{b}\rho^{-\mathfrak{r'}}_{b}U_{R'}(\mathbb{G}_{b}))$ converges in distribution to a bivariate Gaussian distribution and Assumption~\ref{ass:non_degenerate} implies that: $\tilde{\sigma}_{R,R'} < \sqrt{\tilde{\sigma}_{R}\tilde{\sigma}_{R'}}$. 

Therefore, $\beta^2\left[\tilde{\sigma}^2_{R}\frac{\alpha^2}{\beta^2} + 2\tilde{\sigma}_{R,R'}\frac{\alpha}{\beta} +\tilde{\sigma}^2_{R'}\right]$ has no real root, leading to
\begin{equation*}
    \lim_{\substack{b \to \infty}}\var\Big[\alpha\sqrt{b}\rho^{-\mathfrak{r}}_{b}U_{R}(\mathbb{G}_{b})+\beta\sqrt{b}\rho^{-\mathfrak{r'}}_{b}U_{R'}(\mathbb{G}_{b})\Big] > 0,
\end{equation*}
which is satisfied by letting  $c_1 = 1/2\lim_{\substack{b \to \infty}}\var\big[\alpha\sqrt{b}\rho^{-\mathfrak{r}}_{b}U_{R}(\mathbb{G}_{b})+\beta\sqrt{b}\rho^{-\mathfrak{r'}}_{b}U_{R'}(\mathbb{G}_{b})\big]$. 
On the other hand, we set the upper bound 
$c_2 = \max\{4\alpha^2\tilde{\sigma}^2_{R}, 4\beta^2\tilde{\sigma}^2_{R'}\}$ based on 
\begin{equation*}
 \begin{aligned}
    & \lim_{\substack{b \to \infty}}\var\Big[\alpha\sqrt{b}\rho^{-\mathfrak{r}}_{b}U_{R}(\mathbb{G}_{b})+\beta\sqrt{b}\rho^{-\mathfrak{r'}}_{b}U_{R'}(\mathbb{G}_{b})\Big] \\
   \leqslant &4\max\Big\{ \lim_{\substack{b \to \infty }} \alpha^2\var\left[ \sqrt{b}\rho^{-\mathfrak{r}}_{b}U_{R}(\mathbb{G}_{b})\right], \lim_{\substack{b \to \infty}} \beta^2\var\left[ \sqrt{b}\rho^{-\mathfrak{r'}}_{b}U_{R'}(\mathbb{G}_{b})\right]\Big\}.
\end{aligned}   
\end{equation*}

Thus, \eqref{eq:massump2} is verified.
\medskip

More generally, for $m$ motifs, Lemma~\ref{coro:lim_var_subsample} indicates
\begin{equation}
  \lim\limits_{b \to \infty } \var_*\big[\Theta\big]  = \lim_{b \to \infty}\var\Big[a_1\sqrt{b}\rho^{-\mathfrak{r}_1}_{b}U_{R_1}(\mathbb{G}_{b})+ \cdots +a_m\sqrt{b}\rho^{-\mathfrak{r}_m}_{b}U_{R_m}(\mathbb{G}_{b})\Big]
\end{equation}
with probability one. From Lemma~\ref{lem: samplingdistribution_bickel}, we have $$\Big(a_1\sqrt{b}\rho^{-\mathfrak{r}_1}_{b}U_{R_1}(\mathbb{G}_{b}), \cdots ,a_m\sqrt{b}\rho^{-\mathfrak{r}_m}_{b}U_{R_m}(\mathbb{G}_{b})\Big) \stackrel{d}{\longrightarrow} \mathcal{N}\Big(0, \Sigma_{[R_m]}\Big).$$

Let $q$ be a $1 \times m$ vector with all elements equal to one, with the positive definiteness of $\Sigma_{[R_m]}$, we have:
$$\lim_{b \to \infty}\var\Big[a_1\sqrt{b}\rho^{-\mathfrak{r}_1}_{b}U_{R_1}(\mathbb{G}_{b})+ \cdots +a_m\sqrt{b}\rho^{-\mathfrak{r}_m}_{b}U_{R_m}(\mathbb{G}_{b})\Big] = q\Sigma_{[R_m]}q^{\top} > 0.$$ 

Hence, the non-lattice condition in Equation \eqref{eq:massump2} holds by setting
$$c_1 = \frac{1}{2}\lim_{b \to \infty}\var\Big[a_1\sqrt{b}\rho^{-\mathfrak{r}_1}_{b}U_{R_1}(\mathbb{G}_{b})+ \cdots +a_m\sqrt{b}\rho^{-\mathfrak{r}_m}_{b}U_{R_m}(\mathbb{G}_{b})\Big]$$ and 
$c_2 = \max\{m a^2_1\tilde{\sigma}^2_{R_1}, \cdots, ma^2_m\tilde{\sigma}^2_{R_m}\}.$

\item To show \eqref{eq:massump3}, first consider two motifs $R$ and $R'$, so that 
\begin{equation*}
    \begin{aligned}
       g_{1, \Theta^{(\alpha,\beta)}_{R,R'}}\left(\mathbb{V}_1\right) \stackrel{\text{\eqref{eq:linear_combine_sum_g}}}{=}  \alpha\sqrt{b}\rho^{-\mathfrak{r}}_{n}g_{1,R}\left(\mathbb{V}_1\right)+ \beta \sqrt{b}\rho^{-\mathfrak{r'}}_{n}g_{1,R'}\left(\mathbb{V}_1\right).
    \end{aligned}
\end{equation*}
We have, with probability one
$$
      \lim_{b \to \infty} g^2_{1, \Theta^{(\alpha,\beta)}_{R,R'}}\left(\mathbb{V}_1\right) \leqslant  2\lim_{b \to \infty} \alpha^2b\rho^{-2\mathfrak{r}}_{n}g^2_{1,R}\left(\mathbb{V}_1\right)+ 2\lim_{b \to \infty} \beta^2 b\rho^{-2\mathfrak{r'}}_{n}g^2_{1,R'}\left(\mathbb{V}_1\right) \stackrel{\text{\eqref{eq:addref16}}}{=} 0.
$$

Therefore, for any $\epsilon > 0$, there exists a sufficiently large $N$ and $B$ such that for any $n > N$ and $b > B$, with probability one 
$\mathbb{1}_{\{g^2_{1,  \Theta^{(\alpha,\beta)}_{R,R'}} > \epsilon \}} = 0.$ Hence,  with probability one,
\begin{equation}
   \lim_{b \to \infty}  bE_* 
    \big[g^2_{1,  \Theta^{(\alpha,\beta)}_{R,R'}}\left(\mathbb{V}_1\right)]\mathbb{1}_{\big\{g^2_{1,  \Theta^{(\alpha,\beta)}_{R,R'}}\left(\mathbb{V}_1\right) > \epsilon \big\}} = 0.
\end{equation}

For $m$ motifs, condition in Equation \eqref{eq:massump3} also holds as
$$ \lim_{b \to \infty} g^2_{1, \Theta}\left(\mathbb{V}_1\right) \stackrel{\text{\eqref{eq:linear_combine_sum_g}}}{\leqslant}  m\lim_{b \to \infty} a_1^2b\rho^{-2\mathfrak{r}}_{n}g^2_{1,R}\left(\mathbb{V}_1\right)+ \cdots + m\lim_{b \to \infty} a^2_mb\rho^{-2\mathfrak{r'}}_{n}g^2_{1,R'}\left(\mathbb{V}_1\right) \stackrel{\text{\eqref{eq:addref16}}}{=} 0.$$ Lemma \ref{lem: normal_UR_subsample} thus indicates the asymptotic normality of
$a_1\sqrt{b}\rho^{-\mathfrak{r}_1}_{n}U_{R_1}(\mathbb{G}^{*}_{b}) + \cdots + a_m\sqrt{b}\rho^{-\mathfrak{r}_m}_{n}U_{R_m}(\mathbb{G}^{*}_{b})$.
\end{enumerate}

Finally, note that the above argument is for any arbitrary linear combination, by Corollary 4.6.9 of \cite{casella2024statistical}, we have
\begin{equation*}
\begin{aligned}
    &\sqrt{b}\Big\{\big[\rho^{-\mathfrak{r}_1}_{n}U_{R_1}(\mathbb{G}^{*}_{b}), \cdots, \rho^{-\mathfrak{r}_m}_{n}U_{R_m}(\mathbb{G}^{*}_{b})\big]- \big[\rho^{-\mathfrak{r}_1}_{n}U_{R_1}(G), \cdots, \rho^{-\mathfrak{r}_m}_{n}U_{R_m}(G)\big]\Big\} \\ \rightarrow &\mathcal{N}\big[0, \Sigma_{*[R_m]}\big] ~\text{in distribution}
\end{aligned}
\end{equation*}
where $\Sigma_{*[R_m]}$ is the corresponding asymptotic covariance matrix.

\end{proof}

\subsection{Proof of Lemma \ref{coro:lim_var_subsample}}
\label{D1}

\begin{proof}

The following result has been developed in \eqref{eq:more_cov_subsample}.
    \begin{equation*}
 \begin{aligned}
     &\Cov_*\big[\sqrt{b}\rho^{-\mathfrak{r}}_{n}U_{R}(\mathbb{G}^{*}_{b}),\sqrt{b}\rho^{-\mathfrak{r}'}_{n} U_{R'}(\mathbb{G}^{*}_{b})\big] \\\stackrel{\text{\eqref{eq:more_cov_subsample}}}{=}  &\sum_{q=0}^{\min\{r,r'\}} \sum_{S\in \mathcal{S}_{R, R'}^{(q)}} \dfrac{c_Sr!r'!}{s!\rho^{\mathfrak{r}+\mathfrak{r}'}_{n}}\Big[\dfrac{(b-r)!(b-r')!n!(n-s)! - b!(b-s)!(n-r)!(n-r')! }{(b-1)!(b-s)!n!(n-s)!} \Big] U_{S}(G).
\end{aligned}   
\end{equation*}

Notice $c_S$, $r!$, $r'!$ are fixed quantities, while we have to study the limiting behaviors for other quantities under different values of $q$. 

\begin{itemize}
    \item When $q = 0$,  we have $\mathfrak{s} = \mathfrak{r} + \mathfrak{r}'$ and $s = r + r'$.  Thus,
      \begin{align*}
      & \dfrac{(b-r)!(b-r')!n!(n-s)! - b!(b-s)!(n-r)!(n-r')! }{(b-1)!(b-s)!n!(n-s)!}\\
     =   &\dfrac{(b-r)!(b-r')!n!(n-r - r')! - b!(b-r - r')!(n-r)!(n-r')! }{(b-1)!(b-r - r')!n!(n-r - r')!}\\
      =  & \dfrac{n\cdots (n-r+1)(b-r')\cdots (b-r-r'+1) -b(b-1)\cdots (b-r+1) (n-r')\cdots (n-r-r'+1)}{(b-1)\cdots (b-r+1)n(n-1)\cdots (n-r+1)}\\
      = & \dfrac{  \mathrm{IV} - \mathrm{V}}{  \mathrm{VI}}.
    \end{align*}
    
Similar to \eqref{eq:limref}, Now we study each component as follows:
    \begin{align*}
   \mathrm{IV}= & b^{r}n^r - (r'+r'+1+\cdots r'+r-1)b^{r-1}n^r -[1+2+\cdots +(r-1)]b^{r}n^{r-1}\\
   &+O(b^{r-2}n^r) + O(b^{r}n^{r-2}) + O(b^{r-1}n^{r-1})
     \\=& b^{r}n^r - \dfrac{(r+2r'-1)r}{2}b^{r-1}n^r -  \dfrac{(r-1)r}{2}b^{r}n^{r-1}
     \\  &+O(b^{r-2}n^r) + O(b^{r}n^{r-2}) + O(b^{r-1}n^{r-1}).\\
      \mathrm{V}= & b^{r}n^r - (r+r+1+\cdots r+r-1)n^{r-1}b^r -[1+2+\cdots +(r-1)]n^{r}b^{r-1}\\
   &+O(b^{r-2}n^r) + O(b^{r}n^{r-2}) + O(b^{r-1}n^{r-1})
     \\=& b^{r}n^r - \dfrac{(r + 2r'-1)r}{2}n^{r-1}b^r -  \dfrac{(r-1)r}{2}n^{r}b^{r-1}
      \\&+O(b^{r-2}n^r) + O(b^{r}n^{r-2}) + O(b^{r-1}n^{r-1}).\\
\end{align*}
Consequently,
  \begin{equation*}
\begin{aligned}
        \mathrm{IV} - \mathrm{V} = &b^{r}n^r - \dfrac{(r+2r-1)r}{2}b^{r-1}n^r -  \dfrac{(r-1)r}{2}b^{r}n^{r-1}
     \\&+O(b^{r-2}n^r) + O(b^{r}n^{r-2}) + O(b^{r-1}n^{r-1})\\
     & - \big[b^{r}n^r - \dfrac{(r+2r'-1)r}{2}n^{r-1}b^r -  \dfrac{(r-1)r}{2}n^{r}b^{r-1}
     \\&+O(b^{r-2}n^r) + O(b^{r}n^{r-2}) + O(b^{r-1}n^{r-1})\big]\\
     = &\dfrac{-(r+2r'-1)r}{2}n^{r-1}b^{r-1}(n-b) + \dfrac{(r-1)r}{2}n^{r-1}b^{r-1}(n-b)\\
     & + O(b^{r-2}n^r) + O(n^{r-2}b^r)  + O(b^{r-1}n^{r-1}) \\
     = & (-rr')n^{r-1}b^{r-1}(n-b) + O(b^{r-2}n^r) + O(n^{r-2}b^r)  + O(b^{r-1}n^{r-1}).
\end{aligned}
\end{equation*}

Thus, we have
\begin{equation*}
    \begin{aligned}
     \dfrac{  \mathrm{IV} - \mathrm{V}}{  \mathrm{VI}} = & (-rr')\frac{n-b}{n} + o(1),
    \end{aligned}
\end{equation*}

and by Assumption \ref{ass:b},
\begin{equation}
    \begin{aligned}
      \lim\limits_{b \to \infty}  \dfrac{  \mathrm{IV} - \mathrm{V}}{  \mathrm{VI}} = & (-rr')(1 - c_2).
    \end{aligned}
\end{equation}

On the other hand,  since $\rho^{\mathfrak{r}+\mathfrak{r}'}_{n} = \rho^{\mathfrak{s}}_{n} \leqslant \rho^{2\mathfrak{r}_1}_{n} $, by Lemma \ref{lem: lim_netmoment_graphcon}.
\begin{equation*}
\begin{aligned}
        \lim\limits_{b \to \infty } \sum_{S\in\mathcal{S}_{R,R'}^{(0)}}\dfrac{c_Sr!r'!}{s!\rho^{\mathfrak{r}+\mathfrak{r}'}_{n}}\Big(\dfrac{  \mathrm{IV} - \mathrm{V}}{  \mathrm{VI}}\Big) U_{S}(G) & \stackrel{\text{\eqref{eq:limref}}}{=} (1 - c_2) \lim\limits_{b \to \infty} \sum_{S\in\mathcal{S}_{R,R'}^{(0)}}\dfrac{-c_Sr!r'!rr'}{s!} \rho^{-\mathfrak{s}}_{n}U_{S}(G)\\ &\stackrel{\text{\eqref{eq:lim_rUR_subsample}}}{=} (1 - c_2)\sum_{S\in\mathcal{S}_{R,R'}^{(0)}} \dfrac{-c_Sr!r'!rr'}{|\text{Aut}(S)|}P_w(S).\\
\end{aligned}
\end{equation*}

\item When $q = 1$,  we have $\mathfrak{s} = \mathfrak{r} + \mathfrak{r}'$ and $s = r + r'-1$.  Thus,
\begin{equation}
\label{eq:lemmasd1}
    \begin{aligned}
        &\dfrac{(b-r)!(b-r')!n!(n-s)! - b!(b-s)!(n-r)!(n-r')! }{(b-1)!(b-s)!n!(n-s)!}\\
     =   &\dfrac{(b-r)!(b-r')!n!(n-r - r'+1)! - b!(b-r - r'+1)!(n-r)!(n-r')! }{(b-1)!(b-r - r'+1)!n!(n-r - r'+1)!}\\
      =  &  1 - \frac{b}{n} + o(1).
    \end{aligned}
\end{equation}

 Since $\rho^{\mathfrak{r}+\mathfrak{r}'}_{n} = \rho^{\mathfrak{s}}_{n} \leqslant \rho^{2\mathfrak{r}_1}_{n}$,
by Assumption \ref{ass:b},
\begin{equation*}
\begin{aligned}
          & \lim\limits_{b \to \infty } \sum_{S\in\mathcal{S}_{R,R'}^{(1)}}\dfrac{c_Sr!r'!}{s!\rho^{\mathfrak{s}}_{n}}\Big[\dfrac{(b-r)!(b-r')!n!(n-s)! - b!(b-s)!(n-r)!(n-r')! }{(b-1)!(b-s)!n!(n-s)!} \Big]  U_{S}(G) \\ \stackrel{\text{\eqref{eq:lemmasd1}}}{=} &(1 -c_2) \lim\limits_{b \to \infty } \sum_{S\in\mathcal{S}_{R,R'}^{(1)}}\dfrac{c_Sr!r'!}{s!\rho^{\mathfrak{s}}_{n}}U_{S}(G) \stackrel{\text{\eqref{eq:lim_rUR_subsample}}}{=} (1 -c_2)\sum_{S\in\mathcal{S}_{R,R'}^{(1)}}c_S \dfrac{r!r'!}{|\text{Aut}(S)|}P_w(S).
\end{aligned}
\end{equation*}

\item When $q >1$, we have 
\begin{align*}
        \dfrac{[(b-r)!(b-r)!n!(n-s)! - b!(b-s)!(n-r)!(n-r)!] }{(b-1)!(b-s)!n!(n-s)!}= O(\frac{1}{b^{(q-1)}}).
    \end{align*}

 In addition, since $\rho_n^{(\mathfrak{s}-\mathfrak{r}-\mathfrak{r}')}$ is at most $O(\rho_n^{-(q-1)q/2})$ and  $\rho^{q/2}_nb \to \infty$, we have
\begin{equation*}
\begin{aligned}
       &\lim\limits_{\substack{b \to \infty}} \sum_{S\in \mathcal{S}_{R, R'}^{(q)}} \dfrac{c_Sr!r'!}{s!\rho^{\mathfrak{r}+\mathfrak{r}'}_{n}}\Big[\dfrac{(b-r)!(b-r')!n!(n-s)! - b!(b-s)!(n-r)!(n-r')! }{(b-1)!(b-s)!n!(n-s)!} \Big] U_{S}(G)\\
     = & \lim\limits_{\substack{b \to \infty}} \sum_{S\in \mathcal{S}_{R, R'}^{(q)}} \dfrac{c_Sr!r'!}{s!\rho^{\mathfrak{r}+\mathfrak{r}'-\mathfrak{s}}_{n}} O(\frac{1}{b^{(q-1)}}) \rho_n^{-\mathfrak{s}}U_{S}(G) = 
    \lim\limits_{\substack{b \to \infty}} \sum_{S\in \mathcal{S}_{R, R'}^{(q)}} \dfrac{c_Sr!r'!}{s!} O(\frac{1}{b\rho^{q/2}_n})^{(q-1)} \rho_n^{-\mathfrak{s}}U_{S}(G)
       \stackrel{\text{\eqref{eq:lim_rUR_subsample}}}
     =0.
\end{aligned}
\end{equation*}
\end{itemize}

Therefore, with probability one,
\begin{equation*}
\begin{aligned}
&\lim\limits_{ b\to \infty}\rho^{-(\mathfrak{r}+\mathfrak{r}')}_{n}\Cov_*\big[\sqrt{b}U_{R}(\mathbb{G}^{*}_{b}),\sqrt{b}U_{R'}(\mathbb{G}^{*}_{b})  \big]\\ 
    \\ &=\big( 1 - c_2\big) \Big[ \sum_{S\in\mathcal{S}_{R,R'}^{(1)}}c_S \dfrac{r!r'!}{|\text{Aut}(S)|}P_w(S) - \sum_{S\in\mathcal{S}_{R,R'}^{(0)}} \dfrac{c_Sr!r'!rr'}{|\text{Aut}(S)|}P_w(S) \Big] 
    \\& = \big(1 - c_2\big)\lim\limits_{b \to \infty}\rho^{-(\mathfrak{r}+\mathfrak{r}')}_{b}\Cov\big[\sqrt{b}U_{R}(\mathbb{G}_{b}),\sqrt{b}U_{R'}(\mathbb{G}_{b})\big].\\
\end{aligned}
\end{equation*}

As a special, we have
\begin{equation}
\lim\limits_{b \to \infty }\var_*\big[\sqrt{b}\rho^{-\mathfrak{r}}_{n}U_{R}(\mathbb{G}^{*}_{b}) \big] =  \big(1 - c_2\big)\lim\limits_{b \to \infty}\var\big[\sqrt{b}\rho^{-\mathfrak{r}}_{b}U_{R}(\mathbb{G}_{b})\big]
\end{equation}
with probability one.

\end{proof}

\subsection{Proof of Lemma \ref{lem: lim_netmoment_graphcon}}
\label{D2}
\begin{proof}
    
The current proof is adapted from the techniques of \cite{lovasz2006limits} and \cite{zhao2023graph}.  We start with a sequence of graphs $\{\mathbb{G}^{(i)}\}_{i=1}^n$. The first graph  $\mathbb{G}^{(1)}$ is only a node $v_1$ with latent position $\xi_1 \sim \mathrm{Unif}[0,1]$. The second graph $\mathbb{G}^{(2)}$ contains two nodes $v_1, v_2$: $v_1$ is already associated with the latent position $\xi_1$ in $\mathbb{G}^{(1)}$  and we sample $\xi_2 \sim \mathrm{Unif}[0,1]$. The probability of an edge between $v_1, v_2$ is $h_n(\xi_1,\xi_2)$. In this way,  $\{\mathbb{G}^{(i)}\}_{i=1}^n$ is generated by incrementally adding one node and the corresponding edges at a time, and previously selected nodes and edges are not revisited. Furthermore, we have $\mathbb{G}^{(i)} \sim \mathbb{G}^{h_n}_{i}$. 
\medskip

Let $\phi$: $V(R) \rightarrow V(G)$ be an injective mapping, and let $A_{\phi}$ denote the event that $\phi$ is a homomorphism from $R$ to graph $\mathbb{G}^{(n)}$. We define the sequence $\{A_i\}_{i = 1}^n$ as 
\begin{equation}
    A_i =(n)^{-1}_r\sum\limits_{\phi}pr(A_{\phi} \mid G^{(i)}).
\end{equation}

Based on the definition of $A_i$, we have 
\begin{align}
A_n &= (n)^{-1}_r\sum_{\phi}pr(A_{\phi} \mid \mathbb{G}^{(n)}) = t(R,\mathbb{G}^{(n)}), \\
A_0 &= (n)^{-1}_r\sum_{\phi}pr(A_{\phi}) = \int_{[0,1]^r} \prod_{(v_i,v_j) \in \eE(R)} h_n\left(\xi_i, \xi_j\right) \prod_{v_i \in V(R)} d \xi_i\stackrel{\text{\eqref{eq:PR}}}{=}  P_{h_n}(R).
\end{align}

\cite{lovasz2006limits} showed that $\{A_n\}$ is a martingale and $|A_{i} - A_{i-1}| \leqslant r/n$  in their Theorem 2.5. 
Then by invoking Azuma's inequality, they showed that, for every $\delta > 0$,
$$\pr\Big(\big|t(R,\mathbb{G}_n) - P_{h_n}(R)\big| > \delta\Big) \leqslant 2\exp\left(-\delta^2n/2r^2\right).$$
On the other hand, the Proposition 1 of \cite{amini2012counting} implies that
$$X_R(G) = \text{inj}(R,G)/|\mathrm{Aut}(R)|.$$

Therefore, we have
\begin{equation}
    \begin{aligned}
        \rho^{-\mathfrak{r}}_nU_R(G) = \rho^{-\mathfrak{r}}_n\binom{n}{r}^{-1}X_R(G)=\rho^{-\mathfrak{r}}_n\binom{n}{r}^{-1}\frac{\text{inj}(R,G)}{|\mathrm{Aut}(R)|} =\rho^{-\mathfrak{r}}_n\frac{r!t(R,G)}{|\mathrm{Aut}(R)|}
    \end{aligned}
\end{equation}
and 
$$\pr\Big(\big| \rho^{-\mathfrak{r}}_nU_R(\mathbb{G}_n) -\frac{\rho^{-\mathfrak{r}}_nr!}{|\mathrm{Aut}(R)|}P_{h_n}(R)\big| > \frac{r!\rho^{-\mathfrak{r}}_n\delta}{|\mathrm{Aut}(R)|}\Big) \leqslant 2\exp\left(-\frac{\delta^2n}{2r^2}\right).$$ For every $0<\epsilon<1$, let $\delta = \rho^{\mathfrak{r}}_n\epsilon|\mathrm{Aut}(R)|/r!$, then by \eqref{eq:lim_exp_graphon} we have
$$\pr\Big(\big| \rho^{-\mathfrak{r}}_nU_R(\mathbb{G}_n) - E[ \rho^{-\mathfrak{r}}_nU_R(\mathbb{G}_n)]\big| > \epsilon \Big) \leqslant 2\exp\left(-\frac{\epsilon^2|\mathrm{Aut}(R)|^2\rho^{2\mathfrak{r}}_nn}{2(rr!)^2}\right).$$

When  $\rho_n w(u,v) \leq 1$ for all $u,v$, the quantity $\rho^{-\mathfrak{r}}_nP_{h_n}(R) = P_w(R)$ does not depend on $n$. Thus, if there exist some $c_1 > 1$, such that  $n\rho^{2\mathfrak{r}}_n > c_1\log n$, then the sum of the following series  converges:
$$\sum\limits_n 2\exp\left(-\frac{\epsilon^2|\mathrm{Aut}(R)|^2\rho^{2\mathfrak{r}}_nn}{2(rr!)^2}\right).$$
 By Borel-Cantelli lemma, we have $$ \rho^{-\mathfrak{r}}_nU_R(\mathbb{G}_n) \stackrel{\text{a.s}}{\longrightarrow} \frac{\rho^{-\mathfrak{r}}_nr!}{|\mathrm{Aut}(R)|}P_{h_n}(R)=\frac{r!}{|Aut(R)}P_{w}(R).$$ 

As a special case, when motif $R$ is an edge, $\wh{\rho}_{G} = U_R(G)$. Thus, Lemma \ref{lem: lim_netmoment_graphcon} implies that with probability one:
\begin{equation*}
 \lim\limits_{n \to \infty} \rho^{-\mathfrak{r}}_n\wh{\rho}_{G} = \frac{r!}{|\mathrm{Aut}(R)|}P_w(R).
\end{equation*} Since $\mathfrak{r} = 1$, $r = 2$, $|\mathrm{Aut}(R)| = 2$ \citep{rodriguez2014automorphism} and $P_w(R) = 1$ \citep{bickel2011method}, with probability one, we have 
\begin{equation}
    \label{eq: edgeas} 
    \lim\limits_{n \to \infty} \rho^{-1}_n\wh{\rho}_{G}  = 1.
\end{equation}
\end{proof}

%% file: Appendix/Appendix_SuppleA6.tex
\section{Proof of of Theorem \ref{theo:consistent}}
\label{subsec:consis}

\begin{proof}[of Theorem \ref{theo:consistent}] 
We first focus on a single motif $R$, and we want to show that, with probability one:
\begin{equation} 
\label{eq:uni_consistent}
\begin{aligned}
       & \sup\limits_{t \in \mathbb{R}}\big| J^R_{*,n,b}(t) - J^R_{b,(1-\frac{b}{n})}(t)
       \big| \rightarrow 0.
\end{aligned}
\end{equation}

It is easy to see that \eqref{eq:uni_consistent} can be upper bounded as
       \begin{equation}
       \label{eq:ksupper}
    \begin{aligned}
       \sup\limits_{t \in \mathbb{R}}\big| J^R_{*,n,b}(t) - J^R_{b,(1-\frac{b}{n})}(t)
       \big|   \leqslant & \sup\limits_{t \in \mathbb{R}}\big|  J^R_{b,(1-\frac{b}{n})}(t) - J^R_{b,c}(t) \big| + \sup\limits_{t \in \mathbb{R}}\big| J^R_{b,c}(t) -\Phi\big(\frac{t}{\sigma_{c,R}}\big)\big| \\
      &  + \sup\limits_{t \in \mathbb{R}}\big|\Phi\big(\frac{t}{\sigma_{c,R}}\big) -\Phi\big(\frac{t}{\sigma_{*R}}\big)\big|+ \sup\limits_{t \in \mathbb{R}}\big|J^R_{*,n,b}(t) - \Phi\big(\frac{t}{\sigma_{*R}}\big)\big|,
    \end{aligned}
\end{equation}
where $c = 1- c_2$, $\sigma^2_{c,R} = c  \sigma^2_R$ with $\sigma^2_R = \lim\limits_{b \to \infty}\var\big[\sqrt{b}\rho^{-\mathfrak{r}}_{b}U_{R}(\mathbb{G}_{b})\big]$, $\sigma^2_{*R} = \lim\limits_{b \to \infty}\var_*\big[\sqrt{b}\rho^{-\mathfrak{r}}_{b}U_{R}(\mathbb{G}^{*}_{b})\big]$, and $\Phi\big(\frac{t}{\sigma_{c,R}}\big)$ denotes the univariate normal CDF of $\mathcal{N}(0, c\sigma^2_R)$.

Now we are in the position to show that all four components on the right-hand side of \eqref{eq:ksupper} go to zero. By Slutsky's theorem, we have
\begin{equation*}
     \sup\limits_{t \in \mathbb{R}}\big|  J^R_{b,(1-\frac{b}{n})}(t) - J^R_{b,c}(t) \big| \to 0.
\end{equation*}
Lemma~\ref{lem: samplingdistribution_bickel} implies that:
\begin{equation*}
   \sqrt{b}\big\{{\rho}^{-\mathfrak{r}}_{\mathbb{G}_{b}}U_R(\mathbb{G}_{b})-\mE\big[\rho^{-\mathfrak{r}}_{b}U_R(\mathbb{G}_{b})\big]\big\} \rightarrow \mathcal{N}(0, \sigma^2_R)~\text{in distribution},
\end{equation*}
which gives:
\begin{equation*}
\label{eq:weak}
   \sqrt{bc}\big\{{\rho}^{-\mathfrak{r}}_{\mathbb{G}_{b}}U_R(\mathbb{G}_{b})-\mE\big[\rho^{-\mathfrak{r}}_{b}U_R(\mathbb{G}_{b})\big]\big\} \rightarrow \mathcal{N}(0, c\sigma^2_R)~\text{in distribution}.
\end{equation*}
Since $\sigma_R$ is fixed, the continuity of $\Phi$ leads to
\begin{equation*}
 \label{eq:firstterm}
    \sup\limits_{t \in \mathbb{R}}\big|  J^R_{b, c}(t)  -\Phi\big(\frac{t}{\sigma_{c,R}}\big)\big| \to 0.
\end{equation*}

\noindent For the third term, recall that $G \sim \mathbb{G}_n$. Lemma \ref{coro:lim_var_subsample} implies that $\sigma^2_{*R}$ equals to  $c\sigma^2_R$ almost surely. Then, with continuity, we have
\begin{equation*}
    \label{eq:secondterm}
    \sup\limits_{t \in \mathbb{R}}\big|\Phi\big(\frac{t}{\sigma_{*R}}\big) - \Phi\big(\frac{t}{\sigma_{c,R}}\big)\big| \rightarrow 0.
\end{equation*}

\noindent Next, since Assumptions \ref{ass:rho_n_h_n}, \ref{ass:b}, and \ref{ass:non_degenerate} are satisfied, then by  \eqref{eq:unit_subsample_distribution_with_truepn}, with probability one:
        \begin{equation*}
\sqrt{b}\big[\rho_{n}^{-\mathfrak{r}}U_{R}(\mathbb{G}^{*}_{b}) -\rho_{n}^{-\mathfrak{r}}U_{R}(G)\big] \rightarrow  \mathcal{N}(0,\sigma^2_{*R})~\text{in distribution}.
\end{equation*}

In addition, \eqref{eq: edgeas} implies that $ \lim\limits_{n \to \infty} \rho^{-1}_n\wh{\rho}_{G}  = 1$  with probability one. Consequently, by Slutsky's theorem,  with probability one: 
     \begin{equation*}
\sqrt{b}\big[\wh{\rho}^{-\mathfrak{r}}_{G}U_{R}(\mathbb{G}^{*}_{b}) -\wh{\rho}^{-\mathfrak{r}}_{G}U_{R}(G)\big]  \rightarrow \mathcal{N}\big(0,\sigma^2_{*R}\big)~\text{in distribution},
\end{equation*}
which further implies with probability one: 
\begin{equation*}
    \sup\limits_{t \in \mathbb{R}}\big|J^R_{*,n,b}(t) - \Phi\big(\frac{t}{\sigma_{*R}}\big)\big| \to 0.
\end{equation*}
Therefore, with probability one: 
\begin{equation*}
    \sup\limits_{t \in \mathbb{R}}\big|J^R_{*,n,b}(t) - J^R_{b,c}(t)\big| \rightarrow 0.
\end{equation*}

Now we turn to consider $m$ motifs $ \{R_1,\cdots,R_m\}$. For simplicity, let $[t_m] = \{t_1,\cdots,t_m\}$ and $[R_m] = \{R_1,\cdots,R_m\}$. Let $\Sigma_{c[R_m]} = c\cdot \Sigma_{[R_m]}$.
Under Assumption \ref{ass:non_degenerate}, similar as before, we break the Kolmogorov-Smirnov distance into three parts:
\begin{equation*}
    \begin{aligned}
     &\sup\limits_{[t_m] \in \mathbb{R}^m}\big|J^{[R_m]}_{*,n,b}([t_m]) - J^{[R_m]}_{b,c}([t_m])
       \big|  
      \leqslant \sup\limits_{[t_m] \in \mathbb{R}^m}\big| J^{[R_m]}_{b,c}([t_m]) - \Phi_{\Sigma_{c[R_m]}}\big(\frac{t_1}{\sigma_{c,R_1}},\cdots,\frac{t_m}{\sigma_{c,R_m}}\big)\big| \\ +& \sup\limits_{[t_m] \in \mathbb{R}^m}\big|\Phi_{\Sigma_{c[R_m]}}\big(\frac{t_1}{\sigma_{c,R_1}},\cdots,\frac{t_m}{\sigma_{c,R_m}}\big) -\Phi_{\Sigma_{*[R_m]}}\big(\frac{t_1}{\sigma_{*R_1}},\cdots,\frac{t_m}{\sigma_{*R_m}}\big)\big|\\ +&
      \sup\limits_{[t_m] \in \mathbb{R}^m}\big|  J^{[R_m]}_{*,n,b}([t_m]) - \Phi_{\Sigma_{*[R_m]}}\big(\frac{t_1}{\sigma_{*R_1}},\cdots,\frac{t_m}{\sigma_{*R_m}}\big)\big|,
    \end{aligned}
\end{equation*}
where for each $i \in [m]$, $\sigma^2_{c,R_i} = c  \sigma^2_{R_i}$ with $\sigma^2_{R_i} = \lim\limits_{n \to \infty}\var\big[\sqrt{n}\rho^{-\mathfrak{r}}_{n}U_{R_i}(\mathbb{G}_{n})\big]$, and $$\sigma^2_{*R_i} = \lim\limits_{b \to \infty}\var_*\big[\sqrt{b}\rho^{-\mathfrak{r}}_{n}U_{R_i}(\mathbb{G}^{*}_{b})\big].$$

As before, we now want to show that all three components go to zero. By Lemma~\ref{lem: samplingdistribution_bickel}:
\begin{equation*}
\begin{aligned}
     &\sqrt{b}\Big\{\big[{\rho}^{-\mathfrak{r}_1}_{\mathbb{G}_{b}}U_{R_1}(\mathbb{G}_{b}),\cdots,{\rho}^{-\mathfrak{r}_m}_{\mathbb{G}_{b}}U_{R_m}(\mathbb{G}_{b})\big] -  \big[\rho_b^{-\mathfrak{r}_1}\mE[U_{R_1}(\mathbb{G}_{b})],\cdots,\rho_b^{-\mathfrak{r}_m}\mE[U_{R_m}(\mathbb{G}_{b})]\big]\Big\} \\
     \to&~  \mathcal{N}[0, \Sigma_{[R_m]}]~ \text{in distribution}.
\end{aligned}
\end{equation*}

Since $c > 0 $ is a constant, by Slutsky's theorem:
\begin{equation*}
\begin{aligned}
     &\sqrt{bc}\Big\{\big[{\rho}^{-\mathfrak{r}_1}_{\mathbb{G}_{b}}U_{R_1}(\mathbb{G}_{b}),\cdots,{\rho}^{-\mathfrak{r}_m}_{\mathbb{G}_{b}}U_{R_m}(\mathbb{G}_{b})\big] -  \big[\rho_b^{-\mathfrak{r}_1}\mE[U_{R_1}(\mathbb{G}_{b})],\cdots,\rho_b^{-\mathfrak{r}_m}\mE[U_{R_m}(\mathbb{G}_{b})]\big]\Big\} \\
     \to&~  \mathcal{N}[0, \Sigma_{c[R_m]}]~ \text{in distribution}.
\end{aligned}
\end{equation*}
Thus, 
\begin{equation}
\label{eq:cem1}
     \sup\limits_{[t_m] \in \mathbb{R}^m}\big| J^{[R_m]}_{b,c}([t_m]) - \Phi_{\Sigma_{c[R_m]}}\big(\frac{t_1}{\sigma_{c,R_1}},\cdots,\frac{t_m}{\sigma_{c,R_m}}\big)\big| \rightarrow 0.
\end{equation} The second term goes to zero as Lemma \ref{coro:lim_var_subsample}  implies that $\Sigma_{*[R_m]}$ converge to $\Sigma_{c[R_m]}$ almost surely.

Under Assumptions \ref{ass:rho_n_h_n}-\ref{ass:non_degenerate}, by Theorem \ref{theo:subsample_distribution}, for any sequence $\{G^{(n)}\}$, with probability one:
\begin{equation*}
\begin{aligned}
    &\sqrt{b}\Big\{\big[\rho^{-\mathfrak{r}_1}_{n}U_{R_1}(\mathbb{G}^{*}_{b}), \cdots, \rho^{-\mathfrak{r}_m}_{n}U_{R_m}(\mathbb{G}^{*}_{b})\big]- \big[\rho^{-\mathfrak{r}_1}_{n}U_{R_1}(G), \cdots, \rho^{-\mathfrak{r}_m}_{n}U_{R_m}(G)\big]\Big\} \\ \rightarrow &\mathcal{N}\big[0, \Sigma_{*[R_m]}\big] ~\text{in distribution},
\end{aligned}
\end{equation*} Since $ \lim_{n \to \infty} \rho^{-\mathfrak{r}}_n\wh{\rho}_{G}  = 1$ with probability one,
\begin{equation*}
\begin{aligned}
    &\sqrt{b}\Big\{\big[\wh{\rho}^{-\mathfrak{r}_1}_{G}U_{R_1}(\mathbb{G}^{*}_{b}), \cdots,\wh{\rho}^{-\mathfrak{r}_m}_{G}U_{R_m}(\mathbb{G}^{*}_{b})\big]- \big[\wh{\rho}^{-\mathfrak{r}_1}_{G}U_{R_1}(G), \cdots, \wh{\rho}^{-\mathfrak{r}_m}_{G}U_{R_m}(G)\big]\Big\} \\ \rightarrow &\mathcal{N}\big[0, \Sigma_{*[R_m]}\big] ~\text{in distribution with probability one},
\end{aligned}
\end{equation*}
which further implies that with probability one:
\begin{equation}
\label{eq:cem2}
    \sup\limits_{[t_m] \in \mathbb{R}^m}\big|  J^{[R_m]}_{*,n,b}([t_m]) - \Phi_{\Sigma_{*[R_m]}}\big(\frac{t_1}{\sigma_{*R_1}},\cdots,\frac{t_m}{\sigma_{*R_m}}\big)\big| \to 0.
\end{equation} Therefore, with probability one: 
\begin{equation*}
    \sup\limits_{[t_m] \in \mathbb{R}^m}\big|J^{[R_m]}_{*,n,b}([t_m]) - J^{[R_m]}_{b,c}([t_m])
       \big|   \rightarrow 0,
\end{equation*}
which implies 
\begin{equation*}
    \sup\limits_{[t_m] \in \mathbb{R}^m}\big|J^{[R_m]}_{*,n,b}([t_m]) - J^{[R_m]}_{b,(1-\frac{b}{n})}([t_m])
       \big|   \rightarrow 0,
\end{equation*}
as $c = 1 - c_2$.

     \end{proof}

%% file: Appendix/Appendix_SuppleA7.tex
\section{Proof of the empirical consistency}
\label{subsec:consis1}

As before, let $c = 1 - c_2$, $[t_m] = \{t_1,\cdots,t_m\}$ and $[R_m] = \{R_1,\cdots,R_m\}$. 
The following lemmas are used for the proof.
\begin{lemma}[Theorem 1 in \cite{lunde2023subsampling}]
\label{lem: lunde}
Suppose there exists a CDF $J([t_m]) $, such that for  all continuity points of $J(\cdot)$,
\begin{equation*}
\begin{aligned} 
    &|  J^{[R_m]}_{b,c}([t_m]) - J([t_m]) | \to 0,\\
    &|J^{[R_m]}_{*,n,b}([t_m])  - J([t_m]) | \to 0.
\end{aligned}
\end{equation*} Then 
\begin{equation*}
      \widehat{J}^{[R_m]}_{*,n,b}([t_m])  \to J([t_m]) ~\text{in probability}.
\end{equation*}
\end{lemma}

\begin{proof}[of Lemma \ref{theo:consistent_empir}]
To begin with, let $J([t_m]) = \Phi_{\Sigma_{c[R_m]}}\big(t_1\sigma_{c,R_1}^{-1},\ldots, t_m \sigma_{c,R_m}^{-1}\big)$ and  
$$J_*([t_m]) = \Phi_{\Sigma_{*[R_m]}}\big(t_1\sigma_{*R_1}^{-1},\ldots,t_m\sigma_{*R_m}^{-1}\big).$$  
Under  Assumptions \ref{ass:rho_n_h_n}, \ref{ass:b},  and \ref{ass:non_degenerate}, \eqref{eq:cem1} and \eqref{eq:cem2} imply  $|  J^{[R_m]}_{b,c}([t_m]) - J([t_m]) | \to 0$ and $|J^{[R_m]}_{*,n,b}([t_m])  - J_*([t_m]) |   \to 0$, respectively.  

Moreover, Lemma \ref{coro:lim_var_subsample}  implies that $\Sigma_{*[R_m]}$ converge to $\Sigma_{c[R_m]}$ almost surely. Thus, $J_*([t_m])$ converge to $J([t_m])$ almost surely. Therefore, by Lemma \ref{lem: lunde}, with probability one (measure on the random network sequence): 
\begin{equation*}
      \widehat{J}^{[R_m]}_{*,n,b}([t_m])  \to J([t_m]) ~\text{in probability}.
\end{equation*}

Consequently, for all continuity points of $J(\cdot)$,
\begin{equation*}
     \sup\limits_{[t_m] \in \mathbb{R}^m}\big|   \widehat{J}^{[R_m]}_{*,n,b}([t_m])  -  J([t_m]) \big| \rightarrow 0.
\end{equation*}

Finally, based on \eqref{eq:cem1} and \eqref{eq:cem2}, we arrived at
\begin{equation*}
    \begin{aligned}
     &\sup\limits_{[t_m] \in \mathbb{R}^m}\big|  \widehat{J}^{[R_m]}_{*,n,b}([t_m]) 
       - J^{[R_m]}_{*,n,b}([t_m])\big|  
      \leqslant \sup\limits_{[t_m] \in \mathbb{R}^m}\big|   \widehat{J}^{[R_m]}_{*,n,b}([t_m])  - \Phi_{\Sigma_{c[R_m]}}\big(\frac{t_1}{\sigma_{c,R_1}},\cdots,\frac{t_m}{\sigma_{c,R_m}}\big)\big| \\ +& \sup\limits_{[t_m] \in \mathbb{R}^m}\big|\Phi_{\Sigma_{c[R_m]}}\big(\frac{t_1}{\sigma_{c,R_1}},\cdots,\frac{t_m}{\sigma_{c,R_m}}\big) -\Phi_{\Sigma_{*[R_m]}}\big(\frac{t_1}{\sigma_{*R_1}},\cdots,\frac{t_m}{\sigma_{*R_m}}\big)\big|\\ +&
      \sup\limits_{[t_m] \in \mathbb{R}^m}\Big|  J^{[R_m]}_{*,n,b}([t_m]) - \Phi_{\Sigma_{*[R_m]}}\big(\frac{t_1}{\sigma_{*R_1}},\cdots,\frac{t_m}{\sigma_{*R_m}}\big)\big| \to 0.
    \end{aligned}
\end{equation*}
\end{proof}

\setcounter{equation}{0}
\setcounter{figure}{0}
\setcounter{table}{0}

\renewcommand{\theequation}{SI.\arabic{equation}}
\renewcommand{\thefigure}{SI.\arabic{figure}}
\renewcommand{\thetable}{SI.\arabic{table}}
 \section{Proof of Theorem~\ref{them:sparsification_validity}}\label{sec:proof-sparsification}

 \begin{proof}

The proof proceeds in three steps: (1) establishing the conditional generative
model of the sparsified graph; (2) applying the multivariate Delta method to
prove that the asymptotic variance is invariant to the density parameter; and
(3) establishing unconditional convergence via characteristic functions. For
notational simplicity, we always assume $\rho_n, \rho_{b} < 1$ and
$w(\xi_1, \xi_2) \le 1$.

Set $q = b/2$. Recall that in Algorithm~\ref{algo:sparsification_stat},
$\mathbb{G}'$ is partitioned into disjoint subgraphs $\mathbb{G}^{\prime 1}$
and $\mathbb{G}^{\prime 2}$ of size $q$. We have the estimated density
$\widehat{\rho}_{\mathbb{G}^{\prime 1}}$ and the sparsification probability
$\hat{p} = \min(1, \rho^\dagger / \widehat{\rho}_{\mathbb{G}^{\prime 1})}$.
The edges of the second split $\mathbb{G}^{\prime 2}$ are independent of
$\mathbb{G}^{\prime 1}$ conditioned on the true model parameters.

Let $\theta = \hat{p} \rho_{b}$ be the effective sparsity parameter.
Conditioned on $\mathbb{G}^{\prime 1}$, $\mathbb{G}^{\prime 2}$ follows the
sparse graphon model $\theta \cdot w$. By
Lemma~\ref{lem: samplingdistribution_bickel}(a), the density estimator is
consistent, $\widehat{\rho}_{\mathbb{G}^{\prime 1}} / \rho_{b} \to 1$ in probability. 
Consequently, the random parameter $\theta$ converges in probability to the
target density:
\begin{equation}
    \theta \to  \rho^\dagger.
\end{equation}

We now analyze the asymptotic distribution of the statistic
$\bar{\Psi}_{\rho^\dagger}(\mathbb{G}^{\prime 1}, \mathbb{G}^{\prime 2})$
conditioning on the effective density $\theta$. Let
${X}_q \in \mathbb{R}^{m+1}$ denote the vector of the estimated density
and the raw network moments:
\begin{equation}
    {X}_q = \left( \widehat{\rho}_{\mathbb{G}^{\prime 2}},\,
    U_{R_1}(\mathbb{G}^{\prime 2}), \ldots, U_{R_m}(\mathbb{G}^{\prime 2})
    \right)^\top.
\end{equation}
By Lemma~\ref{lem: Exp_Var_graphon}, we have
\begin{equation}\label{eq:mean-value-vec}
{\mu}(\theta) = \mE[{X}_q | \theta]
= \left( \theta,\;
  \theta^{\mathfrak{r}_1}\frac{r_1!}{|\mathrm{Aut}(R_1)|}P_{w}(R_1),\;
  \ldots,\;
  \theta^{\mathfrak{r}_m}\frac{r_m!}{|\mathrm{Aut}(R_m)|}P_{w}(R_m)
  \right)^\top.
\end{equation}
Let ${S}_\theta$ be the $(m+1) \times (m+1)$ diagonal scaling matrix:
\begin{equation}
    {S}_\theta = \mathrm{diag}\!\left(
    \theta^{-1}, \theta^{-\mathfrak{r}_1}, \ldots, \theta^{-\mathfrak{r}_m}
    \right).
\end{equation}
Based on Proposition~\ref{prop: lim_var_graphon} and
Lemma~\ref{lem: samplingdistribution_bickel}(c), for any deterministic sequence
of parameters $\theta$ satisfying Assumption~\ref{ass:rho_n_h_n}, the scaled
deviations converge to a non-degenerate limit:
\begin{equation}\label{eq:scaled_clt_rigorous}
    {Z}_q(\theta) \coloneqq \sqrt{q}\,{S}_\theta
    \left({X}_q - {\mu}(\theta) \right)
   \to  \mathcal{N}({0}, {C}_w)
\end{equation} in distribution, 
where ${C}_w$ is the asymptotic covariance matrix of the normalized network moments. Crucially, as established in
Proposition~\ref{prop: lim_var_graphon}, the matrix ${C}_w$ is defined
solely by the integral properties of the graphon $w$ and is invariant to the
sparsity parameter $\theta$.

We apply the Delta method to $\Psi: \mathbb{R}^{m+1} \to \mathbb{R}^m$ defined
by $\Psi_j({x}) = x_j x_0^{-\mathfrak{r}_j}$ for $j = 1, \ldots, m$,
for the asymptotic distribution of
$\sqrt{q}(\Psi({X}_q) - \Psi({\mu}(\theta)))$. When clear
from context, we overload notation and write
\[
    {\Psi}(\mathbb{G}^{\prime 2}) = \Psi({X}_q).
\]
By the first-order Taylor expansion:
\begin{equation}
    \sqrt{q}\left( \Psi({X}_q) - \Psi({\mu}(\theta)) \right)
    = J({\mu}(\theta))
      \left[ \sqrt{q}({X}_q - {\mu}(\theta)) \right]
      + o_P(1),
\end{equation}
where $J({\mu}(\theta))$ is the Jacobian matrix evaluated at the
mean. To use the non-degenerate limit ${Z}_q(\theta)$ from
\eqref{eq:scaled_clt_rigorous}, we rewrite the dominating term as
\begin{equation}\label{eq:delta_expansion}
   \left[ J({\mu}(\theta))\,{S}_\theta^{-1} \right]
   \underbrace{\left[ \sqrt{q}\,{S}_\theta
     ({X}_q - {\mu}(\theta)) \right]}_{{Z}_q(\theta)}.
\end{equation}
We now compute the Jacobian $J$ evaluated at ${\mu}$. Recalling
\eqref{eq:mean-value-vec}, for the $j$-th row of the Jacobian:
\begin{align}
    \frac{\partial \Psi_j}{\partial x_0}\bigg|_{{\mu}}
    &= -\mathfrak{r}_j \bigl(\theta^{\mathfrak{r}_j}P_w(R_j)\bigr)
       \theta^{-\mathfrak{r}_j - 1}
     = -\mathfrak{r}_j P_w(R_j)\,\theta^{-1}, \\
    \frac{\partial \Psi_j}{\partial x_k}\bigg|_{{\mu}}
    &= \mathbb{1}_{\{k=j\}}\,\theta^{-\mathfrak{r}_j}.
\end{align}
Multiplying this row by ${S}_\theta^{-1}$:
\begin{itemize}
    \item The first column becomes:
          $(-\mathfrak{r}_j P_w(R_j)\,\theta^{-1}) \times \theta
           = -\mathfrak{r}_j P_w(R_j)$.
    \item The $j$-th column becomes:
          $\theta^{-\mathfrak{r}_j} \times \theta^{\mathfrak{r}_j} = 1$.
\end{itemize}
Note that $\theta$ cancels out exactly. Thus the factorization
$J({\mu})\,{S}_\theta^{-1} = \widetilde{J}_w$, where
$\widetilde{J}_w \in \mathbb{R}^{m \times (m+1)}$, is a constant matrix
depending only on $P_w(R_j)$ and $\mathfrak{r}_j$ but invariant to the density:
\begin{equation}
    \widetilde{J}_w = \begin{bmatrix}
    -\mathfrak{r}_1 P_w(R_1) & 1      & 0      & \dots  & 0 \\
    \vdots                   & \vdots & \vdots & \ddots & \vdots \\
    -\mathfrak{r}_m P_w(R_m) & 0      & 0      & \dots  & 1
    \end{bmatrix}.
\end{equation}
Since ${Z}_q(\theta) \xrightarrow{d} \mathcal{N}({0}, {C}_w)$
and $\widetilde{J}_w$ is constant, by the Delta method,
\begin{equation}
    \sqrt{q}\left( \Psi(\mathbb{G}^{\prime 2}) - {\eta}_w \right)
\to \mathcal{N}\!\left({0},\, \Sigma_w\right)
\end{equation} in distribution, 
where
\[\eta_w = \left(
      \frac{r_1!}{|\mathrm{Aut}(R_1)|}P_{w}(R_1),\;
      \ldots,\;
      \frac{r_m!}{|\mathrm{Aut}(R_m)|}P_{w}(R_m)
    \right)^\top
    = {\mE}\!\left[\Psi(\rho^{\dagger}, U_{[R_m]}(\mathbb{G}'))\right],
\] 
and $\Sigma_w \coloneqq \widetilde{J}_w {C}_w \widetilde{J}_w^\top$ is
composed entirely of terms invariant to $\theta$.

Next, we extend the result to unconditional convergence. Let
${W}_q = \sqrt{q}(\Psi(\mathbb{G}^{\prime 2}) - {P}_w)$.
Define $\phi_q({t})$ to be the unconditional characteristic function of
${W}_q$:
\begin{equation}
    \phi_q({t}) = {\mE}_{\theta}\!\left[
    {\mE}\!\left[ e^{i {t}^\top {W}_q} \mid \theta \right]
    \right]
    \eqqcolon {\mE}_{\theta}\!\left[ \varphi_q({t}| \theta) \right].
\end{equation}
We have established that for any \emph{deterministic} sequence of network
densities $\theta$ converging to $\rho^\dagger$ and satisfying
Assumption~\ref{ass:rho_n_h_n}, the conditional characteristic function
converges to the limit
$\Phi({t}) = e^{-\frac{1}{2} {t}^\top \Sigma_w {t}}$.

Now we handle the randomness of $\theta$. Since
$\theta \to \rho^\dagger$ in probability, by the Subsequence Principle, for any
subsequence there exists a further sub-subsequence along which $\theta$
converges almost surely. Along this sub-subsequence, the realization of
$\theta$ acts as a deterministic sequence, so
$\varphi_q({t} | \theta) \to \Phi({t})$ pointwise, and the
limit is unique and invariant to the choice of sub-subsequence. That is, for
any subsequence of $\varphi_q({t} | \theta)$ we can find a
sub-subsequence converging almost surely to $\Phi({t})$. Therefore,
$\varphi_q({t} | \theta)$ converges in probability to the constant
$\Phi({t})$.

Finally, since characteristic functions are uniformly bounded
($|\varphi_q| \le 1$), by the Dominated Convergence Theorem,
\begin{equation}
    \lim_{q \to \infty} \phi_q({t})
    = {\mE}\!\left[ \lim_{q \to \infty} \varphi_q({t} \mid \theta)
      \right]
    = \Phi({t}).
\end{equation}
This establishes the unconditional convergence that
${W}_q \to \mathcal{N}({0}, \Sigma_w)$ in distribution. 

Finally, we verify the subsampling counterpart. The proof of
Theorem~\ref{theo:consistent}, in particular
Theorem~\ref{theo:subsample_distribution}, already established the asymptotic
normality of
$\sqrt{b/2}\,\bar{\Psi}_{\rho^\dagger}(G^{*(i1)}_{b/2}, \mathbb{G}^{*(i2)}_{b/2})$,
conditioning on $\mathbb{G}^{*(i1)}_{b/2}$ and $\mathbb{G}$. The Delta method
argument carries over identically to the subsampling part. The only remaining
thing to check is the scaling. Recall $c' = 1 - \lim_{n\to\infty} b/n$.
\begin{itemize}
    \item The observed statistic
          $\sqrt{c'q}\,\Psi(\mathbb{G}^{\prime 2})$
          has covariance converging to $c'\Sigma_w$.
    \item By Theorem~\ref{theo:consistent} and
          Lemma~\ref{coro:lim_var_subsample}, the covariance of the subsampling
          counterpart is $c'$ times the infinite-population variance $\Sigma_w$.
\end{itemize}
Therefore, the covariances of
$\sqrt{b/2}\bigl(\bar{\Psi}_{\rho^\dagger}(G^{*(i1)}_{b/2},
\mathbb{G}^{*(i2)}_{b/2}) - {\eta}_w\bigr)$ (conditioning on $\mathbb{G}$)
and
$\sqrt{c'b/2}\bigl(\bar{\Psi}_{\rho^\dagger}(\mathbb{G}^{\prime 1},
\mathbb{G}^{\prime 2}) - {\eta}_w\bigr)$
match asymptotically. Since $c'\Sigma_w$ is non-degenerate
(Assumption~\ref{ass:non_degenerate}), the claimed convergence of the CDFs
follows.
\end{proof}


\setcounter{equation}{0}
\setcounter{figure}{0}
\setcounter{table}{0}

\renewcommand{\theequation}{SJ.\arabic{equation}}
\renewcommand{\thefigure}{SJ.\arabic{figure}}
\renewcommand{\thetable}{SJ.\arabic{table}}

\section{Empirical Evidence for the Necessity of Sparsification}\label{sec:necessary-sparsification}

To demonstrate that the sparsification step is indispensable for valid inference,
we conduct an additional simulation experiment in which the small network $G'$ is
compared directly with subsamples from the large network $G$, without any density
matching. Both networks are generated from the same graphon under $H_0$, but with
differing densities. Table~\ref{tab:simu4_TypeI_v2} reports the resulting type~I
error rates.

The results reveal a clear and systematic pattern. When the subsampled density
from $G$ is substantially lower than the density of $G'$, the rejection rates
are severely inflated, reaching as high as $0.483$. Conversely, when the density
of $G'$ is lower than the subsampled density, the tests become overly conservative,
with rejection rates as low as $0.005$. In both cases, the distortion is
bidirectional and depends on the direction of the density mismatch. In contrast,
our method with the sparsification step maintains type~I error rates close to the
nominal level $0.05$ across all settings, confirming that density matching is
necessary for valid two-sample inference.

\begin{table}[htbp]
\centering
\resizebox{0.7\textwidth}{!}{%
\begin{tabular}{lcccc}
\multicolumn{5}{c}{\textbf{Setting 1: $\rho_n = 0.25 \cdot n^{-0.1}$}} \\  
& \multicolumn{2}{c}{$\rho_b = 0.25 \cdot b^{-0.1}$}
& \multicolumn{2}{c}{$\rho_b = 0.1 \cdot b^{-0.1}$} \\ 
& Maha & Cauchy & Maha & Cauchy \\  
Ours (with matching) & 0.042 & 0.049 & 0.057 & 0.050 \\ 
Without matching     & {\bf 0.005} & {\bf 0.012}
                     &  {\bf 0.483} & {\bf 0.285}  \vspace{0.2cm} \\  
                   
\multicolumn{5}{c}{\textbf{Setting 2: $\rho_n = 0.11$}} \\ 
& \multicolumn{2}{c}{$\rho_b = 0.22$}
& \multicolumn{2}{c}{$\rho_b = 0.036$} \\  
& Maha & Cauchy & Maha & Cauchy \\  
Ours (with matching) & 0.044 & 0.052 & 0.048 & 0.052 \\  
Without matching     & {\bf 0.024} & {\bf 0.032}
                     & {\bf 0.324} & {\bf 0.156} \\  
\end{tabular}
}
\caption{Type~I error rates for the two-sample test with and without the
density-matching (sparsification) step, under $H_0$. The two settings
correspond to the $H_0$ configurations in Tables~\ref{tab:simu3gLgS}
and~\ref{tab:simu3PCG} in the main paper, respectively. Bold entries indicate invalid type~I
error control with respect to the nominal level.}
\label{tab:simu4_TypeI_v2}
\end{table}

%% file: Appendix/Appendix_SuppleA8_simu.tex
\section{Additional simulation results}
\label{subsec:addsimu}

In Figures~\ref{fig:simub05p01} and \ref{fig:simub05p02}, we show the counterpart of results in Figures~\ref{fig:simub06p01} and \ref{fig:simub06p02}, but with different subsampling size $b = \lceil 2n^{1/2} \rceil$. The patterns in this setting are consistent with Figures~\ref{fig:simub06p01} and \ref{fig:simub06p02} in the main paper. 

\begin{figure}[ht]
\centering
\begin{subfigure}[t]{0.22\textwidth}
\centering
\includegraphics[width=\textwidth]{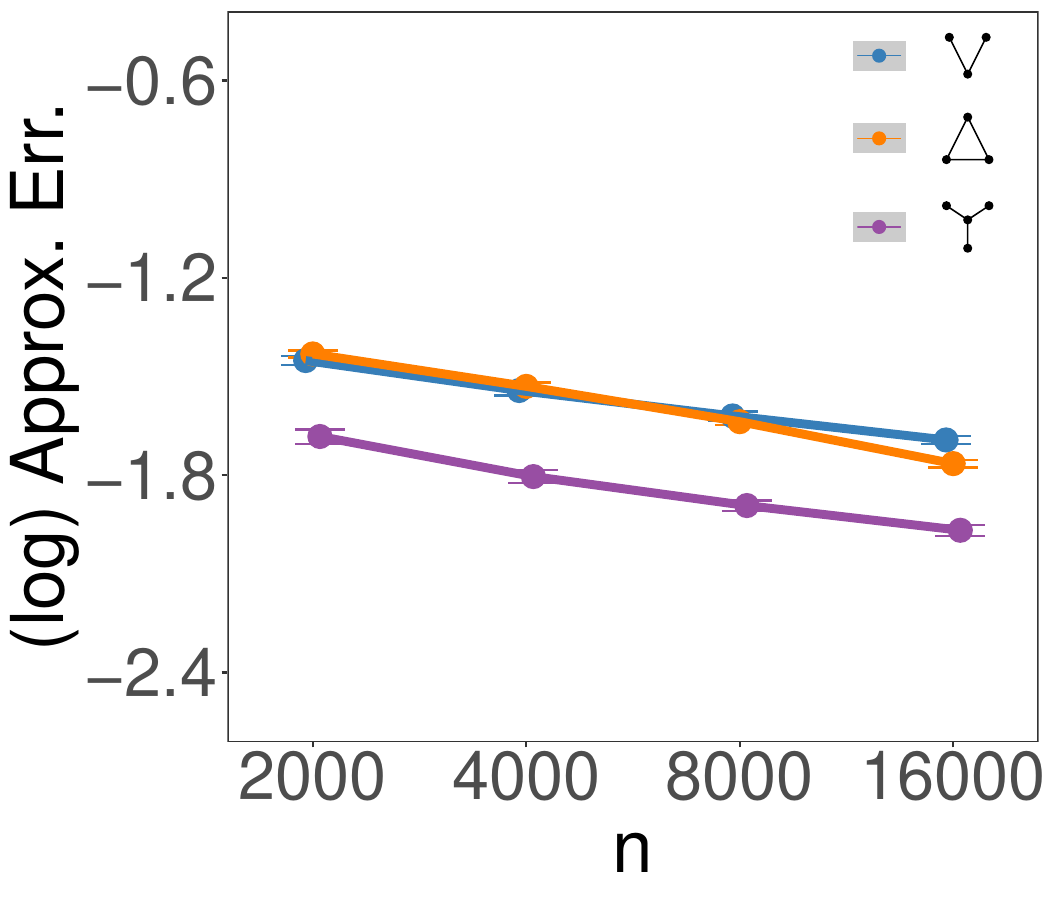}
\captionsetup{justification=centering}
\caption{\scriptsize Graphon 1: marginal cumulative distribution functions}
\label{subfig:uni_g2_b05p01}
\end{subfigure}
\hfill
\begin{subfigure}[t]{0.22\textwidth}
\centering
\includegraphics[width=\textwidth]{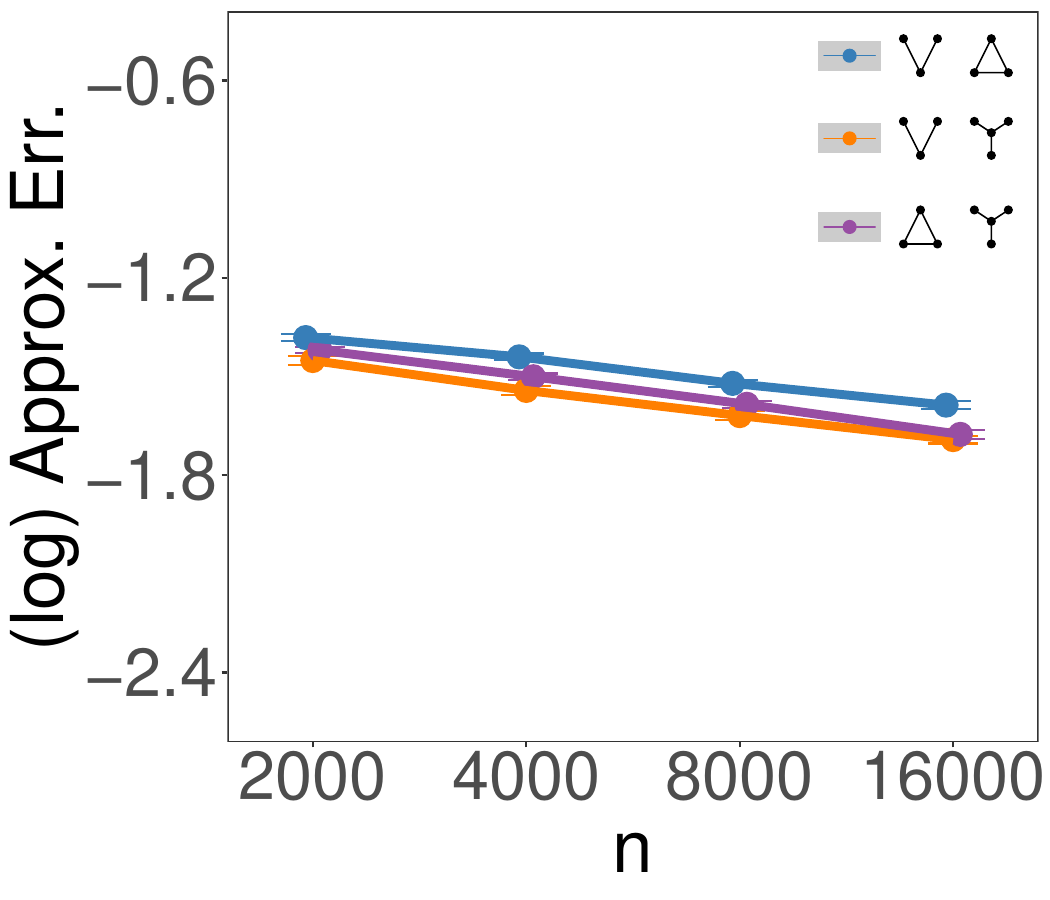}
\captionsetup{justification=centering}
\caption{\scriptsize Graphon 1: joint cumulative distribution functions}
\label{subfig:bi_g2_b05p01}
\end{subfigure}
\hfill
\begin{subfigure}[t]{0.22\textwidth}
\centering
\includegraphics[width=\textwidth]{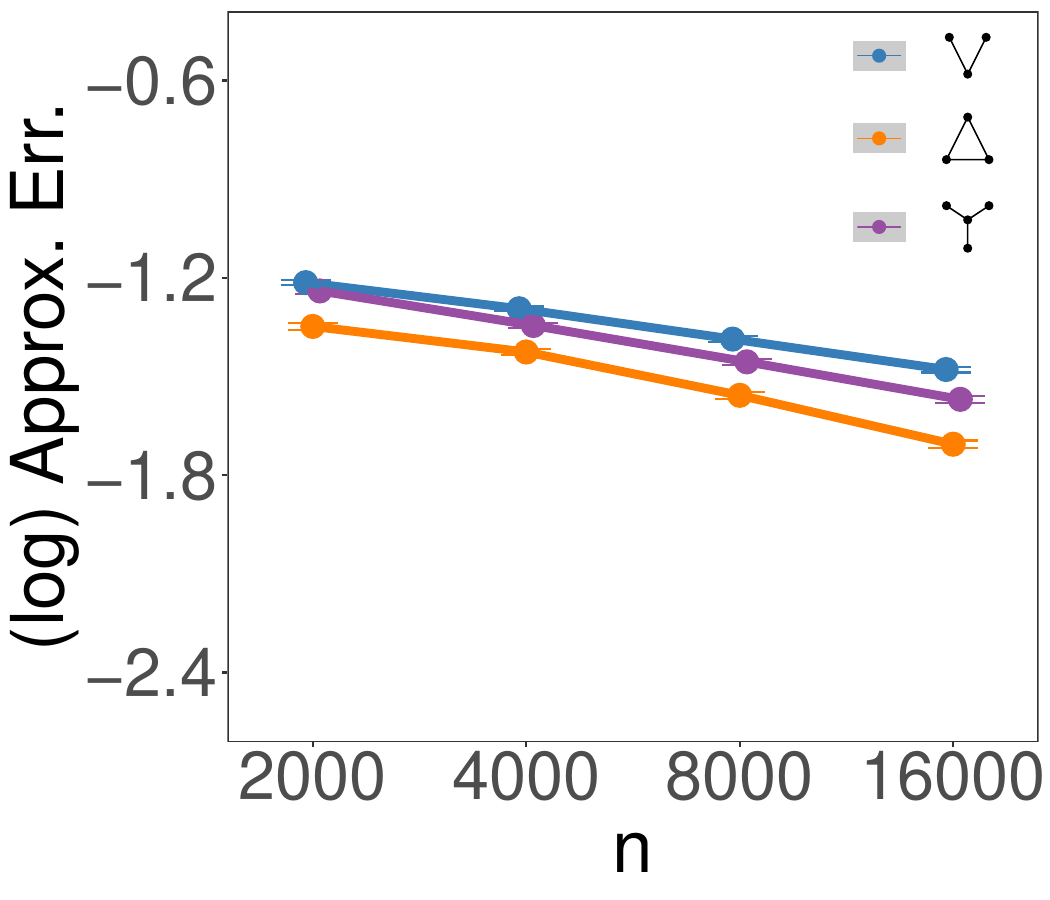}
\captionsetup{justification=centering}
\caption{\scriptsize Graphon 2: marginal cumulative distribution functions}
\label{subfig:uni_g1_b05p01}
\end{subfigure}
\hfill
\begin{subfigure}[t]{0.22\textwidth}
\centering
\includegraphics[width=\textwidth]{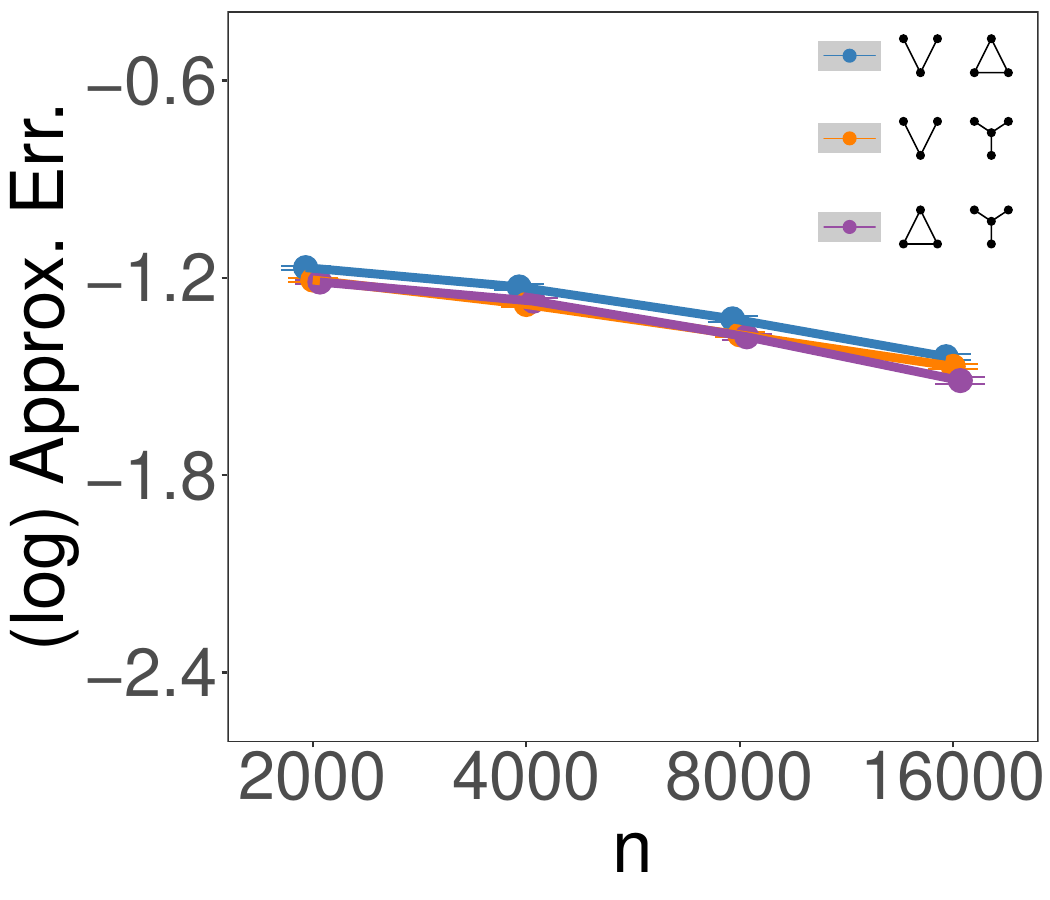}
\captionsetup{justification=centering}
\caption{\scriptsize Graphon 2: joint cumulative distribution functions}
\label{subfig:bi_g1_b05p01}
\end{subfigure}
\caption{Empirical approximation errors of the cumulative distribution functions under $b = \lceil 2n^{1/2} \rceil$ and $\rho_n = 0.25n^{-0.1}$.}
\label{fig:simub05p01}
\end{figure}

\begin{figure}[ht]
\centering
\begin{subfigure}[t]{0.22\textwidth}
\centering
\includegraphics[width=\textwidth]{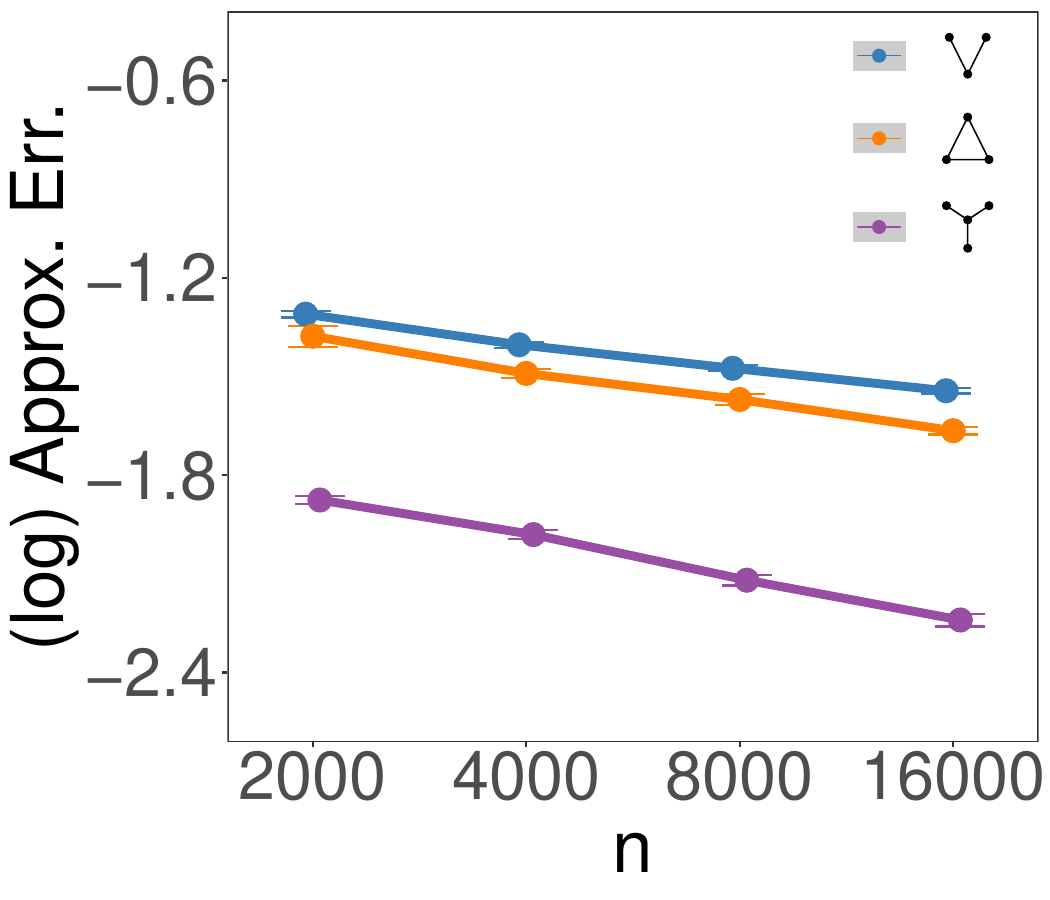}
\captionsetup{justification=centering}
\caption{\scriptsize Graphon 1: marginal cumulative distribution functions}
\label{subfig:uni_g2_b05p02}
\end{subfigure}
\hfill
\begin{subfigure}[t]{0.22\textwidth}
\centering
\includegraphics[width=\textwidth]{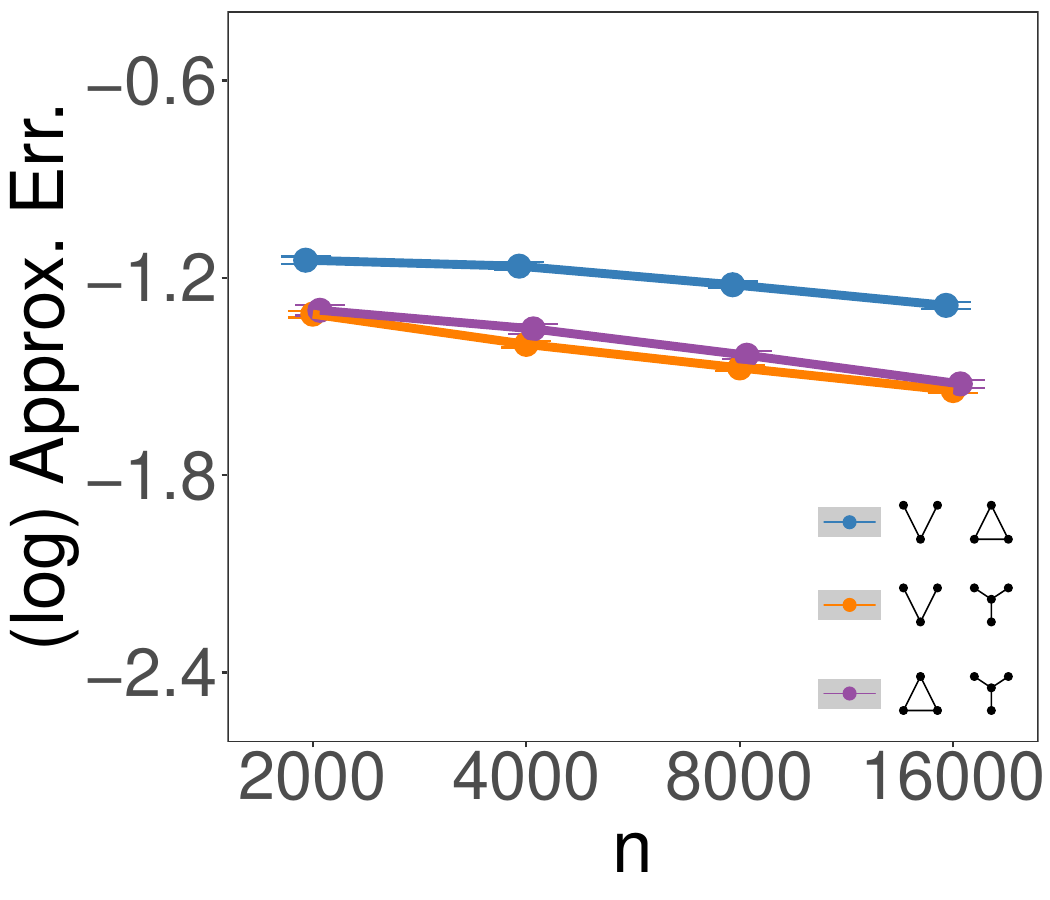}
\captionsetup{justification=centering}
\caption{\scriptsize Graphon 1: joint cumulative distribution functions}
\label{subfig:bi_g2_b05p02}
\end{subfigure}
\hfill
\begin{subfigure}[t]{0.22\textwidth}
\centering
\includegraphics[width=\textwidth]{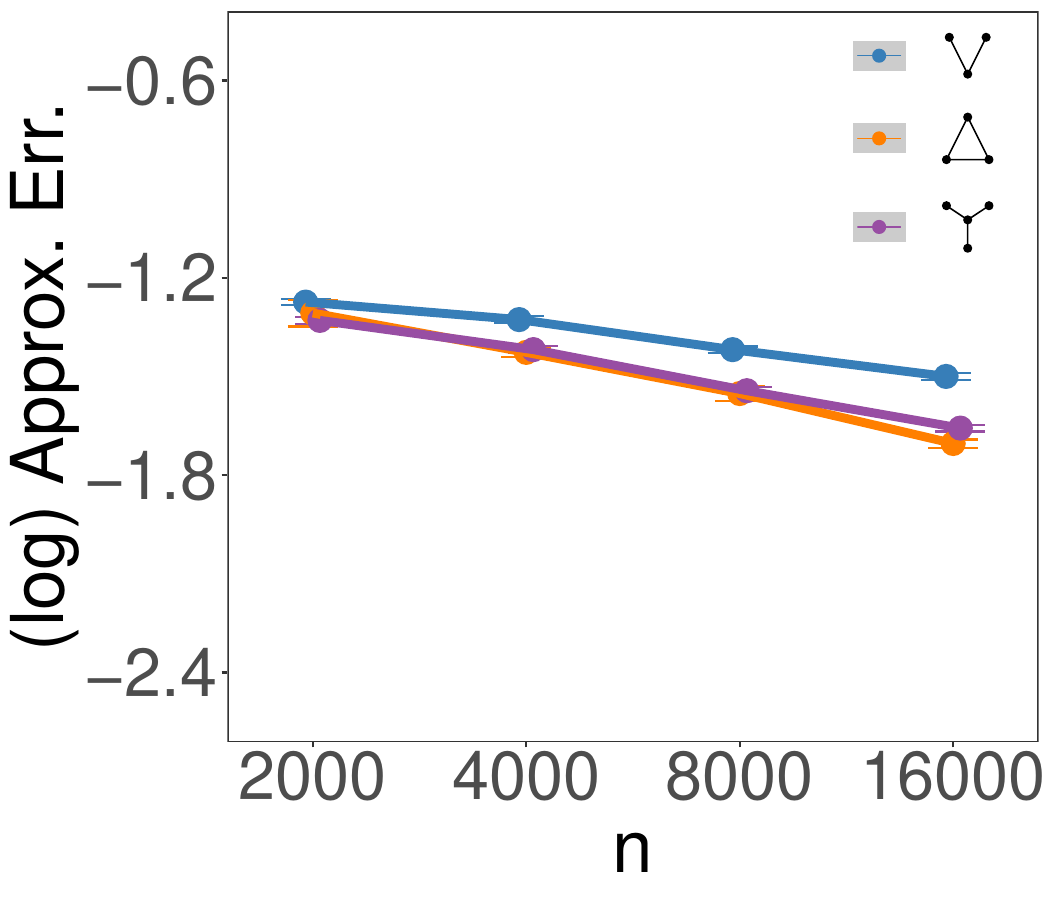}
\captionsetup{justification=centering}
\caption{\scriptsize Graphon 2: marginal cumulative distribution functions}
\label{subfig:uni_g1_b05p02}
\end{subfigure}
\hfill
\begin{subfigure}[t]{0.22\textwidth}
\centering
\includegraphics[width=\textwidth]{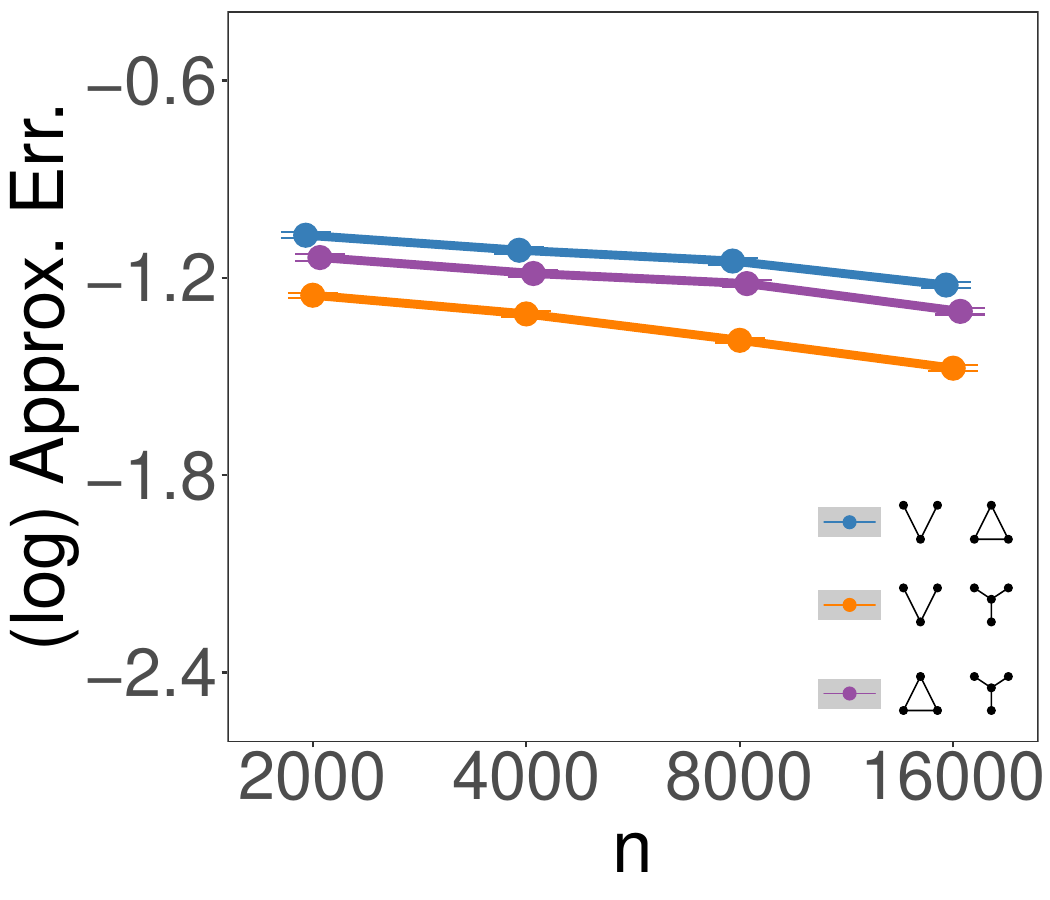}
\captionsetup{justification=centering}
\caption{\scriptsize Graphon 2: joint cumulative distribution functions}
\label{subfig:bi_g1_b05p02}
\end{subfigure}
\caption{Empirical approximation errors of the cumulative distribution functions under $b = \lceil 2n^{1/2} \rceil$ and $\rho_n = 0.25n^{-0.25}$.}
\label{fig:simub05p02}
\end{figure}


Figure~\ref{fig:simub06p03} displays the results in an over-sparse regime with $\rho_n = 0.25n^{-0.5}$. The networks become overly sparse in this scenario, so that the signal-to-noise ratio no longer suffices to support the subsampling inference. This phenomenon is indicated in Assumption~\ref{ass:rho_n_h_n} and has also been observed for other resampling inference methods on network moments \citep{green2022bootstrapping,levin2019bootstrapping,zhang2022edgeworth,lunde2023subsampling}.

\begin{figure}[ht]
\centering
\begin{subfigure}[t]{0.22\textwidth}
\centering
\includegraphics[width=\textwidth]{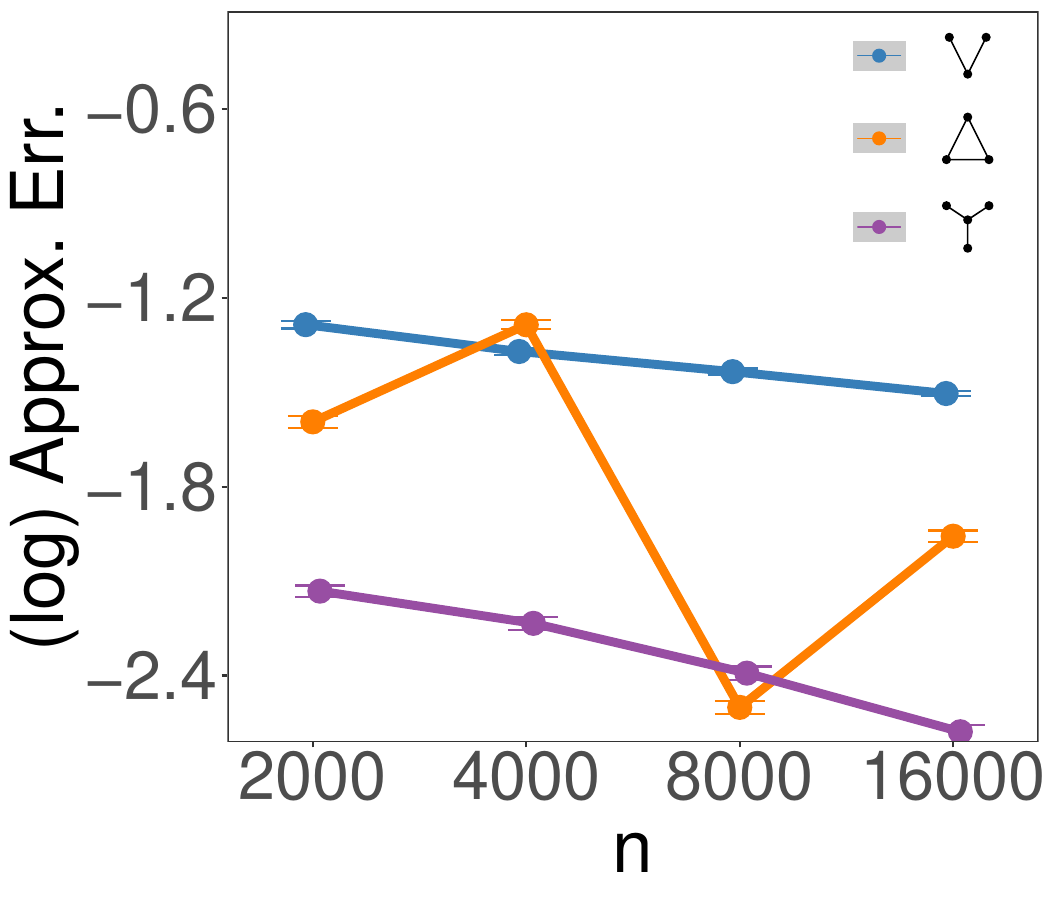}
\captionsetup{justification=centering}
\caption{\scriptsize Graphon 1: marginal cumulative distribution functions}
\label{subfig:uni_g2_b06p03}
\end{subfigure}
\hfill
\begin{subfigure}[t]{0.22\textwidth}
\centering
\includegraphics[width=\textwidth]{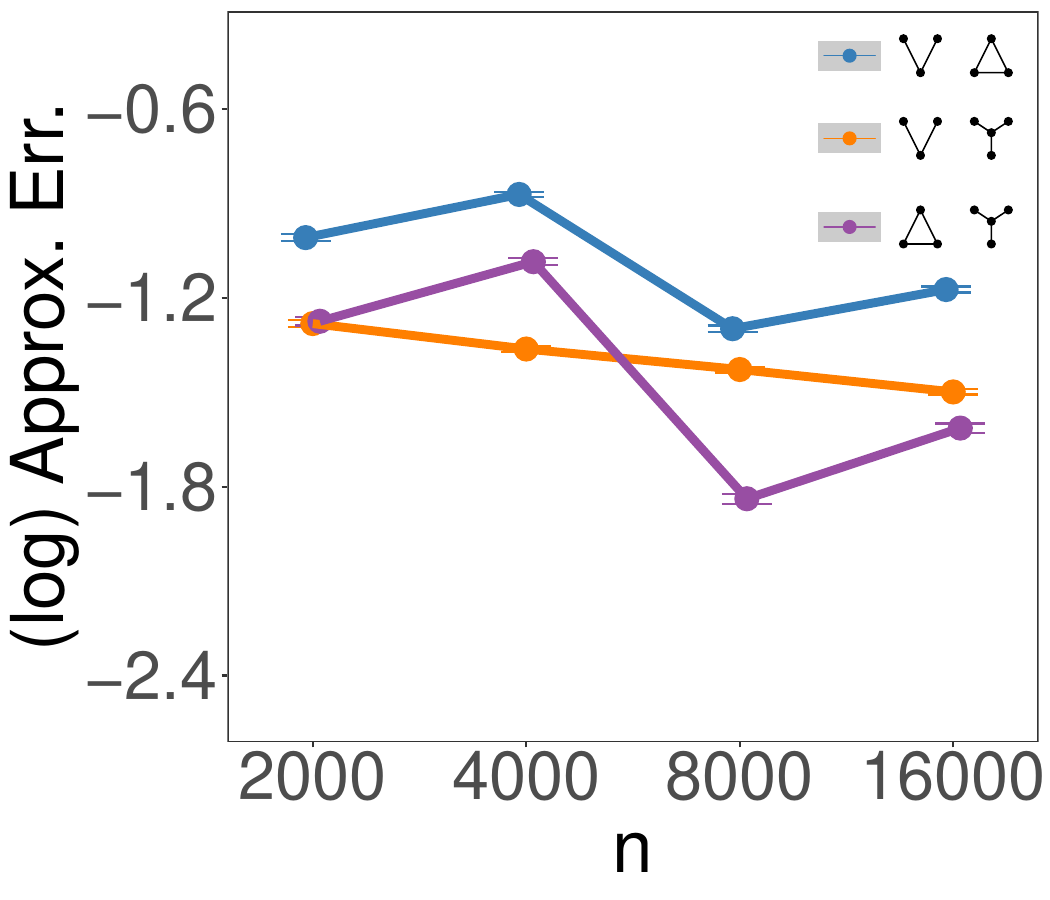}
\captionsetup{justification=centering}
\caption{\scriptsize Graphon 1: joint cumulative distribution functions}
\label{subfig:bi_g2_b06p03}
\end{subfigure}
\hfill
\begin{subfigure}[t]{0.22\textwidth}
\centering
\includegraphics[width=\textwidth]{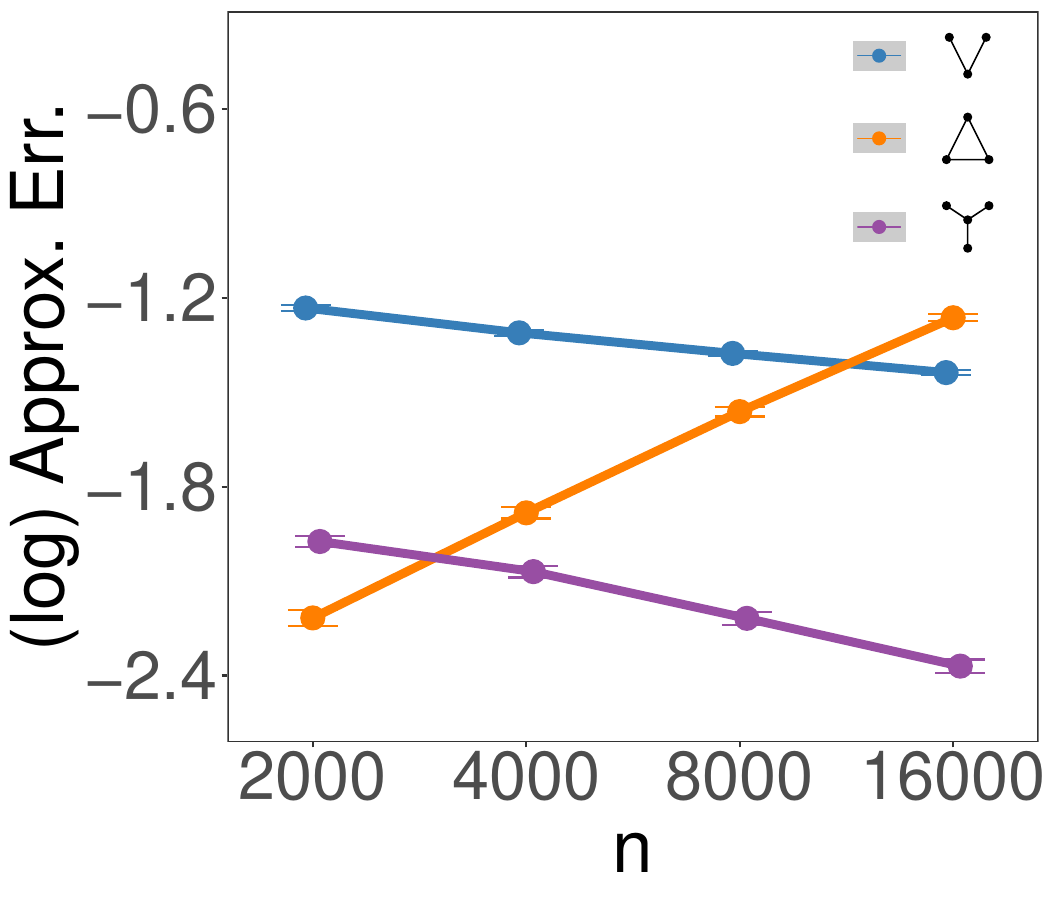}
\captionsetup{justification=centering}
\caption{\scriptsize Graphon 2: marginal cumulative distribution functions}
\label{subfig:uni_g1_b06p03}
\end{subfigure}
\hfill
\begin{subfigure}[t]{0.22\textwidth}
\centering
\includegraphics[width=\textwidth]{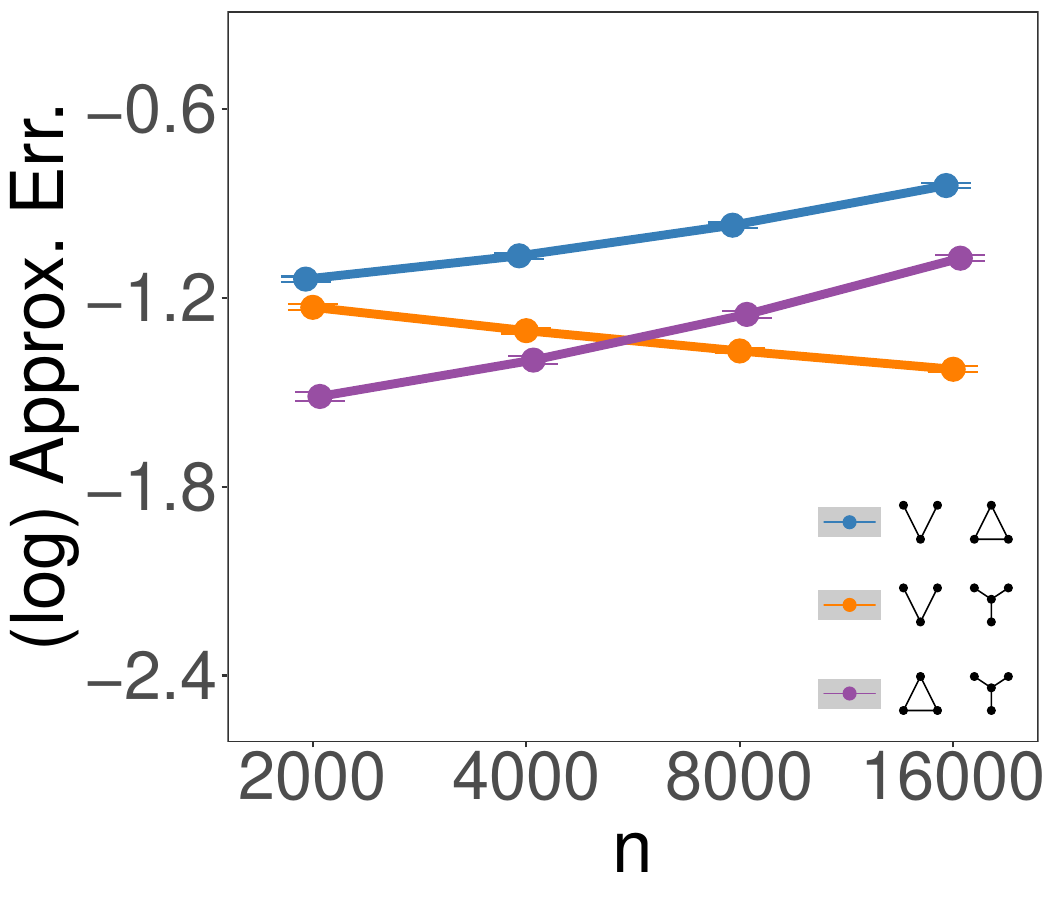}
\captionsetup{justification=centering}
\caption{\scriptsize Graphon 2: joint cumulative distribution functions}
\label{subfig:bi_g1_b06p03}
\end{subfigure}
\caption{Empirical approximation errors of the cumulative distribution functions under $b = \lceil n^{2/3} \rceil$ and $\rho_n = 0.25n^{-0.5}$.}
\label{fig:simub06p03}
\end{figure}